\documentclass[12pt,a4paper]{book}
\usepackage[english]{babel}
\usepackage{amsmath,amsthm}
\usepackage{amssymb}
\usepackage{fullpage}
\usepackage{graphicx}
\usepackage{cite}
\usepackage{fancyhdr}
\usepackage[bf, hang]{caption}
\usepackage[pdftex]{hyperref}

\newtheorem{thm}{Theorem}[chapter]
\newtheorem{cor}[thm]{Corollary}
\newtheorem{lem}[thm]{Lemma}
\newtheorem{prop}[thm]{Proposition}
\theoremstyle{definition}
\newtheorem{defn}[thm]{Definition}
\theoremstyle{remark}
\newtheorem{rem}[thm]{Remark}
\begin{document}
\frontmatter
\begin{titlepage}
\begin{center}
	\includegraphics[width=.4\textwidth]{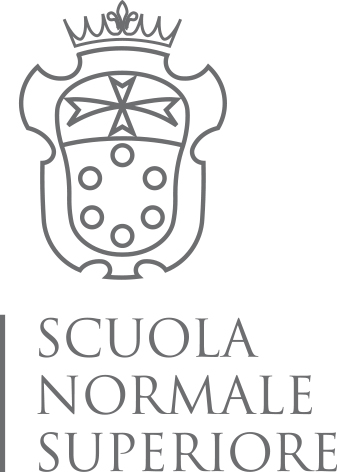}\\
    \bigskip\bigskip
   	\large{\textsc{Faculty of Mathematical and Natural Sciences}}\\
		\rule{5cm}{1pt}\\
	{\small{PhD in \textsc{Physics}}}\\
\bigskip\bigskip\bigskip
	\Huge{\textsc{Gaussian optimizers and other topics in quantum information}}\\	
\bigskip\bigskip\bigskip
	\normalsize{PhD Thesis}\\
	\LARGE{Giacomo De Palma}\\
\bigskip\bigskip\bigskip\bigskip
\end{center}
\begin{large}
\makebox{\parbox[b]{.9\textwidth}{
	\flushleft{\textbf{Supervisor:}}
		\flushleft{
		Prof. Vittorio Giovannetti}
}}
\bigskip\bigskip\bigskip\bigskip
\end{large}
	\begin{center}
	\rule{5cm}{1pt}\\
	Academic Year 2015/2016
	\end{center}
\end{titlepage}
\pagestyle{empty}

\chapter*{Acknowledgments}
First of all, I sincerely thank my supervisor Prof. Vittorio Giovannetti for having introduced me to the fascinating research field of quantum information, and for his constant support and guidance during these three years.
I thank Prof. Luigi Ambrosio for his invaluable support and advices, and my MSc supervisor Prof. Augusto Sagnotti for teaching me how to make scientific research.
I thank Dr. Andrea Mari, that has closely followed most of the work of this thesis, and Dr. Dario Trevisan, for the infinite discussions on the minimum entropy problem.
I thank Dr. Marcus Cramer, Dr. Alessio Serafini, Prof. Seth Lloyd and Prof. Alexander Holevo for giving me the opportunity to work together, and Prof. Giuseppe Toscani and Prof. Giuseppe Savar\'e for their kind ospitality in Pavia.
Last but not least, I thank all my colleagues in the condensed matter and quantum information theory group at Scuola Normale, for the stimulating and lively environment and for all the moments shared together.

\chapter*{Abstract}
Gaussian input states have long been conjectured to minimize the output von Neumann entropy of quantum Gaussian channels for fixed input entropy. We prove the quantum Entropy Power Inequality, that provides an extremely tight lower bound to this minimum output entropy, but is not saturated by Gaussian states, hence it is not sufficient to prove their optimality. Passive states are diagonal in the energy eigenbasis and their eigenvalues decrease as the energy increases. We prove that for any one-mode Gaussian channel, the output generated by a passive state majorizes the output generated by any state with the same spectrum, hence it has a lower entropy. Then, the minimizers of the output entropy of a Gaussian channel for fixed input entropy are passive states. We exploit this result to prove that Gaussian states minimize the output entropy of the one-mode attenuator for fixed input entropy. This result opens the way to the multimode generalization, that permits to determine both the classical capacity region of the Gaussian quantum degraded broadcast channel and the triple trade-off region of the quantum attenuator.

Still in the context of Gaussian quantum information, we determine the classical information capacity of a quantum Gaussian channel with memory effects. Moreover, we prove that any one-mode linear trace-preserving not necessarily positive map preserving the set of Gaussian states is a quantum Gaussian channel composed with the phase-space dilatation. These maps are tests for certifying that a given quantum state does not belong to the convex hull of Gaussian states. Our result proves that phase-space dilatations are the only test of this kind.

In the context of quantum statistical mechanics, we prove that requiring thermalization of a quantum system in contact with a heat bath for any initial uncorrelated state with a well-defined temperature implies the Eigenstate Thermalization Hypothesis for the system-bath Hamiltonian. Then, the ETH constitutes the unique criterion to decide whether a given system-bath dynamics always leads to thermalization.

In the context of relativistic quantum information, we prove that any measurement able to distinguish a coherent superposition of two wavepackets of a massive or charged particle from the corresponding incoherent statistical mixture must require a minimum time. This bound provides an indirect evidence for the existence of quantum gravitational radiation and for the necessity of quantizing gravity.

\cleardoublepage

\pagestyle{fancy}
\addtolength{\headheight}{14.5pt}
\addtolength{\headsep}{12pt}

\renewcommand{\chaptermark}[1]{\markboth{\thechapter.\ #1}{}}
\renewcommand{\sectionmark}[1]{\markright{#1\ \thesection}}

\lhead[\fancyplain{}{\textbf{\footnotesize{\leftmark}}}]{}
\chead{}
\rhead[]{\fancyplain{}{\textbf{\footnotesize{\rightmark}}}}

\cfoot[\footnotesize{PhD Thesis}]{\footnotesize Gaussian optimizers in quantum information}
\lfoot[\thepage]{\footnotesize{\slshape Giacomo De Palma}}
\rfoot[]{\thepage}

\tableofcontents

\mainmatter

\chapter{Introduction}
Quantum information theory \cite{bennett1998quantum,holevo2013quantum,wilde2013quantum,nielsen2010quantum} has had an increasingly large development over the last twenty years.
The interest of the scientific community in this field is twofold.
On one side, quantum communication theory permits to determine the ultimate bounds that quantum mechanics imposes on communication rates \cite{gordon1962quantum,caves1994quantum}.
On the other side, quantum cryptography \cite{gisin2002quantum} permits to design and build devices allowing a perfectly secure communication by distributing to two parties the same secret key, that can be guaranteed not to have been read by any possible eavesdropper.

Most communication devices, such as metal wires, optical fibers and antennas for free space communication, encode the information into pulses of electromagnetic radiation, whose quantum description requires the framework of Gaussian quantum systems.
For this reason, Gaussian quantum information \cite{braunstein2005quantum,caves1994quantum,weedbrook2012gaussian} plays a fundamental role.

In the classical scenario, the general principle ``Gaussian channels have Gaussian optimizers'' has been proven to hold in a wide range of situations \cite{lieb1990gaussian}, and has never been disproved.
This Thesis focuses on the transposition of this principle to the domain of Gaussian quantum information \cite{holevo2015gaussian}, i.e. on the conjecture of the optimality of Gaussian states for the transmission of both classical and quantum information through quantum Gaussian channels; Section \ref{gauss} introduces our results on this topic.
Section \ref{eth} introduces our results on the Eigenstate Thermalization Hypothesis, an application of quantum information ideas to quantum statistical mechanics.
Section \ref{rel} introduces our results on relativistic quantum information.

\section{Gaussian optimizers in quantum information}\label{gauss}
Most communication schemes encode the information into pulses of electromagnetic radiation, that is transmitted through metal wires, optical fibers or free space, and is unavoidably affected by attenuation and environmental noise.
Gauge-covariant Gaussian channels \cite{caves1994quantum,cover2006elements} provide a faithful model for these effects, and a fundamental issue is determining the maximum rate at which information can be transmitted along such channels.
Since the electromagnetic field is ultimately a quantum-mechanical entity, quantum effects must be taken into account \cite{gordon1962quantum}.
They become relevant for low-intensity signals, such as in the case of space probes, that can be reached by only few photons for each bit of information.
These quantum effects are faithfully modeled by gauge-covariant quantum Gaussian channels \cite{chan2006free,braunstein2005quantum,holevo2013quantum,weedbrook2012gaussian}.

The optimality of coherent Gaussian states for the transmission of classical information through gauge-covariant quantum Gaussian channels has been recently proved, hence determining their classical capacity \cite{giovannetti2014ultimate}.
This has been possible thanks to the proof of the so-called minimum output entropy conjecture \cite{giovannetti2004minimum,giovannetti2015solution}, stating that the von Neumann entropy at the output of any gauge-covariant Gaussian channel is minimized when the input is the vacuum state.
Actually this conjecture follows from the more general Gaussian majorization conjecture \cite{giovannetti2015majorization,mari2014quantum}, stating that the output generated by the vacuum majorizes (i.e. it is less noisy than) the output generated by any other state.

However, a sender might want to communicate classical information to two receivers at the same time.
In this scenario, the communication channel is called broadcast channel, and the set of all the couples of simultaneously achievable rates of communication with the two receivers constitutes its classical capacity region \cite{savov2015classical,yard2011quantum}.
The proof of the optimality of coherent Gaussian states for the transmission of classical information through the degraded quantum Gaussian broadcast channel and the consequent determination of its capacity region \cite{guha2007classicalproc,guha2007classical} rely on a constrained minimum output entropy conjecture, stating that Gaussian thermal input states minimize the output entropy of the quantum attenuator for fixed input entropy.

Moreover, a sender might want to transmit both public and private classical information to a receiver, with the possible assistance of a secret key.
This scenario is relevant in the presence of satellite-to-satellite links, used for both public and private communication and quantum key distribution \cite{scarani2009security}.
In this setting the electromagnetic signal travels through free space.
Then, the only effect of the environment is signal attenuation, that is modeled by the Gaussian quantum attenuator.
Dedicating a fraction of the channel uses to public communication, another fraction to private communication and the remaining fraction to key distribution is the easiest strategy, but not the optimal one.
Indeed, significantly higher communication rates can be obtained performing the three tasks at the same time with the so-called trade-off coding \cite{wilde2012public}.
A similar scenario occurs for the simultaneous transmission of classical and quantum information, with the possible assistance of shared entanglement.
The set of all the triples of simultaneously achievable rates for performing the various tasks constitutes the triple trade-off region of the quantum attenuator \cite{wilde2012quantum,wilde2012information}.
Its determination relies on the same unproven constrained minimum output entropy conjecture stated above.

Since a quantum-limited attenuator can be modeled as a beamsplitter that mixes the signal with the vacuum state, this conjecture has been generalized to the so-called entropy photon-number inequality \cite{guha2008capacity,guha2008entropy}, stating that the entropy at the output of a beamsplitter for fixed entropy of each input is minimized by Gaussian inputs with proportional covariance matrices.
So far, none of these two conjectures has been proved.

The classical analog of a quantum state of the electromagnetic radiation is a probability distribution of a random real vector.
The action of a beamsplitter on the two input quantum states is replaced in this setting by a linear combination of the two input random vectors.
The Entropy Power Inequality (see \cite{dembo1991information,gardner2002brunn,shannon2001mathematical,stam1959some,verdu2006simple,rioul2011information} and Chapter 17 of \cite{cover2006elements}) bounds the Shannon differential entropy of a linear combination of real random vectors in terms of their own entropies, and states that it is minimized by Gaussian inputs.
Then, another inequality has been conjectured, the quantum Entropy Power Inequality \cite{konig2013limits}, that keeps the same formal expression of its classical counterpart to give an almost optimal lower bound to the output von Neumann entropy of a beamsplitter in terms of the input entropies.
This inequality has first been proved for the $50:50$ beamsplitter \cite{konig2014entropy}.
In this Thesis we extend it to any beamsplitter and quantum amplifier, and generalize it to the multimode scenario \cite{de2014generalization,de2015multimode}.

Contrarily to its classical counterpart, the quantum Entropy Power Inequality is \emph{not} saturated by quantum Gaussian states, and thus it is not sufficient to prove their conjectured optimality.
As first step toward the proof of the constrained minimum output entropy conjecture and the entropy photon-number inequality, we prove a generalization of the Gaussian majorization conjecture of \cite{mari2014quantum} linking it to the notion of passivity.
A passive state of a quantum system \cite{pusz1978passive,lenard1978thermodynamical,gorecki1980passive,vinjanampathy2015quantum,goold2015role,binder2015quantum} minimizes the average energy among all the states with the same spectrum, and is then diagonal in the Hamiltonian eigenbasis with eigenvalues that decrease as the energy increases.
We prove that the output of any one-mode gauge-covariant quantum Gaussian channel generated by a passive state majorizes (i.e. it is less noisy than) the output generated by any other state with the same spectrum \cite{de2015passive}.
The optimal inputs for the constrained minimum output entropy problem are then to be found among the states diagonal in the energy eigenbasis.
We exploit this result in Chapter \ref{chepni}.
Here we prove that Gaussian thermal input states minimize the output entropy of the one-mode quantum attenuator for fixed input entropy \cite{de2016gaussian}, i.e. the constrained minimum output entropy conjecture for this channel.
In Chapter \ref{chlossy} we extend the majorization result of Chapter \ref{majorization} to a large class of lossy quantum channels, arising from a weak interaction of a small quantum system with a large bath in its ground state \cite{de2016passive}.

In any realistic communication, the pulses of electromagnetic radiation always leave some noise in the channel after their passage.
Since this noise depends on the message sent, the various uses of the channel are no more independent, and memory effects are present.
In this Thesis we consider a particular model of quantum Gaussian channel that implements these memory effects.
We study their influence on the classical information capacity, that we determine analytically \cite{de2014classical}.

Gaussian states of bosonic quantum systems are easy to realize in the laboratory, and so are their convex combinations, belonging to the convex hull of Gaussian states $\mathfrak{C}$.
We explore the set of Gaussian-to-Gaussian superoperators, i.e. the linear trace-preserving not necessarily positive maps preserving the set of Gaussian states.
These maps preserve also $\mathfrak{C}$, and can then be used as a probe to check whether a given quantum state belongs to $\mathfrak{C}$ exactly as a positive but not completely positive map is a test for entanglement.
We prove that for one mode they are all built from the so-called phase-space dilatation, that is hence found to be the only relevant test of this kind \cite{de2015normal}.

\section{Quantum statistical mechanics}\label{eth}
Everyday experience, as well as overwhelming experimental evidence, demonstrates that a small quantum system in contact with a large heat bath at a given temperature evolves toward the state described by the canonical ensemble with the same temperature as the bath.
This state is independent of the details of the initial state of both the system and the bath. This very common behavior, known as thermalization, has proven surprisingly difficult to explain starting from fundamental dynamical laws.
In the quantum-mechanical framework, since 1991 the ``Eigenstate Thermalization Hypothesis'' (ETH) \cite{deutsch1991quantum,gogolin2015equilibration} is known to be a sufficient condition for thermalization. The ETH states that each eigenstate of the global system-bath Hamiltonian locally looks on the system as a canonical state with a temperature that is a smooth function of the energy of the eigenstate.

In this context, we prove that, if a quantum system in contact with a heat bath at a given temperature thermalizes for any initial state with a reasonably sharp energy distribution and without correlations between system and bath, the system-bath Hamiltonian must satisfy the ETH \cite{de2015necessity}.
This results proves that the ETH constitutes the unique criterion to decide whether a given system-bath dynamics always leads to a system equilibrium state described by the canonical ensemble: if the system-bath Hamiltonian satisfy the ETH, the system always thermalizes, while if the ETH is not satisfied, there certainly exists some initial product state not leading to thermalization of the system.

\section{Relativistic quantum information}\label{rel}
The existence of coherent superpositions is a fundamental postulate of quantum mechanics but, apparently, implies very counterintuitive consequences when extended to macroscopic systems, as in the famous Schr\"odinger cat paradox.
However, at least in principle, the standard theory of quantum mechanics is valid at any scale and does not put any limit on the size of the system.
A fundamental still open question is whether quantum superpositions can actually exist also at macroscopic scales, or there is some intrinsic spontaneous collapse mechanism prohibiting them \cite{penrose1996gravity, ghirardi1986unified, diosi1989models,karolyhazy1966gravitation, bassi2013models}.

In this Thesis we study the effect of the static electric or gravitational field generated by a charged or massive particle on the coherence of its own wavefunction \cite{mari2015experiments}.
We show that, without introducing any modification to standard quantum mechanics and quantum field theory, relativistic causality implies that any measurement able to distinguish a coherent superposition of two wavepackets from the corresponding incoherent statistical mixture must require a minimum time.
Indeed, any measurement violating this minimum-time bound is physically forbidden since it would permit a superluminal communication protocol.
In the electromagnetic case, this minimum time can be ascribed to the entanglement with the electromagnetic radiation that is unavoidably emitted in a too fast measurement.
In the gravitational case, this minimum time provides an indirect evidence for the existence of quantum gravitational radiation, and thus for the necessity of quantizing gravity.

\section{Outline of the Thesis}
In Chapter \ref{GQI} we introduce Gaussian quantum information and the problem of the determination of the classical communication capacity of quantum Gaussian channels, and we show the link with the minimum output entropy conjectures.
Chapter \ref{epi} contains the proof of the quantum Entropy Power Inequality.
In Chapter \ref{majorization} we prove the optimality of passive input states for one-mode quantum Gaussian channels, and in Chapter \ref{chepni} we exploit this result to prove the constrained minimum output entropy conjecture for the one-mode quantum attenuator.
In Chapter \ref{chlossy} we extend the majorization result of Chapter \ref{majorization} to a large class of lossy quantum channels.
In Chapter \ref{memory} we determine the classical capacity of a quantum Gaussian channel with memory effects, and in Chapter \ref{normal} we present the classification of Gaussian-to-Gaussian superoperators.

In Chapter \ref{chETH} we prove that the Eigenstate Thermalization Hypothesis is implied by a certain definition of thermalization, and in Chapter \ref{chsuperpos} we prove the minimum-time bound on the measurements able to distinguish coherent superpositions from statistical mixtures.

Finally, the conclusions are in Chapter \ref{concl}.

Appendices \ref{appG} and \ref{appsup} contain some technical details on the properties of Gaussian quantum systems and of quantum electrodynamics, respectively.

\section{References}
This Thesis is based on the following papers:
\begin{enumerate}
\item[\cite{de2014generalization}] G.~De~Palma, A.~Mari, and V.~Giovannetti, ``A generalization of the entropy power inequality to bosonic quantum systems,'' \emph{Nature Photonics}, vol.~8, no.~12, pp. 958--964, 2014.\\ {\small\url{http://www.nature.com/nphoton/journal/v8/n12/full/nphoton.2014.252.html}}
\item[\cite{de2015multimode}] G.~De~Palma, A.~Mari, S.~Lloyd, and V.~Giovannetti, ``Multimode quantum entropy power inequality,'' \emph{Physical Review A}, vol.~91, no.~3, p. 032320, 2015.\\ {\small\url{http://journals.aps.org/pra/abstract/10.1103/PhysRevA.91.032320}}
\item[\cite{de2015passive}] G.~De~Palma, D.~Trevisan, and V.~Giovannetti, ``Passive States Optimize the
  Output of Bosonic Gaussian Quantum Channels,'' \emph{IEEE Transactions on
  Information Theory}, vol.~62, no.~5, pp. 2895--2906, May 2016.\\ {\small\url{http://ieeexplore.ieee.org/document/7442587}}
\item[\cite{de2016gaussian}] G.~De~Palma, D.~Trevisan, and V.~Giovannetti, ``Gaussian states minimize the output entropy of the one-mode quantum
  attenuator,'' \emph{IEEE Transactions on Information Theory}, vol.~63, no.~1,
  pp. 728--737, 2017.\\ {\small\url{http://ieeexplore.ieee.org/document/7707386}}
\item[\cite{de2016passive}] G.~De~Palma, A.~Mari, S.~Lloyd, and V.~Giovannetti, ``Passive states as optimal
  inputs for single-jump lossy quantum channels,'' \emph{Physical Review A},
  vol.~93, no.~6, p. 062328, 2016.\\ {\small\url{http://journals.aps.org/pra/abstract/10.1103/PhysRevA.93.062328}}
\item[\cite{de2014classical}] G.~De~Palma, A.~Mari, and V.~Giovannetti, ``Classical capacity of Gaussian thermal memory channels,'' \emph{Physical Review A}, vol.~90, no.~4, p. 042312, 2014.\\ {\small\url{http://journals.aps.org/pra/abstract/10.1103/PhysRevA.90.042312}}
\item[\cite{de2015normal}] G.~De~Palma, A.~Mari, V.~Giovannetti, and A.~S. Holevo, ``Normal form decomposition for Gaussian-to-Gaussian superoperators,'' \emph{Journal of Mathematical Physics}, vol.~56, no.~5, p. 052202, 2015.\\ {\small\url{http://scitation.aip.org/content/aip/journal/jmp/56/5/10.1063/1.4921265}}
\item[\cite{de2015necessity}] G.~De~Palma, A.~Serafini, V.~Giovannetti, and M.~Cramer, ``Necessity of Eigenstate Thermalization,'' \emph{Physical Review Letters}, vol. 115, no.~22, p. 220401, 2015.\\ {\small\url{http://journals.aps.org/prl/abstract/10.1103/PhysRevLett.115.220401}}
\item[\cite{mari2015experiments}] A.~Mari, G.~De~Palma, and V.~Giovannetti, ``Experiments testing macroscopic quantum superpositions must be slow,'' \emph{Scientific Reports}, vol.~6, p. 22777, 2016.\\ {\small\url{http://www.nature.com/articles/srep22777}}
\end{enumerate}

\chapter{Gaussian optimizers in quantum information}
\label{GQI}
This Chapter introduces Gaussian quantum information and the problem of the determination of the capacity for transmitting classical information through a quantum Gaussian channel.
A more comprehensive presentation can be found in \cite{braunstein2005quantum,weedbrook2012gaussian,ferraro2005gaussian,holevo2013quantum,holevo2015gaussian} and references therein.

We start introducing quantum Gaussian systems, states and channels in Sections \ref{GQSy}, \ref{GQSt} and \ref{GQCh}, respectively.
Then, we define the von Neumann entropy (Section \ref{vNE}), and link it to the classical communication capacity of a quantum channel (Section \ref{clcap}).
In Section \ref{gccapacity} we present the determination of the classical capacity of gauge-covariant quantum Gaussian channels thanks to the proof of a minimum output entropy conjecture, and in Section \ref{secmaj} we show the link with majorization theory.

We then present in Section \ref{broadcasti} the problem of determining the classical capacity region of a degraded quantum broadcast channel, where the sender wants to communicate with multiple parties, and we show how this problem is linked to a constrained minimum output entropy conjecture, i.e. the determination of the minimum output entropy of a quantum channel for fixed input entropy.
Finally, we present in Section \ref{broadcastg} the degraded quantum Gaussian broadcast channel, and in Section \ref{broadcast} its conjectured capacity region and the bounds following from the Entropy Power Inequality that we will prove in Chapter \ref{epi}.
Appendix \ref{appG} contains some technical results we will refer to when needed.

\section{Gaussian quantum systems}\label{GQSy}
A Gaussian quantum system with $n$ modes is the quantum system associated to the Hilbert space of $n$ Harmonic oscillators, i.e. to the representation of the canonical commutation relations
\begin{equation}\label{quadr}
\left[\hat{Q}^i,\;\hat{P}^j\right]=i\;\delta^{ij}\;,\qquad i,\;j=1,\;\ldots,\;n\;,
\end{equation}
where for simplicity, as in the whole Thesis, we have set
\begin{equation}
\hbar=1\;.
\end{equation}
The canonical coordinates $\hat{Q}^i$ and $\hat{P}^i$ are called quadratures.
It is useful to put them collectively in the column vector
\begin{equation}
\hat{\mathbf{R}}=\left(
                   \begin{array}{c}
                     \hat{R}^1 \\
                     \vdots \\
                     \hat{R}^{2n} \\
                   \end{array}
                 \right):=\left(
                   \begin{array}{c}
                     \hat{Q}^1 \\
                     \hat{P}^1\\
                     \vdots \\
                     \hat{Q}^{n} \\
                     \hat{P}^n
                   \end{array}
                 \right)\;,
\end{equation}
with commutation relations
\begin{equation}
\left[\hat{R}^i,\;\hat{R}^j\right]=i\,\Delta^{ij}\;,\qquad i,\;j=1,\;\ldots,\;2n\;,
\end{equation}
where $\Delta$ is the symplectic form given by the antisymmetric matrix
\begin{equation}
\Delta=\bigoplus_{k=1}^n\left(
                          \begin{array}{cc}
                            0 & 1 \\
                            -1 & 0 \\
                          \end{array}
                        \right)\;.
\end{equation}
It is useful to define the ladder operators
\begin{equation}\label{Iladder}
\hat{a}^i=\frac{\hat{Q}^i+i\;\hat{P}^i}{\sqrt{2}}\;,\qquad i=1,\;\ldots,\;n\;,
\end{equation}
satisfying the commutation relations
\begin{equation}\label{CCRa}
\left[\hat{a}^i,\;\left(\hat{a}^j\right)^\dag\right]=\delta^{ij}\;,\qquad i,\;j=1,\;\ldots,\;n\;.
\end{equation}
We can put all the ladder operators together in the column vector
\begin{equation}
\hat{\mathbf{a}}=\left(
                   \begin{array}{c}
                     \hat{a}^1 \\
                     \vdots \\
                     \hat{a}^n \\
                   \end{array}
                 \right)\;.
\end{equation}
We can then define the vacuum as the state annihilated by all the destruction operators:
\begin{equation}
\hat{a}^i|0\rangle=0\;,\qquad i=1,\;\ldots,\;n\;.
\end{equation}

Gaussian quantum systems play a central role in quantum communication theory, since they are the correct framework to represent modes of electromagnetic radiation \cite{barnett2002methods}.
In this interpretation, the ladder operators \eqref{Iladder} and their Hermitian conjugates destroy and create a photon in the corresponding mode, respectively.
The energy is proportional to the number of photons, and the Hamiltonian is then
\begin{equation}\label{Hosc}
\hat{H}=\sum_{i=1}^n\left(\hat{a}^i\right)^\dag\hat{a}^i\;,
\end{equation}
where for simplicity we have set also the frequency equal to $1$.

\section{Quantum Gaussian states}\label{GQSt}
In analogy with classical Gaussian probability distributions, a quantum Gaussian state is a thermal state
\begin{equation}\label{Gstate}
\hat{\rho}_G=e^{-\beta\hat{H}}\left/\mathrm{Tr}\;e^{-\beta\hat{H}}\right.
\end{equation}
of an Hamiltonian that is a generic second-order polynomial in the quadratures, i.e.
\begin{equation}\label{HGauss}
\hat{H}=\left(\hat{\mathbf{R}}-\mathbf{r}\right)^T H\left(\hat{\mathbf{R}}-\mathbf{r}\right)\;,
\end{equation}
where $\mathbf{r}\in\mathbb{R}^{2n}$ is the vector of the expectation values of the quadratures, also called first moment, i.e.
\begin{equation}
r^i=\mathrm{Tr}\left[\hat{R}^i\;\hat{\rho}_G\right]\;,\qquad i=1,\;\ldots,\;2n\;,
\end{equation}
$H$ is a real strictly positive $2n\times2n$ matrix and $\beta>0$ is the inverse temperature (see also Section \ref{Gsta} of Appendix \ref{appG}).
Pure Gaussian states can be recovered in the zero-temperature limit $\beta\to\infty$.
They are the ground states of the quadratic Hamiltonians \eqref{HGauss}.
We stress that the Hamiltonian \eqref{HGauss} does \emph{not} need to be the photon-number Hamiltonian \eqref{Hosc} that governs the evolution of the system, hence a Gaussian state is not necessarily a thermal state in the thermodynamical sense.
The Hamiltonian \eqref{Hosc} can be recovered setting $H=\mathbb{I}_{2n}/2$ and $\mathbf{r}=0$.
In this case $\hat{\rho}_G$ is called a thermal Gaussian state.

As in the classical case, we can define the covariance matrix of $\hat{\rho}_G$ as
\begin{equation}
\sigma^{ij}:=\mathrm{Tr}\left[\left\{\hat{R}^i-r^i,\;\hat{R}^j-r^j\right\}\;\hat{\rho}_G\right]\;,\qquad i,\;j=1,\ldots,2n\;,
\end{equation}
where $\left\{\cdot,\cdot\right\}$ stands for the anticommutator.
As for classical Gaussian probability distributions, the quantum Gaussian state \eqref{Gstate} maximizes the von Neumann entropy among all the states with the same average energy with respect to the Hamiltonian $\hat{H}$ \cite{holevo2013quantum}.

The eigenvalues of $\sigma\Delta^{-1}$ are pure imaginary and come in couples of complex conjugates.
Their absolute values are called the \emph{symplectic eigenvalues} of $\sigma$ \cite{holevo2013quantum}.
The positivity of $\hat{\rho}_G$ implies that all the symplectic eigenvalues are larger or equal than $1$ \cite{holevo2013quantum} (see also Section \ref{secmom} of Appendix \ref{appG}).

If this condition is saturated, the state is pure \cite{holevo2013quantum}.
It is easy to check that the identity matrix has only $1$ as symplectic eigenvalue.
The Gaussian pure states with the identity as covariance matrix are called coherent states \cite{barnett2002methods}, that are the quantum analog of the classical Dirac deltas.
All the other Gaussian pure states are called squeezed.

\section{Quantum Gaussian channels}\label{GQCh}
Quantum channels are the mathematical representation for the most generic physical operation that can be performed in the laboratory on a quantum state.

An operator $\hat{X}$ acting on an Hilbert space $\mathcal{H}$ is called trace-class if its trace norm is finite:
\begin{equation}
\left\|\hat{X}\right\|_1:=\mathrm{Tr}\sqrt{\hat{X}^\dag\hat{X}}<\infty\;.
\end{equation}
We denote with $\mathfrak{T}(\mathcal{H})$ the set of trace-class operator acting on $\mathcal{H}$.
It is easy to check that any density matrix $\hat{\rho}$ has $\left\|\hat{\rho}\right\|_1=1$, and hence belongs to this class.
We denote as $\mathfrak{S}(\mathcal{H})$ the set of the density matrices on $\mathcal{H}$, i.e. the positive operators with trace one.

Given two Hilbert spaces $\mathcal{H}_A$ and $\mathcal{H}_B$ with associated sets of trace-class operators $\mathfrak{T}_A$ and $\mathfrak{T}_B$, a quantum operation from $A$ to $B$ is a continuous linear operator
\begin{equation}
\Phi:\mathfrak{T}_A\to\mathfrak{T}_B
\end{equation}
with the following properties:
\begin{itemize}
  \item it commutes with hermitian conjugation, i.e.
  \begin{equation}
  \Phi\left(\hat{X}^\dag\right)={\Phi\left(\hat{X}\right)}^\dag\;;
  \end{equation}
  \item it is completely positive, i.e.
  \begin{equation}
  \left(\mathbb{I}_{A'}\otimes\Phi\right)\left(\hat{X}\right)\geq0\qquad\forall\;\hat{X}\geq0\;,\quad\hat{X}\in\mathfrak{T}\left(\mathcal{H}_{A'}\otimes\mathcal{H}_A\right)\;.
  \end{equation}
\end{itemize}
If $\Phi$ is also trace-preserving, i.e.
\begin{equation}
\mathrm{Tr}\;\Phi\left(\hat{X}\right)=\mathrm{Tr}\;\hat{X}\;,
\end{equation}
it is called a quantum channel.
These three properties together guarantee that, for any Hilbert space $\mathcal{H}_{A'}$, the channel $\mathbb{I}_{A'}\otimes\Phi$ sends any quantum state on $\mathcal{H}_{A'}\otimes\mathcal{H}_A$ into a proper quantum state on $\mathcal{H}_{A'}\otimes\mathcal{H}_B$.

The displacement operators \cite{holevo2013quantum} are the unitary operators defined by
\begin{equation}
\hat{D}(\mathbf{x}):=e^{i\;\mathbf{x}^T\;\Delta^{-1}\;\hat{\mathbf{R}}}\;,\qquad\mathbf{x}\in\mathbb{R}^{2n}\;.
\end{equation}
It is easy to show that their action on the quadratures is a shift:
\begin{equation}
{\hat{D}(\mathbf{x})}^\dag\;\hat{\mathbf{R}}\;\hat{D}(\mathbf{x})=\hat{\mathbf{R}}+\mathbf{x}\;.
\end{equation}
They are then the quantum analog of the classical translations.

A $2n\times 2n$ real matrix $S$ is called symplectic if it preserves the symplectic form, i.e.
\begin{equation}
S\;\Delta\;S^T=\Delta\;.
\end{equation}
The symplectic $2n\times2n$ matrices form the real symplectic group $\mathrm{Sp}(2n,\mathbb{R})$ \cite{dutta1995real}.
We can associate to any $S\in \mathrm{Sp}(2n,\mathbb{R})$ a symplectic unitary $\hat{U}_S$ \cite{holevo2013quantum} that implements $S$ on the quadratures, i.e.
\begin{equation}
\hat{U}_S^\dag\;\hat{\mathbf{R}}\;\hat{U}_S=S\;\hat{\mathbf{R}}\;.
\end{equation}
The unitary operators $\hat{U}_S$ form a representation of $\mathrm{Sp}(2n,\mathbb{R})$, i.e.
\begin{equation}
\hat{U}_S\;\hat{U}_{S'}=\hat{U}_{SS'}\qquad\forall\;S,\;S'\in \mathrm{Sp}(2n,\mathbb{R})\;.
\end{equation}

It can be proven \cite{demoen1977completely,fannes1976quasi} that all the unitary operators that send any Gaussian state (i.e. any state of the form \eqref{Gstate}) of an $n$-mode Gaussian quantum system into a Gaussian state can be expressed as a displacement composed with a symplectic unitary.

We can now define a quantum Gaussian channel on an $n$-mode quantum Gaussian system as a quantum channel that sends any Gaussian state into a Gaussian state.
Let us add for the moment the additional hypothesis that for any joint $(m+n)$-mode Gaussian quantum system, the channel applied to the subsystem associated to the last $n$ modes sends any joint Gaussian state into another joint Gaussian state.
It has then been proven \cite{holevo2013quantum,giedke2002characterization} that the channel can be implemented as follows: add an auxiliary Gaussian state $\hat{\rho}_G$ on an auxiliary Gaussian quantum system $E$, perform a joint symplectic unitary $\hat{U}_S$, discard the auxiliary system and perform a displacement, i.e.
\begin{equation}\label{Gstinespring}
\Phi\left(\hat{X}\right)= \hat{D}(\mathbf{x})\;\mathrm{Tr}_E\left[\hat{U}_S\left(\hat{X}\otimes\hat{\rho}_G\right)\hat{U}_S^\dag\right]\;{\hat{D}(\mathbf{x})}^\dag\;, \qquad S\in\mathrm{Sp}\left(2(n+m),\mathbb{R}\right)\;,\quad\mathbf{x}\in\mathbb{R}^{2n}\;.
\end{equation}
We have proved (see \cite{de2015normal} and Chapter \ref{normal}) that requiring the channel to send into a Gaussian state any Gaussian state of a joint system is not actually necessary to get the decomposition \eqref{Gstinespring}: it is sufficient to require that the channel sends into a Gaussian state any Gaussian state of the $n$-mode system on which it is naturally defined.

\subsection{The quantum-limited attenuator and amplifier}\label{secattampl}
We present here two particular quantum Gaussian channels, that will be useful in the rest of the Thesis.
Let us consider the $n$-mode Gaussian quantum systems $A$ and $E$, with ladder operators
\begin{equation}
\hat{a}^i\;,\quad\hat{e}^i\;,\qquad i=1,\ldots,n\;.
\end{equation}

The quantum-limited attenuator on $A$ of parameter $0\leq\lambda\leq1$ admits the representation \eqref{Gstinespring}
\begin{equation}\label{attdef}
\mathcal{E}_\lambda\left(\hat{\rho}\right)= \mathrm{Tr}_E\left[\hat{U}_\lambda\left(\hat{\rho}\otimes|0\rangle_E\langle0|\right)\hat{U}_\lambda^\dag\right]\;,
\end{equation}
where the symplectic matrix $S$ is a rotation:
\begin{equation}
S=\left(
    \begin{array}{cc}
      \sqrt{\lambda}\;\mathbb{I}_{2n} & \sqrt{1-\lambda}\;\mathbb{I}_{2n} \\
      -\sqrt{1-\lambda}\;\mathbb{I}_{2n} & \sqrt{\lambda}\;\mathbb{I}_{2n} \\
    \end{array}
  \right)\;,
\end{equation}
such that the unitary operator $\hat{U}_\lambda$ acts on the quadratures as
\begin{eqnarray}
\hat{U}_\lambda^\dag\;\hat{a}^i\;\hat{U}_\lambda &=& \sqrt{\lambda}\;\hat{a}^i+\sqrt{1-\lambda}\;\hat{e}^i\;,\nonumber\\
\hat{U}_\lambda^\dag\;\hat{e}^i\;\hat{U}_\lambda &=& -\sqrt{1-\lambda}\;\hat{a}^i+\sqrt{\lambda}\;\hat{e}^i\;,\qquad i=1,\ldots,n\;.
\end{eqnarray}
It is possible to show \cite{ferraro2005gaussian} that $\hat{U}_\lambda$ is given by a mode mixing:
\begin{equation}\label{mixinga}
\hat{U}_\lambda=\exp\left[\arctan\sqrt{\frac{1-\lambda}{\lambda}}\;\left(\hat{\mathbf{a}}^\dag\hat{\mathbf{e}}-\hat{\mathbf{e}}^\dag\hat{\mathbf{a}}\right)\right]\;,
\end{equation}
and that the quantum-limited attenuators satisfy the multiplicative composition rule
\begin{equation}\label{compatt}
\mathcal{E}_\lambda\circ\mathcal{E}_{\lambda'}=\mathcal{E}_{\lambda\lambda'}\;,\qquad 0\leq\lambda\,,\;\lambda'\leq1\;.
\end{equation}
The quantum-limited attenuator provides a model for the attenuation of an electromagnetic signal travelling through metal wires, optical fibers or  free space, and $\lambda$ is the attenuation coefficient.
More in the spirit of our definition, the quantum-limited attenuator also models the action on a light beam of a beamsplitter with transmissivity $\lambda$.
In this case, the unitary $\hat{U}_\lambda$ implements the splitting of the beam in transmitted and reflected parts, and the partial trace over the environment $E$ represents the discarding of the reflected beam.

The quantum-limited amplifier on $A$ with parameter $\kappa\geq1$ admits the representation \eqref{Gstinespring}
\begin{equation}\label{ampdef}
\mathcal{A}_\kappa\left(\hat{\rho}\right)= \mathrm{Tr}_E\left[\hat{U}_\kappa\left(\hat{\rho}\otimes|0\rangle_E\langle0|\right)\hat{U}_\kappa^\dag\right]\;, \qquad \kappa\geq1\;,
\end{equation}
with
\begin{equation}
S=\left(
    \begin{array}{cc}
      \sqrt{\kappa}\;\mathbb{I}_{2n} & \sqrt{\kappa-1}\;T_{2n} \\
      \sqrt{\kappa-1}\;T_{2n} & \sqrt{\kappa}\;\mathbb{I}_{2n} \\
    \end{array}
  \right)\;,
\end{equation}
where $T_{2n}$ is the $n$-mode time-reversal
\begin{equation}
T_{2n}=\bigoplus_{k=1}^n\left(
                            \begin{array}{cc}
                              1 & 0 \\
                              0 & -1 \\
                            \end{array}
                          \right)
\end{equation}
that flips the sign of each $P^i$, leaving the $Q^i$ unchanged.
The unitary operator $\hat{U}_\kappa$ acts on the quadratures as
\begin{eqnarray}
\hat{U}_\kappa^\dag\;\hat{a}^i\;\hat{U}_\kappa &=& \sqrt{\kappa}\;\hat{a}^i+\sqrt{\kappa-1}\;\left(\hat{e}^i\right)^\dag\;,\nonumber\\
\hat{U}_\kappa^\dag\;\hat{e}^i\;\hat{U}_\kappa &=& \sqrt{\kappa-1}\;\left(\hat{a}^i\right)^\dag+\sqrt{\kappa}\;\hat{e}^i\;,\qquad i=1,\ldots,n\;.
\end{eqnarray}
It is possible to show \cite{ferraro2005gaussian} that $\hat{U}_\kappa$ is given by a squeezing operator:
\begin{equation}\label{squeezinga}
\hat{U}_\kappa=\exp\left[\mathrm{arctanh}\sqrt{\frac{\kappa-1}{\kappa}}\;\left(\hat{\mathbf{a}}^\dag\left(\hat{\mathbf{e}}^\dag\right)^T-\hat{\mathbf{e}}^T\hat{\mathbf{a}}\right)\right]\;,
\end{equation}
that does not conserve energy.
Indeed, its implementation in the laboratory requires active elements.

\section{The von Neumann entropy}\label{vNE}
The concept of entropy is ubiquitous in information theory.

The Shannon entropy of a discrete probability distribution $\{p_i\}_{i\in I}$ is defined as \cite{cover2006elements}
\begin{equation}
H[p]=-\sum_{i\in I}p_i\ln p_i\;,
\end{equation}
and quantifies the randomness of the distribution, i.e. how much information we acquire when the value of $i$ is revealed.
This last property is captured by the data compression theorem \cite{cover2006elements}.
Let us suppose to have a source that transmits a message made of $n$ letters $i_1$, ..., $i_n$, each one taken from an alphabet $I$.
The only a priori knowledge we have about the message is that at each of the $n$ steps, the letter $i$ will be sent with probability $p_i$ without any correlation between the steps.
The theorem then states that, in the large $n$ limit, while the number of possible messages is $|I|^n$, the transmitted message will be contained with probability one in a subset of only $\exp\left(n\,H[p]\right)$ messages.

The Shannon entropy has also a continuous analog for a probability distribution $p(\mathbf{x})$ over $\mathbb{R}^n$, the Shannon differential entropy \cite{cover2006elements}:
\begin{equation}
H[p]=-\int p(\mathbf{x})\ln p(\mathbf{x})\;d^nx\;.
\end{equation}

The generalization of the Shannon entropy to a quantum state $\hat{\rho}$ is the von Neumann entropy \cite{nielsen2010quantum}
\begin{equation}
S\left[\hat{\rho}\right]=-\mathrm{Tr}\left[\hat{\rho}\ln\hat{\rho}\right]\;,
\end{equation}
that coincides with the Shannon entropy of the discrete probability distribution associated to the eigenvalues of $\hat{\rho}$.

Its operational interpretation is provided by the Schumacher's coding theorem \cite{nielsen2010quantum}.
Let us suppose that our source now encodes each letter $i$ in a pure quantum state $|\psi_i\rangle$ taken from an Hilbert space $\mathcal{H}$, and sends the state $|\psi_{i_1}\rangle\otimes\ldots\otimes|\psi_{i_n}\rangle\in\mathcal{H}^{\otimes n}$.
Then, in the large $n$ limit, while the dimension of the global Hilbert space is $\dim\mathcal{H}^{\otimes n}=\left(\dim\mathcal{H}\right)^n$, the state sent will be contained with probability one in a subspace of dimension $\exp\left(n\,S\left[\hat{\rho}\right]\right)$, where $\hat{\rho}$ is the density matrix associated to the ensemble
\begin{equation}
\hat{\rho}=\sum_{i\in I}p_i|\psi_i\rangle\langle\psi_i|\;.
\end{equation}

\section{The classical communication capacity}\label{clcap}
A physically relevant quantity associated to a quantum channel $\Phi$ sending states on the quantum system $A$ into states on the quantum system $B$ is its capacity for transmitting classical information.

Let us suppose that Alice wants to transmit to Bob a message $i$ taken from an alphabet $I$ with the channel $\Phi$.
She then encodes her message into a quantum state $\hat{\rho}_i$ on the Hilbert space $\mathcal{H}_A$, and the state is transmitted to Bob through the quantum channel $\Phi$.
Bob receives the state $\Phi\left(\hat{\rho}_i\right)$, and performing a measurement on it he must guess the transmitted message $i$.
Let
\begin{equation}
\hat{M}_i\geq0\;,\quad i\in I\;,\qquad\sum_{i\in I}\hat{M}_i=\mathbb{I}_B
\end{equation}
be the elements of the POVM performed by Bob, i.e. if he receives the state $\hat{\rho}_B$, he associates to it the message $i$ with probability \begin{equation}
p\left(i\left|\hat{\rho}_B\right.\right)=\mathrm{Tr}\left[\hat{M}_i\;\hat{\rho}_B\right]\;.
\end{equation}
The set $\left\{\hat{\rho}_i,\;\hat{M}_i\right\}_{i\in I}$ of the states sent by Alice and of the POVM elements used by Bob is called a code for the quantum channel $\Phi$.

If Alice has sent the message $i$, Bob correctly guesses it with probability $\mathrm{Tr}\left[\hat{M}_i\;\Phi\left(\hat{\rho}_i\right)\right]$.
We define then the maximum error probability of the code $\mathcal{C}$ as
\begin{equation}
p_e(\mathcal{C})=\max_{i\in I}\left(1-\mathrm{Tr}\left[\hat{M}_i\;\Phi\left(\hat{\rho}_i\right)\right]\right)\;.
\end{equation}

We say that a communication rate $R$ is achievable by the channel $\Phi$ if for any $n\in\mathbb{N}$ there exists an alphabet $I_n$ with
\begin{equation}
|I_n|\geq e^{nR}
\end{equation}
and an associated code $\mathcal{C}^{(n)}$ for the composite channel $\Phi^{\otimes n}$ such that the maximum probability of error tends to zero for $n\to\infty$, i.e.
\begin{equation}
\lim_{n\to\infty}p_e\left(\mathcal{C}^{(n)}\right)=0\;.
\end{equation}

We define then the classical capacity of $\Phi$ as the supremum of all the achievable rates \cite{holevo2013quantum}:
\begin{equation}
C(\Phi)=\sup\left\{R\;|\;R\;\text{achievable}\right\}\;.
\end{equation}

In the standard definition of code $\mathcal{C}^{(n)}$ for the channel $\Phi^{\otimes n}$, Alice is allowed to use for the encoding entangled states on the Hilbert space $\mathcal{H}_A^{\otimes n}$, and Bob is allowed to perform a POVM with entangled elements $\hat{M}_i$ on the Hilbert space $\mathcal{H}_B^{\otimes n}$.
If we change the definition and allow Alice to use only separable states $\hat{\rho}_i$ in the encoding procedure (but we continue to allow Bob to perform any measurement), the capacity of the channel can be explicitely determined.
For any ensemble of states on Alice's Hilbert space $\mathcal{H}_A$
\begin{equation}\label{ensemble}
\mathcal{E}=\left\{p_i,\;\hat{\rho}_i\right\}_{i\in I}\;,
\end{equation}
we define
\begin{equation}\label{ChiHE}
\chi(\mathcal{E},\Phi)=S\left(\sum_{i\in I}p_i\;\Phi\left(\hat{\rho}_i\right)\right)-\sum_{i\in I}p_i\; S\left(\Phi(\hat{\rho}_i)\right)\;,
\end{equation}
where $S$ stands for the von Neumann entropy.
The capacity of the channel is then given by the so-called Holevo information \cite{holevo2013quantum}, given by the supremum of $\chi(\mathcal{E},\Phi)$ over all the possible Alice's ensembles $\mathcal{E}$:
\begin{equation}\label{ChiH}
\chi(\Phi)=\sup_\mathcal{E}\chi(\mathcal{E},\Phi)\;.
\end{equation}
The optimal rate is asymptotically achieved when Alice randomly chooses the states for the encoding according to the ensemble that maximizes \eqref{ChiH}.

If Alice is allowed to use entangled states, the capacity can be larger and involves a regularization over the number of channel uses:
\begin{equation}\label{Ctens}
C(\Phi)=\lim_{n\to\infty}\frac{1}{n}\;\chi\left(\Phi^{\otimes n}\right)\;.
\end{equation}
It is easy to show that for any $n\in\mathbb{N}$
\begin{equation}
\chi\left(\Phi^{\otimes n}\right)\geq n\;\chi(\Phi)\;,
\end{equation}
so that the limit in \eqref{Ctens} is actually a supremum.
If for any $n\in\mathbb{N}$
\begin{equation}
\chi\left(\Phi^{\otimes n}\right)=n\;\chi(\Phi)\;,
\end{equation}
we say that the Holevo information of the channel $\Phi$ is additive.
In this case, the regularization in \eqref{Ctens} is not necessary, and the classical capacity of $\Phi$ coincides with its Holevo information.

It is possible to prove (see \cite{holevo2015gaussian,holevo2013quantum} and references therein) that there exist quantum channels whose Holevo information is not additive.
However, the proof is not constructive, and no explicit example of such channel has been found.

\section{The capacity of Gaussian channels and the minimum output entropy conjecture}\label{gccapacity}
It is easy to show that the optimal ensemble for the Holevo information \eqref{ChiH} must be made of pure states.
Indeed, it is intuitive that sending the least possible noisy input is the best choice for Alice, given the message she wants to communicate.

According to the general principle ``Gaussian channels have Gaussian optimizers'' \cite{holevo2015gaussian}, the Holevo information of a quantum Gaussian channel has then been conjectured to be achieved by a Gaussian ensemble of pure Gaussian states, i.e. in the notation of \eqref{ensemble}
\begin{equation}\label{ensembleG}
\mathcal{E}=\left\{\frac{e^{-\mathbf{x}^T\sigma^{-1}\mathbf{x}}}{\sqrt{\det(\pi\sigma)}}\;d^{2n}x\;,\quad\hat{D}(\mathbf{x})\;\hat{\rho}_0\;{\hat{D}(\mathbf{x})}^\dag\right\}_{\mathbf{x}\in\mathbb{R}^{2n}}\;.
\end{equation}
Here $\hat{\rho}_0$ is a fixed pure Gaussian state, $\sigma$ is a real strictly positive $2n\times2n$ matrix and the probability distribution is continuous with the normalization
\begin{equation}
\int\frac{e^{-\mathbf{x}^T\sigma^{-1}\mathbf{x}}}{\sqrt{\det(\pi\sigma)}}\;d^{2n}x=1\;.
\end{equation}
The optimality of the ensemble \eqref{ensembleG} would mean that the best inputs Alice can use for transmitting information are pure Gaussian states.

The resulting Holevo information is
\begin{equation}\label{ChiHG}
\chi(\Phi)=S\left(\int \Phi\left(\hat{D}(\mathbf{x})\;\hat{\rho}_0\;{\hat{D}(\mathbf{x})}^\dag\right)\;\frac{e^{-\mathbf{x}^T\sigma^{-1}\mathbf{x}}}{\sqrt{\det(\pi\sigma)}}\;d^{2n}x\right) -S\left(\Phi\left(\hat{\rho}_0\right)\right)\;.
\end{equation}
Since for any Gaussian channel the states $\left\{\Phi\left(\hat{D}(\mathbf{x})\;\hat{\rho}_0\;{\hat{D}(\mathbf{x})}^\dag\right)\right\}_{\mathbf{x}\in\mathbb{R}^{2n}}$ are unitarily equivalent, the average over $\mathbf{x}$ in the second term of the right-hand side of \eqref{ChiHG} is not necessary.

It is easy to show that, for any nondegenerate Gaussian channel, sending $\sigma\to\infty$ in \eqref{ChiHG} results in an infinite capacity.
This occurs also in the classical case, and is due to the possibility for Alice of sending an arbitrary number of photons per channel use, allowing her to send an arbitrary amount of information.
However, in any realistic scenario the available input power is limited.
This constraint can be implemented \cite{holevo2013quantum,holevo2016constrained} requiring the input ensemble to have bounded mean energy:
\begin{equation}\label{constrE}
\sum_{i\in I}p_i\;\mathrm{Tr}\left[\hat{H}\;\hat{\rho}_i\right]\leq E\;,
\end{equation}
where $\hat{H}$ is the number Hamiltonian \eqref{Hosc}.

With this constraint, it is natural to consider the class of quantum Gaussian channels that commute with the time evolution generated by $\hat{H}$, i.e. for any trace-class operator $\hat{X}$ and any $t\in\mathbb{R}$,
\begin{equation}
\Phi\left(e^{-i\hat{H}t}\;\hat{X}\;e^{i\hat{H}t}\right)=e^{-i\hat{H}t}\;\Phi\left(\hat{X}\right)\;e^{i\hat{H}t}\;.
\end{equation}
These channels are called gauge-covariant \cite{holevo2013quantum}.
They are the most physically relevant Gaussian channels, since they preserve the class of thermal Gaussian states, and model the effects of signal attenuation and noise addition that affect electromagnetic communications via metal wires, optical fibers and free space \cite{weedbrook2012gaussian}.

Thermal Gaussian states have the maximum entropy among all the states with a given average energy \cite{wolf2006extremality}, and the average energy of the output of a gauge-covariant Gaussian channel is determined by the average energy of the input alone.
It follows that for any gauge-covariant Gaussian channel $\Phi$ the first term in the right-hand side of \eqref{ChiH} under the constraint \eqref{constrE} is maximized by a Gaussian ensemble of coherent states of the form \eqref{ensembleG}, with $\hat{\rho}_0$ the vacuum state and $\sigma$ proportional to the identity \cite{holevo2001evaluating}.
The last step to prove the optimality of the Gaussian ensemble is then to prove that coherent states maximize also the second term in the right-hand side of \eqref{ChiH}, i.e. they minimize the output entropy of the channel \cite{giovannetti2004minimum}.
This minimum output entropy conjecture has been a longstanding problem, only recently solved \cite{giovannetti2015majorization,giovannetti2015solution}.
This result implies that coherent states provide the optimal ensemble for transmitting classical information through any gauge-covariant quantum Gaussian channel, thus permitting the determination of its Holevo information \cite{giovannetti2014ultimate}.
Since the coherent states of a multimode Gaussian quantum system are product states, it follows that entangled input states are not useful, and the Holevo information is additive and then coincides with the classical capacity of the channel.

\section{Majorization}\label{secmaj}
Actually, the minimum output entropy conjecture follows from a stronger property of gauge-covariant quantum Gaussian channels related to majorization theory.

Majorization is the order relation between quantum states induced by random unitary operations: we say that the quantum state $\hat{\rho}$ majorizes the quantum state $\hat{\sigma}$ if there exists a probability measure $\mu$ on the set of unitary operators such that
\begin{equation}\label{majru}
\hat{\sigma}=\int\hat{U}\;\hat{\rho}\;\hat{U}^\dag\;d\mu\left(\hat{U}\right)\;.
\end{equation}
However, since this definition makes the test of the order relation difficult, majorization is usually defined as a property of the spectrum of the states.
The interested reader can find more details in the dedicated book \cite{marshall2010inequalities}, that however deals only with the finite-dimensional case.
\begin{defn}[Majorization]
Let $x$ and $y$ be decreasing summable sequences of positive numbers, i.e. $x_0\geq x_1\geq\ldots\geq0$ and $y_0\geq y_1\geq\ldots\geq0$.
We say that $x$ weakly sub-majorizes $y$, or $x\succ_w y$, iff for any $n\in\mathbb{N}$
\begin{equation}
\sum_{i=0}^n x_i\geq\sum_{i=0}^n y_i\;.
\end{equation}
If they have also the same sum, we say that $x$ majorizes $y$, or $x\succ y$.
\end{defn}
\begin{defn}
Let $\hat{X}$ and $\hat{Y}$ be positive trace-class operators with eigenvalues in decreasing order $\{x_n\}_{n\in\mathbb{N}}$ and $\{y_n\}_{n\in\mathbb{N}}$, respectively.
We say that $\hat{X}$ weakly sub-majorizes $\hat{Y}$, or $\hat{X}\succ_w\hat{Y}$, iff $x\succ_w y$.
We say that $\hat{X}$ majorizes $\hat{Y}$, or $\hat{X}\succ\hat{Y}$, if they have also the same trace.
\end{defn}
The link with the definition in terms of random unitary operation is provided by the following:
\begin{thm}
Given two positive operators $\hat{X}$ and $\hat{Y}$ with the same finite trace, the following conditions are equivalent:
\begin{enumerate}
  \item $\hat{X}\succ\hat{Y}$;
  \item For any continuous nonnegative convex function $f:[0,\infty)\to\mathbb{R}$ with $f(0)=0\,$,
  \begin{equation}\label{Trf}
  \mathrm{Tr}\;f\left(\hat{X}\right)\geq\mathrm{Tr}\;f\left(\hat{Y}\right)\;;
  \end{equation}
  \item For any continuous nonnegative concave function $g:[0,\infty)\to\mathbb{R}$ with $g(0)=0\,$,
  \begin{equation}\label{Trg}
  \mathrm{Tr}\;g\left(\hat{X}\right)\leq\mathrm{Tr}\;g\left(\hat{Y}\right)\;;
  \end{equation}
  \item $\hat{Y}$ can be obtained applying to $\hat{X}$ a convex combination of unitary operators, i.e. there exists a probability measure $\mu$ on unitary operators such that
  \begin{equation}\label{majU}
  \hat{Y}=\int\hat{U}\,\hat{X}\,\hat{U}^\dag\;d\mu\left(\hat{U}\right)\;.
  \end{equation}
\end{enumerate}
\begin{proof}
See Theorems 5, 6 and 7 of \cite{wehrl1974chaotic}.
We notice that Ref. \cite{wehrl1974chaotic} uses the opposite definition of the symbol ``$\succ$'' with respect to most literature (and to Ref. \cite{marshall2010inequalities}), i.e. there $\hat{X}\succ\hat{Y}$ means that $\hat{X}$ is majorized by $\hat{Y}$.
\end{proof}
\end{thm}
\begin{rem}\label{majS}
If $\hat{X}$ and $\hat{Y}$ are quantum states (i.e. their trace is one), \eqref{Trg} with $g(x)=-x\ln x$ implies that the von Neumann entropy of $\hat{X}$ is lower than the von Neumann entropy of $\hat{Y}$, while \eqref{Trf} with $f(x)=x^p$, $p>1$ implies the same for all the R{\'e}nyi entropies \cite{holevo2013quantum}.
\end{rem}

The minimum output entropy conjecture follows exactly from this last property: indeed, in Ref.'s \cite{giovannetti2015majorization,mari2014quantum} it is proven that for any gauge-covariant Gaussian quantum channel the output generated by any coherent state majorizes the output generated by any other state (see Fig. \ref{majfig}).

\begin{figure}[ht]
  \includegraphics[width=\textwidth]{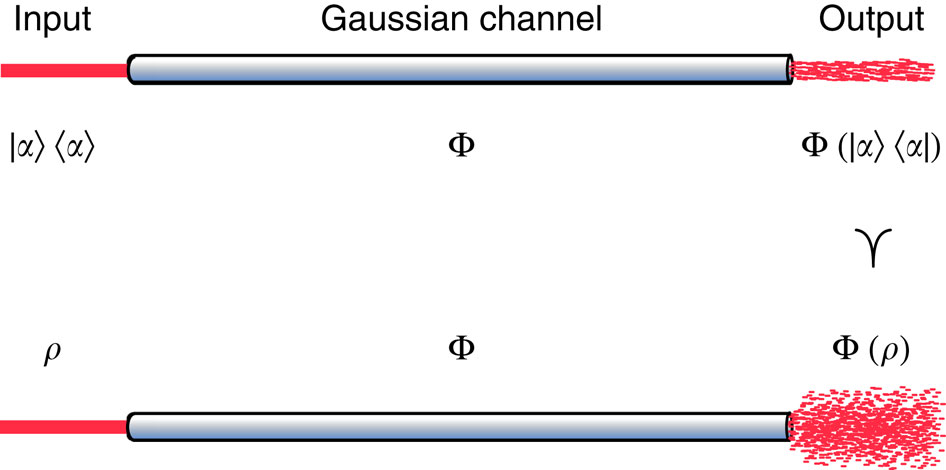}
  \caption{A coherent state $|\alpha\rangle\langle\alpha|$ and an arbitrary state $\hat{\rho}$ are both transmitted through the same gauge-covariant quantum Gaussian channel $\Phi$.
  The respective output states always satisfy the majorization relation $\Phi\left(|\alpha\rangle\langle\alpha|\right)\succ\Phi\left(\hat{\rho}\right)$.
  This means that coherent input states produce less noise at the output of the communication channel.}\label{majfig}
\end{figure}

\section{The capacity of the broadcast channel and the minimum output entropy conjecture}\label{broadcasti}
In Section \ref{secmaj} we have linked the classical capacity of gauge-covariant quantum Gaussian channels to their minimum output entropy.
In this Section we will link the capacity region of the degraded broadcast channel \cite{cover2006elements,guha2007classical,guha2007classicalproc,yard2011quantum,savov2015classical}, where Alice wants to communicate with two parties, to the minimum output entropy of a certain quantum channel for fixed input entropy.

The unconstrained minimum output entropy of gauge-covariant Gaussian quantum channels is achieved by the vacuum input state.
The constrained minimum output entropy for fixed input entropy is conjectured to be achieved by Gaussian thermal input states \cite{guha2007classical,guha2007classicalproc,guha2008entropy}, but a general proof does not exist yet.
In Chapter \ref{epi} we will prove the quantum Entropy Power Inequality, that bounds this constrained minimum output entropy.
In Chapter \ref{chepni} we will prove the conjecture for the one-mode quantum-limited attenuator.
In the remaining part of this Chapter, we define the degraded broadcast channel and its capacity region, and we show the role of the constrained minimum output entropy conjecture in its determination.

Let us suppose that Alice, who can prepare a state on a quantum system $A$, wants to communicate at the same time with Bob and Charlie, who can perform measurements on the quantum systems $B$ and $C$, respectively, with a quantum channel
\begin{equation}
\Phi_{A\to BC}:\mathfrak{T}\left(\mathcal{H}_A\right)\to\mathfrak{T}\left(\mathcal{H}_B\otimes\mathcal{H}_C\right)\;.
\end{equation}
Let us also suppose that Bob and Charlie cannot communicate nor perform joint measurements.
Let $\Phi_{A\to B}$ and $\Phi_{A\to C}$ be the effective quantum channels seen by Bob and Charlie, respectively, i.e. for any trace-class operator $\hat{X}$ on the Hilbert space $\mathcal{H}_A$
\begin{equation}
\Phi_{A\to B}\left(\hat{X}\right)=\mathrm{Tr}_C\Phi_{A\to BC}\left(\hat{X}\right)\;,\qquad\Phi_{A\to C}\left(\hat{X}\right)=\mathrm{Tr}_B\Phi_{A\to BC}\left(\hat{X}\right)\;.
\end{equation}
Let $I$ and $J$ be the sets of possible messages that Alice can send to Bob and Charlie, respectively.
A code $\mathcal{C}$ for the channel $\Phi_{A\to BC}$ is then given by a set of encoding states $\left\{\hat{\rho}_{ij}\right\}_{i\in I,\,j\in J}\subset\mathfrak{S}_A$, and two POVM on $B$ and $C$, respectively:
\begin{align}
&\hat{M}_i^B\geq0\;,\quad i\in I\;,\qquad\sum_{i\in I}\hat{M}_i^B=\hat{\mathbb{I}}_B\nonumber\\
&\hat{M}_j^C\geq0\;,\quad j\in J\;,\qquad\sum_{j\in J}\hat{M}_j^C=\hat{\mathbb{I}}_C\;,
\end{align}
such that, if Bob and Charlie receive the joint state $\hat{\rho}_{BC}$, the joint probability that they associate to it the messages $i$ and $j$, respectively, is
\begin{equation}
p\left(ij\left|\hat{\rho}_{BC}\right.\right)=\mathrm{Tr}_{BC}\left[\left(\hat{M}_i^B\otimes\hat{M}_j^C\right)\;\hat{\rho}_{BC}\right]\;.
\end{equation}
As in the single-party case, the maximum error probability of the code $\mathcal{C}$ is defined as
\begin{equation}
p_e(\mathcal{C})=\max_{i\in I,\,j\in J}\left(1-\mathrm{Tr}_{BC}\left[\left(\hat{M}_i^B\otimes\hat{M}_j^C\right)\Phi_{A\to BC}\left(\hat{\rho}_{ij}\right)\right]\right)\;.
\end{equation}
A couple of rates $(R_B,R_C)$ is said to be achievable if for any $n\in\mathbb{N}$ there exist two alphabets $I_n$ and $J_n$ with
\begin{equation}
|I_n|\geq e^{nR_B}\;,\qquad |J_n|\geq e^{nR_C}\;,
\end{equation}
and an associated code $\mathcal{C}^{(n)}$ for the channel $\Phi_{A\to BC}^{\otimes n}$ with asymptotically vanishing error probability:
\begin{equation}
\lim_{n\to\infty}p_e\left(\mathcal{C}^{(n)}\right)=0\;.
\end{equation}
The capacity region of the channel $\Phi_{A\to BC}$ is then defined as the closure of the set of all the achievable couples of rates.

It is possible to show \cite{guha2007classical} that for any point $(R_B,R_C)$ of the capacity region there exists a sequence of sets $I_n$ and $J_n$, $n\in\mathbb{N}$ with an associated ensemble of \emph{pure} states on the Hilbert space $\mathcal{H}_A^{\otimes n}$
\begin{equation}\label{EBC}
\mathcal{E}^{(n)}=\left\{p^{(n)}_i q^{(n)}_j,\;\hat{\rho}^{(n)}_{ij}\right\}_{i\in I_n,\,j\in J_n}\;,
\end{equation}
such that
\begin{eqnarray}\label{RBlim}
R_B &\leq& \liminf_{n\to\infty} \frac{1}{n}\sum_{j\in J_n}q^{(n)}_j\;\chi\left(\mathcal{E}^{(n)}_B(j),\;\Phi_{A\to B}^{\otimes n}\right)\\
R_C &\leq& \liminf_{n\to\infty} \frac{1}{n}\;\chi\left(\mathcal{E}_C^{(n)},\;\Phi_{A\to C}^{\otimes n}\right)\label{RClim}\;.
\end{eqnarray}
Here $i$ and $j$ represent the messages that Alice wants to send to Bob and Charlie, respectively, and $\mathcal{E}^{(n)}_B(j)$, $j\in J_n$, and $\mathcal{E}^{(n)}_C$ are the ensembles given by
\begin{eqnarray}
\mathcal{E}^{(n)}_B(j) &=& \left\{p^{(n)}_i,\;\hat{\rho}^{(n)}_{ij}\right\}_{i\in I_n}\;,\qquad j\in J_n\nonumber\\
\mathcal{E}^{(n)}_C &=& \left\{q^{(n)}_j,\;\sum_{i\in I_n}p^{(n)}_i\;\hat{\rho}^{(n)}_{ij}\right\}_{j\in J_n}\;.
\end{eqnarray}

In this setup, the energy constraint \eqref{constrE} becomes for the ensemble $\mathcal{E}^{(n)}$
\begin{equation}\label{cEbr}
\mathrm{Tr}\left[\hat{H}_n\;\hat{\omega}^{(n)}\right]\leq n\,E\;,
\end{equation}
where $\hat{\omega}^{(n)}$ is the average state
\begin{equation}
\hat{\omega}^{(n)}=\sum_{i\in I_n,\,j\in J_n}p^{(n)}_{i}q^{(n)}_j\hat{\rho}^{(n)}_{ij}\;,
\end{equation}
and $\hat{H}_n$ is the Hamiltonian on $\mathcal{H}_A^{\otimes n}$
\begin{equation}
\hat{H}_n=\sum_{i=1}^n\hat{\mathbb{I}}_A^{\otimes (i-1)}\otimes\hat{H}\otimes\hat{\mathbb{I}}_A^{\otimes (n-i)}\;.
\end{equation}

The broadcast quantum channel $\Phi_{A\to BC}$ is called \emph{degraded} \cite{cover2006elements,guha2007classical,guha2007classicalproc,yard2011quantum,savov2015classical} if Charlie's output is a degraded version of Bob's output, i.e. there exists a quantum channel $\Phi_{B\to C}$ such that
\begin{equation}\label{degrprop}
\Phi_{A\to C}=\Phi_{B\to C}\circ\Phi_{A\to B}\;.
\end{equation}
In this setup, a bound on the output entropy of the quantum channel $\Phi_{B\to C}$ in terms of its input entropy translates into a bound on the capacity region:
\begin{thm}\label{broadcastthm}
Let us suppose that for any $n\in\mathbb{N}$ and any state $\hat{\sigma}^{(n)}$ on the Hilbert space $\mathcal{H}_B^{\otimes n}$
\begin{equation}\label{lowerbS}
\frac{1}{n}\;S\left(\Phi_{B\to C}^{\otimes n}\left(\hat{\sigma}^{(n)}\right)\right)\geq f\left(\frac{1}{n}\;S\left(\hat{\sigma}^{(n)}\right)\right)\;,
\end{equation}
with $f$ a continuous increasing convex function.
Then any couple $(R_A,R_B)$ of achievable rates for the channel $\Phi_{A\to BC}$ with the energy constraint \eqref{cEbr} must satisfy
\begin{equation}\label{RBC}
f(R_B)+R_C\leq S(E)\;,
\end{equation}
where
\begin{equation}\label{limSE}
S(E)=\sup_{n\in\mathbb{N}}\left\{\left.\frac{1}{n}\;S\left(\Phi_{A\to C}\left(\hat{\omega}^{(n)}\right)\right)\right|\mathrm{Tr}\left[\hat{H}_n\;\hat{\omega}^{(n)}\right]\leq n\,E\right\}\;.
\end{equation}
\begin{proof}
For any ensemble $\mathcal{E}^{(n)}$ and for any $j\in J_n$ for the positivity of the entropy
\begin{equation}\label{chiBlim}
\chi\left(\mathcal{E}^{(n)}_B(j),\;\Phi_{A\to B}^{\otimes n}\right)\leq S\left(\hat{\sigma}_j^{(n)}\right)\;,
\end{equation}
where
\begin{equation}\label{sigmajn}
\hat{\sigma}_j^{(n)}=\sum_{i\in I_n}p^{(n)}_{i}\;\Phi_{A\to B}^{\otimes n}\left(\hat{\rho}^{(n)}_{ij}\right)\;,\qquad j\in J_n\;.
\end{equation}
For the degradability hypothesis \eqref{degrprop}
\begin{equation}\label{chiC}
\chi\left(\mathcal{E}_C^{(n)},\;\Phi_{A\to C}^{\otimes n}\right)=S\left(\Phi_{A\to C}^{\otimes n}\left(\hat{\omega}^{(n)}\right)\right)-\sum_{j\in J_n}q^{(n)}_j\;S\left(\Phi_{B\to C}^{\otimes n}\left(\hat{\sigma}^{(n)}_j\right)\right)\;.
\end{equation}
With \eqref{lowerbS}, \eqref{limSE}, the properties of $f$ and \eqref{chiBlim}, we have
\begin{eqnarray}
\frac{1}{n}\;\chi\left(\mathcal{E}_C^{(n)},\;\Phi_{A\to C}^{\otimes n}\right) &\leq& S\left(E\right)-f\left(\frac{1}{n}\sum_{j\in J_n}q_j^{(n)}\;S\left(\hat{\sigma}^{(n)}_j\right)\right)\leq\nonumber\\
&\leq& S\left(E\right)-f\left(\frac{1}{n}\sum_{j\in J_n}q_j^{(n)}\;\chi\left(\mathcal{E}^{(n)}_B(j),\;\Phi_{A\to B}^{\otimes n}\right)\right)\;.
\end{eqnarray}
Then,
\begin{eqnarray}
R_C &\leq& \liminf_{n\to\infty} \frac{1}{n}\;\chi\left(\mathcal{E}_C^{(n)},\;\Phi_{A\to C}^{\otimes n}\right)\leq\nonumber\\
&\leq& S\left(E\right)-\liminf_{n\to\infty}f\left(\frac{1}{n}\sum_{j\in J_n}q_j^{(n)}\;\chi\left(\mathcal{E}^{(n)}_B(j),\;\Phi_{A\to B}^{\otimes n}\right)\right)=\nonumber\\
&=& S\left(E\right)-f\left(\liminf_{n\to\infty}\frac{1}{n}\sum_{j\in J_n}q_j^{(n)}\;\chi\left(\mathcal{E}^{(n)}_B(j),\;\Phi_{A\to B}^{\otimes n}\right)\right)\leq\nonumber\\
&\leq& S\left(E\right)-f\left(R_B\right)\;,
\end{eqnarray}
where we have used that $f$ is continuous and increasing.
\end{proof}
\end{thm}
From this Theorem it is clear that the exact determination of the capacity region of a degraded broadcast channel requires the determination of the optimal $f$ in \eqref{lowerbS}, i.e. of the minimum output entropy of the channel $\Phi_{B\to C}^{\otimes n}$ for fixed input entropy.
If $\Phi_{B\to C}$ is a gauge-covariant Gaussian channel, this leads to the following conjecture:
\begin{prop}[Constrained minimum output entropy conjecture]\label{CMOE}
Gaussian thermal input states minimize the output entropy of any gauge-covariant Gaussian quantum channel for fixed input entropy.
\end{prop}
Up to now, this conjecture has been proven only for the one-mode quantum-limited attenuator (see Chapter \ref{chepni}).

One may ask whether the inequality \eqref{lowerbS} for $n=1$ is sufficient to derive the bound \eqref{RBC} in the setting where Alice cannot entangle the input state among successive uses of the channel, i.e. when the pure states $\hat{\rho}^{(n)}_{ij}$ are product states.
This would be the case if the bounds \eqref{RBlim}, \eqref{RClim} were additive, i.e. if they did not require the regularization over $n$.
In this case determining them for $n=1$ would be sufficient.
The answer is negative.
Indeed, we can rewrite \eqref{RClim} as
\begin{equation}
R_C \leq\liminf_{n\to\infty} \frac{1}{n}\left(S\left(\Phi_{A\to C}^{\otimes n}\left(\hat{\omega}^{(n)}\right)\right)-\sum_{j\in J_n}q^{(n)}_j\;S\left(\Phi_{B\to C}^{\otimes n}\left(\hat{\sigma}^{(n)}_j\right)\right)\right)\;.
\end{equation}
The subadditivity of the entropy for the terms $S\left(\Phi_{B\to C}^{\otimes n}\left(\hat{\sigma}^{(n)}_j\right)\right)$ goes in the wrong direction.
Additivity would hold if $\hat{\sigma}^{(n)}_j$ were product states, but from \eqref{sigmajn} in general this is not the case.

\section{The Gaussian degraded broadcast channel}\label{broadcastg}
This Section is dedicated to the degraded Gaussian quantum broadcast channel of Ref.'s \cite{guha2007classicalproc,guha2007classical}.

Let us consider the $n$-mode Gaussian quantum systems $A$, $B$, $C$ and $E$, with ladder operators
\begin{equation}
\hat{a}_i\;,\quad\hat{b}_i\;,\quad\hat{c}_i\;,\quad\hat{e}_i\;,\qquad i=1,\ldots,n\;,
\end{equation}
respectively.
Let Alice, Bob and Charlie control the systems $A$, $B$ and $C$, respectively, and let $E$ be the system associated to the environment.
Let also $\hat{U}_\eta$ be the isometry
\begin{equation}
\hat{U}_\eta:\mathcal{H}_A\otimes\mathcal{H}_E\to\mathcal{H}_B\otimes\mathcal{H}_C\;,\qquad \frac{1}{2}\leq\eta\leq1
\end{equation}
that implements the linear mixing of the modes
\begin{eqnarray}
\hat{U}_\eta^\dag\;\hat{b}_i\;\hat{U}_\eta &=& \sqrt{\eta}\;\hat{a}_i+\sqrt{1-\eta}\;\hat{e}_i\nonumber\\
\hat{U}_\eta^\dag\;\hat{c}_i\;\hat{U}_\eta &=& \sqrt{1-\eta}\;\hat{a}_i-\sqrt{\eta}\;\hat{e}_i\;,\qquad i=1,\ldots,n\;,\qquad\frac{1}{2}\leq\eta\leq1\;.
\end{eqnarray}
Upon identifying $B$ with $A$ and $C$ with $E$, and flipping the sign of the $\hat{c}^i$, $\hat{U}_\eta$ is the mode-mixing operator of Eq. \eqref{mixinga}.
Indeed, this channel can be modeled with a beamsplitter with transmission coefficient $\eta$, where Alice sends a signal into the port $A$, that is mixed with the environmental noise coming from $E$ and split into transmitted and reflected parts, that are finally received by Bob and Charlie, respectively (see Fig. \ref{broadcastfig}).
For simplicity, we consider only the case in which the state of the environment is set to be the vacuum, i.e. $\hat{\rho}_E=|0\rangle\langle0|$.
In this case, the beamsplitter has the only action of splitting the signal, and it does not introduce any noise.

\begin{figure}[ht]
\includegraphics[width=\textwidth]{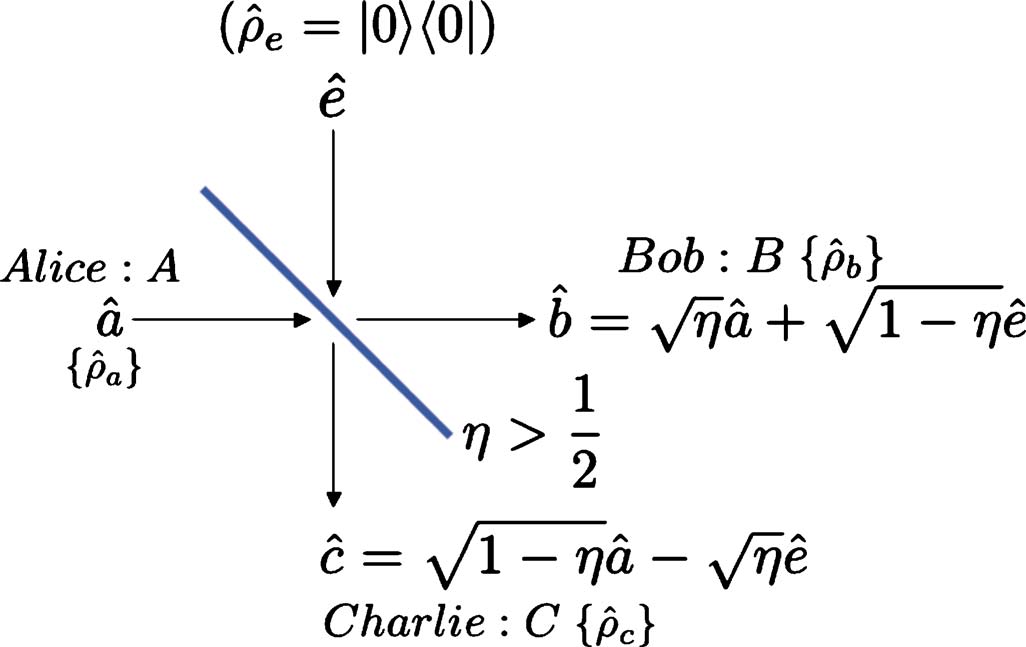}
\caption{Representation of the Gaussian broadcast channel.
Alice sends a signal into the port $A$ of the beamsplitter; Bob and Charlie receive the transmitted and the reflected signals at the ports $B$ and $C$, respectively. The port $E$ represents the action of the environment. In this case, the environment state is chosen to be the vacuum, i.e. the only action of the beamsplitter is splitting the signal into transmitted and reflected parts, without adding any noise.\label{broadcastfig}}
\end{figure}

In the notation of Section \ref{broadcasti}, the channel $\Phi_{A\to BC}$ is in this case the isometry given by the beamsplitter:
\begin{equation}
\Phi_{A\to BC}\left(\hat{\rho}\right)=\hat{U}_\eta\left(\hat{\rho}\otimes|0\rangle_E\langle0|\right)\hat{U}_\eta^\dag\;,
\end{equation}
and hence the reduced channels to $B$ and $C$ alone are given by the quantum-limited attenuators of \eqref{attdef}
\begin{equation}
\Phi_{A\to B}=\mathcal{E}_\eta\;,\qquad\Phi_{A\to C}=\mathcal{E}_{1-\eta}\;.
\end{equation}
Using the composition rule \eqref{compatt}, it is easy to see that this broadcast channel is degraded with
\begin{equation}
\Phi_{B\to C}=\mathcal{E}_\frac{1-\eta}{\eta}\;.
\end{equation}

\section{The capacity region of the Gaussian degraded broadcast channel}\label{broadcast}
We are now ready to apply Theorem \ref{broadcastthm} to the degraded broadcast channel described in Section \ref{broadcastg}.
The energy constraint will be of course imposed with respect to the photon-number Hamiltonian \eqref{Hosc}.

Gaussian thermal states maximize the entropy for fixed average energy, and for any $\hat{\rho}$ and $0\leq\lambda\leq1$
\begin{equation}
\mathrm{Tr}\left[\hat{H}\;\mathcal{E}_\lambda\left(\hat{\rho}\right)\right]=\lambda\;\mathrm{Tr}\left[\hat{H}\;\hat{\rho}\right]\;.
\end{equation}
It is then easy to see that the function $S(E)$ defined in \eqref{limSE} is
\begin{equation}
S(E)=g\left((1-\eta)E\right)\;,
\end{equation}
where $g(E)$ is the entropy of the one-mode Gaussian thermal state with average energy $E$ (see Eq. \eqref{defg} of Appendix \ref{appG}).

Determining the function $f$ in \eqref{lowerbS} requires now to determine the minimum output entropy of a quantum-limited attenuator $\mathcal{E}_\lambda^{\otimes n}$ for fixed input entropy.
Following the constrained minimum output entropy conjecture \ref{CMOE}, in Ref.'s \cite{guha2008capacity,guha2007classical,guha2007classicalproc} Gaussian thermal states are conjectured to minimize the output entropy, and then for any $0\leq\lambda\leq1$
\begin{equation}\label{EPnIatt}
g^{-1}\left(\frac{1}{n}\;S\left(\mathcal{E}_\lambda^{\otimes n}\left(\hat{\rho}\right)\right)\right)\geq \lambda\;g^{-1}\left(\frac{1}{n}\;S\left(\hat{\rho}\right)\right)\;.
\end{equation}
We prove this inequality in Chapter \ref{chepni} for $n=1$; its validity for $n\geq2$ is still an open problem.
Assuming \eqref{EPnIatt}, we can use
\begin{equation}
f(S)=g\left(\frac{1-\eta}{\eta}\;g^{-1}(S)\right)\;,
\end{equation}
that can easily shown to be continuous, increasing and convex.
The resulting bound on the capacity region would be
\begin{equation}\label{rfinale}
R_C+g\left(\frac{1-\eta}{\eta}\;g^{-1}(R_B)\right)\leq g\left((1-\eta)E\right)\;.
\end{equation}
This bound is optimal, in the sense that it can be shown \cite{guha2007classical,guha2007classicalproc} to be achieved by a Gaussian ensemble of coherent states.

The quantum Entropy Power Inequality that we prove in Chapter \ref{epi} provides instead the weaker bound
\begin{equation}\label{EPIatt}
e^{\frac{1}{n}\;S\left(\mathcal{E}_\lambda^{\otimes n}\left(\hat{\rho}\right)\right)}-1\geq \lambda\left(e^{\frac{1}{n}\;S\left(\hat{\rho}\right)}-1\right)\;,
\end{equation}
so that we can take
\begin{equation}
f(S)=\ln\left(\frac{1-\eta}{\eta}\left(e^S-1\right)+1\right)\;,
\end{equation}
that is still continuous, increasing and convex.
The resulting bound on the capacity region is
\begin{equation}\label{rfinale2}
R_C+\ln\left(\frac{1-\eta}{\eta}\left(e^{R_B}-1\right)+1\right)\leq g\left((1-\eta)E\right)\;.
\end{equation}
A comparison between Eq. \eqref{rfinale2} and the conjectured region \eqref{rfinale}  is shown in  Fig. \ref{fig1}: the discrepancy  being small.

\begin{figure}[ht]
\includegraphics[width=0.88\textwidth]{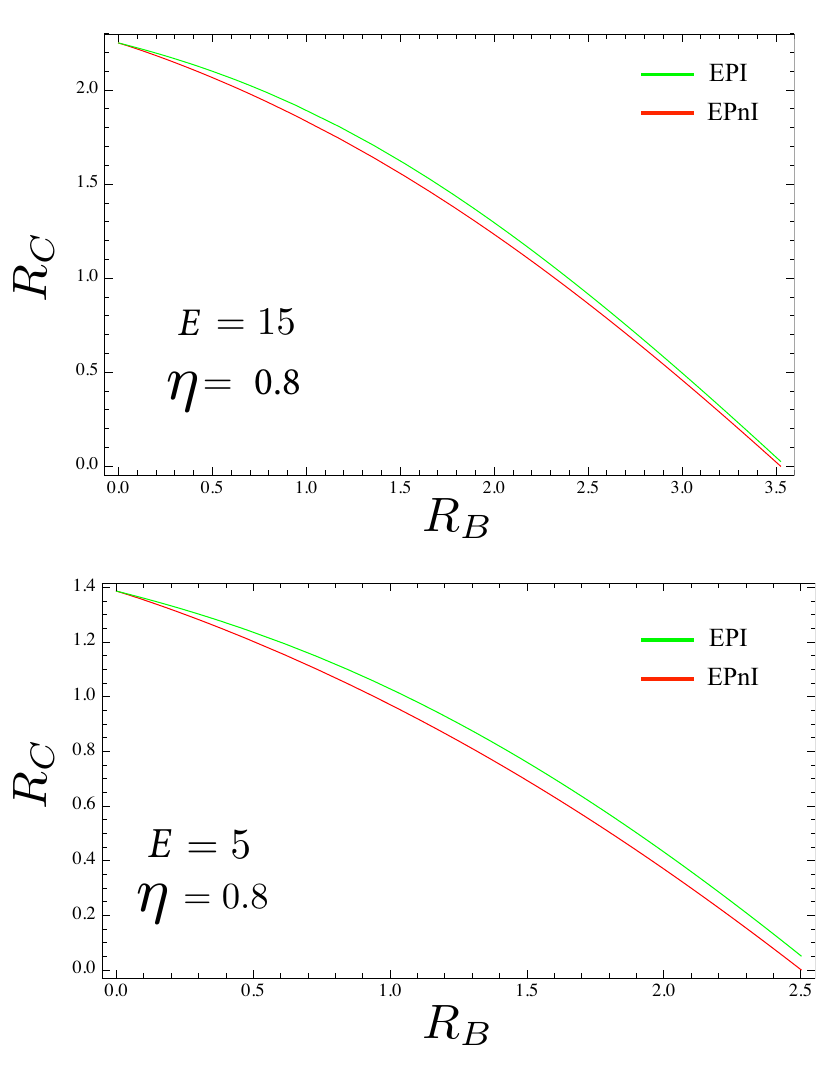}
\caption{Capacity region (expressed in nats per channel uses) for a broadcasting channel \cite{guha2007classical,guha2007classicalproc} in which the sender is communicating simultaneously with two receivers ($B$ and $C$) via a single bosonic mode which
splits at a beam splitter of transmissivity $\eta$  ($B$ receiving the transmitted signals, while $C$ receiving the reflected one), under input energy constraint which limits the
mean photon number of the input messages to be smaller than $E$. The region delimited by the red curve represents the achievable rates $R_B$ and $R_C$ which would apply if the (still unproven) EPnI conjecture \eqref{EPnIatt} held. The green curve instead is the bound one can derive via Eq. \eqref{EPIatt} from the EPI inequality we will prove in Chapter \ref{epi}.}
 \label{fig1}
\end{figure}

\chapter{The quantum Entropy Power Inequality}\label{epi}
In this Chapter we prove the quantum Entropy Power Inequality.
This inequality provides an almost optimal lower bound to the output von Neumann entropy of any linear combination of bosonic input modes in terms of their own entropies.
We have used it in Section \ref{broadcast} to obtain a upper bound to the capacity region of the degraded Gaussian broadcast channel, very close to the conjectured optimal one.

The Chapter is based on
\begin{enumerate}
\item[\cite{de2014generalization}] G.~De~Palma, A.~Mari, and V.~Giovannetti, ``A generalization of the entropy power inequality to bosonic quantum systems,'' \emph{Nature Photonics}, vol.~8, no.~12, pp. 958--964, 2014.\\ {\small\url{http://www.nature.com/nphoton/journal/v8/n12/full/nphoton.2014.252.html}}
\item[\cite{de2015multimode}] G.~De~Palma, A.~Mari, S.~Lloyd, and V.~Giovannetti, ``Multimode quantum entropy power inequality,'' \emph{Physical Review A}, vol.~91, no.~3, p. 032320, 2015.\\ {\small\url{http://journals.aps.org/pra/abstract/10.1103/PhysRevA.91.032320}}
\end{enumerate}

\section{Introduction}
In standard communication schemes, even if based on a digital encoding, the signals which are physically transmitted are intrinsically analogical in the sense that they can assume a continuous set of values. For example, the usual paradigm
is the transmission of information via amplitude and phase modulation of an electromagnetic field.
In general, a continuous signal with $k$ components
can be modeled by a random variable $\mathbf{X}$ with values in $\mathbb R^k$ associated with a probability measure
\begin{equation}
d\mu(\mathbf{x})=p(\mathbf{x})\;d^kx\;,\qquad\mathbf{x}\in\mathbb R^k\;.
\end{equation}
For example, a single mode of electromagnetic radiation is determined by a complex amplitude and therefore it can be classically described by a random variable $\mathbf{X}$ with $k=2$ real
components.
The Shannon differential entropy \cite{shannon2001mathematical,dembo1991information} of a general random variable $\mathbf{X}$ is defined as
\begin{equation}
H(\mathbf{X})=- \int_{\mathbb R^k} p(\mathbf{x}) \ln p(\mathbf{x})\; d^kx\;, \quad \mathbf{x} \in \mathbb R^k\;,
\end{equation}
and plays a fundamental role in information theory. Indeed depending on the context $H(\mathbf{X})$ quantifies the
noise affecting the signal or,
alternatively, the amount of information potentially encoded in the variable $\mathbf{X}$.

Now, let us assume to {\it mix} two random variables $\mathbf{X}_1$ and $\mathbf{X}_2$ and to get the new variable (see Fig.\ \ref{mixing})
\begin{equation}\label{classicalbs}
\mathbf{Y}=\sqrt{\lambda}\;\mathbf{X}_1 + \sqrt{1- \lambda}\;\mathbf{X}_2\;,\qquad 0\leq\lambda\leq1\;.
\end{equation}
\begin{figure}[ht]
\includegraphics[width=0.8\textwidth]{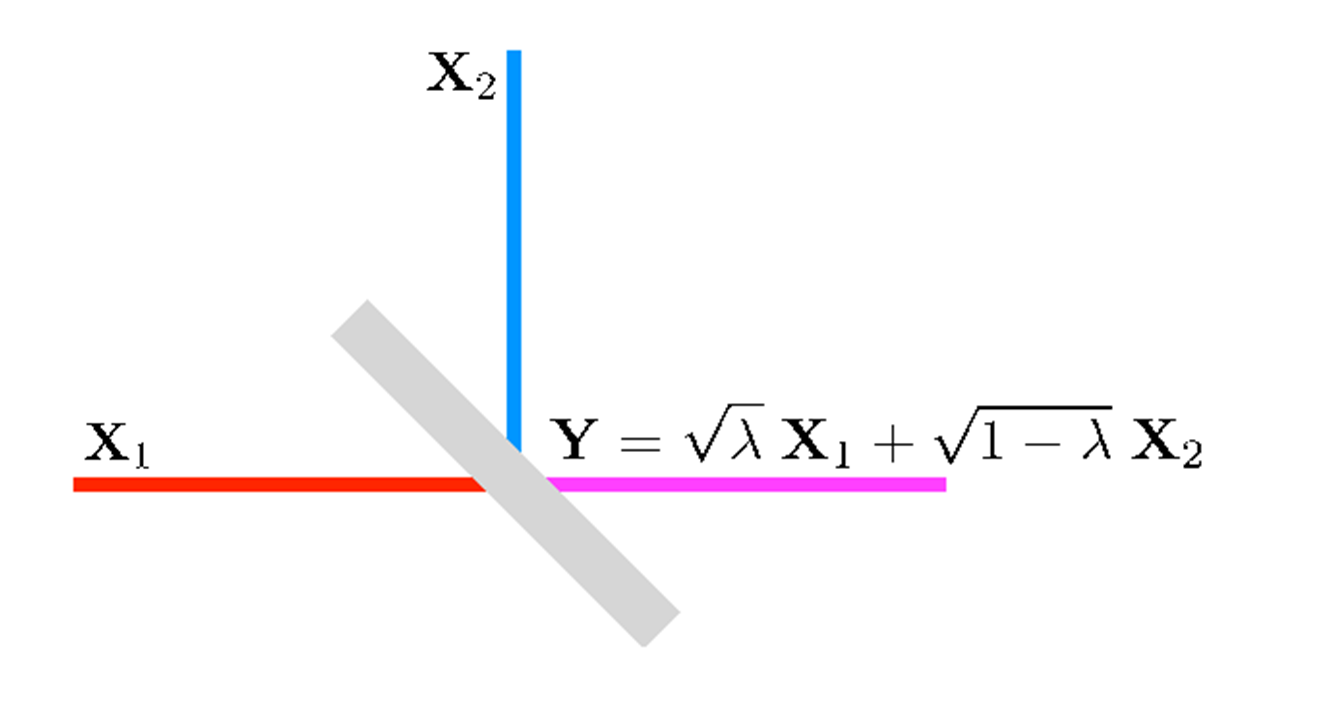}
\caption{Graphical representation of the coherent mixing of the two inputs $\mathbf{X}_1$ and $\mathbf{X}_2$. For the quantum mechanical analog the two input signals correspond to
 electromagnetic modes which  are coherently  mixed at a beamsplitter
of transmissivity $\lambda$. The entropy of the output signal is lower bounded by a function of the input entropies via the quantum Entropy Power Inequality defined in Eq. \eqref{EPI}. } \label{mixing}
\end{figure}
For example this is exactly the situation
in which two optical signals are physically mixed via a beamsplitter of transmissivity $\lambda$. What can be said about the entropy of the output variable $\mathbf{Y}$?
It can be shown that, if the inputs $\mathbf{X}_1$ and $\mathbf{X}_2$ are independent, the following {\it Entropy Power Inequality} (EPI) holds \cite{stam1959some,blachman1965convolution}
 \begin{equation} \label{cEPI}
e^{2 H_Y/k}\ge \lambda\; e^{2 H_1/k}+ (1- \lambda)\;e^{2 H_2/k}\;,
\end{equation}
stating that for fixed $H_1=H(\mathbf{X}_1)$, $H_2=H(\mathbf{X}_2)$, the output entropy $H_Y$ is minimized taking $\mathbf{X}_1$ and $\mathbf{X}_2$ Gaussian with proportional covariance matrices.
This is basically a lower bound on $H_Y$ and the name {\it entropy power} is motivated by the fact that if $p(\mathbf{x})$ is a product of $k$ equal isotropic Gaussians one has
\begin{equation}
\frac{1}{2 \pi e}e^{2 H(\mathbf{X})/k}=\sigma^2\;,
\end{equation}
where $\sigma^2$ is the variance of each Gaussian which is usually identified with the energy or {\it power} of the signal \cite{shannon2001mathematical}.
In the context of (classical) probability theory, several equivalent reformulations \cite{dembo1991information} and
generalizations \cite{verdu2006simple,rioul2011information,guo2006proof}  of Eq. \eqref{cEPI}  have been proposed,
whose proofs have recently renewed the interest in the field.
As a matter of fact,  these inequalities play a fundamental role in classical information theory, by providing computable bounds for the information capacities of various models of noisy channels \cite{shannon2001mathematical,bergmans1974simple,leung1978gaussian}.

The need for a quantum version of the EPI has arisen in the attempt of solving some fundamental problems in quantum communication theory. In particular the EPI has come into play when it has been realized that a suitable generalization to the quantum setting, called \emph{Entropy Photon number Inequality} (EPnI) (see \cite{guha2008entropy,guha2008capacity} and Section \ref{secEPnI}), would directly imply the solution of several optimization problems, including the determination of the classical capacity of Gaussian channels and of the capacity region of the bosonic broadcast channel \cite{guha2007classical,guha2007classicalproc} (see Sections \ref{gccapacity} and \ref{broadcast}).
Up to now the EPnI is still unproven and, while the classical capacity has been recently computed \cite{giovannetti2015solution,giovannetti2014ultimate} by proving the bosonic  minimum output entropy conjecture \cite{giovannetti2004minimum}, the exact capacity region of  the broadcast channel remains undetermined. In 2012 another quantum generalization of the EPI has been proposed, called \emph{quantum Entropy Power Inequality} (qEPI) \cite{konig2014entropy,konig2013limits}, together with its proof valid only for the $50:50$ beamsplitter corresponding to the case $\lambda=1/2$.
Our contribution is to show the validity of this inequality for any beamsplitter, and to extend it to the most general multimode scenario.

The qEPI proved in this Thesis directly gives tight bounds on several entropic quantities and hence constitutes a
potentially powerful tool which could  be used in quantum information theory in the same spirit
in which the classical EPI was instrumental in deriving important classical results like: a bound to the capacity of non-Gaussian channels \cite{shannon2001mathematical}, the convergence of the central limit theorem \cite{barron1986entropy},  the secrecy capacity of the Gaussian wiretap channel \cite{leung1978gaussian}, the capacity region of broadcast channels \cite{bergmans1974simple}, \emph{etc.}.  We consider some of the direct consequences of the qEPI and we hope to stimulate the research of other important implications in the field.

The multimode extension of the qEPI that we present applies to the context
where an arbitrary collection of independent input bosonic modes undergo to a scattering process which mixes them according to some linear coupling --- see Fig. \ref{modelepi} for a schematic representation of the model.
This new inequality permits to put bounds on the MOE inequality, still unproven for non gauge-covariant multimode channels, and then on the classical capacity of any quantum Gaussian channel.
Besides, our finding  can find potential applications in extending the single-mode results on the classical capacity region of the quantum bosonic broadcast channel to the Multiple-Input Multiple-Output setting (see e.g. Ref. \cite{cover2006elements}), providing upper bounds for the associated capacity regions.

\begin{figure}[ht]
\includegraphics[width=0.5\textwidth]{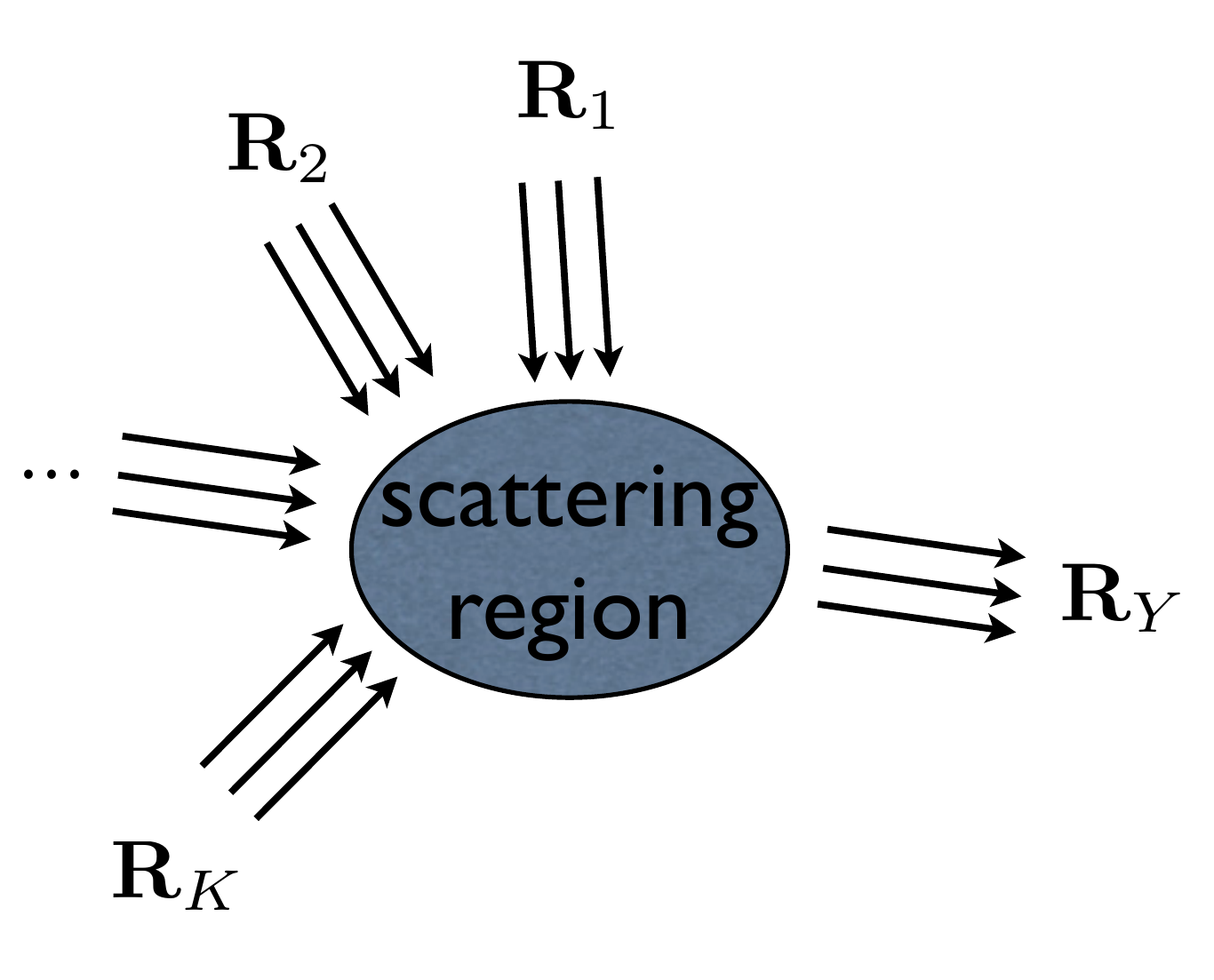}
\caption{Graphical representation of the scheme underlying the multimode qEPI \eqref{EPIrate}: it establishes a lower bound on the von Neumann entropy  emerging from the output port indicated by $\mathbf{R}_Y$ of a multimode scattering process that linearly couples
$K$ independent sets of bosonic input modes (each containing $n$ modes), initialized into factorized density matrices. }\label{modelepi}
\end{figure}

The Chapter is structured as follows.
In Section \ref{secproblem} we precisely define the linear combination of bosonic modes to which the quantum Entropy Power Inequality applies.
In Section \ref{secproofepi} we prove the quantum Entropy Power Inequality.
In Section \ref{secEPnI} we present the Entropy Photon-number Inequality, and in Section \ref{secGMOE} we link it to the generalized minimum output entropy conjecture necessary for determining the capacity of the degraded Gaussian broadcast channel.
Finally, we conclude in Section \ref{seccepi}.

\section{The problem}\label{secproblem}
We present directly the proof of the multimode version of the Entropy Power Inequality, since it includes the single-mode one as a particular case.

The multimode quantum generalization of the EPI we discuss in the present Thesis finds a classical analogous in the multi-variable version of the EPI \cite{shannon2001mathematical,stam1959some,blachman1965convolution,verdu2006simple,rioul2011information,guo2006proof}. The latter applies to a set
of $K$ independent random variables $\mathbf{X}_\alpha,\;\alpha=1,\ldots,K$, valued in $\mathbb{R}^m$ and collectively denoted by $\mathbf{X}$, with factorized probability densities
\begin{equation}
p_X(\mathbf{x})=p_1(\mathbf{x}_1)\ldots p_K(\mathbf{x}_K)\;,
\end{equation}
and with Shannon differential entropies \cite{shannon2001mathematical}
\begin{equation}
H_\alpha=-\left\langle\ln p_\alpha(\mathbf{x}_\alpha)\right\rangle\;,
\end{equation}
(the $\langle \cdots \rangle$ representing the average with respect to the associated probability distribution).
Defining hence the linear combination
\begin{equation}\label{Ycl}
\mathbf{Y}=M\,\mathbf{X}=\sum_{\alpha=1}^K M_{\alpha}\,\mathbf{X}_\alpha\;,
\end{equation}
where $M$ is an $m\times Km$ real matrix made by the $K$ blocks $M_\alpha$, each of dimension $m\times m$, the multi-variable EPI gives an (optimal) lower bound to the Shannon entropy $H_Y$ of $\mathbf{Y}$
\begin{equation}\label{cEPIm}
\exp[{2}H_Y/m]\geq\sum_{\alpha=1}^K|\det M_\alpha|^\frac{2}{m}\;\exp[{2}H_\alpha/m]\;,
\end{equation}
stating that it is minimized by Gaussian inputs.
In the original derivation \cite{shannon2001mathematical,stam1959some,blachman1965convolution,verdu2006simple,rioul2011information,guo2006proof} this inequality  is proved under the assumption that
all the $M_\alpha$ coincide with the identity matrix, i.e. for
\begin{equation}
\mathbf{Y}=\sum_{\alpha=1}^K\widetilde{\mathbf{X}}_\alpha\;.
\end{equation}
From this however Eq. \eqref{cEPIm} can be easily established choosing $\widetilde{\mathbf{X}}_\alpha=M_\alpha\mathbf{X}_\alpha$, and remembering that the entropy $\widetilde{H}_\alpha$ of $\widetilde{\mathbf{X}}_\alpha$ satisfies
\begin{equation}\label{rescS}
\widetilde{H}_\alpha=H_\alpha+\ln|\det M_\alpha|\;.
\end{equation}
It is also worth observing that for Gaussian variables the exponentials of
the entropies $H_\alpha$ and $H_\mathbf{Y}$
 are proportional to the determinant of the corresponding  covariance matrices, i.e.
\begin{equation}\label{Salpha}
H_\alpha=\frac{1}{2}\ln\det\left(\pi e\,\sigma_\alpha\right)
\end{equation}
and
\begin{equation}\label{SY}
H_\mathbf{Y}=\frac{1}{2}\ln\det\left(\pi e\,\sigma_\mathbf{Y}\right)\;,
\end{equation}
with
$$\sigma_\alpha=2\left\langle\Delta \mathbf{x}_\alpha \,\Delta \mathbf{x}_\alpha^T\right\rangle\;,\qquad\sigma_Y=2\left\langle\Delta \mathbf{y} \,\Delta \mathbf{y}^T\right\rangle$$
and
$$\Delta \mathbf{x}_\alpha = \mathbf{x}_\alpha - \left\langle \mathbf{x}_\alpha \right\rangle\;,\qquad\Delta \mathbf{y} = \mathbf{y} - \left\langle \mathbf{y} \right\rangle\;.$$
 Accordingly  in this special case  Eq. \eqref{cEPIm} can be seen as
an  instance  of the Minkowski's determinant inequality \cite{hardy1952inequalities}, stating that for any $K$ real $m\times m$ positive matrices
\begin{equation}\label{mink}
\left(\det\sum_{i=1}^K A_i\right)^\frac{1}{m}\geq\sum_{i=1}^K\left(\det A_i\right)^\frac{1}{m}\;,
\end{equation}
with equality iff all the $A_i$ are proportional.
Eq. \eqref{cEPIm} indeed follows from applying \eqref{mink}, \eqref{Salpha} and \eqref{SY} to the identity
\begin{equation}\label{cmphi}
\sigma_Y=\sum_{\alpha=1}^KM_\alpha\,\sigma_\alpha\,M_\alpha^T\;,
\end{equation}
and it saturates under the assumption that the matrices entering the sum are  all proportional to a given matrix $\sigma$, i.e.
\begin{eqnarray} \label{propcondition}
A_\alpha:=M_\alpha\,\sigma_\alpha\,M_\alpha^T = c_\alpha\,\sigma\;,
\end{eqnarray}
with $c_\alpha$ being arbitrary (real) coefficients.

In the quantum setting the random variables get replaced by $n=m/2$ bosonic modes (for each mode there are two quadratures, $Q$ and $P$), and instead of probability distributions over $\mathbb{R}^{2n}$, we have the quantum density matrices $\hat{\rho}_\alpha$ on the Hilbert space $L^2(\mathbb{R}^n)$ (see Sections \ref{GQSy} and \ref{GQSa} for the details).
For each $\alpha$, let $\hat{\mathbf{R}}_\alpha$ be the column vector (see \eqref{Rdef}) that collectively denotes all the quadratures of the $\alpha$-th subsystem.

Let us then consider totally factorized input states
\begin{equation}
\hat{\rho}_X=\bigotimes_{\alpha=1}^K\hat{\rho}_\alpha\;,
\end{equation}
where $\hat{\rho}_\alpha$ is the density matrix of the $\alpha$-th input, with associated characteristic function $\chi_\alpha(\mathbf{k}_\alpha)$ (see Section \ref{chia} in Appendix \ref{appG}).
The characteristic function of the global input state is then
\begin{equation}
\chi_X(\mathbf{k}_X)=\prod_{\alpha=1}^K\chi_\alpha(\mathbf{k}_\alpha)\;,
\end{equation}
with
\begin{equation}
\mathbf{k}_X=\left(\mathbf{k}_1,\;\ldots,\;\mathbf{k}_K\right)\;.
\end{equation}
The quantum analog of \eqref{Ycl} is defined imposing the same transformation law on the characteristic functions:
\begin{equation}\label{channelbs}
\chi_Y(\mathbf{k}_Y)=\chi_X\left(\mathbf{k}_XM\right)=\prod_{\alpha=1}^K\chi_\alpha\left(\mathbf{k}M_\alpha\right)\;,
\end{equation}
where as before, $M$ is a $2n\times 2Kn$ real matrix made by the $2n\times 2n$ square blocks $M_\alpha$.
The channel defined in \eqref{channelbs} can be recovered from the general expression of a Gaussian channel in Eq. \eqref{channelchi} of Appendix \ref{appG} putting $\alpha=0$ and $\mathbf{y}=\mathbf{0}$.
The complete-positivity condition \eqref{CP} imposes the constraint
\begin{equation}\label{condDelta}
M\;\Delta_X\;M^T=\sum_{\alpha=1}^K M_\alpha\Delta_\alpha M_\alpha^T=\Delta_Y\;,
\end{equation}
where $\Delta_Y$ is the symplectic form associated to the output $Y$, while
\begin{equation}
\Delta_X=\bigoplus_{\alpha=1}^K\Delta_\alpha
\end{equation}
is the form associated to the input $X$.

The channel \eqref{channelbs} can be implemented by an isometry (see \cite{holevo2013quantum} and Section \ref{QGCa})
\begin{equation}
\hat{U}:\mathcal{H}_X\longrightarrow\mathcal{H}_Y\otimes\mathcal{H}_Z
\end{equation}
between the input Hilbert space $\mathcal{H}_X$ and the tensor product of the output Hilbert space $\mathcal{H}_Y$ with an ancilla Hilbert space $\mathcal{H}_Z$:
\begin{equation}\label{channelU}
\hat{\rho}_Y=\Phi\left(\hat{\rho}_X\right)=\mathrm{Tr}_Z\left(\hat{U}\;\hat{\rho}_X\;\hat{U}^\dag\right)\;,
\end{equation}
where $\hat{U}$ satisfies
\begin{equation}\label{Yq}
\hat{U}^\dag\;\hat{\mathbf{R}}_Y\;\hat{U}=M\,\hat{\mathbf{R}}_X=\sum_{\alpha=1}^K M_{\alpha}\,\hat{\mathbf{R}}_\alpha\;.
\end{equation}
With this representation, the CP condition \eqref{condDelta} can be easily shown to arise from the preservation of the canonical commutation relations between the quadratures.
The isometry $\hat{U}$ in \eqref{channelU} does not necessarily conserve energy, i.e. it can contain active elements, so that even if the input $\hat{\rho}_X$ is the vacuum on all its $K$ modes, the output $\hat{\rho}_Y$ can be thermal with a nonzero temperature.

For $K=2$, the beamsplitter \cite{walls2012quantum} of parameter $0\leq\lambda\leq1$ is easily recovered with
\begin{equation}\label{bsdef}
M_1=\sqrt{\lambda}\;\mathbb{I}_{2n}\;,\qquad M_2=\sqrt{1-\lambda}\;\mathbb{I}_{2n}\;.
\end{equation}
In this case, upon identifying the output Hilbert space $\mathcal{H}_Y$ with the Hilbert space of the first input $\mathcal{H}_1$, the isometry $\hat{U}$ implements the same mode mixing of \eqref{mixinga}, i.e.
\begin{equation}\label{mixingbs}
\hat{U}=\exp\left[\arctan\sqrt{\frac{1-\lambda}{\lambda}}\;\left(\hat{\mathbf{a}}_1^\dag\hat{\mathbf{a}}_2-\hat{\mathbf{a}}_2^\dag\hat{\mathbf{a}}_1\right)\right]\;,
\end{equation}
where $\hat{\mathbf{a}}_\alpha$ is the vector of the ladder operators (see Eq. \eqref{Iladder}) associated to the $\alpha$-th subsystem.
Eq. \eqref{Yq} becomes then of the same form as \eqref{classicalbs}:
\begin{eqnarray}\label{beamsplitter}
\hat{U}^\dag\;\hat{\mathbf{Y}}\;\hat{U} &=& \sqrt{\lambda}\;\hat{\mathbf{X}}_1+\sqrt{1-\lambda}\;\hat{\mathbf{X}}_2\;,\nonumber\\
\hat{U}^\dag\;\hat{\mathbf{a}}_Y\;\hat{U} &=& \sqrt{\lambda}\;\hat{\mathbf{a}}_1+\sqrt{1-\lambda}\;\hat{\mathbf{a}}_2\;,
\end{eqnarray}
i.e. the output quadratures are a weighted sum of the corresponding input quadratures.

To get the quantum amplifier \cite{walls2012quantum} (see also Section \ref{secattampl}) of parameter $\kappa\geq1$, we must take instead
\begin{equation}
M_1=\sqrt{\kappa}\;\mathbb{I}_{2n}\;,\qquad M_2=\sqrt{\kappa-1}\;T_{2n}\;,
\end{equation}
where $T_{2n}$ is the $n$-mode time-reversal
\begin{equation}
T_{2n}=\bigoplus_{k=1}^n\left(
                            \begin{array}{cc}
                              1 & 0 \\
                              0 & -1 \\
                            \end{array}
                          \right)\;.
\end{equation}
In this case, with the same identification between $\mathcal{H}_Y$ and $\mathcal{H}_{1}$, the unitary $\hat{U}$ implements a squeezing \cite{barnett2002methods}, i.e.
\begin{equation}
\hat{U}=\exp\left[\mathrm{arctanh}\sqrt{\frac{\kappa-1}{\kappa}}\;\left(\hat{\mathbf{a}}_1^\dag\left(\hat{\mathbf{a}}_2^\dag\right)^T-\hat{\mathbf{a}}_1^T\hat{\mathbf{a}}_2\right)\right]\;,
\end{equation}
and acts on the ladder operators as
\begin{equation}\label{squeezedl}
\hat{U}^\dag\;\hat{\mathbf{a}}_Y\;\hat{U} = \sqrt{\kappa}\;\hat{\mathbf{a}}_1+\sqrt{\kappa-1}\;\hat{\mathbf{a}}_2^\dag\;.
\end{equation}
We notice in Eq. \eqref{squeezedl} the dagger on $\hat{\mathbf{a}}_2^\dag$, signaling that $\hat{U}$ does not conserve energy, and therefore it requires active elements to be implemented in the laboratory.

We can now state the multimode qEPI: the von Neumann entropies of the inputs $S_\alpha$ and the output $S_Y$ satisfy the analog of \eqref{cEPIm}
\begin{equation}
\exp[{S_Y}/{n}]\geq\sum_{\alpha=1}^K \lambda_\alpha\;\exp[{S_\alpha}/{n}]\;,\label{EPIm}
\end{equation}
where we have defined
\begin{equation}
\lambda_\alpha:=\left|\det M_\alpha\right|^\frac{1}{n}\;.
\end{equation}
For a beamsplitter of parameter $0\leq\lambda\leq1$ with inputs $\mathbf{X}_1$ and $\mathbf{X}_2$ and output $\mathbf{Y}$, \eqref{EPIm} reduces to
\begin{equation}\label{EPI}
\exp[S_Y/n]\geq\lambda\exp[S_1/n]+(1-\lambda)\exp[S_2/n]\;.
\end{equation}
For a quantum amplifier of parameter $\kappa\geq1$, we have instead
\begin{equation}
\exp[S_Y/n]\geq\kappa\exp[S_1/n]+(\kappa-1)\exp[S_2/n]\;.
\end{equation}

\section{The proof}\label{secproofepi}
The proof of Eq. \eqref{EPIm} proceeds along the same lines of its classical counterpart \cite{blachman1965convolution}.
We expect that the  qEPI should be saturated by quantum Gaussian states  with high entropy and whose covariance matrices $\sigma_\alpha$
fulfill the condition \eqref{propcondition} (the high entropy limit being necessary to ensure that the associated quantum Gaussian states behave as classical Gaussian probability distributions).
Let us hence suppose to apply a transformation on the input modes of the system which depends on a real parameter $t$ that plays the role of an effective temporal coordinate,  and which is constructed in  such a way that, starting from $t=0$ from the input state $\hat{\rho}_X$  it will drive the modes towards such optimal  Gaussian configurations in the asymptotic limit  $t\rightarrow\infty$ --- see Section \ref{evol}.
Accordingly for each $t\geq 0$ we will have an associated value for the entropies $S_\alpha$ and $S_Y$ which, if the qEPI is correct, should still fulfill the bound \eqref{EPIm}.
To verify this it is useful to put the qEPI \eqref{EPIm} in the rate form
\begin{equation}
\frac{\sum_{\alpha=1}^K \lambda_\alpha\;\exp[{S_\alpha}/{n}]}{\exp[{S_Y}/{n}]}\leq1\label{EPIrate}\;.
\end{equation}
We will then study  the  left-hand-side of Eq. \eqref{EPIrate}  showing that its
parametric  derivative is always positive  (see Section \ref{sec:para}) and that that for $t\rightarrow \infty$ it tends to 1  (see Section \ref{appscaling}).

\subsection{The Liouvillian}
The parametric evolution suitable for the proof will be given in terms of a quantum generalization of the classical Laplacian, that we define in this Section.

Let $\gamma\geq0$ be a positive semi-definite real matrix.
We define the Liouvillian
\begin{equation}\label{liouville}
\mathcal{L}_\gamma\left(\hat{X}\right):=\frac{1}{4}\gamma^{ij}\left.\frac{\partial^2}{\partial x^i\partial x^j}\hat{D}(\mathbf{x})\;\hat{X}\;{\hat{D}(\mathbf{x})}^\dag\right|_{\mathbf{x}=\mathbf{0}}\;,
\end{equation}
where the sum over the repeated indices is implicit.
$\mathcal{L}_\gamma$ is linear in $\gamma$, commutes with hermitian conjugation:
\begin{equation}
\mathcal{L}_\gamma\left(\hat{X}^\dag\right)=\left(\mathcal{L}_\gamma\left(\hat{X}\right)\right)^\dag\;,
\end{equation}
and is self-adjoint with respect to the Hilbert-Schmidt product:
\begin{equation}\label{Ldag}
\mathrm{Tr}\left(\mathcal{L}_\gamma\left(\hat{X}\right)\;\hat{Y}\right)=\mathrm{Tr}\left(\hat{X}\;\mathcal{L}_\gamma\left(\hat{Y}\right)\right)\;.
\end{equation}
Taking the characteristic function of both sides of \eqref{liouville}, and recalling Eq. \eqref{chiDD} of Appendix \ref{appG}, we get
\begin{equation}\label{chiL}
\chi_{\mathcal{L}_\gamma\left(\hat{X}\right)}(\mathbf{k})=-\frac{1}{4}\mathbf{k}\,\gamma\,\mathbf{k}^T\;\chi_{\hat{X}}(\mathbf{k})\;.
\end{equation}
If we formally define the exponential of $\mathcal{L}_\gamma$, Eq. \eqref{chiL} can be easily integrated into
\begin{equation}\label{chiaddn}
\chi_{e^{\mathcal{L}(\gamma)}\left(\hat{X}\right)}(\mathbf{k})=e^{-\frac{1}{4}\mathbf{k}\,\gamma\,\mathbf{k}^T}\;\chi_{\hat{X}}(\mathbf{k})\;,
\end{equation}
and $e^{\mathcal{L}(\gamma)}$ can be easily recognized as the additive-noise channel that can be recovered from \eqref{channelchi} with $M=\mathbb{I}$, $\alpha=\gamma$ and $\mathbf{y}=\mathbf{0}$.
This channel adds to the state noise with covariance matrix $\gamma$, acts on the moments as
\begin{eqnarray}
\sigma&\mapsto&\sigma+\gamma\label{addnoise}\\
\mathbf{r}&\mapsto&\mathbf{r}\;,
\end{eqnarray}
and hence on the Gaussian state $\hat{\rho}_G(\sigma,\,\mathbf{x})$ as
\begin{equation}\label{Lgauss}
e^{\mathcal{L}(\gamma)}\left(\hat{\rho}_G(\sigma,\,\mathbf{x})\right)=\hat{\rho}_G(\sigma+\gamma,\,\mathbf{x})\;.
\end{equation}

\subsection{Useful properties}
In the proof, we will need some properties of the channel defined in \eqref{channelbs} and of the Liouvillian \eqref{liouville}.

From \eqref{channelbs} and \eqref{chiDD} the action of $\Phi$ on translations follows:
\begin{equation}\label{PhiD}
\Phi\left(\hat{D}(\mathbf{x})\;\hat{X}\;{\hat{D}(\mathbf{x})}^\dag\right)=\hat{D}(M\mathbf{x})\;\Phi\left(\hat{X}\right)\;{\hat{D}(M\mathbf{x})}^\dag\;.
\end{equation}
We can now use \eqref{PhiD} and \eqref{liouville} to compute the action of $\Phi$ on $\mathcal{L}_\gamma$:
\begin{equation}
\Phi\left(\mathcal{L}(\gamma)\left(\hat{X}\right)\right) = \mathcal{L}\left(M\gamma M^T\right)\left(\Phi\left(\hat{X}\right)\right)\;,\label{LPhi}
\end{equation}
and hence
\begin{equation}\label{LePhi}
\Phi\left(e^{\mathcal{L}(\gamma)}\left(\hat{\rho}\right)\right)=e^{\mathcal{L}\left(M\gamma M^T\right)}\left(\Phi\left(\hat{\rho}\right)\right)\;.
\end{equation}

\subsection{The evolution}\label{evol}
The idea of the proof is to evolve the inputs (and consequently the output) toward Gaussian states with very high entropies and with covariance matrices satisfying \eqref{propcondition}.
For this purpose, we use the additive-noise channel that we have just defined in \eqref{chiaddn}.
Let us fix a positive matrix $\gamma$, and define for each $\alpha$
\begin{equation}\label{Malpha}
\gamma_\alpha:=\left\{
           \begin{array}{ll}
             \lambda_\alpha\;M_\alpha^{-1}\;\gamma\;M_\alpha^{-T}\qquad & \text{if}\;\lambda_\alpha>0 \\
             0 & \text{if}\;\lambda_\alpha=0 \\
           \end{array}
         \right.\;,
\end{equation}
such that
\begin{equation}
M_\alpha\;\gamma_\alpha\;M_\alpha^T=\lambda_\alpha\;\gamma\;.
\end{equation}
Let $t$ be the time of the evolution.
We apply to the $\alpha$-th input the additive-noise channel $e^{t_\alpha(t)\;\mathcal{L}(\gamma_\alpha)}$, with a time-dependent coefficient $t_\alpha(t)$ to be determined:
\begin{equation}\label{rhoit}
\hat{\rho}_\alpha(t):=e^{t_\alpha(t)\;\mathcal{L}(\gamma_\alpha)}(\hat{\rho}_\alpha)\;.
\end{equation}
We notice that, if some $\lambda_\alpha=0$, we are not evolving at all the corresponding state $\hat{\rho}_\alpha$.

From \eqref{LePhi} and \eqref{Malpha}, the evolution \eqref{rhoit} of the input mode induces the temporal
evolution  of the output modes
\begin{equation}\label{rhoyt}
\Phi\left(\left(\bigotimes_{\alpha=1}^Ke^{t_\alpha(t)\mathcal{L}(\gamma_\alpha)}\right)\hat{\rho}\right)=e^{t_Y(t)\mathcal{L}\left(\gamma\right)}\left(\Phi\left(\hat{\rho}\right)\right)\;,
\end{equation}
where
\begin{equation}\label{tY}
t_Y(t)=\sum_{\alpha=1}^K\lambda_\alpha t_\alpha(t)\;.
\end{equation}

From \eqref{chiaddn}, the characteristic functions evolve as
\begin{equation}
\chi_\alpha(\mathbf{k})(t)=e^{-\frac{1}{4}t_\alpha(t)\,\mathbf{k}\,\gamma_\alpha\,\mathbf{k}^T}\;\chi_\alpha(\mathbf{k})(0)\;,
\end{equation}
so that if $\lambda_\alpha>0$ and $t_\alpha(t)\to\infty$ for $t\to\infty$, the evolved state $\hat{\rho}_\alpha(t)$ is asymptotic to the Gaussian state $\hat{\rho}_G\left(t_\alpha(t)\,\gamma_\alpha\right)$, that satisfies \eqref{propcondition} with $c_\alpha(t)=\lambda_\alpha t_\alpha(t)$ and $\sigma=\gamma$ for any choice of $t_\alpha(t)$.

However, for initial Gaussian states that almost saturate the EPI, i.e.
\begin{equation}\label{rhotest}
\hat{\rho}_\alpha(0)=\hat{\rho}_G(\sigma_\alpha)
\end{equation}
with the $\sigma_\alpha$ having large symplectic eigenvalues and satisfying \eqref{propcondition}, the evolved $\sigma_\alpha(t)$ must still almost saturate the EPI and then satisfy \eqref{propcondition} also for finite $t$, i.e. the time-evolved version of the $A_\alpha$
\begin{equation}\label{propt}
A_\alpha(t):=M_\alpha\;\sigma_\alpha(t)\;M_\alpha^T=c_\alpha\;\sigma+\lambda_\alpha\;t_\alpha(t)\;\gamma
\end{equation}
must remain proportional (we have used \eqref{Lgauss} to get the time evolution).
For this purpose, we use the freedom in the choice of $t_\alpha(t)$, defining them as the solutions of
\begin{eqnarray}
    \frac{d}{dt}t_\alpha(t) &=& \mu_\alpha(t)\nonumber\\
    t_\alpha(0)&=&0\;,\label{tidot}
\end{eqnarray}
where we have defined
\begin{equation}\label{mualpha}
\mu_\alpha(t):=e^{S\left(\hat{\rho}_\alpha(t)\right)/n}=\exp\left(\frac{1}{n}S\left(e^{t_\alpha(t)\;\mathcal{L}\left(\gamma_\alpha\right)}\left(\hat{\rho}_\alpha\right)\right)\right)\;.
\end{equation}
This is a first-order differential equation for the functions $t_\alpha(t)$, and under reasonable assumptions on the regularity of the function
\begin{equation}
t_\alpha\mapsto S\left(e^{t_\alpha\;\mathcal{L}\left(\gamma_\alpha\right)}\left(\hat{\rho}_\alpha\right)\right)
\end{equation}
always admits a unique solution.
Let us check that the evolution defined by \eqref{tidot} has the required properties.
First, since quantum entropies are nonnegative we have
\begin{equation}
\frac{d}{dt}t_\alpha(t)\geq1\;,
\end{equation}
so that
\begin{equation}
\lim_{t\to\infty}t_\alpha(t)=\infty\;.
\end{equation}
The differential equation \eqref{tidot} allows us to define equivalently the $A_\alpha(t)$ as the solutions of
\begin{eqnarray}
\frac{d}{dt}A_\alpha(t)&=&\lambda_\alpha\;\mu_\alpha(t)\;\gamma\nonumber\\
A_\alpha(0)&=&c_\alpha\;\sigma\label{Aidot}\;,
\end{eqnarray}
where we have used \eqref{propt}.
Using \eqref{entas} to approximate the entropy of a Gaussian state with a large covariance matrix, the coefficients $\mu_\alpha(t)$ are given by
\begin{equation}
\mu_\alpha(t)\simeq\frac{e}{2}\left(\det\sigma_\alpha(t)\right)^\frac{1}{2n}\;.
\end{equation}
Let us put into \eqref{Aidot} the ansatz of proportional $A_\alpha(t)$:
\begin{equation}
A_\alpha(t):=c_\alpha\;\sigma(t)\;.\label{Ai}
\end{equation}
Then, the system of $K$ differential equations in \eqref{Aidot} reduces to only one equation for $\sigma(t)$:
\begin{eqnarray}
\frac{d}{dt}\sigma(t)&=&\frac{e}{2}\left(\det\sigma(t)\right)^\frac{1}{2n}\;\gamma\\
\sigma(0)&=&\sigma\;,
\end{eqnarray}
that always admits a solution.
Therefore as required, if the covariance matrices $\sigma_\alpha$ fulfill \eqref{propcondition} at $t=0$, they will fulfill it at any time.

\subsection{Relative entropy}
In order to prove the positivity of the time derivative of the right-hand side of \eqref{EPIrate} along the evolution described in Section \ref{evol}, we will link the time derivative of the entropy of a given quantum state to the relative entropy of this state with respect to a displaced version of it.

The relative entropy of a state $\hat{\rho}$ with respect to a state $\hat{\sigma}$ is defined as
\begin{equation}
S\left(\hat{\rho}\|\hat{\sigma}\right)=\mathrm{Tr}\left[\hat{\rho}\left(\ln\hat{\rho}-\ln\hat{\sigma}\right)\right]\;.
\end{equation}
The probability of confusing $n$ copies of $\hat{\sigma}$ with $n$ copies of $\hat{\rho}$ scales as $\exp\left(-n\,S\left(\hat{\rho}\|\hat{\sigma}\right)\right)$ in the large $n$ limit \cite{vedral2002role}, so the relative entropy provides a (not symmetric) measure of the distinguishability  of two states.

Since any physical operation on states cannot increase distinguishability, the relative entropy decreases under the application of any quantum channel $\Phi$:
\begin{equation}\label{dataprintro}
S\left(\Phi\left(\hat{\rho}\right)\|\Phi\left(\hat{\sigma}\right)\right)\leq S\left(\hat{\rho}\|\hat{\sigma}\right)\;.
\end{equation}
This is called the data-processing inequality.
Many proof of it are known, but none of them is simple.
They can be found in \cite{wilde2013quantum,nielsen2010quantum,holevo2013quantum}.

\subsection{Quantum Fisher information}\label{secQFI}
The proof of the positivity of the time-derivative of the rate in \eqref{EPIrate} requires the introduction of a quantity that has an importance by its own: the quantum Fisher information.

We define the quantum Fisher information matrix $J$ of a state $\hat{\rho}$ (see \cite{konig2014entropy,de2014generalization} for the single mode and \cite{de2015multimode} for the multimode case) as the Hessian with respect to $\mathbf{x}$ of the relative entropy \cite{holevo2013quantum}
\begin{equation}\label{relent}
S\left(\hat{\rho}\|\hat{\sigma}\right)=\mathrm{Tr}\left[\hat{\rho}\left(\ln\hat{\rho}-\ln\hat{\sigma}\right)\right]
\end{equation}
between the original state $\hat{\rho}$ and its version displaced by $\mathbf{x}$:
\begin{equation}\label{fisher}
J_{ij}(\hat{\rho}):=\left.\frac{\partial^2}{\partial x^i\partial x^j} S\left(\hat{\rho}\left\|\hat{D}(\mathbf{x})\;\hat{\rho}\;{\hat{D}(\mathbf{x})}^\dag\right.\right)\right|_{\mathbf{x}=\mathbf{0}}\;.
\end{equation}
The quantum Fisher information generalizes the classical Fisher information of \cite{stam1959some}, and measures how much the displaced state $\hat{D}(\mathbf{x})\hat{\rho}{\hat{D}(\mathbf{x})}^\dag$ is distinguishable from the original one.
For the comparison with the quantum Fisher information of the quantum Cram\'er-Rao bound \cite{helstrom1967minimum,paris2009quantum,cramer2016mathematical}, see Section \ref{appQFI} of Appendix \ref{appG}.

We can get a more explicit expression plugging into \eqref{fisher} the definition \eqref{relent} of the relative entropy:
\begin{equation}\label{Jlog}
J_{ij}(\hat{\rho})=-\mathrm{Tr}\left(\hat{\rho}\;\left.\frac{\partial^2}{\partial x^i\partial x^j}\hat{D}(\mathbf{x})\;\ln\hat{\rho}\;{\hat{D}(\mathbf{x})}^\dag\right|_{\mathbf{x}=\mathbf{0}}\right)\;.
\end{equation}

\subsection{De Bruijn identity}
The quantum Fisher information is intimately linked to the derivative of the entropy of a state under the evolution induced by the Liouvillian \eqref{liouville}.
Let us consider indeed an infinitesimal variation
\begin{equation}
d\hat{\rho}=\mathcal{L}_{d\gamma}(\hat{\rho})\;.
\end{equation}
Then, using \eqref{Ldag} and comparing with \eqref{Jlog} the variation of the entropy of $\hat{\rho}$ is
\begin{equation}\label{debr}
dS(\hat{\rho})=\frac{1}{4}d\gamma^{ij}\;J_{ij}(\hat{\rho})\;.
\end{equation}
From its classical analog, this equation takes the name of de Bruijn identity.

\subsection{Stam inequality}
The positivity the derivative of the rate \eqref{EPIrate} will follow from an inequality on the quantum Fisher information, called quantum Stam inequality from its classical analog \cite{stam1959some,kagan2008some}.

The core of its proof is the data-processing inequality for the relative entropy \cite{holevo2013quantum}, stating that it decreases under the action of any completely-positive trace-preserving map:
\begin{eqnarray}
S\left(\hat{\rho}\left\|\hat{D}(\mathbf{x})\;\hat{\rho}\;{\hat{D}(\mathbf{x})}^\dag\right.\right) &\geq& S\left(\Phi\left(\hat{\rho}\right)\left\|\Phi\left(\hat{D}(\mathbf{x})\;\hat{\rho}\;{\hat{D}(\mathbf{x})}^\dag\right)\right.\right)=\nonumber\\
&=& S\left(\Phi\left(\hat{\rho}\right)\left\|\hat{D}\left(M\mathbf{x}\right)\;\Phi\left(\hat{\rho}\right)\;{\hat{D}\left(M\mathbf{x}\right)}^\dag\right.\right)\;,\label{dataproc}
\end{eqnarray}
where we have used \eqref{PhiD}.
Since both members of \eqref{dataproc} are always nonnegative and vanish for $\mathbf{x}=\mathbf{0}$, this point is a minimum for both, and the inequality translates to the Hessians:
\begin{align}
\left.\frac{\partial^2}{\partial x^i\partial x^j}S\left(\hat{\rho}\left\|\hat{D}(\mathbf{x})\;\hat{\rho}\;{\hat{D}(\mathbf{x})}^\dag\right.\right)\right|_{\mathbf{x}=\mathbf{0}} &\geq&  \left.\frac{\partial^2}{\partial x^i\partial x^j}S\left(\Phi\left(\hat{\rho}\right)\left\|\hat{D}\left(M\mathbf{x}\right)\;\Phi\left(\hat{\rho}\right)\;{\hat{D}\left(M\mathbf{x}\right)}^\dag\right.\right)\right|_{\mathbf{x}=\mathbf{0}}=\nonumber\\
&=& \left.M^k_{\phantom{k}i}M^l_{\phantom{l}j}\;\frac{\partial^2}{\partial y^k\partial y^l}S\left(\Phi\left(\hat{\rho}\right)\left\|\hat{D}(\mathbf{y})\;\Phi\left(\hat{\rho}\right)\;{\hat{D}(\mathbf{y})}^\dag\right.\right)\right|_{\mathbf{y}=\mathbf{0}}\;,\label{hessin}
\end{align}
where the inequalities are meant for the whole matrices (and not for their entries), and we have made the change of variable
\begin{equation}
\mathbf{y}=M\mathbf{x}\;.
\end{equation}
Recalling the definition of Fisher information matrix \eqref{fisher}, inequality \eqref{hessin} becomes
\begin{equation}\label{stamlin}
J\left(\hat{\rho}\right)\geq M^T\;J\left(\Phi\left(\hat{\rho}\right)\right)\;M\;.
\end{equation}
Inequality \eqref{stamlin} is equivalent to
\begin{equation}\label{stam}
{J\left(\Phi\left(\hat{\rho}\right)\right)}^{-1}\geq M\;{J\left(\hat{\rho}\right)}^{-1}\;M^T\;.
\end{equation}
To see this, it is sufficient to choose bases in $X$ and $Y$ such that
\begin{eqnarray}
J\left(\hat{\rho}\right) &=& \mathbb{I}_X\\
J\left(\Phi\left(\hat{\rho}\right)\right) &=& \mathbb{I}_Y\;.
\end{eqnarray}
Then, \eqref{stamlin} and \eqref{stam} read
\begin{eqnarray}
\mathbb{I}_X&\geq& M^T\;M\\
\mathbb{I}_Y&\geq& M\;M^T\;,
\end{eqnarray}
that are equivalent since $M^T\,M$ and $M\,M^T$ have the same spectrum, except for the multiplicity of the eigenvalue zero.

Inequality \eqref{stam} is called the quantum Stam inequality.
In the particular case of the beamsplitter, it has already appeared in \cite{konig2014entropy,de2014generalization}, while in the multimode scenario it is an original result of this Thesis.

\subsection{Positivity of the time-derivative}\label{sec:para}
We have now all the instruments to prove that the time-derivative of the rate \eqref{EPIrate} is positive.
Recalling the definition \eqref{mualpha}, we can write the inequality to be proved as
\begin{equation}\label{murate}
\frac{d}{dt}\frac{\sum_{\alpha=1}^K\lambda_\alpha\;\mu_\alpha(t)}{\mu_Y(t)}\geq0\;.
\end{equation}
Let us now define the functions
\begin{eqnarray}
J_\alpha(t) &:=& J\left(\hat{\rho}_\alpha(t)\right)\\
J_Y(t) &:=& J\left(\Phi\left(\hat{\rho}(t)\right)\right)\;.
\end{eqnarray}
Combining the de Bruijn identity \eqref{debr} and the definition of the time evolution in \eqref{rhoit} and \eqref{tidot}, the time-derivative of the entropy of each input can be linked to its quantum Fisher information matrix:
\begin{equation}
\frac{d}{dt}S\left(\hat{\rho}_\alpha(t)\right) = \frac{\mu_\alpha(t)}{4}\;\gamma_\alpha^{ij}\;J^\alpha_{ij}(t)\;,
\end{equation}
and consequently
\begin{equation}
\frac{d}{dt}\mu_\alpha(t) = \frac{{\mu_\alpha(t)}^2}{4n}\;\gamma_\alpha^{ij}\;J^\alpha_{ij}(t)\;.
\end{equation}
With also \eqref{rhoyt} and \eqref{tY}, the analog for the output is
\begin{equation}
\frac{d}{dt}S\left(\Phi\left(\hat{\rho}(t)\right)\right)=\frac{\sum_{\alpha=1}^K\lambda_\alpha\;\mu_\alpha(t)}{4}\;\gamma^{ij}\;J^Y_{ij}(t)\;,
\end{equation}
and
\begin{equation}
\frac{d}{dt}\mu_Y(t)=\mu_Y(t)\;\frac{\sum_{\alpha=1}^K\lambda_\alpha\;\mu_\alpha(t)}{4n}\;\gamma^{ij}\;J^Y_{ij}(t)\;.
\end{equation}
Then \eqref{murate} becomes
\begin{equation}\label{ineqmu}
\left(\sum_{\alpha=1}^K\lambda_\alpha\;\mu_\alpha(t)\right)^2\gamma^{ij}\;J^Y_{ij}(t)\leq \sum_{\alpha=1}^K\lambda_\alpha\;{\mu_\alpha(t)}^2\;\gamma^{ij}_\alpha\;J^\alpha_{ij}(t)\;.
\end{equation}
To prove \eqref{ineqmu}, we use the quantum Stam inequality in the form \eqref{stamlin}, that for our $K$-partite input reads
\begin{equation}\label{stammatrix}
\left(
  \begin{array}{ccc}
    J_1(t) &  &  \\
     & \ddots &  \\
     &  & J_K(t) \\
  \end{array}
\right)\geq
\left(
  \begin{array}{c}
    M_1^T \\
    \vdots \\
    M_K^T \\
  \end{array}
\right)\;J_Y(t)\;
\left(
  \begin{array}{ccc}
    M_1 & \hdots & M_K \\
  \end{array}
\right)\;.
\end{equation}
Multiplying on the left by $\left(\mu_1(t)\lambda_1M_1^{-T}\;\ldots\;\mu_K(t)\lambda_KM_K^{-T}\right)$ and on the right by its transpose, we get
\begin{equation}
\sum_{\alpha=1}^K\lambda_\alpha^2\;{\mu_\alpha(t)}^2\;M_\alpha^{T}\;J_\alpha(t)\;M_\alpha^{-1}\geq\left(\sum_{\alpha=1}^K\lambda_\alpha\;\mu_\alpha(t)\right)^2J_Y(t)\;,
\end{equation}
and \eqref{ineqmu} follows upon taking the trace with $\gamma$ and recalling \eqref{Malpha}.

\subsection{Asymptotic scaling}\label{appscaling}
In this Section we show that the rate \eqref{EPIrate} tends to $1$ for $t\to\infty$, concluding then the proof of the EPI.

For this purpose, we first prove that for any strictly positive matrix $\gamma>0$ the entropy of $e^{t\;\mathcal{L}(\gamma)}(\hat{\rho})$ for $t\to\infty$ is asymptotically
\begin{equation}\label{scaling}
S\left(e^{t\;\mathcal{L}(\gamma)}(\hat{\rho})\right)=n\ln\frac{t}{2}+\frac{1}{2}\ln\det\gamma+n+\mathcal{O}\left(\frac{1}{t}\right)\;.
\end{equation}

\subsubsection{A lower bound for the entropy}
A lower bound for the entropy follows on expressing the state $\hat{\rho}$ in terms of its generalized  Husimi function (see Section \ref{sec:Husimi} of Appendix \ref{appG}).

We define
\begin{equation}
t_1=\frac{1}{\nu_{\min}}\;,
\end{equation}
where $\nu_{\min}$ is the minimum symplectic eigenvalue of $\gamma$. We have then
\begin{equation}
t_1\,\gamma\geq\pm i\,\Delta\;,
\end{equation}
and we can exploit the generalized Husimi representation \eqref{husimi} associated to the matrix  $t_1\,\gamma$:
\begin{equation}
\hat{\rho}=\int Q_{\hat{\rho}}(\mathbf{x})\;\hat{\rho}_G(-t_1\gamma,\,\mathbf{x})\;d^{2n}x\;.
\end{equation}
For the linearity of the evolution \eqref{liouville}, we can take the super-operator $e^{t\;\mathcal{L}(\gamma)}$ inside the integral, and remembering \eqref{Lgauss} we get
\begin{equation}
e^{t\;\mathcal{L}(\gamma)}(\hat{\rho})=\int Q_{\hat{\rho}}(\mathbf{x})\;\hat{\rho}_G((t-t_1)\gamma,\,\mathbf{x})\;d^{2n}x\;.
\end{equation}
For $t\geq2t_1$, we have
\begin{equation}
(t-t_1)\gamma\geq t_1\,\gamma\geq\pm i\Delta\;,
\end{equation}
i.e. $\hat{\rho}_G((t-t_1)\gamma)$ is a proper quantum state.
Since $Q_{\hat{\rho}}(\mathbf{x})$ is a probability distribution, the concavity of the von Neumann entropy implies
\begin{equation}\label{lowerboundS}
S\left(e^{t\;\mathcal{L}(\gamma)}(\hat{\rho})\right)\geq S\left(\hat{\rho}_G((t-t_1)\gamma)\right)= n\ln\frac{t}{2}+\frac{1}{2}\ln\det\gamma+n+\mathcal{O}\left(\frac{1}{t}\right)\;,
\end{equation}
where we have used Eq. \eqref{entas} of Appendix \ref{appG}.

\subsubsection{An upper bound for the entropy}
Given a state $\hat{\rho}$, let $\hat{\rho}_G$ be the centered Gaussian state with the same covariance matrix. It is then possible to prove \cite{wolf2006extremality} that $S\left(\hat{\rho}_G\right)\geq S\left(\hat{\rho}\right)$.
Let $\sigma$ be the covariance matrix of $\hat{\rho}$, respectively. Then, \eqref{addnoise} implies
\begin{equation}
\left(e^{t\;\mathcal{L}(\gamma)}\hat{\rho}\right)_G=\hat{\rho}_G(\sigma+t\gamma)\;,
\end{equation}
so that
\begin{equation}
S\left(e^{t\;\mathcal{L}(\gamma)}\hat{\rho}\right)\leq S\left(\hat{\rho}_G(\sigma+t\gamma)\right)\;.
\end{equation}
Let $t_2$ be the maximum eigenvalue of $\sigma\,\gamma^{-1}$ ($\gamma$ and $\sigma$ are strictly positive, so $\gamma^{-1}$ exists and $t_2$ is finite and strictly positive). Then,
\begin{equation}
\sigma\leq t_2\,\gamma
\end{equation}
(to see this, it is sufficient to choose a basis in which $\gamma=\mathbb{I}_{2n}$).
We remind that given two covariance matrices $\sigma'\leq\sigma''$, the Gaussian state $\hat{\rho}_{\sigma''}$ can be obtained applying an additive noise channel to $\hat{\rho}_{\sigma'}$. Since such channel is unital, it always increases the entropy, so  we have $S(\hat{\rho}_{\sigma'})\leq S(\hat{\rho}_{\sigma''})$. Applying this to $\sigma+t\gamma\leq(t_2+t)\gamma$, we get again
\begin{equation}
S\left(\hat{\rho}_G(\sigma+t\gamma)\right)\leq S\left(\hat{\rho}_G((t_2+t)\gamma)\right)=n\ln\frac{t}{2}+\frac{1}{2}\ln\det\gamma+n+\mathcal{O}\left(\frac{1}{t}\right)\;,\label{upper}
\end{equation}
where we have used \eqref{entas} again.

\subsubsection{Scaling of the rate}
From Section \ref{evol} we can see that for our evolutions if $M_\alpha$ is invertible $\det\gamma_\alpha=1$, so
\begin{equation}
\mu_\alpha(t)=\frac{e}{2}t_\alpha(t)+\mathcal{O}\left(1\right)\;,
\end{equation}
and similarly
\begin{equation}
\mu_Y(t)=\frac{e}{2}t_Y(t)+\mathcal{O}\left(1\right)\;.
\end{equation}
Replacing this into \eqref{EPIrate}, and remembering that if $M_\alpha$ is not invertible, then $\lambda_\alpha=0$ and the corresponding terms vanish,
from  \eqref{tY}  it easily follows that such quantity tends to $1$ in the $t\rightarrow \infty$ limit.

\section{The Entropy Photon-number Inequality}\label{secEPnI}
The quantum EPI \eqref{EPI} is not saturated by Gaussian states with proportional covariance matrices, and then it is not sufficient to determine the minimum entropy of $\mathbf{Y}$ for fixed entropies of $\mathbf{X}_1$ and $\mathbf{X}_2$.
However, as in the classical case, Gaussian states with proportional covariance matrices are conjectured to be the solution to this optimization problem.
This belief has led to conjecture the {\it Entropy Photon number Inequality} (EPnI) \cite{guha2008capacity,guha2008entropy}:
\begin{equation} \label{EPnI}
N(\hat{\rho}_Y)  \overset{?}{\ge}  \lambda\; N(\hat{\rho}_1) + (1- \lambda)\;N(\hat{\rho}_2)\;.
\end{equation}
Here
\begin{equation}
g(N)=(N+1)\ln(N+1)-N\ln N
\end{equation}
is the entropy of a single mode thermal Gaussian state with mean photon number $N$ (see Section \ref{entrGa} of Appendix \ref{appG}), and
\begin{equation}
N(\hat{\rho})= g^{-1}\left(S(\hat{\rho})/n\right)
\end{equation}
is the mean photon number per mode of an $n$-mode thermal Gaussian state with the same entropy of $\hat{\rho}$.
Indeed, the EPnI states exactly that fixing the input entropies $S_1$, $S_2$, the output entropy $S_Y$ is minimum when the inputs are Gaussian with proportional covariance matrices.
Since the qEPI \eqref{EPI} is \emph{not} saturated by Gaussian states with proportional covariance matrices (unless they have also the same entropy), it is weaker than (and it is actually implied by) the EPnI \eqref{EPnI}, so our proof of qEPI does not imply the EPnI, which still remains an open conjecture.

We show in Section \ref{EPnIG} of Appendix \ref{appG} that the EPnI \eqref{EPnI} holds for any couple of $n$-mode Gaussian states.
In Chapter \ref{chepni} we prove the EPnI \eqref{EPnI} in the one-mode case when the second input is chosen to be the vacuum.
Moreover, as we are going to show, the validity of the qEPI imposes a very tight bound (of the order of $0.132$) on the maximum allowed violation of the EPnI \eqref{EPnI}.

The map  $e^{S(\hat{\rho})/n} \mapsto N(\hat{\rho})$ from the entropy power to the entropy photon-number is the function $f(x)\equiv g^{-1}(\ln(x))$ defined on the interval $[1,\infty]$. Unfortunately it is convex and we cannot
obtain the EPnI \eqref{EPnI} from \eqref{EPI}. Fortunately however, $f(x)$ is {\it not too convex} and is well approximated by a linear function. It is easy to show indeed (see Section \ref{proofsec}) that
\begin{equation}
f(x) =-1/2 + x/e +\delta(x)\;,
\end{equation}
 where
\begin{equation}
0\leq\delta(x)\le \delta(1)=1/2-1/e \simeq 0.132\;.
\end{equation}
This directly implies that the entropy photon number inequality is valid up to such a small error,
\begin{equation}
N(\rho_Y)-\lambda\;N(\rho_1) - (1- \lambda)\;N(\rho_2) \ge 1/e -1/2 \;.\label{epnib}
\end{equation}

We also conjecture that the EPnI can be extended to the most general multimode scenario, i.e. that upon fixing the entropy of each input of the channel defined in \eqref{channelbs}, the output entropy is still minimized by Gaussian input states.
This multimode EPnI would be an improvement of our EPI \eqref{EPIm}, since it would imply it.
However, it cannot be written with elementary functions as an inequality on the output entropy, since the optimization over all the Gaussian input states cannot be performed analytically.

\subsection{Proof of the bound}\label{proofsec}
We want to evaluate how close is our qEPI \eqref{EPI} to the EPnI \eqref{EPnI} and prove \eqref{epnib}.
The qEPI \eqref{EPI} implies for the output entropy photon number
\begin{equation}
N_Y\geq g^{-1}\left(\ln\left(\lambda\;e^{g(N_1)}+(1-\lambda)\;e^{g(N_2)}\right)\right)\label{nc}\;.
\end{equation}
The EPnI \eqref{EPnI} is stronger than the EPI \eqref{EPI}, and in fact
\begin{equation}
g^{-1}\left(\ln\left(\lambda\; e^{g(N_1)}+(1-\lambda)\;e^{g(N_2)}\right)\right)\leq \lambda\;N_1+(1-\lambda)\;N_2\;,
\end{equation}
since the function $g^{-1}\left(\ln\left(x\right)\right)$ is increasing and convex.
Since $e^{g(N)}$ for $N\to\infty$ goes like
\begin{equation}
e^{g(N)}=e\left(N+\frac{1}{2}\right)+\mathcal{O}\left(\frac{1}{N}\right)\;,
\end{equation}
we have for $x\to\infty$
\begin{equation}
g^{-1}\left(\ln x\right)=\frac{x}{e}-\frac{1}{2}+\mathcal{O}\left(\frac{1}{x}\right)\;.
\end{equation}
If we define
\begin{equation}
\delta(x)\equiv g^{-1}\left(\ln x\right)-\frac{x}{e}+\frac{1}{2}\;,
\end{equation}
$\delta$ is convex, decreasing and
\begin{equation}
\lim_{x\to\infty}\delta(x)=0\;.
\end{equation}
We can also evaluate
\begin{equation}
\delta(1)=\frac{1}{2}-\frac{1}{e}\;,
\end{equation}
and for any $x_1,\;x_2\geq1$ we have
\begin{equation}
\delta(\lambda\;x_1+(1-\lambda)\;x_2)\geq\lambda\;\delta(x_1)+(1-\lambda)\;\delta(x_2)-\left(\frac{1}{2}-\frac{1}{e}\right)\;.
\end{equation}
Since
\begin{align}
&g^{-1}\left(\ln\left(\lambda\;x_1+(1-\lambda)\;x_2\right)\right)-\lambda\;g^{-1}\left(\ln x_1\right)-(1-\lambda)\;g^{-1}\left(\ln x_2\right)=\nonumber\\
&=\delta\left(\lambda\;x_1+(1-\lambda)\;x_2\right)-\lambda\;\delta(x_1)-(1-\lambda)\;\delta(x_2)\;,
\end{align}
in the case $x_1=e^{S_1}$, $x_2=e^{S_2}$ we get
\begin{align}
&g^{-1}\left(\ln\left(\lambda\;e^{S_1}+(1-\lambda)\;e^{S_2}\right)\right)-\lambda\;N_1-(1-\lambda)\;N_2=\nonumber\\
&=\delta\left(\lambda\;e^{S_1}+(1-\lambda)\;e^{S_2}\right)-\lambda\;\delta(e^{S_1})-(1-\lambda)\;\delta(e^{S_2})\;,
\end{align}
and we can conclude from \eqref{nc} that
\begin{eqnarray}
N_Y & \geq & \lambda\;N_1+(1-\lambda)\;N_2+\delta(\lambda\;e^{S_1}+(1-\lambda)\;e^{S_2})-\lambda\;\delta(e^{S_1})-(1-\lambda)\;\delta(e^{S_2})\geq\nonumber\\
&\geq& \lambda\;N_1+(1-\lambda)\;N_2-\left(\frac{1}{2}-\frac{1}{e}\right)\;,
\end{eqnarray}
so the \eqref{EPnI} violation can be at most
\begin{equation}
\frac{1}{2}-\frac{1}{e}\simeq 0.132\;.
\end{equation}

\section{The constrained minimum output entropy conjecture}\label{secGMOE}
Recently the so called {\it minimum output entropy conjecture} has been proved (see \cite{giovannetti2015solution,mari2014quantum,giovannetti2014ultimate} and Section \ref{gccapacity}).
It claims that the output entropy of a gauge-covariant Gaussian channel is minimum when the input is the vacuum.
A large class of physically relevant gauge-covariant Gaussian channels can be constructed with the beamsplitter defined in Eqs. \eqref{mixingbs} and \eqref{beamsplitter} taking as second input $\hat{\rho}_2$ a fixed Gaussian thermal state.
In this setup, the EPnI implies that the entropy of the output $S_Y$ is minimum when the first input $\hat{\rho}_1$ is the vacuum, i.e. the MOE conjecture.
A more general problem \cite{guha2008entropy,guha2008capacity} is to determine what is the minimum output entropy $S_Y$ with the constraint that the entropy of the first input $S_1$ is fixed to some value $\bar{S}>0$.
For simplicity, we concentrate on the one-mode case and we fix the second input to be the vacuum:
\begin{equation}
\hat{\rho}_2=|0\rangle\langle0|\;.
\end{equation}
It is easy to show that the EPnI \eqref{EPnI} implies that the minimum of $S_Y$ is achieved by the Gaussian thermal state with entropy $\bar{S}$, corresponding to an output entropy of
\begin{equation}\label{epni0}
S_Y=g\left(\lambda\;g^{-1}\left(\bar S\right)\right)\;.
\end{equation}
We will prove \eqref{epni0} in Chapter \ref{chepni}.
Here we use our qEPI to obtain a tight lower bound
on $S_Y$.  The bound follows directly from \eqref{EPI} for $S_2=0$ and can be expressed as
\begin{equation} \label{bound}
S_Y \ge \ln \left[\lambda\; e^{\bar S}+ (1-\lambda)\right]\;.
\end{equation}
The RHS of  \eqref{bound} is extremely close to the conjectured minimum $g\left(\lambda g^{-1}\left(\bar S\right)\right)$. Indeed the error between the two quantities
\begin{equation}
\Delta(\bar S,\lambda)=g\left(\lambda  g^{-1}\left(\bar S\right)\right)- \ln \left[\lambda e^{\bar S}+ (1-\lambda)\right]
\end{equation}
is bounded by $\sim 0.107$ and moreover it decays to zero in large part of the parameter space $(\bar S,\lambda)$ (see Fig. \ref{delta}).

\begin{figure}[ht]
\includegraphics[width=\textwidth]{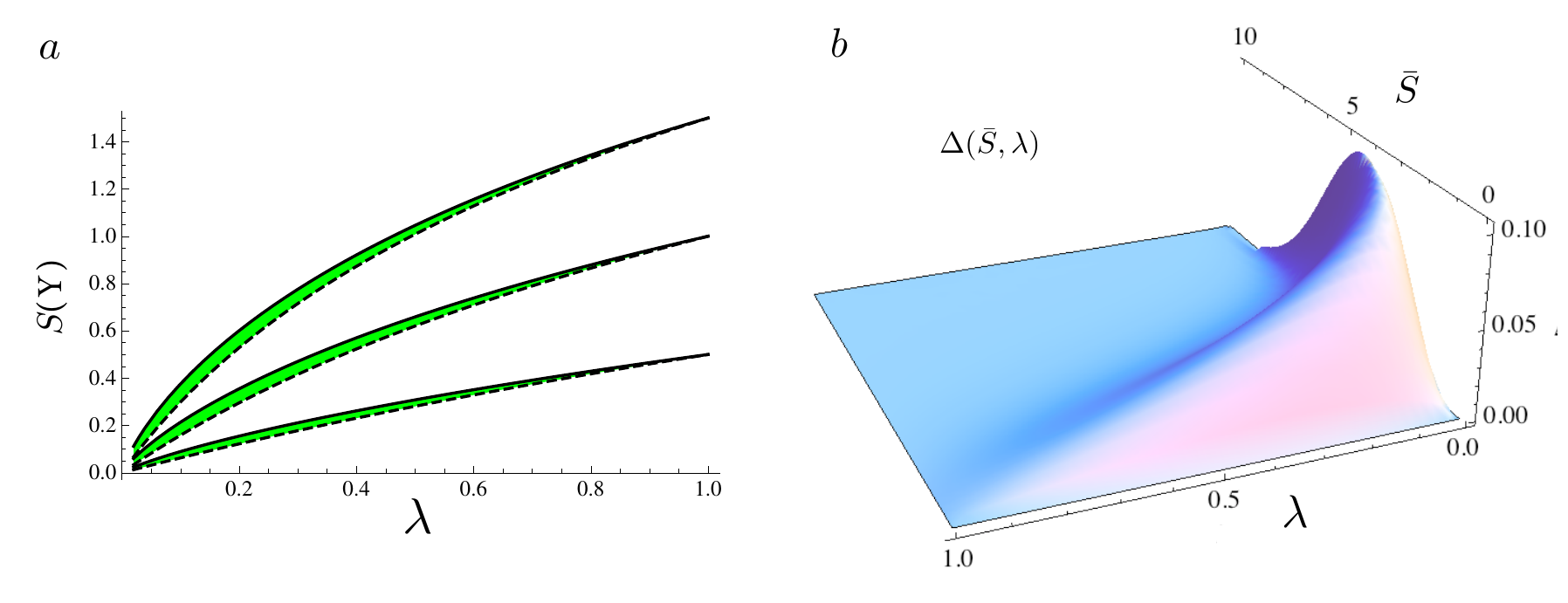}
\caption{{\bf a} Plot of the output entropies as functions of $\lambda$ and for different input entropies $\bar S=0.5,1,1.5$. In full lines are the entropy achievable with a Gaussian input state while the dotted lines represent the lower bound \eqref{bound}. The corresponding minimum output entropies are necessarily constrained within the green regions. We notice that larger values of input entropies $\bar S$ are not considered in this plot because the Gaussian ansatz and the bound becomes practically indistinguishable. {\bf b} Maximum allowed violation $\Delta(\bar S,\lambda)$ of the generalized minimum output entropy conjecture. The two axes are the input entropy  $\bar S$ and the beamsplitter transmissivity $\lambda$.} \label{delta}
\end{figure}

\section{Conclusion}\label{seccepi}

Understanding the complex physics of continuous variable quantum systems \cite{braunstein2005quantum} represents a fundamental challenge of modern science which is crucial
for developing an information technology  capable of taking full advantage of quantum effects \cite{caves1994quantum,weedbrook2012gaussian}.
This task appears now to be within our grasp due to a series of very recent works which have solved a collection  of  long standing conjectures. Specifically,
the minimum output entropy and output majorization conjectures (proposed in Ref. \cite{giovannetti2004minimum} and solved in Ref.'s \cite{giovannetti2015solution} and \cite{mari2014quantum} respectively),
the optimal Gaussian ensemble and the additivity conjecture (proposed in \cite{holevo2001evaluating} and solved in Ref. \cite{giovannetti2015solution}),
the optimality of Gaussian decomposition in the calculation of entanglement of formation \cite{giedke2003entanglement} and of Gaussian discord \cite{pirandola2014optimality,modi2012classical}
for two-mode gaussian states (both solved in Ref. \cite{giovannetti2014ultimate}), the proof of the strong converse of the classical capacity theorem \cite{bardhan2015strong}.

This result represents a fundamental further step in this direction by extending the proof of \cite{konig2014entropy} for the qEPI conjecture to the most general multimode scenario.

\chapter{Optimal inputs: passive states}\label{majorization}
The passive states of a quantum system minimize the average energy for fixed spectrum.
In this Chapter we prove that these states are the optimal inputs of one-mode gauge-covariant Gaussian quantum channels, in the sense that the output generated by a passive state majorizes the output generated by any other state with the same spectrum.
This result reduces the constrained quantum minimum output entropy conjecture (Proposition \ref{CMOE}) to a problem on discrete classical probability distributions.

The Chapter is based on
\begin{enumerate}
\item[\cite{de2015passive}] G.~De~Palma, D.~Trevisan, and V.~Giovannetti, ``Passive States Optimize the
  Output of Bosonic Gaussian Quantum Channels,'' \emph{IEEE Transactions on
  Information Theory}, vol.~62, no.~5, pp. 2895--2906, May 2016.\\ {\small\url{http://ieeexplore.ieee.org/document/7442587}}
\end{enumerate}
\section{Introduction}
The minimum von Neumann entropy at the output of a quantum communication channel can be crucial for the determination of its classical communication capacity (see \cite{holevo2013quantum} and Section \ref{gccapacity}).

Most communication schemes encode the information into pulses of electromagnetic radiation, that travels through metal wires, optical fibers or free space and is unavoidably affected by attenuation and noise.
The gauge-covariant quantum Gaussian channels \cite{holevo2013quantum} presented in Chapter \ref{GQI} provide a faithful model for these effects, and are characterized by the property of preserving the thermal states of electromagnetic radiation.

It has been recently proved (see \cite{mari2014quantum,giovannetti2015solution,giovannetti2015majorization,holevo2015gaussian} and Section \ref{gccapacity}) that the output entropy of any gauge-covariant Gaussian quantum channel is minimized when the input state is the vacuum.
This result has permitted the determination of the classical information capacity of this class of channels \cite{giovannetti2014ultimate}.

However, it is not sufficient to determine the triple trade-off region of the same class of channels \cite{wilde2012quantum,wilde2012information}, nor the capacity region of the Gaussian quantum broadcast channel.
Indeed, the solutions of these problems both rely on the still unproven constrained minimum output entropy conjecture \ref{CMOE}, stating that Gaussian thermal input states minimize the output von Neumann entropy of a quantum-limited attenuator among all the states with a given entropy (see \cite{guha2007classicalproc,guha2007classical} and Sections \ref{broadcasti}, \ref{broadcastg} and \ref{broadcast}).
This still unproven result would follow from a stronger conjecture, the Entropy Photon-number Inequality (EPnI) (see \cite{guha2008entropy} and Section \ref{secEPnI}), stating that Gaussian states with proportional covariance matrices minimize the output von Neumann entropy of a beamsplitter among all the couples of input states, each one with a given entropy.

Actually, Ref.'s \cite{mari2014quantum,giovannetti2015majorization,holevo2015gaussian} do not only prove that the vacuum minimizes the output entropy of any gauge-covariant quantum Gaussian channel.
They also prove that the output generated by the vacuum majorizes the output generated by any other state, i.e. applying a convex combination of unitary operators to the former, we can obtain any of the latter states (see Section \ref{secmaj}).
In this Chapter we go in the same direction, and prove a generalization of this result valid for any one-mode gauge-covariant quantum Gaussian channel.
Our result states that the output generated by any quantum state is majorized by the output generated by the state with the same spectrum diagonal in the Fock basis and with decreasing eigenvalues, i.e. by the state which is {\it passive} \cite{pusz1978passive,lenard1978thermodynamical,gorecki1980passive} with respect to the number operator (see \cite{vinjanampathy2015quantum,goold2015role,binder2015quantum} for the use of passive states in the context of quantum thermodynamics).
This can be understood as follows: among all the states with a given spectrum, the one diagonal in the Fock basis with decreasing eigenvalues produces the less noisy output.
All the states with a given spectrum have the same von Neumann entropy.
Then, our result implies that the input state minimizing the output entropy for fixed input entropy is certainly diagonal in the Fock basis.
This reduces the minimum output entropy quantum problem to a problem on discrete classical probability distributions.

We will solve this reduced problem in Chapter \ref{chepni}, where we prove that Gaussian thermal input states minimize the output entropy of the one-mode quantum attenuator for fixed input entropy.

Thanks to the classification of one-mode Gaussian channels in terms of unitary equivalence \cite{holevo2007one,holevo2013quantum}, we extend the result of this Chapter to the channels that are not gauge-covariant with the exception of the singular cases $A_2)$ and $B_1)$, for which we show that an optimal basis does not exist.

We also point out that the classical channel acting on discrete probability distributions associated to the restriction of the quantum-limited attenuator to states diagonal in the Fock basis coincides with the channel already known in the probability literature under the name of thinning.
First introduced by R{\'e}nyi \cite{renyi1956characterization} as a discrete analog of the rescaling of a continuous random variable, the thinning has been recently involved in discrete versions of the central limit theorem \cite{harremoes2007thinning,yu2009monotonic,harremoes2010thinning}
and of the Entropy Power Inequality \cite{yu2009concavity,johnson2010monotonicity}.
In particular, the Restricted Thinned Entropy Power Inequality \cite{johnson2010monotonicity} states that the Poisson probability distribution minimizes the output Shannon entropy of the thinning among all the ultra log-concave input probability distributions with a given Shannon entropy.

The Chapter is organized as follows.
In Section \ref{defs} we introduce the Gaussian quantum channels.
The Fock rearrangement is defined in Section \ref{rearrangement}, while Section \ref{secoptimal} defines the notion of Fock optimality and proves some of its properties.
The main theorem is proved in Section \ref{mainproof}, and the case of a generic not gauge-covariant Gaussian channel is treated in Section \ref{generic}.
Section \ref{secthinning} links our result to the thinning operation, and we conclude in Section \ref{secconclmaj}.

\section{Preliminaries}\label{defs}
In this Section we recall some properties of Gaussian quantum channels.
For more details, see Sections \ref{GQCh} and \ref{QGCa}, and the books \cite{holevo2013quantum,barnett2002methods,holevo2011probabilistic}.

We consider a one-mode Gaussian quantum system (see Section \ref{GQSy}), i.e. the Hilbert space $\mathcal{H}$ of one harmonic oscillator.
$\mathcal{H}$ has a countable orthonormal basis
\begin{equation}\label{fockdef}
\{|n\rangle\}_{n\in\mathbb{N}}\;,\qquad \langle m|n\rangle=\delta_{mn}
\end{equation}
called the Fock basis, on which the ladder operator $\hat{a}$ acts as
\begin{equation}\label{acta}
\hat{a}\;|n\rangle = \sqrt{n}\;|n-1\rangle\;,\qquad \hat{a}^\dag\;|n\rangle = \sqrt{n+1}\;|n+1\rangle\;.
\end{equation}
For one mode, the Hamiltonian \eqref{Hosc} reduces to
\begin{equation}\label{Nosc}
\hat{N}=\hat{a}^\dag\hat{a}\;,
\end{equation}
satisfying
\begin{equation}
\hat{N}\;|n\rangle=n\;|n\rangle\;.
\end{equation}

\begin{lem}
The quantum-limited attenuator of parameter $0\leq\lambda\leq1$ (see Section \ref{secattampl}) admits the explicit representation
\begin{equation}\label{kraus}
\Phi_\lambda\left(\hat{X}\right)=\sum_{l=0}^\infty\frac{(1-\lambda)^l}{l!}\;\lambda^\frac{\hat{N}}{2}\;\hat{a}^l\;\hat{X}\;\left(\hat{a}^\dag\right)^l\;\lambda^\frac{\hat{N}}{2}
\end{equation}
for any trace-class operator $\hat{X}$.
Then, if $\hat{X}$ is diagonal in the Fock basis, $\Phi_\lambda\left(\hat{X}\right)$ is diagonal in the same basis for any $0\leq\lambda\leq1$ also.
\begin{proof}
The channel $\Phi_\lambda$ admits the Kraus decomposition (see Eq. (4.5) of \cite{ivan2011operator})
\begin{equation}
\Phi_\lambda\left(\hat{X}\right)=\sum_{l=0}^\infty\hat{B}_l\;\hat{X}\;\hat{B}_l^\dag\;,
\end{equation}
where
\begin{equation}
\hat{B}_l=\sum_{m=0}^\infty\sqrt{\binom{m+l}{l}}\;(1-\lambda)^\frac{l}{2}\;\lambda^\frac{m}{2}\;|m\rangle\langle m+l|\;,\qquad l\in\mathbb{N}\;.
\end{equation}
Using \eqref{acta}, we have
\begin{equation}
\hat{a}^l=\sum_{m=0}^\infty\sqrt{l!\;\binom{m+l}{l}}\;|m\rangle\langle m+l|\;,
\end{equation}
and the claim easily follows.
\end{proof}
\end{lem}
\begin{lem}\label{lemL}
The quantum-limited attenuator of parameter $\lambda=e^{-t}$ with $t\geq0$ can be written as the exponential of a Lindbladian $\mathcal{L}$, i.e. $\Phi_\lambda=e^{t\mathcal{L}}$, where
\begin{equation}\label{lindblad}
\mathcal{L}\left(\hat{X}\right)=\hat{a}\;\hat{X}\;\hat{a}^\dag-\frac{1}{2}\hat{a}^\dag\hat{a}\;\hat{X}-\frac{1}{2}\hat{X}\;\hat{a}^\dag\hat{a}
\end{equation}
for any trace-class operator $\hat{X}$.
\begin{proof}
Putting $\lambda=e^{-t}$ into \eqref{kraus} and differentiating with respect to $t$ we have for any trace-class operator $\hat{X}$
\begin{equation}
\frac{d}{dt}\Phi_{\lambda}\left(\hat{X}\right)=\mathcal{L}\left(\Phi_{\lambda}\left(\hat{X}\right)\right)\;,
\end{equation}
where $\mathcal{L}$ is the Lindbladian given by \eqref{lindblad}.
\end{proof}
\end{lem}

\begin{lem}
Let
\begin{equation}\label{Xdiag}
\hat{X}=\sum_{k=0}^\infty x_k\;|\psi_k\rangle\langle\psi_k|\;,\quad\langle\psi_k|\psi_l\rangle=\delta_{kl}\;,\quad x_0\geq x_1\geq\ldots
\end{equation}
be a self-adjoint Hilbert-Schmidt operator.
Then, the projectors
\begin{equation}\label{Pin}
\hat{\Pi}_n=\sum_{k=0}^n |\psi_k\rangle\langle\psi_k|
\end{equation}
satisfy
\begin{equation}
\mathrm{Tr}\left[\hat{\Pi}_n\;\hat{X}\right]=\sum_{k=0}^n x_k\;.
\end{equation}
\begin{proof}
Easily follows from an explicit computation.
\end{proof}
\end{lem}
\begin{lem}[Ky Fan's Maximum Principle]\label{sumeig}
Let $\hat{X}$ be a positive Hilbert-Schmidt operator with eigenvalues $\{x_k\}_{k\in\mathbb{N}}$ in decreasing order, i.e. $x_0\geq x_1\geq\ldots\;$,
and let $\hat{P}$ be a projector of rank $n+1$.
Then
\begin{equation}\label{TrPiX}
\mathrm{Tr}\left[\hat{P}\;\hat{X}\right]\leq\sum_{k=0}^n x_k\;.
\end{equation}
\begin{proof}
(See also \cite{bhatia2013matrix,fan1951maximum}).
Let us diagonalize $\hat{X}$ as in \eqref{Xdiag}.
The proof proceeds by induction on $n$.
Let $\hat{P}$ have rank one.
Since
\begin{equation}
\hat{X}\leq x_0\;\hat{\mathbb{I}}\;,
\end{equation}
we have
\begin{equation}
\mathrm{Tr}\left[\hat{P}\;\hat{X}\right]\leq x_0\;.
\end{equation}
Suppose now that \eqref{TrPiX} holds for any rank-$n$ projector.
Let $\hat{P}$ be a projector of rank $n+1$.
Its support then certainly contains a vector $|\psi\rangle$ orthogonal to the support of $\hat{\Pi}_{n-1}$, that has rank $n$.
We can choose $|\psi\rangle$ normalized (i.e. $\langle\psi|\psi\rangle=1$), and define the rank-$n$ projector
\begin{equation}
\hat{Q}=\hat{P}-|\psi\rangle\langle\psi|\;.
\end{equation}
By the induction hypothesis on $\hat{Q}$,
\begin{equation}\label{ineqQpsi}
\mathrm{Tr}\left[\hat{P}\,\hat{X}\right]=\mathrm{Tr}\left[\hat{Q}\,\hat{X}\right]+\langle\psi|\hat{X}|\psi\rangle\leq\sum_{k=0}^{n-1}x_k+\langle\psi|\hat{X}|\psi\rangle\;.
\end{equation}
Since $|\psi\rangle$ is in the support of $\hat{\mathbb{I}}-\hat{\Pi}_{n-1}$, and
\begin{equation}
\left(\hat{\mathbb{I}}-\hat{\Pi}_{n-1}\right)\hat{X}\left(\hat{\mathbb{I}}-\hat{\Pi}_{n-1}\right)\leq x_n\;\hat{\mathbb{I}}\;,
\end{equation}
we have
\begin{equation}\label{ineqpsi}
\langle\psi|\hat{X}|\psi\rangle\leq x_n\;,
\end{equation}
and this concludes the proof.
\end{proof}
\end{lem}

\begin{lem}\label{HS*}
Let $\hat{X}$ and $\hat{Y}$ be positive Hilbert-Schmidt operators (see \eqref{HSN} in Appendix \ref{appG}) with eigenvalues in decreasing order $\{x_n\}_{n\in\mathbb{N}}$ and $\{y_n\}_{n\in\mathbb{N}}$, respectively.
Then,
\begin{equation}
\sum_{n=0}^\infty (x_n-y_n)^2\leq\left\|\hat{X}-\hat{Y}\right\|_2^2\;.
\end{equation}
\begin{proof}
We have
\begin{equation}\label{TrXYr}
\left\|\hat{X}-\hat{Y}\right\|_2^2-\sum_{n=0}^\infty (x_n-y_n)^2=2\sum_{n=0}^\infty x_ny_n-2\mathrm{Tr}\left[\hat{X}\hat{Y}\right]\geq0\,.
\end{equation}
To prove the inequality in \eqref{TrXYr}, let us diagonalize $\hat{X}$ as in \eqref{Xdiag}.
We then also have
\begin{equation}
\hat{X}=\sum_{n=0}^\infty\left(x_n-x_{n+1}\right)\hat{\Pi}_n\;,
\end{equation}
where
\begin{equation}
\hat{\Pi}_n=\sum_{k=0}^n|\psi_k\rangle\langle\psi_k|\;.
\end{equation}
We then have
\begin{equation}
\mathrm{Tr}\left[\hat{X}\;\hat{Y}\right] = \sum_{n=0}^\infty\left(x_n-x_{n+1}\right)\mathrm{Tr}\left[\hat{\Pi}_n\;\hat{Y}\right]\leq \sum_{n=0}^\infty\left(x_n-x_{n+1}\right)\sum_{k=0}^n y_k= \sum_{n=0}^\infty x_n\;y_n\;,
\end{equation}
where we have used Ky Fan's Maximum Principle (Lemma \ref{sumeig}) and rearranged the sum (see also the Supplemental Material of \cite{koenig2009strong}).
\end{proof}
\end{lem}

\section{Fock rearrangement}\label{rearrangement}
In order to state our main theorem, we need the following:
\begin{defn}[Fock rearrangement]\label{defrearr}
Let $\hat{X}$ be a positive trace-class operator with eigenvalues $\{x_n\}_{n\in\mathbb{N}}$ in decreasing order.
We define its Fock rearrangement as
\begin{equation}
\hat{X}^\downarrow:=\sum_{n=0}^\infty x_n\;|n\rangle\langle n|\;.
\end{equation}
If $\hat{X}$ coincides with its own Fock rearrangement, i.e. $\hat{X}=\hat{X}^\downarrow$, we say that it is {\it passive} \cite{pusz1978passive,lenard1978thermodynamical,gorecki1980passive} with respect to the Hamiltonian $\hat{N}$.
For simplicity, in the following we will always assume $\hat{N}$ to be the reference Hamiltonian, and an operator with $\hat{X}=\hat{X}^\downarrow$ will be called simply passive.
\end{defn}
\begin{rem}
The Fock rearrangement of any projector $\hat{\Pi}_n$ of rank $n+1$ is the projector onto the first $n+1$ Fock states:
\begin{equation}\label{Pin*}
\hat{\Pi}_n^\downarrow=\sum_{i=0}^n|i\rangle\langle i|\;.
\end{equation}
\end{rem}

We recall that a quantum operation has the same definition of a quantum channel, but it is not required to be trace-preserving (see Section \ref{GQCh}).
We define the notion of passive-preserving quantum operation, that will be useful in the following.
\begin{defn}[Passive-preserving quantum operation]\label{*pres}
We say that a quantum operation $\Phi$ is passive-preserving if $\Phi\left(\hat{X}\right)$ is passive for any passive positive trace-class operator $\hat{X}$.
\end{defn}

We will also need these lemmata:
\begin{lem}\label{PXPlem}
For any self-adjoint trace-class operator $\hat{X}$,
\begin{equation}
\lim_{N\to\infty}\left\|\hat{\Pi}_N^\downarrow\;\hat{X}\;\hat{\Pi}_N^\downarrow-\hat{X}\right\|_2=0\;,
\end{equation}
where the $\hat{\Pi}_N^\downarrow$ are the projectors onto the first $N+1$ Fock states defined in \eqref{Pin*}.
\begin{proof}
We have
\begin{eqnarray}\label{PXP}
\left\|\hat{\Pi}_N^\downarrow\hat{X}\hat{\Pi}_N^\downarrow-\hat{X}\right\|_2^2 &=& \mathrm{Tr}\left[\hat{X}\left(\hat{\mathbb{I}}+\hat{\Pi}_N^\downarrow\right)\hat{X}\left(\hat{\mathbb{I}}-\hat{\Pi}_N^\downarrow\right)\right]\leq 2\;\mathrm{Tr}\left[\hat{X}^2\left(\hat{\mathbb{I}}-\hat{\Pi}_N^\downarrow\right)\right]=\nonumber\\
&=& 2\sum_{n=N+1}^\infty\langle n|\hat{X}^2|n\rangle\;,
\end{eqnarray}
where we have used that $\hat{\mathbb{I}}+\hat{\Pi}_N^\downarrow\leq 2\;\hat{\mathbb{I}}\,$.
Since $\hat{X}$ has finite trace-norm, also its Hilbert-Schmidt norm is finite, the sum in \eqref{PXP} converges, and its tail tends to zero for $N\to\infty$.
\end{proof}
\end{lem}
\begin{lem}\label{*symtr}
A positive trace-class operator $\hat{X}$ is passive iff for any finite-rank projector $\hat{P}$
\begin{equation}\label{PP*X}
\mathrm{Tr}\left[\hat{P}\;\hat{X}\right]\leq\mathrm{Tr}\left[\hat{P}^\downarrow\;\hat{X}\right]\;.
\end{equation}
\begin{proof}
First, let us suppose that $\hat{X}$ is passive with eigenvalues $\{x_n\}_{n\in\mathbb{N}}$ in decreasing order, and let $\hat{P}$ have rank $n+1$.
Then, by Lemma \ref{sumeig}
\begin{equation}
\mathrm{Tr}\left[\hat{P}\;\hat{X}\right]\leq\sum_{i=0}^n x_i=\mathrm{Tr}\left[\hat{P}^\downarrow\;\hat{X}\right]\;.
\end{equation}

Let us now suppose that \eqref{PP*X} holds for any finite-rank projector.
Let us diagonalize $\hat{X}$ as in \eqref{Xdiag}.
Putting into \eqref{PP*X} the projectors $\hat{\Pi}_n$ defined in \eqref{Pin},
\begin{equation}
\sum_{i=0}^n x_i=\mathrm{Tr}\left[\hat{\Pi}_n\;\hat{X}\right]\leq\mathrm{Tr}\left[\hat{\Pi}_n^\downarrow\;\hat{X}\right]\leq\sum_{i=0}^n x_i\;,
\end{equation}
where we have again used Lemma \ref{sumeig}.
It follows that for any $n\in\mathbb{N}$
\begin{equation}
\mathrm{Tr}\left[\hat{\Pi}_n^\downarrow\;\hat{X}\right]=\sum_{i=0}^n x_i\;,
\end{equation}
and
\begin{equation}
\langle n|\hat{X}|n\rangle=x_n\;.
\end{equation}
It is then easy to prove by induction on $n$ that
\begin{equation}
\hat{X}=\sum_{n=0}^\infty x_n\;|n\rangle\langle n|\;,
\end{equation}
i.e. $\hat{X}$ is passive.
\end{proof}
\end{lem}
\begin{lem}\label{sumstar}
Let $\left\{\hat{X}_n\right\}_{n\in\mathbb{N}}$ be a sequence of positive trace-class operators with $\hat{X}_n$ passive for any $n\in\mathbb{N}$.
Then also $\sum_{n=0}^\infty\hat{X}_n$ is passive, provided that its trace is finite.
\begin{proof}
Follows easily from the definition of Fock rearrangement.
\end{proof}
\end{lem}
\begin{lem}\label{phistar}
Let $\Phi$ be a quantum operation.
Let us suppose that $\Phi\left(\hat{\Pi}\right)$ is passive for any passive finite-rank projector $\hat{\Pi}$.
Then, $\Phi$ is passive-preserving.
\begin{proof}
Choose a passive operator
\begin{equation}
\hat{X}=\sum_{n=0}^\infty x_n\,|n\rangle\langle n|\;,
\end{equation}
with $\{x_n\}_{n\in\mathbb{N}}$ positive and decreasing.
We then also have
\begin{equation}
\hat{X}=\sum_{n=0}^\infty z_n\;\hat{\Pi}_n^\downarrow\;,
\end{equation}
where the $\hat{\Pi}_n^\downarrow$ are defined in \eqref{Pin*}, and
\begin{equation}
z_n=x_n-x_{n+1}\geq0\;.
\end{equation}
Since by hypothesis $\Phi\left(\hat{\Pi}_n^\downarrow\right)$ is passive for any $n\in\mathbb{N}$, according to Lemma \ref{sumstar} also
\begin{equation}
\Phi\left(\hat{X}\right)=\sum_{n=0}^\infty z_n\;\Phi\left(\hat{\Pi}_n^\downarrow\right)
\end{equation}
is passive.
\end{proof}
\end{lem}

\begin{lem}\label{majprojlem}
Let $\hat{X}$ and $\hat{Y}$ be positive trace-class operators.
\begin{enumerate}
\item Let us suppose that for any finite-rank projector $\hat{\Pi}$
\begin{equation}\label{majproj}
\mathrm{Tr}\left[\hat{\Pi}\,\hat{X}\right]\leq\mathrm{Tr}\left[\hat{\Pi}^\downarrow\,\hat{Y}\right]\;.
\end{equation}
Then $\hat{X}\prec_w\hat{Y}$ (see Section \ref{secmaj} for the definition of weak submajorization).
\item Let $\hat{Y}$ be passive, and let us suppose that $\hat{X}\prec_w\hat{Y}$.
Then \eqref{majproj} holds for any finite-rank projector $\hat{\Pi}$.
\end{enumerate}
\begin{proof}
Let $\{x_n\}_{n\in\mathbb{N}}$ and $\{y_n\}_{n\in\mathbb{N}}$ be the eigenvalues in decreasing order of $\hat{X}$ and $\hat{Y}$, respectively, and let us diagonalize $\hat{X}$ as in \eqref{Xdiag}.
\begin{enumerate}
\item Let us first suppose that \eqref{majproj} holds for any finite-rank projector $\hat{\Pi}$.
For any $n\in\mathbb{N}$ we have
\begin{equation}\label{eqtr}
\sum_{i=0}^n x_i=\mathrm{Tr}\left[\hat{\Pi}_n\,\hat{X}\right]\leq\mathrm{Tr}\left[\hat{\Pi}_n^\downarrow\,\hat{Y}\right]\leq\sum_{i=0}^n y_i\;,
\end{equation}
where the $\hat{\Pi}_n$ are defined in \eqref{Pin} and we have used Lemma \ref{sumeig}.
Then $x\prec_w y$, and $\hat{X}\prec_w\hat{Y}$.

\item Let us now suppose that $\hat{X}\prec_w\hat{Y}$ and $\hat{Y}=\hat{Y}^\downarrow$.
Then, for any $n\in\mathbb{N}$ and any projector $\hat{\Pi}$ of rank $n+1$,
\begin{equation}
\mathrm{Tr}\left[\hat{\Pi}\,\hat{X}\right]\leq \sum_{i=0}^n x_i\leq\sum_{i=0}^n y_i=\mathrm{Tr}\left[\hat{\Pi}^\downarrow\,\hat{Y}\right]\;,
\end{equation}
where we have used Lemma \ref{sumeig} again.
\end{enumerate}
\end{proof}
\end{lem}
\begin{lem}\label{corTr}
Let $\hat{Y}$ and $\hat{Z}$ be positive trace-class operators with $\hat{Y}\prec_w\hat{Z}=\hat{Z}^\downarrow$.
Then, for any positive trace-class operator $\hat{X}$,
\begin{equation}
\mathrm{Tr}\left[\hat{X}\;\hat{Y}\right]\leq\mathrm{Tr}\left[\hat{X}^\downarrow\;\hat{Z}\right]\;.
\end{equation}
\begin{proof}
Let us diagonalize $\hat{X}$ as in \eqref{Xdiag}.
Then, it can be rewritten as
\begin{equation}\label{XPi}
\hat{X}=\sum_{n=0}^\infty d_n\,\hat{\Pi}_n\;,
\end{equation}
where the projectors $\hat{\Pi}_n$ are as in \eqref{Pin} and
\begin{equation}
d_n=x_n-x_{n+1}\geq0\;.
\end{equation}
The Fock rearrangement of $\hat{X}$ is
\begin{equation}\label{X*Pi}
\hat{X}^\downarrow=\sum_{n=0}^\infty d_n\,\hat{\Pi}_n^\downarrow\;.
\end{equation}
We then have from Lemma \ref{majprojlem}
\begin{equation}
\mathrm{Tr}\left[\hat{X}\;\hat{Y}\right]= \sum_{n=0}^\infty d_n\;\mathrm{Tr}\left[\hat{\Pi}_n\;\hat{Y}\right]\leq\sum_{n=0}^\infty d_n\;\mathrm{Tr}\left[\hat{\Pi}_n^\downarrow\;\hat{Z}\right] =\mathrm{Tr}\left[\hat{X}^\downarrow\;\hat{Z}\right]\;.
\end{equation}
\end{proof}
\end{lem}

\begin{lem}\label{majsum}
Let $\left\{\hat{X}_n\right\}_{n\in\mathbb{N}}$ and $\left\{\hat{Y}_n\right\}_{n\in\mathbb{N}}$ be two sequences of positive trace-class operators, with $\hat{Y}_n=\hat{Y}_n^\downarrow$ and $\hat{X}_n\prec_w\hat{Y}_n$ for any $n\in\mathbb{N}$.
Then
\begin{equation}
\sum_{n=0}^\infty\hat{X}_n\prec_w\sum_{n=0}^\infty\hat{Y}_n\;,
\end{equation}
provided that both sides have finite traces.
\begin{proof}
Let $\hat{P}$ be a finite-rank projector.
Since $\hat{X}_n\prec_w\hat{Y}_n$ and $Y_n=Y_n^\downarrow$, by the second part of Lemma \ref{majprojlem}
\begin{equation}
\mathrm{Tr}\left[\hat{P}\;\hat{X}_n\right]\leq\mathrm{Tr}\left[\hat{P}^\downarrow\;\hat{Y}_n\right]\qquad\forall\;n\in\mathbb{N}\;.
\end{equation}
Then,
\begin{equation}
\mathrm{Tr}\left[\hat{P}\;\sum_{n=0}^\infty\hat{X}_n\right]\leq\mathrm{Tr}\left[\hat{P}^\downarrow\;\sum_{n=0}^\infty\hat{Y}_n\right]\;,
\end{equation}
and the submajorization follows from the first part of Lemma \ref{majprojlem}.
\end{proof}
\end{lem}
\begin{lem}\label{HS**}
The Fock rearrangement is continuous in the Hilbert-Schmidt norm.
\begin{proof}
Let $\hat{X}$ and $\hat{Y}$ be trace-class operators, with eigenvalues in decreasing order $\{x_n\}_{n\in\mathbb{N}}$ and $\{y_n\}_{n\in\mathbb{N}}$, respectively.
We then have
\begin{equation}
\left\|\hat{X}^\downarrow-\hat{Y}^\downarrow\right\|_2^2=\sum_{n=0}^\infty(x_n-y_n)^2\leq\left\|\hat{X}-\hat{Y}\right\|_2^2\;,
\end{equation}
where we have used Lemma \ref{HS*}.
\end{proof}
\end{lem}

\section{Fock-optimal quantum operations}\label{secoptimal}
We will prove that any  gauge-covariant Gaussian quantum channel satisfies this property:
\begin{defn}[Fock-optimal quantum operation]
We say that a quantum operation $\Phi$ is Fock-optimal if for any positive trace-class operator $\hat{X}$
\begin{equation}\label{conjectureeq}
\Phi\left(\hat{X}\right)\prec_w\Phi\left(\hat{X}^\downarrow\right)\;,
\end{equation}
i.e. Fock-rearranging the input always makes the output less noisy, or among all the quantum states with a given spectrum, the passive one generates the least noisy output (see Section \ref{secmaj} for the definitions of majorization and weak submajorization).
\end{defn}
\begin{rem}
If $\Phi$ is trace-preserving, weak sub-majorization in \eqref{conjectureeq} can be equivalently replaced by majorization.
\end{rem}
We can now state the main result of the Chapter:
\begin{thm}\label{maintheorem}
Any one-mode gauge-covariant Gaussian quantum channel is passive-preserving and Fock-optimal.
\begin{proof}
See Section \ref{mainproof}.
\end{proof}
\end{thm}
\begin{cor}
Any linear combination with positive coefficients of gauge-covariant quantum Gaussian channels is Fock-optimal.
\begin{proof}
Follows from Theorem \ref{maintheorem} and Lemma \ref{convexhull}.
\end{proof}
\end{cor}

In the remainder of this Section, we prove some general properties of Fock-optimality that will be needed in the main proof.

\begin{lem}\label{majprojprop}
Let $\Phi$ be a passive-preserving quantum operation.
If for any finite-rank projector $\hat{P}$
\begin{equation}\label{hypproj}
\Phi\left(\hat{P}\right)\prec_w\Phi\left(\hat{P}^\downarrow\right)\;,
\end{equation}
then $\Phi$ is Fock-optimal.
\begin{proof}
Let $\hat{X}$ be a positive trace-class operator as in \eqref{XPi}, with Fock rearrangement as in \eqref{X*Pi}.
Since $\Phi$ is passive-preserving, for any $n\in\mathbb{N}$
\begin{equation}
\Phi\left(\hat{\Pi}_n\right)\prec_w\Phi\left(\hat{\Pi}_n^\downarrow\right)=\Phi\left(\hat{\Pi}_n^\downarrow\right)^\downarrow\;.
\end{equation}
Then we can apply Lemma \ref{majsum} to
\begin{equation}
\Phi\left(\hat{X}\right)=\sum_{n=0}^\infty d_n\;\Phi\left(\hat{\Pi}_n\right)\prec_w\sum_{n=0}^\infty d_n\;\Phi\left(\hat{\Pi}_n^\downarrow\right)=\Phi\left(\hat{X}^\downarrow\right)\;,
\end{equation}
and the claim follows.
\end{proof}
\end{lem}

\begin{lem}\label{conjPQprop}
A quantum operation $\Phi$ is passive-preserving and Fock-optimal iff
\begin{equation}\label{conjPQeq}
\mathrm{Tr}\left[\hat{Q}\;\Phi\left(\hat{P}\right)\right]\leq\mathrm{Tr}\left[\hat{Q}^\downarrow\;\Phi\left(\hat{P}^\downarrow\right)\right]
\end{equation}
for any two finite-rank projectors $\hat{Q}$ and $\hat{P}$.
\begin{proof}
Let us first suppose that $\Phi$ is passive-preserving and Fock-optimal, and let $\hat{P}$ and $\hat{Q}$ be finite-rank projectors.
Then
\begin{equation}
\Phi\left(\hat{P}\right)\prec_w\Phi\left(\hat{P}^\downarrow\right)=\Phi\left(\hat{P}^\downarrow\right)^\downarrow\;,
\end{equation}
and \eqref{conjPQeq} follows from Lemma \ref{majprojlem}.

Let us now suppose that \eqref{conjPQeq} holds for any finite-rank projectors $\hat{P}$ and $\hat{Q}$.
Choosing $\hat{P}$ passive, we get
\begin{equation}
\mathrm{Tr}\left[\hat{Q}\;\Phi\left(\hat{P}\right)\right]\leq\mathrm{Tr}\left[\hat{Q}^\downarrow\;\Phi\left(\hat{P}\right)\right]\;,
\end{equation}
and from Lemma \ref{*symtr} also $\Phi\left(\hat{P}\right)$ is passive, so from Lemma \ref{phistar} $\Phi$ is passive-preserving.
Choosing now a generic $\hat{P}$, by Lemma \ref{majprojlem}
\begin{equation}
\Phi\left(\hat{P}\right)\prec_w\Phi\left(\hat{P}^\downarrow\right)\;,
\end{equation}
and from Lemma \ref{majprojprop} $\Phi$ is also Fock-optimal.
\end{proof}
\end{lem}
We can now prove the two fundamental properties of Fock-optimality:
\begin{thm}\label{Phidag}
Let $\Phi$ be a quantum operation with the restriction of its Hilbert-Schmidt dual $\Phi^\dag$ (see \eqref{dagdef} in Appendix \ref{appG}) to trace-class operators continuous in the trace norm.
Then, $\Phi$ is passive-preserving and Fock-optimal iff $\Phi^\dag$ is passive-preserving and Fock-optimal.
\begin{proof}
Condition \eqref{conjPQeq} can be rewritten as
\begin{equation}\label{PQdag}
\mathrm{Tr}\left[\Phi^\dag\left(\hat{Q}\right)\hat{P}\right]\leq\mathrm{Tr}\left[\Phi^\dag\left(\hat{Q}^\downarrow\right)\hat{P}^\downarrow\right]\;,
\end{equation}
and is therefore symmetric for $\Phi$ and $\Phi^\dag$.
\end{proof}
\end{thm}
\begin{thm}\label{qcirc}
Let $\Phi_1$ and $\Phi_2$ be passive-preserving and Fock-optimal quantum operations with the restriction of $\Phi_2^\dag$ to trace-class operators continuous in the trace norm.
Then, their composition $\Phi_2\circ\Phi_1$ is also passive-preserving and Fock-optimal.
\begin{proof}
Let $\hat{P}$ and $\hat{Q}$ be finite-rank projectors.
Since $\Phi_2$ is Fock-optimal and passive-preserving,
\begin{equation}
\Phi_2\left(\Phi_1\left(\hat{P}\right)\right)\prec_w\Phi_2\left(\Phi_1\left(\hat{P}\right)^\downarrow\right)=\Phi_2\left(\Phi_1\left(\hat{P}\right)^\downarrow\right)^\downarrow\;,
\end{equation}
and by Lemma \ref{majprojlem}
\begin{equation}\label{eqphi12}
\mathrm{Tr}\left[\hat{Q}\;\Phi_2\left(\Phi_1\left(\hat{P}\right)\right)\right] \leq \mathrm{Tr}\left[\hat{Q}^\downarrow\;\Phi_2\left(\Phi_1\left(\hat{P}\right)^\downarrow\right)\right]= \mathrm{Tr}\left[\Phi_2^\dag\left(\hat{Q}^\downarrow\right)\Phi_1\left(\hat{P}\right)^\downarrow\right]\;.
\end{equation}
Since $\Phi_1$ is Fock-optimal and passive-preserving,
\begin{equation}
\Phi_1\left(\hat{P}\right)^\downarrow\prec_w\Phi_1\left(\hat{P}^\downarrow\right)=\Phi_1\left(\hat{P}^\downarrow\right)^\downarrow\;.
\end{equation}
From Theorem \ref{Phidag} also $\Phi_2^\dag$ is passive-preserving, and $\Phi_2^\dag\left(\hat{Q}^\downarrow\right)$ is passive.
Lemma \ref{corTr} implies then
\begin{equation}\label{eqphi3}
\mathrm{Tr}\left[\Phi_2^\dag\left(\hat{Q}^\downarrow\right)\Phi_1\left(\hat{P}\right)^\downarrow\right] \leq \mathrm{Tr}\left[\Phi_2^\dag\left(\hat{Q}^\downarrow\right)\Phi_1\left(\hat{P}^\downarrow\right)\right]= \mathrm{Tr}\left[\hat{Q}^\downarrow\;\Phi_2\left(\Phi_1\left(\hat{P}^\downarrow\right)\right)\right]\;,
\end{equation}
and the claim follows from Lemma \ref{conjPQprop} combining \eqref{eqphi3} with \eqref{eqphi12}.
\end{proof}
\end{thm}

\begin{lem}\label{Phifinite}
Let $\Phi$ be a quantum operation continuous in the Hilbert-Schmidt norm.
Let us suppose that for any $N\in\mathbb{N}$ its restriction to the span of the first $N+1$ Fock states is passive-preserving and Fock-optimal, i.e. for any positive operator $\hat{X}$ supported on the span of the first $N+1$ Fock states
\begin{equation}
\Phi\left(\hat{X}\right)\prec_w\Phi\left(\hat{X}^\downarrow\right)=\Phi\left(\hat{X}^\downarrow\right)^\downarrow\;.
\end{equation}
Then, $\Phi$ is passive-preserving and Fock-optimal.
\begin{proof}
Let $\hat{P}$ and $\hat{Q}$ be two generic finite-rank projectors.
Since the restriction of $\Phi$ to the support of $\hat{\Pi}_N^\downarrow$ is Fock-optimal and passive-preserving,
\begin{equation}
\Phi\left(\hat{\Pi}_N^\downarrow\;\hat{P}\;\hat{\Pi}_N^\downarrow\right) \prec_w \Phi\left(\left(\hat{\Pi}_N^\downarrow\;\hat{P}\;\hat{\Pi}_N^\downarrow\right)^\downarrow\right)= \left(\Phi\left(\left(\hat{\Pi}_N^\downarrow\;\hat{P}\;\hat{\Pi}_N^\downarrow\right)^\downarrow\right)\right)^\downarrow\;.
\end{equation}
Then, from Lemma \ref{majprojlem}
\begin{equation}\label{TrPQN}
\mathrm{Tr}\left[\hat{Q}\;\Phi\left(\hat{\Pi}_N^\downarrow\;\hat{P}\;\hat{\Pi}_N^\downarrow\right)\right] \leq \mathrm{Tr}\left[\hat{Q}^\downarrow\;\Phi\left(\left(\hat{\Pi}_N^\downarrow\;\hat{P}\;\hat{\Pi}_N^\downarrow\right)^\downarrow\right)\right]\;.
\end{equation}
From Lemma \ref{PXPlem},
\begin{equation}
\left\|\hat{\Pi}_N^\downarrow\;\hat{P}\;\hat{\Pi}_N^\downarrow-\hat{P}\right\|_2\to0\qquad\text{for}\;N\to\infty\;,
\end{equation}
and since $\Phi$, the Fock rearrangement (see Lemma \ref{HS**}) and the Hilbert-Schmidt product are continuous in the Hilbert-Schmidt norm, we can take the limit $N\to\infty$ in \eqref{TrPQN} and get
\begin{equation}
\mathrm{Tr}\left[\hat{Q}\;\Phi\left(\hat{P}\right)\right] \leq \mathrm{Tr}\left[\hat{Q}^\downarrow\;\Phi\left(\hat{P}^\downarrow\right)\right]\;.
\end{equation}
The claim now follows from Lemma \ref{conjPQprop}.
\end{proof}
\end{lem}
\begin{lem}\label{convexhull}
Let $\Phi_1$ and $\Phi_2$ be Fock-optimal and passive-preserving quantum operations.
Then, also $\Phi_1+\Phi_2$ is Fock-optimal and passive-preserving.
\begin{proof}
Easily follows from Lemma \ref{conjPQprop}.
\end{proof}
\end{lem}

\section{Proof of the main Theorem}\label{mainproof}
First, we can reduce the problem to the quantum-limited attenuator:
\begin{lem}\label{att->all}
If the  quantum-limited attenuator is passive-preserving and Fock-optimal, the property extends to any  gauge-covariant quantum Gaussian channel.
\begin{proof}
From Section \ref{secatta} of Appendix \ref{appG}, any quantum  gauge-covariant Gaussian channel can be obtained composing a quantum-limited attenuator with a quantum-limited amplifier.
Moreover, the Hilbert-Schmidt dual of a quantum-limited amplifier is proportional to a quantum-limited attenuator, and from Lemma \ref{Phidag} also the amplifier is passive-preserving and Fock-optimal.
Finally, the claim follows from Theorem \ref{qcirc}.
\end{proof}
\end{lem}

By Lemma \ref{Phifinite}, we can restrict to quantum states $\hat{\rho}$ supported on the span of the first $N+1$ Fock states.
Let now
\begin{equation}
\hat{\rho}(t)=e^{t\mathcal{L}}\left(\hat{\rho}\right)\;,
\end{equation}
where $\mathcal{L}$ is the generator of the quantum-limited attenuator defined in \eqref{lindblad}.
From the explicit representation \eqref{kraus}, it is easy to see that $\hat{\rho}(t)$ remains supported on the span of the first $N+1$ Fock states for any $t\geq0$.
In finite dimension, the quantum states with nondegenerate spectrum are dense in the set of all quantum states.
Besides, the spectrum is a continuous function of the operator, and any linear map is continuous.
Then, without loss of generality we can suppose that $\hat{\rho}$ has nondegenerate spectrum.
Let
\begin{equation}
p(t)=\left(p_0(t),\ldots,p_N(t)\right)
\end{equation}
be the vectors of the eigenvalues of $\hat{\rho}(t)$ in decreasing order, and let
\begin{equation}
s_n(t)=\sum_{i=0}^n p_i(t)\;,\qquad n=0,\ldots,\,N\;,
\end{equation}
their partial sums, that we similarly collect into the vector $s(t)$.
Let instead
\begin{equation}\label{pndt}
p_n^\downarrow(t)=\langle n|e^{t\mathcal{L}}\left(\hat{\rho}^\downarrow\right)|n\rangle\;,\qquad n=0,\,\ldots,\,N
\end{equation}
be the eigenvalues of $e^{t\mathcal{L}}\left(\hat{\rho}^\downarrow\right)$ (recall that it is diagonal in the Fock basis for any $t\geq0$), and
\begin{equation}
s_n^\downarrow(t)=\sum_{i=0}^n p_i^\downarrow(t)\;,\qquad n=0,\,\ldots,\,N\;,
\end{equation}
their partial sums.
We notice that $p(0)=p^\downarrow(0)$ and then $s(0)=s^\downarrow(0)$.
Combining \eqref{pndt} with the expression for the Lindbladian \eqref{lindblad}, with the help of \eqref{acta} it is easy to see that the eigenvalues $p_n^\downarrow(t)$ satisfy
\begin{equation}
\frac{d}{dt}p_n^\downarrow(t)=\left(n+1\right)p_{n+1}^\downarrow(t)-n\,p_n^\downarrow(t)\;,
\end{equation}
implying
\begin{equation}
\frac{d}{dt}s_n^\downarrow(t)=(n+1)\left(s^\downarrow_{n+1}(t)-s^\downarrow_n(t)\right)
\end{equation}
for their partial sums.
The proof of Theorem \ref{maintheorem} is a consequence of:
\begin{lem}\label{deg}
The spectrum of $\hat{\rho}(t)$ can be degenerate at most in isolated points.
\end{lem}
\begin{lem}\label{lemma1}
$s(t)$ is continuous in $t$, and for any $t\geq0$ such that $\hat{\rho}(t)$ has nondegenerate spectrum it satisfies
\begin{equation}\label{sdot}
\frac{d}{dt}s_n(t)\leq(n+1)(s_{n+1}(t)-s_n(t))\;,\qquad n=0,\,\ldots,\,N-1\;.
\end{equation}
\end{lem}
\begin{lem}\label{lemma2}
If $s(t)$ is continuous in $t$ and satisfies \eqref{sdot}, then
\begin{equation}
s_n(t)\leq s_n^\downarrow(t)
\end{equation}
for any $t\geq0$ and $n=0,\,\ldots,\,N$.
\end{lem}
Lemma \ref{lemma2} implies that the quantum-limited attenuator is passive-preserving.
Indeed, let us choose $\hat{\rho}$ passive.
Since $e^{t\mathcal{L}}\left(\hat{\rho}\right)$ is diagonal in the Fock basis, $s_n^\downarrow(t)$ is the sum of the eigenvalues corresponding to the first $n+1$ Fock states $|0\rangle,\;\ldots,\;|n\rangle$.
Since $s_n(t)$ is the sum of the $n+1$ greatest eigenvalues, $s_n^\downarrow(t)\leq s_n(t)$.
However, Lemma \ref{lemma2} implies $s_n(t)=s_n^\downarrow(t)$ for $n=0,\,\ldots,\,N$.
Thus $p_n(t)=p_n^\downarrow(t)$, so the operator $e^{t\mathcal{L}}\left(\hat{\rho}\right)$ is passive for any $t$, and the channel $e^{t\mathcal{L}}$ is passive-preserving.

Then from the definition of majorization and Lemma \ref{lemma2} again,
\begin{equation}
e^{t\mathcal{L}}\left(\hat{\rho}\right)\prec_w e^{t\mathcal{L}}\left(\hat{\rho}^\downarrow\right)
\end{equation}
for any $\hat{\rho}$, and the quantum-limited attenuator is also Fock-optimal.

\subsection{Proof of Lemma \ref{deg}}
The matrix elements of the operator $e^{t\mathcal{L}}\left(\hat{\rho}\right)$ are analytic functions of $t$.
The spectrum of $\hat{\rho}(t)$ is degenerate iff the function
\begin{equation}
\phi(t)=\prod_{i\neq j}\left(p_i(t)-p_j(t)\right)
\end{equation}
vanishes.
This function is a symmetric polynomial in the eigenvalues of $\hat{\rho}(t)=e^{t\mathcal{L}}\left(\hat{\rho}\right)$.
Then, for the Fundamental Theorem of Symmetric Polynomials (see e.g Theorem 3 in Chapter 7 of \cite{cox2015ideals}), $\phi(t)$ can be written as a polynomial in the elementary symmetric polynomials in the eigenvalues of $\hat{\rho}(t)$.
However, these polynomials coincide with the coefficients of the characteristic polynomial of $\hat{\rho}(t)$, that are in turn polynomials in its matrix elements.
It follows that $\phi(t)$ can be written as a polynomial in the matrix elements of the operator $\hat{\rho}(t)$.
Since each of these matrix element is an analytic function of $t$, also $\phi(t)$ is analytic.
Since by hypothesis the spectrum of $\hat{\rho}(0)$ is nondegenerate, $\phi$ cannot be identically zero, and its zeroes are isolated points.

\subsection{Proof of Lemma \ref{lemma1}}
The matrix elements of the operator $e^{t\mathcal{L}}\left(\hat{\rho}\right)$ are analytic (and hence continuous and differentiable) functions of $t$.
Then for Weyl's Perturbation Theorem $p(t)$ is continuous in $t$, and also $s(t)$ is continuous (see e.g. Corollary III.2.6 and the discussion at the beginning of Chapter VI  of \cite{bhatia2013matrix}).
Let $\hat{\rho}(t_0)$ have nondegenerate spectrum.
Then, $\hat{\rho}(t)$ has nondegenerate spectrum for any $t$ in a suitable neighbourhood of $t_0$.
In this neighbourhood, we can diagonalize $\hat{\rho}(t)$ with
\begin{equation}
\hat{\rho}(t)=\sum_{n=0}^N p_n(t) |\psi_n(t)\rangle\langle\psi_n(t)|\;,
\end{equation}
where the eigenvalues in decreasing order $p_n(t)$ are differentiable functions of $t$ (see Theorem 6.3.12 of \cite{horn2012matrix}),
and
\begin{equation}
\frac{d}{dt}p_n(t)=\langle\psi_n(t)|\mathcal{L}\left(\hat{\rho}(t)\right)|\psi_n(t)\rangle\;.
\end{equation}
We then have
\begin{equation}
\frac{d}{dt}s_n(t)=\mathrm{Tr}\left[\hat{\Pi}_n(t)\;\mathcal{L}\left(\hat{\rho}(t)\right)\right]\;,
\end{equation}
where
\begin{equation}
\hat{\Pi}_n(t)=\sum_{i=0}^n|\psi_i(t)\rangle\langle\psi_i(t)|\;.
\end{equation}
We can write
\begin{equation}
\hat{\rho}(t)=\sum_{n=0}^N d_n(t)\;\hat{\Pi}_n(t)\;,
\end{equation}
where
\begin{equation}
d_n(t)=p_n(t)-p_{n+1}(t)\geq0\;,
\end{equation}
so that
\begin{equation}
\frac{d}{dt}s_n(t)=\sum_{k=0}^N d_k(t)\;\mathrm{Tr}\left[\hat{\Pi}_n(t)\;\mathcal{L}\left(\hat{\Pi}_k(t)\right)\right]\;.
\end{equation}
With the explicit expression \eqref{lindblad} for $\mathcal{L}$, it is easy to prove that
\begin{equation}
\sum_{k=0}^N d_k(t)\;\mathrm{Tr}\left[\hat{\Pi}_n^\downarrow\;\mathcal{L}\left(\hat{\Pi}_k^\downarrow\right)\right]=(n+1)(s_{n+1}(t)-s_n(t))\;,
\end{equation}
so it would be sufficient to show that
\begin{equation}\label{PL}
\mathrm{Tr}\left[\hat{\Pi}_n(t)\;\mathcal{L}\left(\hat{\Pi}_k(t)\right)\right]\overset{?}{\leq} \mathrm{Tr}\left[\hat{\Pi}_n^\downarrow\;\mathcal{L}\left(\hat{\Pi}_k^\downarrow\right)\right]\;.
\end{equation}
We write explicitly the left-hand side of \eqref{PL}:
\begin{equation}\label{PLext}
\mathrm{Tr}\left[\hat{\Pi}_n(t)\;\hat{a}\;\hat{\Pi}_k(t)\;\hat{a}^\dag-\hat{\Pi}_n(t)\;\hat{\Pi}_k(t)\;\hat{a}^\dag\hat{a}\right]\;,
\end{equation}
where we have used that $\hat{\Pi}_n(t)$ and $\hat{\Pi}_k(t)$ commute.
\begin{itemize}
  \item Let us suppose $n\geq k$.
  Then
  \begin{equation}
  \hat{\Pi}_n(t)\;\hat{\Pi}_k(t)=\hat{\Pi}_k(t)\;.
  \end{equation}
  Using that $\hat{\Pi}_n(t)\leq\hat{\mathbb{I}}$ in the first term of \eqref{PLext}, we get
  \begin{equation}
  \mathrm{Tr}\left[\hat{\Pi}_n(t)\;\hat{a}\;\hat{\Pi}_k(t)\;\hat{a}^\dag-\hat{\Pi}_n(t)\;\hat{\Pi}_k(t)\;\hat{a}^\dag\hat{a}\right]\leq0\;.
  \end{equation}
  On the other hand, since the support of $\hat{a}\,\hat{\Pi}_k^\downarrow\,\hat{a}^\dag$ is contained in the support of $\hat{\Pi}_{k-1}^\downarrow$, and hence in the one of $\hat{\Pi}_n^\downarrow$, we have also
  \begin{equation}
  \hat{\Pi}_n^\downarrow\;\hat{a}\;\hat{\Pi}_k^\downarrow\;\hat{a}^\dag=\hat{a}\;\hat{\Pi}_k^\downarrow\;\hat{a}^\dag\;,
  \end{equation}
  so that
  \begin{equation}
  \mathrm{Tr}\left[\hat{\Pi}_n^\downarrow\;\hat{a}\;\hat{\Pi}_k^\downarrow\;\hat{a}^\dag-\hat{\Pi}_n^\downarrow\;\hat{\Pi}_k^\downarrow\;\hat{a}^\dag\hat{a}\right]=0\;.
  \end{equation}
  \item Let us now suppose that $k\geq n+1$.
  Then
  \begin{equation}
  \hat{\Pi}_n(t)\;\hat{\Pi}_k(t)=\hat{\Pi}_n(t)\;.
  \end{equation}
  Using that $\hat{\Pi}_k(t)\leq\hat{\mathbb{I}}$ in the first term of \eqref{PLext}, together with the commutation relation \eqref{CCR}, we get
  \begin{equation}
  \mathrm{Tr}\left[\hat{\Pi}_n(t)\;\hat{a}\;\hat{\Pi}_k(t)\;\hat{a}^\dag-\hat{\Pi}_n(t)\;\hat{\Pi}_k(t)\;\hat{a}^\dag\hat{a}\right]\leq n+1\;.
  \end{equation}
  On the other hand, since the support of $\hat{a}^\dag\,\hat{\Pi}_n^\downarrow\,\hat{a}$ is contained in the support of $\hat{\Pi}_{n+1}^\downarrow$ and hence in the one of $\hat{\Pi}_k^\downarrow$, we have also
  \begin{equation}
  \hat{\Pi}_k^\downarrow\;\hat{a}^\dag\;\hat{\Pi}_n^\downarrow\;\hat{a}=\hat{a}^\dag\;\hat{\Pi}_n^\downarrow\;\hat{a}\;,
  \end{equation}
  so that
  \begin{equation}
  \mathrm{Tr}\left[\hat{\Pi}_n^\downarrow\;\hat{a}\;\hat{\Pi}_k^\downarrow\;\hat{a}^\dag-\hat{\Pi}_n^\downarrow\;\hat{\Pi}_k^\downarrow\;\hat{a}^\dag\hat{a}\right]=n+1\;.
  \end{equation}
\end{itemize}
\subsection{Proof of Lemma \ref{lemma2}}
Since the quantum-limited attenuator is trace-preserving, we have
\begin{equation}
s_N(t)=\mathrm{Tr}\left[\hat{\rho}(t)\right]=1=s_N^\downarrow(t)\;.
\end{equation}
We will use induction on $n$ in the reverse order: let us suppose to have proved
\begin{equation}
s_{n+1}(t)\leq s_{n+1}^\downarrow(t)\;.
\end{equation}
We then have from \eqref{sdot}
\begin{equation}
\frac{d}{dt}s_n(t)\leq(n+1)\left(s_{n+1}^\downarrow(t)-s_n(t)\right)\;,
\end{equation}
while
\begin{equation}
\frac{d}{dt}s_n^\downarrow(t)=(n+1)\left(s_{n+1}^\downarrow(t)-s_n^\downarrow(t)\right)\;.
\end{equation}
Defining
\begin{equation}
f_n(t)=s_n^\downarrow(t)-s_n(t)\;,
\end{equation}
we have $f_n(0)=0$, and
\begin{equation}
\frac{d}{dt}f_n(t)\geq-(n+1)f_n(t)\;.
\end{equation}
This can be rewritten as
\begin{equation}
e^{-(n+1)t}\;\frac{d}{dt}\left(e^{(n+1)t}\;f_n(t)\right)\geq0\;,
\end{equation}
and implies
\begin{equation}
f_n(t)\geq0\;.
\end{equation}

\section{Generic one-mode Gaussian channels}\label{generic}
In this Section we extend Theorem \ref{maintheorem} to any one-mode quantum Gaussian channel.

\begin{defn}
We say that two quantum channels $\Phi$ and $\Psi$ are equivalent if there are a unitary operator $\hat{U}$ and a unitary or anti-unitary $\hat{V}$ such that
\begin{equation}\label{Psi}
\Psi\left(\hat{X}\right)=\hat{V}\;\Phi\left(\hat{U}\;\hat{X}\;\hat{U}^\dag\right)\;\hat{V}^\dag
\end{equation}
for any trace-class operator $\hat{X}$.
\end{defn}
Clearly, a channel equivalent to a Fock-optimal channel is also Fock-optimal with a suitable redefinition of the Fock rearrangement:
\begin{lem}\label{U*}
Let $\Phi$ be a Fock-optimal quantum channel, and $\Psi$ be as in \eqref{Psi}.
Then, for any positive trace-class operator $\hat{X}$,
\begin{equation}
\Psi\left(\hat{X}\right)\prec_w\Psi\left(\hat{U}^\dag\left(\hat{U}\;\hat{X}\;\hat{U}^\dag\right)^\downarrow\hat{U}\right)\;.
\end{equation}
\end{lem}
The problem of analyzing any Gaussian quantum channel from the point of view of majorization is then reduced to the equivalence classes.

\subsection{Quadratures and squeezing}
In this Section, differently from the rest of the Chapter, $\hat{Q}$ and $\hat{P}$ will denote the quadratures \eqref{quadr}, and not generic projectors.
We can define a continuous basis of not normalizable vectors $\left\{|q\rangle\right\}_{q\in\mathbb{R}}$ with
\begin{eqnarray}
\hat{Q}|q\rangle &=& q|q\rangle\;,\\
\langle q|q'\rangle &=& \delta(q-q')\;,\\
\int_{\mathbb{R}}|q\rangle\langle q|\;dq &=& \hat{\mathbb{I}}\;,\\
e^{-iq\hat{P}}|q'\rangle &=& |q'+q\rangle\;,\qquad q,\,q'\in\mathbb{R}\;.
\end{eqnarray}
For any $\kappa>0$ we define the squeezing unitary operator \cite{barnett2002methods} $\hat{S}_\kappa$ with
\begin{equation}
\hat{S}_\kappa |q\rangle=\sqrt{\kappa}\;|\kappa q\rangle
\end{equation}
for any $q\in\mathbb{R}$.
It satisfies also
\begin{equation}
\hat{S}_\kappa^\dag\;\hat{P}\;\hat{S}_\kappa = \frac{1}{\kappa}\;\hat{P}\;.
\end{equation}

\subsection{Classification theorem}
Then, the following classification theorem holds \cite{holevo2007one,holevo2013quantum}:
\begin{thm}
Any one-mode trace-preserving quantum Gaussian channel is equivalent to one of the following:
\begin{enumerate}
\item a gauge-covariant Gaussian channel, i.e. a channel commuting with the time evolution generated by the photon-number Hamiltonian \eqref{Nosc} (cases $A_1)$, $B_2)$, $C)$ and $D)$ of \cite{holevo2007one});
\item a measure-reprepare channel $\Phi$ of the form
\begin{equation}\label{class2}
\Phi\left(\hat{X}\right)=\int_{\mathbb{R}}\langle q|\hat{X}|q\rangle\;e^{-iq\hat{P}}\;\hat{\rho}_0\;e^{iq\hat{P}}\;dq
\end{equation}
for any trace-class operator $\hat{X}$, where $\rho_0$ is a given Gaussian state (case $A_2)$ of \cite{holevo2007one});
\item a random unitary channel $\Phi_\sigma$ of the form
\begin{equation}\label{Phieta}
\Phi_\sigma\left(\hat{X}\right)=\int_{\mathbb{R}}e^{-iq\hat{P}}\;\hat{X}\;e^{iq\hat{P}}\;\frac{e^{-\frac{q^2}{2\sigma}}}{\sqrt{2\pi\sigma}}\;dq
\end{equation}
for any trace-class operator $\hat{X}$, with $\sigma>0$ (case $B_1)$ of \cite{holevo2007one}).
\end{enumerate}
\end{thm}
From Lemma \ref{U*}, with a suitable redefinition of Fock rearrangement all the channels of the first class are Fock-optimal.
On the contrary, for both the second and the third classes the optimal basis would be an infinitely squeezed version of the Fock basis:

\subsection{Class 2}
We will show that the channel \eqref{class2} does not have optimal inputs.

Let $\hat{\omega}$ be a generic quantum state.
Since $\Phi$ applies a random displacement to the state $\hat{\rho}_0$,
\begin{equation}\label{PhiX0}
\Phi\left(\hat{\omega}\right)\prec\hat{\rho}_0\;.
\end{equation}
Moreover, $\Phi\left(\hat{\omega}\right)$ and $\hat{\rho}_0$ cannot have the same spectrum unless the probability distribution $\langle q|\hat{\omega}|q\rangle$ is a Dirac delta, but this is never the case for any quantum state $\hat{\omega}$, so the majorization in \eqref{PhiX0} is always strict.
Besides, in the limit of infinite squeezing the output tends to $\hat{\rho}_0$ in trace norm:
\begin{eqnarray}
\left\|\Phi\left(\hat{S}_\kappa\;\hat{\omega}\;\hat{S}_\kappa^\dag\right)-\hat{\rho}_0\right\|_1 &=& \left\|\int_{\mathbb{R}}\langle q|\hat{\omega}|q\rangle\left(e^{-i\kappa q\hat{P}}\;\hat{\rho}_0\;e^{i\kappa q\hat{P}}-\hat{\rho}_0\right)dq\right\|_1 \leq\nonumber\\
&\leq& \int_{\mathbb{R}}\langle q|\hat{\omega}|q\rangle\left\|e^{-i\kappa q\hat{P}}\;\hat{\rho}_0\;e^{i\kappa q\hat{P}}-\hat{\rho}_0\right\|_1dq\;,
\end{eqnarray}
and the last integral tends to zero for $\kappa\to0$ since the integrand is dominated by the integrable function $2\langle q|\hat{\omega}|q\rangle$, and tends to zero pointwise.
It follows that the majorization relation
\begin{equation}
\Phi\left(\hat{S}_\kappa\;\hat{\omega}\;\hat{S}_\kappa\right)\prec\Phi\left(\hat{\omega}\right)
\end{equation}
will surely not hold for some positive $\kappa$ in a neighbourhood of $0$, and $\hat{\omega}$ is not an optimal input for $\Phi$.

\subsection{Class 3}
For the channel \eqref{Phieta}, squeezing the input always makes the output strictly less noisy.
Indeed, it is easy to show that for any positive $\sigma$ and $\sigma'$
\begin{equation}
\Phi_\sigma\circ\Phi_{\sigma'}=\Phi_{\sigma+\sigma'}\;.
\end{equation}
Then, for any $\kappa>1$ and any positive trace-class $\hat{X}$
\begin{equation}
\hat{S}_\kappa\;\Phi_\sigma\left(\hat{X}\right)\;\hat{S}_\kappa^\dag = \Phi_{\kappa^2\sigma}\left(\hat{S}_\kappa\;\hat{X}\;\hat{S}_\kappa^\dag\right)= \Phi_{(\kappa^2-1)\sigma}\left(\Phi_{\sigma}\left(\hat{S}_\kappa\;\hat{X}\;\hat{S}_\kappa^\dag\right)\right)\;,
\end{equation}
hence, recalling that $\Phi$ applies a random displacement,
\begin{equation}
\Phi_\sigma\left(\hat{X}\right)\prec \Phi_{\sigma}\left(\hat{S}_\kappa\;\hat{X}\;\hat{S}_\kappa^\dag\right)\;.
\end{equation}

\section{The thinning}\label{secthinning}
The thinning \cite{renyi1956characterization} is the map acting on classical probability distributions on the set of natural numbers that is the discrete analogue of the continuous rescaling operation on positive real numbers.

In this Section we show that the thinning coincides with the restriction of the Gaussian quantum-limited attenuator to quantum states diagonal in the Fock basis, and we hence extend Theorem \ref{maintheorem} to the discrete classical setting.

\begin{defn}[$\ell^1$ norm]
The $\ell^1$ norm of a sequence $\{x_n\}_{n\in\mathbb{N}}$ is
\begin{equation}
\|x\|_1=\sum_{n=0}^\infty |x_n|\;.
\end{equation}
We say that $x$ is summable if $\|x\|_1<\infty$.
\end{defn}
\begin{defn}
A discrete classical channel is a linear positive map on summable sequences that is continuous in the $\ell^1$ norm and preserves the sum, i.e. for any summable sequence $x$
\begin{equation}
\sum_{n=0}^\infty\left[\Phi(x)\right]_n=\sum_{n=0}^\infty x_n\;.
\end{equation}
\end{defn}
The definitions of passive-preserving and Fock-optimal channels can be easily extended to the discrete classical case:
\begin{defn}
Given a summable sequence of positive numbers $\{x_n\}_{n\in\mathbb{N}}$, we denote with $x^\downarrow$ its decreasing rearrangement.
\end{defn}
\begin{defn}
We say that a discrete classical channel $\Phi$ is passive-preserving if for any decreasing summable sequence $x$ of positive numbers $\Phi(x)$ is still decreasing.
\end{defn}
\begin{defn}
We say that a discrete classical channel $\Phi$ is Fock-optimal if for any summable sequence $x$ of positive numbers
\begin{equation}\label{optimalcl}
\Phi(x)\prec\Phi\left(x^\downarrow\right)\;.
\end{equation}
\end{defn}
Let us now introduce the thinning.
\begin{defn}[Thinning]
Let $N$ be a random variable with values in $\mathbb{N}$.
The thinning with parameter $0\leq\lambda\leq1$ is defined as
\begin{equation}
T_\lambda(N)=\sum_{i=1}^N B_i\;,
\end{equation}
where the $\{B_n\}_{n\in\mathbb{N}^+}$ are independent Bernoulli variables with parameter $\lambda$, i.e. each $B_i$ is one with probability $\lambda$, and zero with probability $1-\lambda$.
\end{defn}
From a physical point of view, the thinning can be understood as follows:
consider a beam-splitter of transmissivity $\lambda$, where each incoming photon has probability $\lambda$ of being transmitted, and $1-\lambda$ of being reflected, and suppose that what happens to a photon is independent from what happens to the other ones.
Let $N$ be the random variable associated to the number of incoming photons, and $\{p_n\}_{n\in\mathbb{N}}$ its probability distribution, i.e. $p_n$ is the probability that $N=n$ (i.e. that $n$ photons are sent).
Then, $T_\lambda(p)$ is the probability distribution of the number of transmitted photons.
It is easy to show that
\begin{equation}\label{Tn}
\left[T_\lambda(p)\right]_n=\sum_{k=0}^\infty r_{n|k}\;p_k\;,
\end{equation}
where the transition probabilities $r_{n|k}$ are given by
\begin{equation}\label{rnk}
r_{n|k}=\binom{k}{n}\lambda^n(1-\lambda)^{k-n}\;,
\end{equation}
and vanish for $k<n$.

The map \eqref{Tn} can be uniquely extended by linearity to the set of summable sequences:
\begin{equation}\label{Tne}
\left[T_\lambda(x)\right]_n=\sum_{k=0}^\infty r_{n|k}\;x_k\;,\qquad \|x\|_1<\infty\;.
\end{equation}
\begin{prop}
The map $T_\lambda$ defined in \eqref{Tne} is continuous in the $\ell^1$ norm and sum-preserving.
\begin{proof}
For any summable sequence $x$ we have
\begin{equation}
\sum_{n=0}^\infty\left|T_\lambda(x)\right|_n\leq\sum_{n=0}^\infty\sum_{k=0}^\infty r_{n|k}\;|x_k|=\sum_{k=0}^\infty|x_k|\;,
\end{equation}
where we have used that for any $k\in\mathbb{N}$
\begin{equation}
\sum_{n=0}^\infty r_{n|k}=1\;.
\end{equation}
Then, $T_\lambda$ is continuous in the $\ell^1$ norm.

An analogous proof shows that $T_\lambda$ is sum-preserving.
\end{proof}
\end{prop}

\begin{thm}\label{thinatt}
Let $\Phi_\lambda$ and $T_\lambda$ be the quantum-limited attenuator and the thinning of parameter $0\leq\lambda\leq1$, respectively.
Then for any summable sequence $x$
\begin{equation}
\Phi_\lambda\left(\sum_{n=0}^\infty x_n\;|n\rangle\langle n|\right)=\sum_{n=0}^\infty \left[T_\lambda(x)\right]_n\;|n\rangle\langle n|\;.
\end{equation}
\begin{proof}
Easily follows from the representation \eqref{kraus}, \eqref{Tn} and \eqref{rnk}.
\end{proof}
\end{thm}
As easy consequence of Theorem \ref{thinatt} and Theorem \ref{maintheorem}, we have
\begin{thm}
The thinning is passive-preserving and Fock-optimal.
\end{thm}

\section{Conclusion}\label{secconclmaj}
We have proved that for any one-mode gauge-covariant bosonic Gaussian channel, the output generated by any state diagonal in the Fock basis and with decreasing eigenvalues majorizes the output generated by any other input state with the same spectrum.
Then, the input state with a given entropy minimizing the output entropy is certainly diagonal in the Fock basis and has decreasing eigenvalues.
The non-commutative quantum constrained minimum output entropy conjecture \ref{CMOE} is hence reduced to a problem in classical discrete probability, that we will solve in Chapter \ref{chepni}.

Exploiting unitary equivalence we also extend our results to one-mode trace-preserving bosonic Gaussian channel which are not gauge-covariant, with the notable exceptions of those special maps admitting normal forms $A_2)$ and $B_1)$ \cite{holevo2007one} for which we show that no general majorization ordering is possible.

\chapter{Gaussian states minimize the output entropy of the attenuator}\label{chepni}
In this Chapter we exploit the majorization result of Chapter \ref{majorization} to prove that Gaussian thermal input states minimize the output entropy of the one-mode Gaussian quantum-limited attenuator for fixed input entropy.

The Chapter is based on
\begin{enumerate}
\item[\cite{de2016gaussian}] G.~De~Palma, D.~Trevisan, and V.~Giovannetti, ``Gaussian states minimize the output entropy of the one-mode quantum
  attenuator,'' \emph{IEEE Transactions on Information Theory}, vol.~63, no.~1,
  pp. 728--737, 2017.\\ {\small\url{http://ieeexplore.ieee.org/document/7707386}}
\end{enumerate}

\section{Introduction}
Most communication schemes encode the information into pulses of electromagnetic radiation, that is transmitted through metal wires, optical fibers or free space, and is unavoidably affected by signal attenuation.
The maximum achievable communication rate of a channel depends on the minimum noise achievable at its output.
A continuous classical signal can be modeled by a real random variable $X$.
Signal attenuation corresponds to a rescaling $X\mapsto\sqrt{\lambda}\,X$, where $0\leq\lambda\leq1$ is the attenuation coefficient (the power of the signal is proportional to $X^2$ and gets rescaled by $\lambda$).
The noise of a real random variable is quantified by its Shannon differential entropy  $H$ \cite{cover2006elements}.
The Shannon entropy of the rescaled signal is a simple function of the entropy of the original signal \cite{cover2006elements}:
\begin{equation}\label{epniscaling}
H\left(\sqrt{\lambda}\;X\right)=H\left(X\right)+\ln\sqrt{\lambda}\;.
\end{equation}
This property is ubiquitous in classical information theory.
For example, it lies at the basis of the proof of the Entropy Power Inequality \cite{dembo1991information,gardner2002brunn,shannon2001mathematical,stam1959some,verdu2006simple,rioul2011information,cover2006elements} (see also Section \ref{secproblem} and Equation \eqref{rescS}).

In the quantum regime the role of the classical Shannon entropy is played by the von Neumann entropy \cite{wilde2013quantum,holevo2013quantum} and signal attenuation is modeled by the Gaussian quantum-limited attenuator (see \cite{chan2006free,braunstein2005quantum,holevo2013quantum,weedbrook2012gaussian,holevo2015gaussian} and Section \ref{secattampl}).

A striking consequence of the quantization of the energy is that the output entropy of the quantum-limited attenuator is not a function of the input entropy alone.
A fundamental problem in quantum communication is then determining the minimum output entropy of the attenuator for fixed input entropy.
According to the constrained minimum output entropy conjecture \ref{CMOE}, Gaussian thermal input states achieve this minimum output entropy \cite{guha2007classicalproc,guha2007classical,guha2008entropy,guha2008capacity,wilde2012information,wilde2012quantum}.
The first attempt of a proof has been the quantum Entropy Power Inequality (qEPI) (see \cite{konig2013limits,konig2014entropy,de2014generalization,de2015multimode} and Chapter \ref{epi}), that provides the lower bound
\begin{equation}\label{epniqEPI}
S\left(\Phi_\lambda\left(\hat{\rho}\right)\right)\geq n\;\ln\left(\lambda\left(e^{\left.S\left(\hat{\rho}\right)\right/n}-1\right)+1\right)
\end{equation}
to the output entropy of the $n$-mode quantum-limited attenuator $\Phi_\lambda$ in terms of the entropy of the input state $\hat{\rho}$.
However, the qEPI \eqref{epniqEPI} is \emph{not} saturated by thermal Gaussian states, and thus it is not sufficient to prove their conjectured optimality.

Here we prove that Gaussian thermal input states minimize the output entropy of the one-mode quantum-limited attenuator for fixed input entropy (Theorem \ref{epnithmmain}).
The proof starts from the recent majorization result on one-mode Gaussian quantum channels that we have proved in Chapter \ref{majorization} (see also \cite{de2015passive}), that reduces the problem to input states diagonal in the Fock basis.
The key point of the proof is a new isoperimetric ineqeuality (Theorem \ref{epnithmiso}), that provides a lower bound to the derivative of the output entropy of the attenuator with respect to the attenuation coefficient.

The restriction of the one-mode quantum-limited attenuator to input states diagonal in the Fock basis is the map acting on discrete classical probability distributions on $\mathbb{N}$ known in the probability literature under the name of thinning \cite{de2015passive}.
The thinning has been introduced by R\'enyi \cite{renyi1956characterization} as a discrete analogue of the rescaling of a continuous real random variable.
The thinning has been involved with this role in discrete versions of the central limit theorem \cite{harremoes2007thinning,yu2009monotonic,harremoes2010thinning}
and of the Entropy Power Inequality \cite{yu2009concavity,johnson2010monotonicity}.
All these results require the ad hoc hypothesis of the ultra log-concavity (ULC) of the input state.
In particular, the Restricted Thinned Entropy Power Inequality \cite{johnson2010monotonicity} states that the Poisson input probability distribution minimizes the output Shannon entropy of the thinning among all the ULC input probability distributions with a given Shannon entropy.
We prove (Theorem \ref{epnithmthin}) that the geometric distribution minimizes the output entropy of the thinning among all the input probability distributions with a given entropy, without the ad hoc ULC constraint.

Theorem \ref{epnithmmain} constitutes a strong evidence for the validity of the conjecture in the multimode scenario, whose proof could exploit a multimode generalization of the isoperimetric inequality \eqref{epnilogs}.
The multimode generalization of Theorem \ref{epnithmmain} would finally permit to conclude the proof of the optimality of coherent Gaussian states for two communication tasks.
The first is the triple trade-off coding for public communication, private communication and secret key distribution through the Gaussian quantum-limited attenuator \cite{wilde2012public,wilde2012information,wilde2012quantum}.
The second is the transmission of classical information to two receivers through the Gaussian degraded quantum broadcast channel \cite{guha2007classicalproc,guha2007classical}, that we have discussed in Section \ref{broadcast}.
Moreover, it would permit to determine the triple trade-off region for the simultaneous transmission of both classical and quantum information with assistance or generation of shared entanglement through the Gaussian quantum-limited attenuator \cite{wilde2012public,wilde2012information,wilde2012quantum}.

The Chapter is structured as follows.
In Section \ref{epnisetup} we state the main result (Theorem \ref{epnithmmain}).
Section \ref{epnisecproof} contains the proof of Theorem \ref{epnithmmain} and the statement of the isoperimetric inequality (Theorem \ref{epnithmiso}); Sections \ref{epniseclem} and \ref{secprooflem} contain the proof of Theorem \ref{epnithmiso}.
Section \ref{epnisecthinning} links these results to the thinning operation, and Sections \ref{epnifinitesproof} and \ref{epniauxlemmata} contain the proof of some auxiliary lemmata.
Finally, the conclusions are in Section \ref{epnisecconcl}.

\section{Main result}\label{epnisetup}
The Gaussian thermal state with respect to the photon-number Hamiltonian \eqref{Nosc} and with average energy $E\geq0$ is
\begin{equation}\label{epniomegaE}
\hat{\omega}_E=\sum_{n=0}^\infty \frac{1}{E+1}\left(\frac{E}{E+1}\right)^n\;|n\rangle\langle n|\;,\quad\mathrm{Tr}\left[\hat{H}\;\hat{\omega}_E\right]=E\;,
\end{equation}
where $|n\rangle_{n\in\mathbb{N}}$ are the states of the Fock basis \eqref{fockdef}.
$\hat{\omega}$ corresponds to a geometric probability distribution of the energy, and has von Neumann entropy
\begin{equation}\label{epnidefg}
S\left(\hat{\omega}_E\right)=\left(E+1\right)\ln\left(E+1\right)-E\ln E:=g(E)\;.
\end{equation}
The quantum-limited attenuator sends thermal states into themselves, i.e. $\Phi_\lambda\left(\hat{\omega}_E\right)=\hat{\omega}_{\lambda E}$, hence
\begin{equation}
S\left(\Phi_\lambda\left(\hat{\omega}_E\right)\right)=g(\lambda E)=g\left(\lambda\;g^{-1}\left(S\left(\hat{\omega}_E\right)\right)\right)\;.
\end{equation}

We can now state our main result.
\begin{thm}\label{epnithmmain}
Gaussian thermal input states \eqref{epniomegaE} minimize the output entropy of the quantum-limited attenuator among all the input states with a given entropy, i.e. for any input state $\hat{\rho}$ and any $0\leq\lambda\leq1$
\begin{equation}
S\left(\Phi_\lambda\left(\hat{\rho}\right)\right)\geq g\left(\lambda\;g^{-1}\left(S\left(\hat{\rho}\right)\right)\right)\;.
\end{equation}
\begin{proof}
See Section \ref{epnisecproof}.
\end{proof}
\end{thm}

\section{Proof of Theorem \ref{epnithmmain}}\label{epnisecproof}
The starting point of the proof is the result of Chapter \ref{majorization} and Ref. \cite{de2015passive}, that links the constrained minimum output entropy conjecture to the notions of passive states.
The passive states of a quantum system \cite{pusz1978passive,lenard1978thermodynamical,gorecki1980passive,vinjanampathy2015quantum,goold2015role,binder2015quantum} minimize the average energy for a given spectrum.
They are diagonal in the energy eigenbasis, and their eigenvalues decrease as the energy increases.
The passive rearrangement $\hat{\rho}^\downarrow$ of a quantum state $\hat{\rho}$ is the only passive state with the same spectrum of $\hat{\rho}$.
The result is the following:
\begin{thm}\label{epnithmmaj}
The passive rearrangement of the input decreases the output entropy, i.e. for any quantum state $\hat{\rho}$ and any $0\leq\lambda\leq1$
\begin{equation}\label{epnipassiveS}
S\left(\Phi_\lambda\left(\hat{\rho}\right)\right)\geq S\left(\Phi_\lambda\left(\hat{\rho}^\downarrow\right)\right)\;.
\end{equation}
\begin{proof}
Follows from Theorem \ref{maintheorem} and Remark \ref{majS}.
\end{proof}
\end{thm}
Then, it is sufficient to prove Theorem \ref{epnithmmain} for passive states, i.e. states of the form
\begin{equation}
\hat{\rho}=\sum_{n=0}^\infty p_n\;|n\rangle\langle n|\;,\qquad p_0\geq p_1\geq\ldots\geq0\;.
\end{equation}

\begin{lem}\label{epnifinites}
If Theorem \ref{epnithmmain} holds for any passive state with finite support, then it holds for any passive state.
\begin{proof}
See Section \ref{epnifinitesproof}.
\end{proof}
\end{lem}
From Lemma \ref{epnifinites}, we can suppose $\hat{\rho}$ to be a passive state with finite support.
\begin{lem}\label{epnilemcomp}
The quantum-limited attenuator $\Phi_\lambda$ satisfies the composition rule $\Phi_\lambda\circ\Phi_{\lambda'}=\Phi_{\lambda\,\lambda'}$.
\begin{proof}
Follows from Lemma \ref{lemL}.
\end{proof}
\end{lem}
The function $g(x)$ defined in \eqref{epnidefg} is differentiable for $x>0$, and continuous and strictly increasing for $x\geq0$, and its image is the whole interval $[0,\,\infty)$.
Then, its inverse $g^{-1}(S)$ is defined for any $S\geq0$, it is continuous and strictly increasing for $S\geq0$, and differentiable for $S>0$.
We define for any $t\geq0$ the functions
\begin{equation}
\phi(t)=S\left(\Phi_{e^{-t}}\left(\hat{\rho}\right)\right)\;,\qquad\phi_0(t)=g\left(e^{-t}\;g^{-1}\left(S\left(\hat{\rho}\right)\right)\right)\;.
\end{equation}
It is easy to show that
\begin{equation}\label{epnistart}
\phi(0)=\phi_0(0)\;,
\end{equation}
and
\begin{equation}\label{epniphi0}
\frac{d}{dt}\phi_0(t)=f\left(\phi_0(t)\right)\;,
\end{equation}
where
\begin{equation}\label{epnideff}
f(S)=-g^{-1}(S)\;g'\left(g^{-1}(S)\right)
\end{equation}
is defined for any $S\geq0$, and differentiable for $S>0$.
\begin{lem}
$f$ is differentiable for any $S\geq0$.
\begin{proof}
We have
\begin{equation}
f'(S)=\frac{1}{\left(1+g^{-1}(S)\right)\ln\left(1+\frac{1}{g^{-1}(S)}\right)}-1\;,
\end{equation}
hence $\lim_{S\to0}f'(S)=-1$.
\end{proof}
\end{lem}

The key point of the proof is
\begin{thm}[Isoperimetric inequality]\label{epnithmiso}
For any quantum state $\hat{\rho}$ with finite support
\begin{equation}\label{epnilogs}
\left.\frac{d}{dt}S\left(\Phi_{e^{-t}}\left(\hat{\rho}\right)\right)\right|_{t=0}:=-F\left(\hat{\rho}\right)\geq f\left(S\left(\hat{\rho}\right)\right)\;.
\end{equation}
\begin{proof}
See Section \ref{epniseclem}.
\end{proof}
\end{thm}
Since the quantum-limited attenuator sends the set of passive states with finite support into itself \cite{de2015passive}, we can replace $\hat{\rho}\to \Phi_{e^{-t}}\left(\hat{\rho}\right)$ in equation \eqref{epnilogs}, and with the help of Lemma \ref{epnilemcomp} we get
\begin{equation}\label{epniphi}
\frac{d}{dt}\phi(t)\geq f\left(\phi(t)\right)\;.
\end{equation}
The claim then follows from
\begin{thm}[Comparison theorem for first-order ordinary differential equations]
Let $\phi,\,\phi_0:[0,\infty)\to[0,\infty)$ be differentiable functions satisfying \eqref{epnistart}, \eqref{epniphi0} and \eqref{epniphi} with $f:[0,\infty)\to\mathbb{R}$ differentiable.
Then, $\phi(t)\geq\phi_0(t)$ for any $t\geq0$.
\begin{proof}
See e.g. Theorem 2.2.2 of \cite{ames1997inequalities}.
\end{proof}
\end{thm}

\section{Proof of Theorem \ref{epnithmiso}}\label{epniseclem}
Let us fix $S\left(\hat{\rho}\right)=S$.

If $S=0$, for the positivity of the entropy we have for any quantum state $-F\left(\hat{\rho}\right)\geq0=f(0)$,
and the inequality \eqref{epnilogs} is proven.

We can then suppose $S>0$.
Taking the derivative of \eqref{epnipassiveS} with respect to $t$ for $t=0$ we get
\begin{equation}\label{epnipassiveF}
F\left(\hat{\rho}\right)\leq F\left(\hat{\rho}^\downarrow\right)\;,
\end{equation}
hence it is sufficient to prove Theorem \ref{epnithmiso} for passive states with finite support.

Let us fix $N\in\mathbb{N}$, and consider a quantum state $\hat{\rho}$ with entropy $S$ of the form
\begin{equation}\label{epnidefp}
\hat{\rho}=\sum_{n=0}^N p_n\;|n\rangle\langle n|\;.
\end{equation}
Let $\mathcal{D}_N$ be the set of decreasing probability distributions on $\left\{0,\ldots,\,N\right\}$ with Shannon entropy $S$.
We recall that the Shannon entropy of $p$ coincides with the von Neumann entropy of $\hat{\rho}$.
The state in \eqref{epnidefp} is passive if $p\in\mathcal{D}_N$.

\begin{lem}\label{epnilemcpt}
$\mathcal{D}_N$ is compact.
\begin{proof}
The set of decreasing probability distributions on $\left\{0,\ldots,\,N\right\}$ is a closed bounded subset of $\mathbb{R}^{N+1}$, hence it is compact.
The Shannon entropy $H$ is continuous on this set.
$\mathcal{D}_N$ is the counterimage of the point $S$, hence it is closed.
Since $\mathcal{D}_N$ is contained in a compact set, it is compact, too.
\end{proof}
\end{lem}

\begin{defn}[Connected support]
A probability distribution $p$ on $\left\{0,\ldots,\,N\right\}$ has connected support iff
$p_n>0$ for $n=0,\ldots,\,N'$, and $p_{N'+1}=\ldots=p_N=0$, where $0\leq N'\leq N$ can depend on $p$
($N'=N$ means $p_n>0$ for any $n$).
We call $\mathcal{P}_N$ the set of probability distributions on $\left\{0,\ldots,\,N\right\}$ with connected support and Shannon entropy $S$.
\end{defn}

We relax the passivity hypothesis, and consider all the states as in \eqref{epnidefp} with $p\in\mathcal{P}_N$.
We notice that any decreasing $p$ has connected support, i.e. $\mathcal{D}_N\subset\mathcal{P}_N$.

From Equations \eqref{rnk}, \eqref{Tne} and Theorem \ref{thinatt}, we have for any $t\geq0$
\begin{equation}
\Phi_{e^{-t}}\left(\hat{\rho}\right)=\sum_{n=0}^N p_n(t)\;|n\rangle\langle n|\;,
\end{equation}
where
\begin{equation}\label{epnipnt}
p_n(t)=\sum_{k=n}^N\binom{k}{n}e^{-nt}\left(1-e^{-t}\right)^{k-n}p_k
\end{equation}
satisfies $p_n'(0)=\left(n+1\right)p_{n+1}-n\,p_n$ for $n=0,\ldots,\,N$, and we have set for simplicity $p_{N+1}=0$.

Since $p_{N'+1}=\ldots=p_N=0$, from \eqref{epnipnt} we get $p_{N'+1}(t)=\ldots=p_N(t)=0$ for any $t\geq0$.
We then have
\begin{equation}
S\left(\Phi_{e^{-t}}\left(\hat{\rho}\right)\right)=-\sum_{n=0}^{N'}p_n(t)\ln p_n(t)\;,
\end{equation}
and
\begin{equation}\label{epnidefF}
F\left(\hat{\rho}\right)=\sum_{n=0}^{N'}p_n'(0)\left(\ln p_n+1\right)=\sum_{n=1}^{N'} n\,p_n\ln\frac{p_{n-1}}{p_n}\;.
\end{equation}

Let $F_N$ be the $\sup$ of $F(p)$ for $p\in\mathcal{P}_N$, where with a bit of abuse of notation we have defined $F(p)=F\left(\hat{\rho}\right)$ for any $\hat{\rho}$ as in \eqref{epnidefp}.
For \eqref{epnipassiveF}, $F_N$ is also the $\sup$ of $F(p)$ for $p\in\mathcal{D}_N$.
For Lemma \ref{epnilemcpt} $\mathcal{D}_N$ is compact.
Since $F$ is continuous on $\mathcal{D}_N$, the $\sup$ is achieved in a point $p^{(N)}\in\mathcal{D}_N$.
This point satisfies the Karush-Kuhn-Tucker (KKT) necessary conditions \cite{kuhn1951} for the maximization of $F$ with the entropy constraint.
The proof then comes from
\begin{lem}\label{epnilemmain}
There is a subsequence $\left\{N_k\right\}_{k\in\mathbb{N}}$ such that
\begin{equation}
\lim_{k\to\infty}F_{N_k}=-f(S)\;.
\end{equation}
\begin{proof}
See Section \ref{secprooflem}.
\end{proof}
\end{lem}

Then, since $\mathcal{D}_N\subset\mathcal{D}_{N+1}$ for any $N$, $F_N$ is increasing in $N$, and for any $p\in\mathcal{P}_N$
\begin{equation}
F(p)\leq F_N\leq \sup_{N\in\mathbb{N}}F_N=\lim_{N\to\infty}F_N=\lim_{k\to\infty}F_{N_k}=-f(S).
\end{equation}

\section{Proof of Lemma \ref{epnilemmain}}\label{secprooflem}
The point $p^{(N)}$ is the maximum of $F$ for $p\in\mathcal{P}_N$.
The constraints read
\begin{equation}
p_0,\ldots,\,p_N\geq0\;,\quad \sum_{n=0}^{N} p_n=1\;,\quad-\sum_{n=0}^{N} p_n\ln p_n=S\;.
\end{equation}
$p^{(N)}$ must then satisfy the associated KKT necessary conditions \cite{kuhn1951}.
We build the functional
\begin{equation}
\tilde{F}(p)=F(p)-\lambda_N\sum_{n=0}^N p_n+\mu_N\sum_{n=0}^N p_n\ln p_n\;.
\end{equation}
Let $N'$ be such that
\begin{equation}\label{epnipdecr}
p^{(N)}_0\geq\ldots\geq p^{(N)}_{N'}>p^{(N)}_{N'+1}=\ldots=p^{(N)}_N=0\;.
\end{equation}
\begin{rem}\label{epniremN}
We must have $N'\geq1$.
\begin{proof}
If $N'=0$, we must have $p_0^{(N)}=1$ and $p^{(N)}_1=\ldots=p^{(N)}_N=0$, hence $S=0$, contradicting the hypothesis $S>0$.
\end{proof}
\end{rem}

The KKT stationarity condition for $n=0,\ldots,\,N'$ reads
\begin{equation}\label{epniprec}
\left.\frac{\partial}{\partial p_n}\tilde{F}\right|_{p=p^{(N)}} =n\ln\frac{p_{n-1}^{(N)}}{p_n^{(N)}}-n+(n+1)\frac{p_{n+1}^{(N)}}{p_n^{(N)}}-\lambda_N+\mu_N\ln p_n^{(N)}+\mu_N=0\;.
\end{equation}
If $N'<N$, $p^{(N)}$ satisfies the KKT dual feasibility condition associated to $p^{(N)}_{N'+1}$.
To avoid the singularity of the logarithm in $0$, we make the variable change
\begin{equation}
y=-p_{N'+1}\ln p_{N'+1}\;,\qquad p_{N'+1}=\psi(y)\;,
\end{equation}
where $\psi$ satisfies
\begin{equation}\label{epnidefpsi}
\psi\left(-x\ln x\right)=x\qquad\forall\;0\leq x\leq\frac{1}{e}\;.
\end{equation}
Since $\psi(0)=0$, the point $p_{N'+1}=0$ corresponds to $y=0$.
Differentiating \eqref{epnidefpsi} with respect to $x$, we get
\begin{equation}
\psi'\left(-x\ln x\right)=-\frac{1}{1+\ln x}\qquad\forall\;0<x<\frac{1}{e}\;,
\end{equation}
and taking the limit for $x\to0$ we get that $\psi'(y)$ is continuous in $y=0$ with $\psi'(0)=0$.

For hypothesis $p^{(N)}\in\mathcal{P}_{N'}\subset\mathcal{P}_{N'+1}\subset\mathcal{P}_N$.
Then, $p^{(N)}$ is a maximum point for $F(p)$ also if we restrict to $p\in\mathcal{P}_{N'+1}$.
We can then consider the restriction of the functional $\tilde{F}$ on $\mathcal{P}_{N'+1}$:
\begin{eqnarray}
\tilde{F}(p) &=& \sum_{n=1}^{N'}n\,p_n\ln\frac{p_{n-1}}{p_n}+\left(N'+1\right)\psi(y)\ln p_{N'}+\left(N'+1\right)y\nonumber\\
&&-{\lambda_N}\sum_{n=0}^{N'}p_n-\lambda_N\;\psi(y)+{\mu_N}\sum_{n=0}^{N'}p_n\ln p_n-\mu_N\,y\;.
\end{eqnarray}
The condition is then
\begin{equation}\label{epnimuN}
\left.\frac{\partial}{\partial y}\tilde{F}\right|_{p=p^{(N)}}=N'+1-\mu_N\leq0\;,
\end{equation}
where we have used that $\psi'(0)=0$.

We define for any $n=0,\ldots,\,N'$
\begin{equation}
z_n^{(N)}=\frac{p_{n+1}^{(N)}}{p_n^{(N)}}\;.
\end{equation}
Condition \eqref{epnipdecr} implies
\begin{equation}\label{epniz01}
0<z_n^{(N)}\leq1\qquad \forall\;n=0,\ldots,\,N'-1\;,\qquad z_{N'}^{(N)}=0\;.
\end{equation}
For Remark \ref{epniremN} $N'\geq1$, hence $z_0^{(N)}>0$.

Taking the difference of \eqref{epniprec} for two consecutive values of $n$ we get for any $n=0,\ldots,\,N'-1$
\begin{equation}\label{epniznrec}
\left(n+2\right)z_{n+1}^{(N)} = \left(n+2\right)z_n^{(N)}+1-z_n^{(N)}+\left(1-\mu_N\right)\ln z_n^{(N)}+n\ln\frac{z_n^{(N)}}{z_{n-1}^{(N)}}\;.
\end{equation}

\begin{lem}\label{epnilemincr}
We must have
\begin{equation}\label{epnimudecr}
1-\mu_N\geq\frac{z_0^{(N)}-1}{\ln z_0^{(N)}}\geq0\;.
\end{equation}
Moreover, $z^{(N)}_n$ is decreasing in $n$ and $N'=N$, i.e.
\begin{equation}
1\geq z^{(N)}_0\geq\ldots\geq z^{(N)}_{N-1}>z^{(N)}_{N}=0\;.
\end{equation}
\begin{proof}
Let us suppose $1-\mu_N<\left.\left(z_0^{(N)}-1\right)\right/\ln z_0^{(N)}$.
We will prove by induction on $n$ that the sequence $z^{(N)}_n$ is increasing in $n$.
The inductive hypothesis is $0<z_0^{(N)}\leq\ldots\leq z_n^{(N)}\leq1$, true for $n=0$.
Since the function $\left.\left(z-1\right)\right/\ln z$ is strictly increasing for $0\leq z\leq1$, we have
\begin{equation}
1-\mu_N<\frac{z_0^{(N)}-1}{\ln z_0^{(N)}}\leq\frac{z_n^{(N)}-1}{\ln z_n^{(N)}}\;,
\end{equation}
and hence $\left(1-\mu_N\right)\ln z_n^{(N)}\geq z_n^{(N)}-1$.
Since $z_{n-1}^{(N)}\leq z_n^{(N)}$, from \eqref{epniznrec} we have
\begin{equation}
\left(n+2\right)\left(z^{(N)}_{n+1}-z^{(N)}_n\right) = 1-z_n^{(N)}+\left(1-\mu_N\right)\ln z_n^{(N)}+n\ln\frac{z_n^{(N)}}{z_{n-1}^{(N)}}\geq0\;,
\end{equation}
and hence $z_{n+1}^{(N)}\geq z_n^{(N)}$.
However, this is in contradiction with the hypothesis $z^{(N)}_{N'}=0$.

We must then have $1-\mu_N\leq\left.\left(z_0^{(N)}-1\right)\right/\ln z_0^{(N)}$.
We will prove by induction on $n$ that the sequence $z^{(N)}_n$ is decreasing in $n$.
The inductive hypothesis is now $1\geq z_0^{(N)}\geq\ldots\geq z_n^{(N)}>0$, true for $n=0$.
If $n+1=N'$, since $z^{(N)}_{N'}=0$ there is nothing to prove.
We can then suppose $n+1<N'$.
We have
\begin{equation}
1-\mu_N\geq\frac{z_0^{(N)}-1}{\ln z_0^{(N)}}\geq\frac{z_n^{(N)}-1}{\ln z_n^{(N)}}\;,
\end{equation}
and hence $\left(1-\mu_N\right)\ln z_n^{(N)}\leq z_n^{(N)}-1$.
Since $z_{n-1}^{(N)}\geq z_n^{(N)}$, from \eqref{epniznrec} we have
\begin{equation}
\left(n+2\right)\left(z^{(N)}_{n+1}-z^{(N)}_n\right) = 1-z_n^{(N)}+\left(1-\mu_N\right)\ln z_n^{(N)}+n\ln\frac{z_n^{(N)}}{z_{n-1}^{(N)}}\leq0\;,
\end{equation}
and hence $z_{n+1}^{(N)}\leq z_n^{(N)}$.
Since $n+1< N'$, we also have $z_{n+1}^{(N)}>0$, and the claim is proven.

Finally, if $N'<N$, combining \eqref{epnimudecr} with \eqref{epnimuN} we get $N'\leq\mu-1\leq0$, in contradiction with $N'\geq1$.
We must then have $N'=N$.
\end{proof}
\end{lem}

\begin{lem}\label{epnilemzbar}
$\limsup_{N\to\infty}z^{(N)}_{\bar{n}}<1$, where $\bar{n}=\min\left\{n\in\mathbb{N}:n+2>e^S\right\}$ does not depend on $N$.
\begin{proof}
We recall that $z_n^{(N)}\leq1$ for any $n$ and $N$, hence $\limsup_{N\to\infty}z^{(N)}_{\bar{n}}\leq1$.
Let us suppose that $\limsup_{N\to\infty}z^{(N)}_{\bar{n}}=1$.
Then, there is a subsequence $\left\{N_k\right\}_{k\in\mathbb{N}}$ such that $\lim_{k\to\infty}z^{(N_k)}_{\bar{n}}=1$.
Since $z^{(N)}_n$ is decreasing in $n$ for any $N$, we also have
\begin{equation}\label{epnilimbar}
\lim_{k\to\infty}z^{(N_k)}_n=1\qquad \forall\;n=0,\ldots,\,\bar{n}\;.
\end{equation}
Let us define for any $N$ the probability distribution $q^{(N)}\in\mathcal{D}_{\bar{n}+1}$ as
\begin{equation}
q^{(N)}_n=\frac{p^{(N)}_n}{\sum_{k=0}^{\bar{n}+1}p^{(N)}_k}\;,\qquad n=0,\ldots,\,\bar{n}+1\;.
\end{equation}
From \eqref{epnilimbar} we get for any $n=0,\ldots,\,\bar{n}+1$
\begin{equation}
\lim_{k\to\infty}\frac{q^{(N_k)}_n}{q^{(N_k)}_0}=\lim_{k\to\infty}z^{(N_k)}_0\ldots z^{(N_k)}_{n-1}=1\;.
\end{equation}
For any $k$
\begin{equation}\label{epnisumq}
\sum_{n=0}^{\bar{n}+1}q^{(N_k)}_n=1\;.
\end{equation}
Dividing both members of \eqref{epnisumq} by $q^{(N_k)}_0$ and taking the limit $k\to\infty$ we get
\begin{equation}
\lim_{k\to\infty}q^{(N_k)}_0=\frac{1}{\bar{n}+2}\;,
\end{equation}
hence
\begin{equation}
\lim_{k\to\infty}q^{(N_k)}_n=\frac{1}{\bar{n}+2}\;,\qquad n=0,\ldots,\,\bar{n}+1\;,
\end{equation}
and
\begin{equation}
\lim_{k\to\infty}H\left(q^{(N_k)}\right)=\ln\left(\bar{n}+2\right)>S\;.
\end{equation}
However, for Lemma \ref{epnilemtrunc} we have $H\left(q^{(N)}\right)\leq H\left(p^{(N)}\right)=S$.
\end{proof}
\end{lem}
\begin{cor}\label{epnicorz}
There exists $0\leq\bar{z}<1$ (that does not depend on $N$) such that $z^{(N)}_{\bar{n}}\leq\bar{z}$ for any $N\geq\bar{n}$.
\end{cor}

\begin{lem}
The sequence $\left\{\mu_N\right\}_{N\in\mathbb{N}}$ is bounded.
\begin{proof}
An upper bound for $\mu_N$ is provided by \eqref{epnimudecr}.
Let us then prove a lower bound.

For any $N\geq\bar{n}+1$ we must have $z_{\bar{n}+1}^{(N)}\geq0$.
The recursive equation \eqref{epniznrec} for $n=\bar{n}$ gives
\begin{equation}\label{epniznbar}
0 \leq \left(\bar{n}+2\right)z_{\bar{n}+1}^{(N)}= \left(\bar{n}+1\right)z_{\bar{n}}^{(N)}+1+\left(1-\mu_N\right)\ln z_{\bar{n}}^{(N)}+\bar{n}\ln\frac{z_{\bar{n}}^{(N)}}{z_{\bar{n}-1}^{(N)}}\;.
\end{equation}
Since $z^{(N)}_n$ is decreasing in $n$, we have $z_{\bar{n}}^{(N)}\leq z_{\bar{n}-1}^{(N)}$.
Recalling from \eqref{epnimudecr} that $1-\mu_N\geq0$, and from Corollary \ref{epnicorz} that $z^{(N)}_{\bar{n}}\leq\bar{z}<1$, \eqref{epniznbar} implies
\begin{equation}
0 \leq \left(\bar{n}+1\right)z_{\bar{n}}^{(N)}+1+\left(1-\mu_N\right)\ln z_{\bar{n}}^{(N)}\leq\left(\bar{n}+1\right)\bar{z}+1+\left(1-\mu_N\right)\ln \bar{z}\;,
\end{equation}
hence $1-\mu_N\leq\left.-\left(\left(\bar{n}+1\right)\bar{z}+1\right)\right/\ln\bar{z}<\infty$.
\end{proof}
\end{lem}
The sequence $\left\{\mu_N\right\}_{N\in\mathbb{N}}$ has then a converging subsequence $\left\{\mu_{N_k}\right\}_{k\in\mathbb{N}}$ with
\begin{equation}
\lim_{k\to\infty}\mu_{N_k}=\mu\;.
\end{equation}

Since the sequences $\left\{z^{(N)}_0\right\}_{N\in\mathbb{N}}$ and $\left\{p^{(N)}_0\right\}_{N\in\mathbb{N}}$ are constrained between $0$ and $1$, we can also assume $\lim_{k\to\infty}z^{\left(N_k\right)}_0=z_0$ and $\lim_{k\to\infty}p^{\left(N_k\right)}_0=p_0$.
Taking the limit of \eqref{epnimudecr} we get
\begin{equation}\label{epnimulim0}
1-\mu\geq\frac{z_0-1}{\ln z_0}\geq0\;.
\end{equation}

\begin{lem}\label{epnilemzn}
$\lim_{k\to\infty}z^{\left(N_k\right)}_n=z_n$ for any $n\in\mathbb{N}$,
where the $z_n$ are either all $0$ or all strictly positive, and in the latter case they satisfy for any $n$ in $\mathbb{N}$ the recursive relation \eqref{epniznrec} with $\mu_N$ replaced by $\mu$:
\begin{equation}\label{epnizreclim}
\left(n+2\right)z_{n+1}= \left(n+2\right)z_n+1-z_n+\left(1-\mu\right)\ln z_n+n\ln\frac{z_n}{z_{n-1}}\;.
\end{equation}
\begin{proof}
If $z_0=0$, since $z^{(N)}_n$ is decreasing in $n$ we have for any $n$ in $\mathbb{N}$
\begin{equation}
\limsup_{k\to\infty}z^{(N_k)}_n\leq\limsup_{k\to\infty}z^{(N_k)}_0=z_0=0\;,
\end{equation}
hence $\lim_{k\to\infty}z^{(N_k)}_n=0$.

Let us now suppose $z_0>0$, and proceed by induction on $n$.
For the inductive hypothesis, we can suppose
\begin{equation}
z_0=\lim_{k\to\infty}z^{\left(N_k\right)}_0\geq\ldots\geq\lim_{k\to\infty}z^{\left(N_k\right)}_n=z_n>0\;.
\end{equation}
Then, taking the limit in \eqref{epniznrec} we get
\begin{equation}
z_{n+1} = \lim_{k\to\infty}z^{\left(N_k\right)}_{n+1}=z_n+\frac{1-z_n+\left(1-\mu\right)\ln z_n+n\ln\frac{z_n}{z_{n-1}}}{n+2}\;.
\end{equation}
If $z_{n+1}>0$, the claim is proven.
Let us then suppose $z_{n+1}=0$.
From \eqref{epniznrec} we get then
\begin{equation}
0\leq\lim_{k\to\infty}z^{\left(N_k\right)}_{n+2}=\frac{1+\left(n+2-\mu\right)\ln0-\left(n+1\right)\ln z_n}{n+3}\;,
\end{equation}
that implies $\mu\geq n+2\geq2$.
However, \eqref{epnimulim0} implies $\mu\leq1$.
\end{proof}
\end{lem}

\begin{lem}\label{epnizlim}
There exists $0\leq z<1$ such that $z_n=\lim_{k\to\infty}z^{\left(N_k\right)}_n=z$ for any $n\in\mathbb{N}$.
\begin{proof}
If $z_0=0$, Lemma \ref{epnilemzn} implies the claim with $z=0$.
Let us then suppose $z_0>0$.

If $z_0=1$, with \eqref{epnizreclim} it is easy to prove that $z_n=1$ for any $n\in\mathbb{N}$.
However, from Lemma \ref{epnilemzn} and Corollary \ref{epnicorz} we must have $z_{\bar{n}}=\lim_{k\to\infty}z^{\left(N_k\right)}_{\bar{n}}\leq\bar{z}<1$.
Then, it must be $0<z_0<1$.

Since the sequence $\left\{z^{(N)}_n\right\}_{n\in\mathbb{N}}$ is decreasing for any $N$, also the sequence $\left\{z_n\right\}_{n\in\mathbb{N}}$ is decreasing.
Since it is also positive, it has a limit $\lim_{n\to\infty}z_n=\inf_{n\in\mathbb{N}}z_n=z$, that satisfies $0\leq z\leq z_0<1$.
Since $z_n\leq z_{n-1}\leq z_0<1$, \eqref{epnizreclim} implies
\begin{equation}
\left(n+2\right)\left(z_n-z_{n+1}\right)+1-z_n+\left(1-\mu\right)\ln z_n\geq0\;,
\end{equation}
hence
\begin{equation}\label{epnizdecrineq}
1-\mu\leq\frac{\left(n+2\right)\left(z_n-z_{n+1}\right)}{-\ln z_n}+\frac{z_n-1}{\ln z_n}\;.
\end{equation}
The sequence $\left\{z_n-z_{n+1}\right\}_{n\in\mathbb{N}}$ is positive and satisfies $\sum_{n=0}^\infty\left(z_n-z_{n+1}\right)=z_0-z<\infty$.
Then, for Lemma \ref{epnilemnx} $\liminf_{n\to\infty}\left(n+2\right)\left(z_n-z_{n+1}\right)=0$, and since $-\ln z_n\geq-\ln z_0>0$, also
\begin{equation}
\liminf_{n\to\infty}\frac{\left(n+2\right)\left(z_n-z_{n+1}\right)}{-\ln z_n}=0\;.
\end{equation}
Then, taking the $\liminf$ of \eqref{epnizdecrineq} we get $1-\mu\leq\left.\left(z-1\right)\right/\ln z$.
Combining with \eqref{epnimulim0} and recalling that $z\leq z_0$ we get
\begin{equation}\label{epnimulim}
\frac{z-1}{\ln z}\leq\frac{z_0-1}{\ln z_0}\leq1-\mu\leq\frac{z-1}{\ln z}\;,
\end{equation}
that implies $z=z_0$.
Since $z_n$ is decreasing and $z=\inf_{n\in\mathbb{N}}z_n$, we have $z_0=z\leq z_n\leq z_0$ for any $n$, hence $z_n=z$.
\end{proof}
\end{lem}

\begin{lem}\label{epnilemp0lim}
$\lim_{k\to\infty}p^{\left(N_k\right)}_n=p_0\,z^n$ for any $n\in\mathbb{N}$.
\begin{proof}
The claim is true for $n=0$.
The inductive hypothesis is $\lim_{k\to\infty}p^{\left(N_k\right)}_{n'}=p_0\,z^{n'}$ for $n'=0,\ldots,\,n$.
We then have $\lim_{k\to\infty}p^{\left(N_k\right)}_{n+1}=\lim_{k\to\infty}p^{\left(N_k\right)}_n\,z^{\left(N_k\right)}_n=p_0\,z^{n+1}$,
where we have used the inductive hypothesis and Lemma \ref{epnizlim}.
\end{proof}
\end{lem}
\begin{lem}\label{epnilimp}
$p_0=1-z$, hence $\lim_{k\to\infty}p^{\left(N_k\right)}_n=\left(1-z\right)z^n$ for any $n\in\mathbb{N}$.
\begin{proof}
We have $\sum_{n=0}^N p^{(N)}_n=1$ for any $N\in\mathbb{N}$.
Moreover, since $z^{(N)}_n$ is decreasing in $n$, we also have
\begin{equation}
p^{(N)}_n=p^{(N)}_0\,z^{(N)}_0\ldots\,z^{(N)}_{n-1}\leq p^{(N)}_0\left(z^{(N)}_0\right)^n\;.
\end{equation}
Since $\lim_{k\to\infty}z_0^{\left(N_k\right)}=z<1$, for sufficiently large $k$ we have $z_0^{\left(N_k\right)}\leq\left(1+z\right)/2$, and since $p^{(N)}_0\leq1$,
\begin{equation}\label{epniboundp}
p^{\left(N_k\right)}_n\leq\left(\frac{1+z}{2}\right)^n\;.
\end{equation}
The sums $\sum_{n=0}^N p^{\left(N_k\right)}_n$ are then dominated for any $N$ in $\mathbb{N}$ by $\sum_{n=0}^\infty \left(\frac{1+z}{2}\right)^n<\infty$, and for the dominated convergence theorem we have
\begin{equation}
1 = \lim_{k\to\infty}\sum_{n=0}^{N_k}p^{\left(N_k\right)}_n=\sum_{n=0}^\infty\lim_{k\to\infty}p^{\left(N_k\right)}_n=p_0\sum_{n=0}^\infty z^n=\frac{p_0}{1-z}\;,
\end{equation}
where we have used Lemma \ref{epnilemp0lim}.
\end{proof}
\end{lem}

\begin{lem}\label{epniz(S)}
$z=g^{-1}(S)\left/\left(g^{-1}(S)+1\right)\right.$.
\begin{proof}
The function $-x\ln x$ is increasing for $0\leq x\leq1/e$.
Let us choose $n_0$ such that $\left(\left.\left(1+z\right)\right/2\right)^{n_0}\leq1/e$.
Recalling \eqref{epniboundp}, the sums $-\sum_{n=n_0}^N p^{(N)}_n\ln p^{(N)}_n$ are dominated for any $N$ in $\mathbb{N}$ by $-\sum_{n=n_0}^\infty n\left(\frac{1+z}{2}\right)^n\ln\frac{1+z}{2}<\infty$.
For any $N$ we have $S=-\sum_{n=0}^N p^{(N)}_n\ln p^{(N)}_n$.
Then, for the dominated convergence theorem and Lemma \ref{epnilimp} we have
\begin{equation}\label{epniSz}
S = -\sum_{n=0}^\infty\lim_{k\to\infty}p^{(N)}_n\ln p^{(N)}_n= -\sum_{n=0}^\infty\left(1-z\right)z^n\left(\ln\left(1-z\right)+n\ln z\right)=g\left(\frac{z}{1-z}\right)\;,
\end{equation}
where we have used the definition of $g$ \eqref{epnidefg}.
Finally, the claim follows solving \eqref{epniSz} with respect to $z$.
\end{proof}
\end{lem}

It is convenient to rewrite $F_{N_k}=F\left(p^{(N_k)}\right)$ as
\begin{equation}\label{epniFN}
F_{N_k}=-\sum_{n=0}^{N_k-1}\left(n+1\right)p^{(N_k)}_n\,z^{(N_k)}_n\ln z^{(N_k)}_n\;.
\end{equation}
Since $z^{(N_k)}_n\leq1$, each term of the sum is positive.
Since $-x\ln x\leq1/e$ for $0\leq x\leq1$, and recalling \eqref{epniboundp}, the sum is dominated by $\sum_{n=0}^\infty \frac{n+1}{e}\left(\frac{1+z}{2}\right)^n<\infty$.
We then have for the dominated convergence theorem, recalling Lemmata \ref{epnilimp} and \ref{epnizlim},
\begin{eqnarray}
\lim_{k\to\infty}F_{N_k} &=& -\sum_{n=0}^\infty\left(n+1\right)\lim_{k\to\infty}p^{(N_k)}_n\,z^{(N_k)}_n\ln z^{(N_k)}_n=-\sum_{n=0}^\infty\left(n+1\right)\left(1-z\right)z^{n+1}\ln z=\nonumber\\
&=&\frac{z\ln z}{z-1}=g^{-1}(S)\ln\left(1+\frac{1}{g^{-1}(S)}\right)=-f(S)\;,
\end{eqnarray}
where we have used Lemma \ref{epniz(S)} and the definitions of $f$ \eqref{epnideff} and $g$ \eqref{epnidefg}.

\section{The thinning}\label{epnisecthinning}
The thinning \cite{renyi1956characterization} is the map acting on classical probability distributions on the set of natural numbers that is the discrete analogue of the continuous rescaling operation on positive real numbers.
We have introduced it in Section \ref{secthinning}.
Thanks to Theorem \ref{thinatt}, our main results Theorems \ref{epnithmmain} and \ref{epnithmiso} apply also to the thinning:
\begin{thm}\label{epnithmthin}
For any probability distribution $p$ on $\mathbb{N}$ and any $0\leq\lambda\leq 1$ we have
\begin{equation}
H\left(T_\lambda(p)\right)\geq g\left(\lambda\;g^{-1}\left(H(p)\right)\right)\;,
\end{equation}
i.e. geometric input probability distributions minimize the output Shannon entropy of the thinning for fixed input entropy.
\end{thm}
\begin{thm}
For any probability distribution $p$ on $\mathbb{N}$
\begin{equation}
\left.\frac{d}{dt}H\left(T_{e^{-t}}(p)\right)\right|_{t=0}\geq f\left(H(p)\right)\;.
\end{equation}
\end{thm}

\section{Proof of Lemma \ref{epnifinites}}\label{epnifinitesproof}
Let $\hat{\rho}$ be a passive state.
If $S\left(\Phi_\lambda\left(\hat{\rho}\right)\right)=\infty$, there is nothing to prove.
We can then suppose $S\left(\Phi_\lambda\left(\hat{\rho}\right)\right)<\infty$.

We can associate to $\hat{\rho}$ the probability distribution $p$ on $\mathbb{N}$ such that
\begin{equation}
\hat{\rho}=\sum_{n=0}^\infty p_n\;|n\rangle\langle n|\;,
\end{equation}
satisfying $-\sum_{n=0}^\infty p_n\ln p_n=S\left(\hat{\rho}\right)$.
Let us define for any $N\in\mathbb{N}$ the quantum state
\begin{equation}
\hat{\rho}_N=\sum_{n=0}^N \frac{p_n}{s_N}\;|n\rangle\langle n|\;,
\end{equation}
where $s_N=\sum_{n=0}^N p_n$.
We have
\begin{equation}
\left\|\hat{\rho}_N-\hat{\rho}\right\|_1 = \frac{1-s_N}{s_N}\sum_{n=0}^Np_n+\sum_{n=N+1}^\infty p_n\;,
\end{equation}
where $\left\|\cdot\right\|_1$ denotes the trace norm \cite{wilde2013quantum,holevo2013quantum}.
Since $\lim_{N\to\infty}s_N=1$ and $\sum_{n=0}^\infty p_n=1$, we have $\lim_{N\to\infty}\left\|\hat{\rho}_N-\hat{\rho}\right\|_1 =0$.
Since $\Phi_\lambda$ is continuous in the trace norm, we also have
\begin{equation}\label{epnilimPhiN}
\lim_{N\to\infty}\left\|\Phi_\lambda\left(\hat{\rho}_N\right)-\Phi_\lambda\left(\hat{\rho}\right)\right\|_1=0\;.
\end{equation}
Moreover,
\begin{equation}\label{epnilimSN}
\lim_{N\to\infty}S\left(\hat{\rho}_N\right)=\lim_{N\to\infty}\left(\ln s_N-\sum_{n=0}^N \frac{p_n}{s_N}\ln p_n\right)=S\left(\hat{\rho}\right)\;.
\end{equation}
Notice that \eqref{epnilimSN} holds also if $S\left(\hat{\rho}\right)=\infty$.

Let us now define the probability distribution $q$ on $\mathbb{N}$ as
\begin{equation}
\Phi_\lambda\left(\hat{\rho}\right)=\sum_{n=0}^\infty q_n\;|n\rangle\langle n|\;,
\end{equation}
satisfying
\begin{equation}\label{epniSq}
S\left(\Phi_\lambda\left(\hat{\rho}\right)\right)=-\sum_{n=0}^\infty q_n\ln q_n\;.
\end{equation}
From Equation \eqref{kraus}, the channel $\Phi_\lambda$ sends the set of states supported on the span of the first $N+1$ Fock states into itself.
Then, for any $N\in\mathbb{N}$ there is a probability distribution $q^{(N)}$ on $\left\{0,\ldots,\,N\right\}$ such that
\begin{equation}
\Phi_\lambda\left(\hat{\rho}_N\right)=\sum_{n=0}^N q_n^{(N)}\;|n\rangle\langle n|\;.
\end{equation}
From \eqref{epnilimPhiN} we get for any $n\in\mathbb{N}$
\begin{equation}\label{epnilimqN}
\lim_{N\to\infty}q^{(N)}_n=q_n\;.
\end{equation}
Since $\Phi_\lambda$ is trace preserving, we have $\sum_{n=0}^\infty q_n=1$, hence $\lim_{n\to\infty} q_n=0$.
Then, there is $n_0\in\mathbb{N}$ (that does not depend on $N$) such that for any $n\geq n_0$ we have $q_n\leq p_0/e$.
Since $s_N\;\hat{\rho}_N\leq\hat{\rho}$ and the channel $\Phi_\lambda$ is positive, we have $s_N\;\Phi_\lambda\left(\hat{\rho}_N\right)\leq\Phi_\lambda\left(\hat{\rho}\right)$.
Then, for any $n\geq n_0$
\begin{equation}
q_n^{(N)}\leq\frac{q_n}{s_N}\leq\frac{q_n}{p_0}\leq\frac{1}{e}\;,
\end{equation}
where we have used that $s_N\geq p_0>0$.
Since the function $-x\ln x$ is increasing for $0\leq x\leq1/e$, the sums $-\sum_{n=n_0}^N q^{(N)}_n\ln q^{(N)}_n$ are dominated by
\begin{equation}
\sum_{n=n_0}^\infty \frac{q_n\ln p_0-q_n\ln q_n}{p_0}\leq\frac{\ln p_0+S\left(\Phi_\lambda\left(\hat{\rho}\right)\right)}{p_0}<\infty\;,
\end{equation}
where we have used \eqref{epniSq}.
Then, for the dominated convergence theorem we have
\begin{equation}
\lim_{N\to\infty}S\left(\Phi_\lambda\left(\hat{\rho}_N\right)\right) = -\lim_{N\to\infty}\sum_{n=0}^N q^{(N)}_n\ln q^{(N)}_n= -\sum_{n=0}^\infty\lim_{N\to\infty}q^{(N)}_n\ln q^{(N)}_n= S\left(\Phi_\lambda\left(\hat{\rho}\right)\right)\;,
\end{equation}
where we have also used \eqref{epnilimqN}.

If Theorem \ref{epnithmmain} holds for passive states with finite support, for any $N$ in $\mathbb{N}$ we have
\begin{equation}
S\left(\Phi_\lambda\left(\hat{\rho}_N\right)\right)\geq g\left(\lambda\;g^{-1}\left(S\left(\hat{\rho}_N\right)\right)\right)\;.
\end{equation}
Then, the claim follows taking the limit $N\to\infty$.

\section{Auxiliary Lemmata}\label{epniauxlemmata}
\begin{lem}\label{epnilemtrunc}
Let us choose a probability distribution $p\in\mathcal{D}_N$, fix $0\leq N'\leq N$, and define the probability distribution $q\in\mathcal{D}_{N'}$ as
\begin{equation}
q_n=\frac{p_n}{\sum_{k=0}^{N'}p_k}\;,\qquad n=0,\ldots,\,N'\;.
\end{equation}
Then, $H(q)\leq H(p)$.
\begin{proof}
We have for any $n=0,\ldots,\,N'$
\begin{equation}
\sum_{k=0}^n q_k=\frac{\sum_{k=0}^n p_k}{\sum_{l=0}^{N'}p_l}\geq\sum_{k=0}^n p_k\;,
\end{equation}
Then, $q\succ p$ and the claim follows from Remark \ref{majS}.
\end{proof}
\end{lem}

\begin{lem}\label{epnilemnx}
Let $\left\{x_n\right\}_{n\in\mathbb{N}}$ be a positive sequence with finite sum.
Then
\begin{equation}
\liminf_{n\to\infty}n\,x_n=0\;.
\end{equation}
\begin{proof}
Let us suppose
\begin{equation}
\liminf_{n\to\infty}n\,x_n=c>0\;.
\end{equation}
Then, there exists $n_0\in\mathbb{N}$ such that $n\,x_n\geq c/2$ for any $n\geq n_0$.
Then,
\begin{equation}
\sum_{n=0}^\infty x_n\geq\sum_{n=n_0}^\infty\frac{c}{2n}=\infty\;,
\end{equation}
contradicting the hypothesis.
\end{proof}
\end{lem}

\section{Conclusion}\label{epnisecconcl}
We have proved that Gaussian thermal input states minimize the output von Neumann entropy of the Gaussian quantum-limited attenuator for fixed input entropy (Theorem \ref{epnithmmain}).
The proof is based on a new isoperimetric inequality (Theorem \ref{epnithmiso}).
Theorem \ref{epnithmmain} implies that geometric input probability distributions minimize the output Shannon entropy of the thinning for fixed input entropy (Theorem \ref{epnithmthin}).
The multimode generalization of the isoperimetric inequality \eqref{epnilogs} would prove Theorem \ref{epnithmmain} in the multimode scenario.
This multimode extension permits to determine both the triple trade-off region of the Gaussian quantum-limited attenuator \cite{wilde2012public,wilde2012information,wilde2012quantum} and the classical capacity region of the Gaussian quantum degraded broadcast channel \cite{guha2007classicalproc,guha2007classical}.

\chapter{Lossy channels}\label{chlossy}
In this Chapter we extend the majorization results of Chapter \ref{majorization} to a wide class of quantum lossy channels, emerging from a weak interaction of a small quantum system with a large bath in its ground state.

The Chapter is based on
\begin{enumerate}
\item[\cite{de2016passive}] G.~De~Palma, A.~Mari, S.~Lloyd, and V.~Giovannetti, ``Passive states as optimal
  inputs for single-jump lossy quantum channels,'' \emph{Physical Review A},
  vol.~93, no.~6, p. 062328, 2016.\\ {\small\url{http://journals.aps.org/pra/abstract/10.1103/PhysRevA.93.062328}}
\end{enumerate}
\section{Introduction}
The passive states \cite{pusz1978passive,lenard1978thermodynamical} of a quantum system are the states diagonal in the eigenbasis of the Hamiltonian, with eigenvalues decreasing as the energy increases.
They minimize the average energy among all the states with a given spectrum, and hence no work can be extracted from them on average with unitary operations \cite{janzing2006computational}. For this reason they play a key role in the recently emerging field of quantum thermodynamics (see \cite{vinjanampathy2015quantum,goold2015role} for a review).

Majorization (see \cite{marshall2010inequalities} and Section \ref{secmaj}) is the order relation between quantum states induced by random unitary operations, i.e. a state $\hat{\sigma}$ is majorized by a state $\hat{\rho}$ iff $\hat{\sigma}$ can be obtained applying random unitaries to $\hat{\rho}$.
Majorization theory is ubiquitous in quantum information.
Its very definition suggests applications in quantum thermodynamics \cite{goold2015role,gour2015resource,horodecki2013fundamental}, where the goal is determining the set of final states that can be obtained from a given initial state with a given set of operations.
In the context of quantum entanglement, it also determines whether it is possible to convert a given bipartite pure state into another given pure state by means of local operations and classical communication \cite{nielsen1999conditions,nielsen2001majorization}.
Majorization has proven to be crucial in the longstanding problem of the determination of the classical communication capacity of quantum gauge-covariant bosonic Gaussian channels \cite{giovannetti2014ultimate}, and the consequent proof of the optimality of Gaussian states for the information encoding.
Indeed, a turning point has been the proof of a majorization property: the output of any of these channels generated by any input state is majorized by the output generated by the vacuum \cite{mari2014quantum,giovannetti2015majorization} (see also \cite{holevo2015gaussian} for a review).
In Chapter \ref{majorization} this fundamental result has been extended and linked to the notion of passive states (see also \cite{de2015passive}). We proved that these states optimize the output of any one-mode quantum Gaussian channel, in the sense that the output generated by a passive state majorizes the output generated by any other state with the same spectrum. Moreover, the same channels preserve the majorization relation when applied to passive states \cite{jabbour2015majorization}.

Here we extend the result of Chapter \ref{majorization} to a large class of lossy quantum channels.
Lossy quantum channels arise from a weak interaction of the quantum system of interest with a large Markovian bath in its zero-temperature (i.e. ground) state.
We prove that passive states are the optimal inputs of these channels.
Indeed, we prove that the output $\Phi\left(\hat{\rho}\right)$ generated by any input state $\hat{\rho}$ majorizes the output $\Phi\left(\hat{\rho}^\downarrow\right)$ generated by the passive input state $\hat{\rho}^\downarrow$ with the same spectrum of $\hat{\rho}$.
Then, $\Phi\left(\hat{\rho}\right)$ can be obtained applying a random unitary operation to $\Phi\left(\hat{\rho}^\downarrow\right)$, and it is more noisy than $\Phi\left(\hat{\rho}^\downarrow\right)$.
Moreover, $\Phi\left(\hat{\rho}^\downarrow\right)$ is still passive, i.e. the channel maps passive states into passive states.

In the context of quantum thermodynamics, this result puts strong constraints on the possible spectrum of the output of lossy channels.
It can then be useful to determine which output states can be obtained from an input state with a given spectrum in a resource theory with the lossy channel among the allowed operations.
The Gaussian analogue of this result has been crucial for proving that Gaussian input states minimize the output entropy of the one-mode Gaussian quantum attenuator for fixed input entropy (see Chapter \ref{chepni} and \cite{de2016gaussian}).
The result of this Chapter can find applications in the proof of similar entropic inequalities on the output states of lossy channels in the same spirit of the quantum Entropy Power Inequalities of \cite{konig2014entropy,de2014generalization,de2015multimode,audenaert2015entropy}, and then determine their classical capacity.

Our result applies to all the interactions of a quantum system with a heat bath such that the reduced system dynamics can be modeled by a master equation \cite{schaller2014open,breuer2007theory} and the following hypotheses are satisfied:
\begin{enumerate}
\item The Hamiltonian of the system is nondegenerate.
\item The system-bath interaction Hamiltonian couples only consecutive eigenstates of the Hamiltonian of the system alone.
\item If the system starts in its maximally mixed state, its reduced state remains passive.
\item The bath starts in its ground (i.e. zero temperature) state.
\end{enumerate}
The first assumption is satisfied by a large class of quantum systems, and it is usually taken for granted in both quantum thermodynamics and quantum statistical mechanics \cite{gogolin2015equilibration}.
The second assumption is also satisfied by a large class of quantum systems.
The third assumption means that the interaction cannot generate population inversion if the system is initialized in the infinite-temperature state, as it is for most physical systems.
The fourth assumption is for example satisfied by the interaction of a quantum system with an optical bath at room temperature.
Indeed, $\hbar\omega\gg k_BT$ for $\omega$ in the optical range and $T\approx300^\circ K$, hence the state of the bath at room temperature is indistinguishable from the vacuum.

These assumptions turn out to be necessary.
Indeed, dropping any of them it is possible to find explicit  counterexamples for which passive inputs are not optimal choices for output majorization.

The Chapter is organized as follows.
The main result is presented in Sec.~\ref{Lsecopt} where we first define in a rigorous way the class of lossy maps we are interested in and then proceed with a formal proof the optimality for passive states.
Section~\ref{Lseccount} is instead devoted to counterexamples.
In particular in Sec.~\ref{Lseccatt} we show that for the two-mode bosonic Gaussian quantum-limited attenuator, whose associated Hamiltonian is degenerate, no majorization relations can be ascribed to the passive states.
In Sec.\ref{Lsec2step} instead a counterexample is provided for a two-qubit lossy map with two different choices of the Hamiltonian.
In the first case the Hamiltonian is nondegenerate, but the process involves quantum jumps of more than one energy step.
In the second case only quantum jumps of one energy step are allowed, but the Hamiltonian becomes degenerate.
In Sec.~\ref{Lsecfin} we analyze the case of a map where the bath temperature is not zero.
We show that the optimal input states are a pure coherent superposition of the Hamiltonian eigenstates, hence non passive. Conclusions and comments are presented in Sec.~\ref{Lseccon} while technical derivations are presented in the appendices.

\section{Passive states}\label{Lrearrangement}
We consider a $d$-dimensional quantum system with nondegenerate Hamiltonian
\begin{equation}\label{LHpass}
\hat{H}=\sum_{i=1}^{d} E_i\;|i\rangle\langle i|\;,\qquad \langle i|j\rangle=\delta_{ij}\;,\qquad E_1<\ldots<E_{d}\;.
\end{equation}
A self-adjoint operator is \emph{passive} \cite{pusz1978passive,lenard1978thermodynamical} if it is diagonal in the eigenbasis of the Hamiltonian and its eigenvalues decrease as the energy increases.
\begin{defn}[Passive rearrangement]
Let $\hat{X}$ be a self-adjoint operator with eigenvalues $x_1\geq\ldots\geq x_d$.
As we did with Definition \ref{defrearr} for quantum Gaussian systems, we define its passive rearrangement as
\begin{equation}
\hat{X}^\downarrow:=\sum_{i=1}^d x_i\;|i\rangle\langle i|\;,
\end{equation}
where $\left\{|i\rangle\right\}_{i=1,\ldots,n}$ is the eigenbasis of the Hamiltonian \eqref{LHpass}.
Of course, $\hat{X}=\hat{X}^\downarrow$ for any passive operator.
\end{defn}
\begin{rem}
The passive rearrangement of any rank-$n$ projector $\hat{\Pi}_n$ is the projector onto the first $n$ energy eigenstates:
\begin{equation}\label{LPin*}
\hat{\Pi}_n^\downarrow=\sum_{i=1}^n|i\rangle\langle i|\;.
\end{equation}
\end{rem}
\begin{rem}
It is easy to show that passive quantum states minimize the average energy among all the states with a given spectrum, i.e.
\begin{equation}
\mathrm{Tr}\left[\hat{H}\;\hat{U}\;\hat{\rho}\;\hat{U}^\dag\right]\geq\mathrm{Tr}\left[\hat{H}\;\hat{\rho}^\downarrow\right]\qquad\forall\;\hat{U}\;\text{unitary}\;.
\end{equation}
\end{rem}

\section{Optimality of passive states for lossy channels} \label{Lsecopt}
The most general master equation that induces a completely positive Markovian dynamics is \cite{breuer2007theory,schaller2014open}
\begin{equation}
\frac{d}{dt}\hat{\rho}(t)=\mathcal{L}\left(\hat{\rho}(t)\right)\;,
\end{equation}
where the generator $\mathcal{L}$ has the Lindblad form
\begin{equation}\label{LLdef}
\mathcal{L}\left(\hat{\rho}\right)=-i\left[\hat{H}_{LS},\;\hat{\rho}\right]+\sum_{\alpha=1}^{\alpha_0}\left(\hat{L}_\alpha\;\hat{\rho}\;\hat{L}_\alpha^\dag-\frac{1}{2}\left\{\hat{L}_\alpha^\dag\hat{L}_\alpha,\;\hat{\rho}\right\}\right)\;,
\end{equation}
where $\alpha_0\in\mathbb{N}$.
This dynamics arises from a weak interaction with a large Markovian bath in the rotating-wave approximation \cite{breuer2007theory,schaller2014open}.
In this case, $\hat{H}_{LS}$ commutes with the Hamiltonian $\hat{H}$, i.e. $\hat{H}_{LS}$ only shifts the energies of $\hat{H}$:
\begin{equation}\label{LHLS}
\hat{H}_{LS}=\sum_{i=1}^d \delta E_i\;|i\rangle\langle i|\;.
\end{equation}
As anticipated in the introduction, we suppose that the bath starts in its ground state and that the interaction Hamiltonian $\hat{V}_{SB}$ couples only neighbouring energy levels of the system:
\begin{equation}
\hat{V}_{SB}=\sum_{i=1}^d|i\rangle_S\langle i|\otimes \hat{V}^B_i+\sum_{i=1}^{d-1}\left(|i\rangle_S\langle i+1|\otimes \hat{W}^B_i+\text{h.c.}\right)\;.
\end{equation}
Here the $\hat{V}_i^B$ are generic self-adjoint operators, while the $\hat{W}_i^B$ are completely generic operators.
In the rotating-wave approximation only the transitions that conserve the energy associated to the noninteracting Hamiltonian are allowed.
If the bath is in its ground state, it cannot transfer energy to the system, and only the transitions that decrease its energy are possible.
Then, each Lindblad operator $\hat{L}_\alpha$ can induce either dephasing in the energy eigenbasis:
\begin{equation}\label{Ldephase}
\hat{L}_\alpha=\sum_{i=1}^d a_i^\alpha\;|i\rangle\langle i|\;,\qquad a_i^\alpha\in\mathbb{C}\;,\qquad\alpha=1,\ldots,\,\alpha_0\;,
\end{equation}
or decay toward the ground state with quantum jumps of one energy level:
\begin{equation}\label{Ljump}
\hat{L}_\alpha=\sum_{i=1}^{d-1} b_i^\alpha\;|i\rangle\langle i+1|\;,\qquad b_i^\alpha\in\mathbb{C}\;,\qquad\alpha=1,\ldots,\,\alpha_0\;.
\end{equation}

It is easy to show that, if $\hat{\rho}$ is diagonal in the energy eigenbasis, also $\mathcal{L}\left(\hat{\rho}\right)$ is diagonal in the same basis, hence $e^{t\mathcal{L}}\left(\hat{\rho}\right)$ remains diagonal for any $t$.

As anticipated in the introduction, we also suppose that the quantum channel $e^{t\mathcal{L}}\left(\hat{\rho}\right)$ sends the maximally mixed state into a passive state.
As a consequence, the generator $\mathcal{L}$ maps the identity into a passive operator (see Section \ref{Lpasst}).

To see explicitly how this last condition translates on the coefficients $b_i^\alpha$, we compute
\begin{equation}
\mathcal{L}\left(\hat{\mathbb{I}}\right)=\sum_{i=1}^d\left(\sum_\alpha\left(\left|b_i^\alpha\right|^2-\left|b_{i-1}^\alpha\right|^2\right)\right)|i\rangle\langle i|\;,
\end{equation}
where for simplicity we have set $b_0^\alpha=b_d^\alpha=0$, and the operator is passive iff the function
\begin{equation}\label{Lr_i}
r_i:=\sum_\alpha\left|b_i^\alpha\right|^2\;,\qquad i=0,\,\ldots,\,d
\end{equation}
is concave in $i$.

The main result of this Chapter is that passive states optimize the output of the quantum channel generated by any dissipator of the form \eqref{LLdef} satisfying \eqref{LHLS} and with Lindblad operators of the form \eqref{Ldephase} or \eqref{Ljump} such that the function \eqref{Lr_i} is concave.
We will prove that the output $e^{t\mathcal{L}}\left(\hat{\rho}\right)$ generated by any input state $\hat{\rho}$ majorizes the output $e^{t\mathcal{L}}\left(\hat{\rho}^\downarrow\right)$ generated by the passive state $\hat{\rho}^\downarrow$ with the same spectrum of $\hat{\rho}$, i.e. for any $t\geq0$
\begin{equation}\label{Loptimaldef}
e^{t\mathcal{L}}\left(\hat{\rho}\right)\prec e^{t\mathcal{L}}\left(\hat{\rho}^\downarrow\right)\;.
\end{equation}
Moreover, for any $t\geq0$ the state $e^{t\mathcal{L}}\left(\hat{\rho}^\downarrow\right)$ is still passive, i.e. the quantum channel $e^{t\mathcal{L}}$ preserves the set of passive states.
The proof closely follows \cite{de2015passive} and Chapter \ref{majorization}, and it is contained in the next section.

\subsection{Proof of the main result}\label{Lmainproof}
Let us define
\begin{equation}
\hat{\rho}(t)=e^{t\mathcal{L}}\left(\hat{\rho}\right)\;.
\end{equation}
The quantum states with nondegenerate spectrum are dense in the set of all quantum states.
Besides, the spectrum is a continuous function of the operator, and any linear map is continuous.
Then, without loss of generality we can suppose that $\hat{\rho}$ has nondegenerate spectrum.
Let $p_1(t)\geq\ldots\geq p_{d}(t)$ be the eigenvalues of $\hat{\rho}(t)$, and let
\begin{equation}
s_n(t)=\sum_{i=1}^n p_i(t)\;,\qquad n=1,\ldots,\,d\;.
\end{equation}
Let instead
\begin{equation}
p_i^\downarrow(t)=\langle i|e^{t\mathcal{L}}\left(\hat{\rho}^\downarrow\right)|i\rangle\;,\qquad i=1,\,\ldots,\,d
\end{equation}
be the eigenvalues of $e^{t\mathcal{L}}\left(\hat{\rho}^\downarrow\right)$, and
\begin{equation}
s_n^\downarrow(t)=\sum_{i=1}^n p_i^\downarrow(t)\;,\qquad n=1,\,\ldots,\,d\;.
\end{equation}
We notice that $p(0)=p^\downarrow(0)$ and then $s(0)=s^\downarrow(0)$, where
\begin{equation}
p(t)=\left(p_1(t),\ldots,p_d(t)\right)\;,
\end{equation}
and similarly for $s(t)$.
The proof comes from:
\begin{lem}\label{Ldeg}
The spectrum of $\hat{\rho}(t)$ can be degenerate at most in isolated points.
\begin{proof}
See Section \ref{Lproofdeg}.
\end{proof}
\end{lem}
\begin{lem}\label{Llemma1}
$s(t)$ is continuous in $t$, and for any $t\geq0$ such that $\hat{\rho}(t)$ has nondegenerate spectrum it satisfies
\begin{equation}\label{Lsdot}
\frac{d}{dt}s_n(t)\leq\lambda_n(s_{n+1}(t)-s_n(t))\;,\qquad n=1,\,\ldots,\,d\;,
\end{equation}
where
\begin{equation}
\lambda_n=\mathrm{Tr}\left[\hat{\Pi}_n^\downarrow\;\mathcal{L}\left(\hat{\mathbb{I}}\right)\right]\geq0\;.
\end{equation}
\begin{proof}
See Section \ref{Lprooflemma1}.
\end{proof}
\end{lem}
\begin{lem}\label{Llemma2}
If $s(t)$ is continuous in $t$ and satisfies \eqref{Lsdot}, then $s_n(t)\leq s_n^\downarrow(t)$ for any $t\geq0$ and $n=1,\,\ldots,\,d$.
\begin{proof}
See Section \ref{Lprooflemma2}.
\end{proof}
\end{lem}
Lemma \ref{Llemma2} implies that for any $t\geq0$ the quantum channel $e^{t\mathcal{L}}$ preserves the set of passive states.
Indeed, let us choose the initial state $\hat{\rho}$ already passive.
Then, $s_n(t)$ is the sum of the $n$ largest eigenvalues of $e^{t\mathcal{L}}\left(\hat{\rho}\right)$.
Recalling that $e^{t\mathcal{L}}\left(\hat{\rho}\right)$ is diagonal in the Hamiltonian eigenbasis, $s_n^\downarrow(t)$ is the sum of the eigenvalues corresponding to the first $n$ eigenstates of the Hamiltonian $|1\rangle,\;\ldots,\;|n\rangle$, so that $s_n^\downarrow(t)\leq s_n(t)$.
However, Lemma \ref{Llemma2} implies $s_n(t)=s_n^\downarrow(t)$ for $n=1,\,\ldots,\,d$, then $p_n(t)=p_n^\downarrow(t)$ and $e^{t\mathcal{L}}\left(\hat{\rho}\right)$ preserves the set of passive states for any $t$.

Then, for the definition of majorization and Lemma \ref{Llemma2} again,
\begin{equation}
e^{t\mathcal{L}}\left(\hat{\rho}\right)\prec e^{t\mathcal{L}}\left(\hat{\rho}^\downarrow\right)
\end{equation}
for any $\hat{\rho}$, and the passive states are the optimal inputs for the channel.

\section{Counterexamples} \label{Lseccount}
In this Section we show that by dropping the hypothesis introduced at the beginning of Section \ref{Lsecopt} one can find counterexamples of maps for which Eq. \eqref{Loptimaldef} does not hold.

\subsection{Gaussian attenuator with degenerate Hamiltonian}
\label{Lseccatt}
The hypothesis of nondegenerate Hamiltonian is necessary for the optimality of passive states.
Indeed, in this Section we provide an explicit counterexample with degenerate Hamiltonian: the two-mode bosonic Gaussian quantum-limited attenuator (see Section \ref{secattampl}).

Let us consider the Hamiltonian of an harmonic oscillator
\begin{equation}
\hat{H}=\sum_{i=1}^\infty i\;|i\rangle\langle i|\;,\qquad\langle i|j\rangle=\delta_{ij}\;,
\end{equation}
and the Lindbladian
\begin{equation}\label{LLatt}
\mathcal{L}\left(\hat{\rho}\right)=\hat{a}\;\hat{\rho}\;\hat{a}^\dag-\frac{1}{2}\left\{\hat{a}^\dag\hat{a},\;\hat{\rho}\right\}\;,
\end{equation}
where $\hat{a}$ is the ladder operator
\begin{equation}
\hat{a}=\sum_{i=1}^\infty\sqrt{i}\;|i-1\rangle\langle i|\;.
\end{equation}
The quantum-limited attenuator is the channel $e^{t\mathcal{L}}$ generated by the Lindbladian \eqref{LLatt}.
We have proved in Chapter \ref{majorization} that this quantum channel preserves the set of passive states, and they are its optimal inputs in the sense of Eq. \eqref{Loptimaldef}.
Here we will show that this last property does no more hold for the two-mode attenuator
\begin{equation}
\mathcal{E}_t:=e^{t\mathcal{L}}\otimes e^{t\mathcal{L}}\;.
\end{equation}
In this case, the Hamiltonian becomes degenerate:
\begin{equation}
\hat{H}_2=\hat{H}\otimes\hat{\mathbb{I}}+\hat{\mathbb{I}}\otimes\hat{H}=\sum_{k=1}^\infty k \sum_{i+j=k} |i,j\rangle\langle i,j|\;.
\end{equation}
However, the two Lindblad operators $\hat{a}\otimes\hat{\mathbb{I}}$ and $\hat{\mathbb{I}}\otimes\hat{a}$ can still induce only jumps between a given energy level and the immediately lower one, and there are no ambiguities in the definition of the passive rearrangement of quantum states with the same degeneracies of the Hamiltonian.
Let us consider for example
\begin{equation}
\hat{\rho}=\frac{1}{6}\sum_{i+j\leq2} |i,j\rangle\langle i,j|\;,\qquad\mathrm{Tr}\left[\hat{H}_2\;\hat{\rho}\right]=\frac{4}{3}\;.
\end{equation}
It is easy to show that it minimizes the average energy among the states with the same spectrum, i.e. it is passive.
Moreover, there are no other states with the same spectrum and the same average energy, i.e. its passive rearrangement is unique.
Let us consider instead
\begin{equation}
\hat{\sigma}=\frac{1}{6}\sum_{i=0}^5|0,i\rangle\langle 0,i|\;,\qquad\mathrm{Tr}\left[\hat{H}_2\;\hat{\sigma}\right]=\frac{5}{2}\;,
\end{equation}
that has the same spectrum of $\hat{\rho}$, but it has a higher average energy and it is not passive.
The three largest eigenvalues of $\mathcal{E}_t\left(\hat{\rho}\right)$ are associated with the eigenvectors $|0,0\rangle$, $|0,1\rangle$ and $|1,0\rangle$, and their sum is
\begin{equation}
s_3(t)=1-\frac{e^{-2t}}{2}\;.
\end{equation}
On the other side, the three largest eigenvalues of $\mathcal{E}_t\left(\hat{\sigma}\right)$ are associated with the eigenvectors $|0,0\rangle$, $|0,1\rangle$ and $|0,2\rangle$, and their sum is
\begin{equation}
\tilde{s}_3(t)=1-e^{-3t}\frac{5-6e^{-t}+2e^{-2t}}{2}\;.
\end{equation}
It is then easy to see that for
\begin{equation}
e^{-t}<1-\frac{1}{\sqrt{2}}\;,
\end{equation}
i.e.
\begin{equation}
t>\ln\left(2+\sqrt{2}\right):=t_0\;,
\end{equation}
we have
\begin{equation}
s_3(t)<\tilde{s}_3(t)\;,
\end{equation}
i.e. the passive state $\hat{\rho}$ is not the optimal input.
Let $p_1(t)$ and $\tilde{p}_1(t)$ be the largest eigenvalues of $\mathcal{E}_t\left(\hat{\rho}\right)$ and $\mathcal{E}_t\left(\hat{\sigma}\right)$, respectively.
They are both associated to the eigenvector $|0,0\rangle$, and
\begin{align}
&p_1(t) = \frac{6-8e^{-t}+3e^{-2t}}{6}\\
&\tilde{p}_1(t) = \frac{\left(2-e^{-t}\right)\left(3-3e^{-t}+e^{-2t}\right)\left(1-e^{-t}+e^{-2t}\right)}{6}\;.
\end{align}
For any $t>0$
\begin{equation}
p_1(t)>\tilde{p}_1(t)\;,
\end{equation}
so that $\hat{\sigma}$ is not the optimal input, and for $t>t_0$ no majorization relation holds between $\mathcal{E}_t\left(\hat{\rho}\right)$ and $\mathcal{E}_t\left(\hat{\sigma}\right)$.

This counterexample cannot be ascribed to the infinite dimension of the Hilbert space, since it is easy to see that both the supports of $\mathcal{E}_t\left(\hat{\rho}\right)$ and $\mathcal{E}_t\left(\hat{\sigma}\right)$ do not depend on $t$ and have dimension $6$.

\subsection{Two-qubit lossy channel} \label{Lsec2step}
We consider a quantum lossy channel acting on the quantum system of two qubits with two possible choices for the Hamiltonian, and we show that passive states are not the optimal inputs in the sense of \eqref{Loptimaldef}.
In one case (Section \ref{L+jumps}) the Hamiltonian is nondegenerate, but the channel involves quantum jumps of more than one energy step.
In the other case (Section \ref{LdegH}), only quantum jumps of one energy step are allowed, but the Hamiltonian becomes degenerate.

Let us consider the Hilbert space of two distinguishable spins with Hamiltonian
\begin{equation}\label{LHspin}
\hat{H}=E_1\;|1\rangle\langle1|\otimes\hat{\mathbb{I}}+E_2\;\hat{\mathbb{I}}\otimes|1\rangle\langle 1|\;.
\end{equation}
We notice that $\hat{H}$ is not symmetric under the exchange of the two spins, i.e. the spins are different, though the same Hilbert space $\mathbb{C}^2$ is associated to both of them.
Let us suppose that
\begin{equation}
0< E_2\leq E_1\;,
\end{equation}
so that the eigenvectors of $\hat{H}$ are, in order of increasing energy,
\begin{eqnarray}
\hat{H}|0,0\rangle &=& 0\nonumber\\
\hat{H}|0,1\rangle &=& E_2|0,1\rangle\nonumber\\
\hat{H}|1,0\rangle &=& E_1|1,0\rangle\nonumber\\
\hat{H}|1,1\rangle &=& (E_1+E_2)|1,1\rangle\;,
\end{eqnarray}
with the only possible degeneracy between $|0,1\rangle$ and $|1,0\rangle$ if $E_1=E_2$.

Let $\mathcal{L}$ be the generator of the form \eqref{LLdef} with the two Lindblad operators
\begin{eqnarray}
\hat{L}_1&=& |0,0\rangle\langle1,0|\nonumber\\
\hat{L}_2 &=& |0,0\rangle\langle0,1|+\sqrt{2}\;|0,1\rangle\langle 1,1|\;,
\end{eqnarray}
and let
\begin{equation}
\mathcal{E}_t=e^{t\mathcal{L}}\;,\qquad t\geq0\;,
\end{equation}
be the associated quantum channel.

\subsubsection{Jumps of more than one energy step}\label{L+jumps}
If $E_2<E_1$ the Hamiltonian \eqref{LHspin} is nondegenerate, but the Lindblad operator $\hat{L}_2$ can induce a transition from $|1,1\rangle$ to $|0,1\rangle$, that are not consecutive eigenstates.

For simplicity, we parameterize a state diagonal in the Hamiltonian eigenbasis with
\begin{equation}\label{Lrhopar}
\hat{\rho}=\sum_{i,j=0}^1 p_{ij}\;|i,j\rangle\langle i,j|\;.
\end{equation}
First, let
\begin{equation}\label{Lrho0}
\hat{\rho}^{(0)}(t)=\mathcal{E}_t\left(\frac{\hat{\mathbb{I}}}{4}\right)
\end{equation}
be the output of the channel applied to the maximally mixed state.
Then, we can compute
\begin{eqnarray}\label{Lpij}
p^{(0)}_{00}(t) &=& 1-e^{-t}+\frac{e^{-2t}}{4}\nonumber\\
p^{(0)}_{01}(t) &=& e^{-t}\;\frac{3-2e^{-t}}{4}\nonumber\\
p^{(0)}_{10}(t) &=& \frac{e^{-t}}{4}\nonumber\\
p^{(0)}_{11}(t) &=& \frac{e^{-2t}}{4}\;.
\end{eqnarray}
It is easy to check that, for any $t>0$,
\begin{equation}\label{Lpij>}
p^{(0)}_{00}(t)>p^{(0)}_{01}(t)>p^{(0)}_{10}(t)>p^{(0)}_{11}(t)\;,
\end{equation}
so that $\hat{\rho}^{(0)}(t)$ is passive, and the channel $\mathcal{E}_t$ satisfies the hypothesis of Lemma \ref{Lpassl}.
Let us instead compare
\begin{align}\label{Lrho1}
&\hat{\rho}^{(1)}(t) = \mathcal{E}_t\left(\frac{|0,0\rangle\langle0,0|+|0,1\rangle\langle0,1|+|1,0\rangle\langle1,0|}{3}\right)\\
&\hat{\rho}^{(2)}(t) = \mathcal{E}_t\left(\frac{|0,0\rangle\langle0,0|+|0,1\rangle\langle0,1|+|1,1\rangle\langle1,1|}{3}\right)\;.\label{Lrho2}
\end{align}
It is easy to see that $\hat{\rho}^{(1)}(0)$ is passive, while $\hat{\rho}^{(2)}(0)$ is not, and they have the same spectrum.
Moreover, there are no other states with the same spectrum and the same average energy of $\hat{\rho}^{(1)}(0)$, i.e. its passive rearrangement is unique.
We can now compute
\begin{align}\label{Lp12t}
&p^{(1)}_{00}(t) = 1-\frac{2}{3}e^{-t}\qquad && p^{(2)}_{00}(t) = 1-e^{-t}+\frac{e^{-2t}}{3}\nonumber\\
&p^{(1)}_{01}(t) = \frac{e^{-t}}{3}\qquad && p^{(2)}_{01}(t) = e^{-t}\left(1-\frac{2}{3}e^{-t}\right)\nonumber\\
&p^{(1)}_{10}(t) = \frac{e^{-t}}{3}\qquad && p^{(2)}_{10}(t) =0\nonumber\\
&p^{(1)}_{11}(t) = 0\qquad && p^{(2)}_{11}(t) =\frac{e^{-2t}}{3}\;.
\end{align}
It is easy to see that for any $t>0$
\begin{align}\label{Lp12>}
&p^{(1)}_{00}(t)>p^{(1)}_{01}(t)=p^{(1)}_{10}(t)\nonumber\\
&p^{(2)}_{00}(t)>p^{(2)}_{01}(t)>p^{(2)}_{11}(t)\;,
\end{align}
so that $\hat{\rho}^{(1)}(t)$ remains always passive.
However, on one hand
\begin{equation}
p^{(1)}_{00}(t)>p^{(2)}_{00}(t)\;,
\end{equation}
but on the other hand
\begin{equation}
p^{(1)}_{00}(t)+p^{(1)}_{01}(t)<p^{(2)}_{00}(t)+p^{(2)}_{01}(t)\;,
\end{equation}
so that no majorization relation can exist between $\hat{\rho}^{(1)}(t)$ and $\hat{\rho}^{(2)}(t)$.

\subsubsection{Degenerate Hamiltonian}\label{LdegH}
If $E_1=E_2$, the eigenstates $|0,1\rangle$ and $|1,0\rangle$ of the Hamiltonian \eqref{LHspin} become degenerate, but both $\hat{L}_1$ and $\hat{L}_2$ induce only transitions between consecutive energy levels.

We use the parametrization \eqref{Lrhopar} as before.
Let $\hat{\rho}^{(0)}(t)$ be the output of the channel applied to the maximally mixed state as in \eqref{Lrho0}.
Since the generator $\mathcal{L}$ is the same of Section \ref{L+jumps}, the probabilities $p^{(0)}_{ij}(t)$, $i,j=0,1$, are still given by \eqref{Lpij}.
Eq. \eqref{Lpij>} still holds for any $t>0$, so that $\hat{\rho}^{(0)}(t)$ is passive, and the channel $\mathcal{E}_t$ satisfies the hypothesis of Lemma \ref{Lpassl}.

Let us instead compare $\hat{\rho}^{(1)}(t)$ and $\hat{\rho}^{(2)}(t)$ defined as in \eqref{Lrho1} and \eqref{Lrho2}, respectively.
The state $\hat{\rho}^{(1)}(0)$ is passive, while $\hat{\rho}^{(2)}(0)$ is not, and they have the same spectrum.
Moreover, there are no other states with the same spectrum and the same average energy of $\hat{\rho}^{(1)}(0)$, i.e. its passive rearrangement is unique.
The probabilities $p^{(1)}_{ij}(t)$ and $p^{(2)}_{ij}(t)$, $i,j=0,1$, are still given by \eqref{Lp12t}.
Eq. \eqref{Lp12>} still holds for any $t>0$, and $\hat{\rho}^{(1)}(t)$ remains always passive.
However, on one hand
\begin{equation}
p^{(1)}_{00}(t)>p^{(2)}_{00}(t)\;,
\end{equation}
but on the other hand
\begin{equation}
p^{(1)}_{00}(t)+p^{(1)}_{01}(t)<p^{(2)}_{00}(t)+p^{(2)}_{01}(t)\;,
\end{equation}
so that no majorization relation can exist between $\hat{\rho}^{(1)}(t)$ and $\hat{\rho}^{(2)}(t)$.

\subsection{Optimal states for a finite-temperature two-level system are nonclassical}\label{Lsecfin}
In this Section we show that at finite temperature, already for a two-level system the optimal states are no more passive, and include coherent superpositions of the energy eigenstates.

An intuitive explanation is that a dissipator with only energy-raising Lindblad operators keeps fixed the maximum-energy eigenstate, that is hence optimal for the generated channel.
Then, it is natural to expect that the optimal pure state in the presence of both energy-lowering and energy-raising Lindblad operators will interpolate between the ground and the maximum energy state, and will hence be a coherent superposition of different eigenstates of the Hamiltonian.

The simplest example is a two-level system with Hamiltonian
\begin{equation}
\hat{H}=\frac{1}{2}E_0\;\hat{\sigma}_z=\frac{E_0}{2}\;|1\rangle\langle 1|-\frac{E_0}{2}\;|0\rangle\langle0|\;,\qquad E_0>0\;,
\end{equation}
undergoing the quantum optical master equation \cite{breuer2007theory}, describing the weak coupling with a thermal bath of one mode of bosonic excitations in the rotating-wave approximation.
This is the simplest extension of the evolutions considered in Section \ref{Lsec2step} to an interaction with a finite-temperature bath.

Its generator is
\begin{equation}\label{LqubitL}
\mathcal{L}\left(\hat{\rho}\right) = \gamma_0(N+1)\left(\hat{\sigma}_-\;\hat{\rho}\;\hat{\sigma}_+-\frac{1}{2}\left\{\hat{\sigma}_+\hat{\sigma}_-,\;\hat{\rho}\right\}\right)+ \gamma_0N\left(\hat{\sigma}_+\;\hat{\rho}\;\hat{\sigma}_--\frac{1}{2}\left\{\hat{\sigma}_-\hat{\sigma}_+,\;\hat{\rho}\right\}\right)\;,
\end{equation}
where
\begin{equation}
\hat{\sigma}_\pm=\frac{\hat{\sigma}_x\pm i\hat{\sigma}_y}{2}
\end{equation}
are the ladder operators, $\gamma_0>0$ is the coupling strength and $N>0$ is the average number of photons or phonons in the bosonic mode of the bath coupled to the system. We also notice that
for $N=0$ the process becomes a lossy map fulfilling the condition discussed at the beginning of Section \ref{Lsecopt}.

We will now show that, for the quantum channel associated to the master equation \eqref{LqubitL}, the output generated by a certain coherent superposition of the two energy eigenstates majorizes the output generated by any other state.

It is convenient to use the Bloch representation
\begin{equation}\label{Lbloch}
\hat{\rho}=\frac{\hat{\mathbb{I}}+x\;\hat{\sigma}_x+y\;\hat{\sigma}_y+z\;\hat{\sigma}_z}{2}\;,\qquad x^2+y^2+z^2\leq1\;.
\end{equation}
The master equation \eqref{LqubitL} induces the differential equations
\begin{equation}\label{Lvdot}
\frac{dx}{dt} = -\frac{\gamma}{2}\;x\;,\qquad \frac{dy}{dt} = -\frac{\gamma}{2}\;y\;,\qquad \frac{dz}{dt} = -\gamma\left(z-z_\infty\right)\;,
\end{equation}
where
\begin{equation}
\gamma=\gamma_0(2N+1)\qquad\text{and}\qquad z_\infty=-\frac{1}{2N+1}\;.
\end{equation}
The solution of \eqref{Lvdot} is
\begin{eqnarray}
x(t) &=& e^{-\frac{\gamma}{2}t}\;x_0\;,\nonumber\\
y(t) &=& e^{-\frac{\gamma}{2}t}\;y_0\;,\nonumber\\
z(t) &=& z_\infty+e^{-\gamma t}\left(z_0-z_\infty\right)\;,
\end{eqnarray}
and its asymptotic state is the canonical state with inverse temperature $\beta$
\begin{equation}
\hat{\rho}_\infty=\frac{e^\frac{\beta E_0}{2}\;|0\rangle\langle0|+e^{-\frac{\beta E_0}{2}}\;|1\rangle\langle 1|}{2\cosh\frac{\beta E_0}{2}}\;,
\end{equation}
satisfying
\begin{equation}
z_\infty=-\tanh\frac{\beta\;E_0}{2}\;.
\end{equation}
Since the density matrix of a two-level system has only two eigenvalues, the purity is a sufficient criterion for majorization, i.e. for any two quantum states $\hat{\rho}$ and $\hat{\sigma}$,
\begin{equation}
\hat{\rho}\prec\hat{\sigma}\qquad\text{iff}\qquad\mathrm{Tr}\;\hat{\rho}^2\leq\mathrm{Tr}\;\hat{\sigma}^2\;.
\end{equation}
We recall that in the Bloch representation \eqref{Lbloch}
\begin{equation}
\mathrm{Tr}\;\hat{\rho}^2=\frac{1+x^2+y^2+z^2}{2}\;.
\end{equation}
We have then
\begin{equation}\label{Lpurity}
\mathrm{Tr}\;{\hat{\rho}(t)}^2 = \frac{1+e^{-\gamma t}\left(x_0^2+y_0^2+z_0^2\right)}{2}+ \frac{1-e^{-\gamma t}}{2}\left(z_\infty^2-e^{-\gamma t}\left(z_0-z_\infty\right)^2\right)\;.
\end{equation}
The right-hand side of \eqref{Lpurity} is maximized by
\begin{equation}
x_0^2+y_0^2 = 1-z_\infty^2\qquad\text{and}\qquad z_0 = z_\infty\;,
\end{equation}
i.e. when the initial state is a pure coherent superposition of the energy eigenstates $|0\rangle$ and $|1\rangle$ with the same average energy of the asymptotic state:
\begin{equation}
|\psi\rangle=e^{i\varphi_0}\sqrt{\frac{1-z_\infty}{2}}\;|0\rangle+e^{i\varphi_1}\sqrt{\frac{1+z_\infty}{2}}\;|1\rangle\;,
\end{equation}
where $\varphi_0$ and $\varphi_1$ are arbitrary real phases.

\section{Auxiliary lemmata}\label{lossylem}
\subsection{Passivity of the evolved maximally mixed state}\label{Lpasst}
\begin{lem}\label{Lpassl}
Let $\mathcal{L}$ be a Lindblad generator such that for any $t\geq0$ the operator $e^{t\mathcal{L}}\left(\hat{\mathbb{I}}\right)$ is passive.
Then, also $\mathcal{L}\left(\hat{\mathbb{I}}\right)$ is passive.
\begin{proof}
Recalling the Hamiltonian eigenbasis \eqref{LHpass}, for any $t\geq0$ it must hold
\begin{equation}
e^{t\mathcal{L}}\left(\hat{\mathbb{I}}\right)=\sum_{i=1}^d c_i(t)\;|i\rangle\langle i|
\end{equation}
with
\begin{equation}
c_1(t)\geq\ldots\geq c_d(t)\;,\qquad c_1(0)=\ldots c_d(0)=1\;,
\end{equation}
and each $c_i(t)$ is an analytic function of $t$.
It follows that
\begin{equation}
c_1'(0)\geq\ldots\geq c_d'(0)\;.
\end{equation}
However, we have also
\begin{equation}
\mathcal{L}\left(\hat{\mathbb{I}}\right)=\left.\frac{d}{dt}e^{t\mathcal{L}}\left(\hat{\mathbb{I}}\right)\right|_{t=0}=\sum_{i=1}^d c_i'(0)\;|i\rangle\langle i|\;,
\end{equation}
hence the thesis.
\end{proof}
\end{lem}

\subsection{Proof of Lemma \ref{Ldeg}}\label{Lproofdeg}
The matrix elements of the operator $e^{t\mathcal{L}}\left(\hat{\rho}\right)$ are analytic functions of $t$.
The spectrum of $\hat{\rho}(t)$ is degenerate iff the function
\begin{equation}
\phi(t)=\prod_{i\neq j}\left(p_i(t)-p_j(t)\right)
\end{equation}
vanishes.
This function is a symmetric polynomial in the eigenvalues of $\hat{\rho}(t)=e^{t\mathcal{L}}\left(\hat{\rho}\right)$.
Then, for the Fundamental Theorem of Symmetric Polynomials (see e.g Theorem 3 in Chapter 7 of \cite{cox2015ideals}), $\phi(t)$ can be written as a polynomial in the elementary symmetric polynomials in the eigenvalues of $\hat{\rho}(t)$.
However, these polynomials coincide with the coefficients of the characteristic polynomial of $\hat{\rho}(t)$, that are in turn polynomials in its matrix elements.
It follows that $\phi(t)$ can be written as a polynomial in the matrix elements of the operator $\hat{\rho}(t)$.
Since each of these matrix element is an analytic function of $t$, also $\phi(t)$ is analytic.
Since by hypothesis the spectrum of $\hat{\rho}(0)$ is nondegenerate, $\phi$ cannot be identically zero, and its zeroes are isolated points.

\subsection{Proof of Lemma \ref{Llemma1}}\label{Lprooflemma1}
The matrix elements of the operator $e^{t\mathcal{L}}\left(\hat{\rho}\right)$ are analytic (and hence continuous and differentiable) functions of $t$.
Then for Weyl's Perturbation Theorem $p(t)$ is continuous in $t$, and also $s(t)$ is continuous (see e.g. Corollary III.2.6 and the discussion at the beginning of Chapter VI  of \cite{bhatia2013matrix}).
Let $\hat{\rho}(t_0)$ have nondegenerate spectrum.
Then, $\hat{\rho}(t)$ has nondegenerate spectrum for any $t$ in a suitable neighbourhood of $t_0$.
In this neighbourhood, we can diagonalize $\hat{\rho}(t)$ with
\begin{equation}
\hat{\rho}(t)=\sum_{i=1}^d p_i(t) |\psi_i(t)\rangle\langle\psi_i(t)|\;,
\end{equation}
where the eigenvalues in decreasing order $p_i(t)$ are differentiable functions of $t$ (see Theorem 6.3.12 of \cite{horn2012matrix}).
We then have
\begin{equation}
\frac{d}{dt}p_i(t)=\langle\psi_i(t)|\mathcal{L}\left(\hat{\rho}(t)\right)|\psi_i(t)\rangle\;,\qquad i=1,\ldots,d\;,
\end{equation}
and
\begin{equation}
\frac{d}{dt}s_n(t)=\mathrm{Tr}\left[\hat{\Pi}_n(t)\;\mathcal{L}\left(\hat{\rho}(t)\right)\right]\;,
\end{equation}
where
\begin{equation}
\hat{\Pi}_n(t)=\sum_{i=1}^n|\psi_i(t)\rangle\langle\psi_i(t)|\;.
\end{equation}
We can write
\begin{equation}
\hat{\rho}(t)=\sum_{n=1}^d d_n(t)\;\hat{\Pi}_n(t)\;,
\end{equation}
where
\begin{equation}
d_n(t)=p_n(t)-p_{n+1}(t)\geq0\;,
\end{equation}
and for simplicity we have set $p_{d+1}(t)=0$, so that
\begin{equation}\label{Lsndot}
\frac{d}{dt}s_n(t)=\sum_{k=1}^d d_k(t)\;\mathrm{Tr}\left[\hat{\Pi}_n(t)\;\mathcal{L}\left(\hat{\Pi}_k(t)\right)\right]\;.
\end{equation}
We have now
\begin{equation}\label{LPLext}
\mathrm{Tr}\left[\hat{\Pi}_n(t)\;\mathcal{L}\left(\hat{\Pi}_k(t)\right)\right]= \sum_\alpha\mathrm{Tr}\left[\hat{\Pi}_n(t)\;\hat{L}_\alpha\;\hat{\Pi}_k(t)\;\hat{L}_\alpha^\dag-\hat{\Pi}_{k\land n}(t)\;\hat{L}_\alpha^\dag\hat{L}_\alpha\right]\;,
\end{equation}
where $k\land n=\min(k,n)$ and we have used that
\begin{equation}
\hat{\Pi}_n(t)\;\hat{\Pi}_k(t)=\hat{\Pi}_k(t)\;\hat{\Pi}_n(t)=\hat{\Pi}_{k\land n}(t)\;.
\end{equation}
\begin{itemize}
  \item Let us suppose $n\geq k$.
  Using that $\hat{\Pi}_n(t)\leq\hat{\mathbb{I}}$ in the first term of \eqref{LPLext}, we get
  \begin{equation}\label{Ln>k}
  \mathrm{Tr}\left[\hat{\Pi}_n(t)\;\mathcal{L}\left(\hat{\Pi}_k(t)\right)\right]\leq0\;.
  \end{equation}
  On the other hand, recalling the structure of the Lindblad operators \eqref{Ldephase} and \eqref{Ljump}, for any $\alpha$ the support of $\hat{L}_\alpha\,\hat{\Pi}_k^\downarrow\,\hat{L}_\alpha^\dag$ is contained into the support of $\hat{\Pi}_{k}^\downarrow$, and hence into the one of $\hat{\Pi}_n^\downarrow$, and we have also
  \begin{equation}\label{Ln>k*}
  \mathrm{Tr}\left[\hat{\Pi}^\downarrow_n\;\mathcal{L}\left(\hat{\Pi}_k^\downarrow\right)\right]=0\;.
  \end{equation}
  \item Let us now suppose $k>n$.
  Using that $\hat{\Pi}_k(t)\leq\hat{\mathbb{I}}$ in the first term of \eqref{LPLext}, we get
  \begin{equation}\label{Lk>n}
  \mathrm{Tr}\left[\hat{\Pi}_n(t)\;\mathcal{L}\left(\hat{\Pi}_k(t)\right)\right] \leq \mathrm{Tr}\left[\hat{\Pi}_n(t)\;\mathcal{L}\left(\hat{\mathbb{I}}\right)\right] \leq\mathrm{Tr}\left[\hat{\Pi}_n^\downarrow\;\mathcal{L}\left(\hat{\mathbb{I}}\right)\right]=\lambda_n\;,
  \end{equation}
  where in the last step we have used Ky Fan's maximum principle (Lemma \ref{sumeig}) and the passivity of $\mathcal{L}\left(\hat{\mathbb{I}}\right)$.
  On the other hand, from \eqref{Ldephase} and \eqref{Ljump} the support of $\hat{L}_\alpha^\dag\,\hat{\Pi}_n^\downarrow\,\hat{L}_\alpha$ is contained into the support of $\hat{\Pi}_{n+1}^\downarrow$, and hence into the one of $\hat{\Pi}_k^\downarrow$, and we have also
  \begin{equation}\label{Ln<k*}
  \mathrm{Tr}\left[\hat{\Pi}_n^\downarrow\;\mathcal{L}\left(\hat{\Pi}_k^\downarrow\right)\right]=\lambda_n\;.
  \end{equation}
\end{itemize}
Plugging \eqref{Ln>k} and \eqref{Lk>n} into \eqref{Lsndot}, we get
\begin{equation}
\frac{d}{dt}s_n(t)\leq\lambda_n\;p_{n+1}(t)=\lambda_n\left(s_{n+1}(t)-s_n(t)\right)\;.
\end{equation}
From \eqref{Ln>k*} and \eqref{Ln<k*} we get instead
\begin{equation}
\frac{d}{dt}s_n^\downarrow(t)=\lambda_n\;p_{n+1}^\downarrow(t)=\lambda_n\left(s^\downarrow_{n+1}(t)-s_n^\downarrow(t)\right)\;.
\end{equation}
See Lemma \ref{Llambdan} for the positivity of the coefficients $\lambda_n$.

\subsection{Proof of Lemma \ref{Llemma2}}\label{Lprooflemma2}
Since the quantum channel $e^{t\mathcal{L}}$ is trace-preserving, we have
\begin{equation}
s_d(t)=\mathrm{Tr}\;\hat{\rho}(t)=1=s_d^\downarrow(t)\;.
\end{equation}
We will use induction on $n$ in the reverse order: let us suppose to have proved
\begin{equation}
s_{n+1}(t)\leq s_{n+1}^\downarrow(t)\;.
\end{equation}
Since $\lambda_n\geq0$ for Lemma \ref{Llambdan}, we have from \eqref{Lsdot}
\begin{equation}
\frac{d}{dt}s_n(t)\leq\lambda_n\left(s_{n+1}^\downarrow(t)-s_n(t)\right)\;,
\end{equation}
while
\begin{equation}
\frac{d}{dt}s_n^\downarrow(t)=\lambda_n\left(s_{n+1}^\downarrow(t)-s_n^\downarrow(t)\right)\;.
\end{equation}
Defining
\begin{equation}
f_n(t)=s_n^\downarrow(t)-s_n(t)\;,
\end{equation}
we have $f_n(0)=0$, and
\begin{equation}
\frac{d}{dt}f_n(t)\geq-\lambda_n\;f_n(t)\;.
\end{equation}
This can be rewritten as
\begin{equation}
e^{-\lambda_nt}\;\frac{d}{dt}\left(e^{\lambda_nt}\;f_n(t)\right)\geq0\;,
\end{equation}
and implies
\begin{equation}
f_n(t)\geq0\;.
\end{equation}

\subsection{Proof of Lemma \ref{Llambdan}}\label{Lprooflambdan}
\begin{lem}\label{Llambdan}
$\lambda_n\geq0$ for $n=1,\,\ldots,\,d$.
\begin{proof}
For Ky Fan's maximum principle (Lemma \ref{sumeig}), for any unitary $\hat{U}$
\begin{equation}\label{LlambdaU}
\lambda_n = \mathrm{Tr}\left[\hat{\Pi}_n^\downarrow\;\mathcal{L}\left(\hat{\mathbb{I}}\right)\right]\geq \mathrm{Tr}\left[\hat{U}\;\hat{\Pi}_n^\downarrow\;\hat{U}^\dag\;\mathcal{L}\left(\hat{\mathbb{I}}\right)\right]\;.
\end{equation}
The thesis easily follows taking the average over the Haar measure $\mu$ of the right-hand side of \eqref{LlambdaU}, since
\begin{equation}
\int\hat{U}^\dag\;\mathcal{L}\left(\hat{\mathbb{I}}\right)\;\hat{U}\;d\mu\left(\hat{U}\right)= \frac{\hat{\mathbb{I}}}{d}\;\mathcal{L}\left(\hat{\mathbb{I}}\right)=0\;.
\end{equation}
\end{proof}
\end{lem}

\section{Conclusion} \label{Lseccon}
In this Chapter we have extended the proof of the optimality of passive states of Chapter \ref{majorization} to a large class of lossy channels, showing that they preserve the set of passive states, that are the optimal inputs in the sense that the output generated by a passive state majorizes the output generated by any other state with the same spectrum.
Then, thanks to the equivalent definition of majorization in terms of random unitary operations \eqref{majru}, the output generated by a passive state minimizes any concave functional among the outputs generated by any unitary equivalent state.
Since the class of concave functionals includes the von Neumann and all the R\'enyi entropies, the solution to any entropic optimization problem has to be found among passive states.
This result can then lead to entropic inequalities on the output of a lossy channel, and can be crucial in the determination of its information capacity.
Moreover, in the context of quantum thermodynamics this result can be useful to determine which quantum states can be obtained from an initial state with a given spectrum in a resource theory with lossy channels among the allowed operations.

The optimality of passive states crucially depends on the assumptions of nondegenerate Hamiltonian, quantum jumps of only one energy step and zero temperature.
Indeed, the two-mode bosonic Gaussian quantum-limited attenuator provides a counterexample with degenerate Hamiltonian.
Moreover, two-qubit systems can provide counterexamples both with degenerate Hamiltonian or with quantum jumps of more than one energy step.
Finally, at finite temperature this optimality property fails already for a two-level system, where the best input is a coherent superposition of the two energy eigenstates.
This shows that even the quantum channels that naturally arise from a weak interaction with a thermal bath can have a very complex entropic behaviour, and that coherence can play a crucial role in the optimal encoding of information.

\chapter{Memory effects}\label{memory}
In this Chapter we determine the capacity for transmitting classical information over a model of Gaussian channel with memory effects.

The Chapter is based on
\begin{enumerate}
\item[\cite{de2014classical}] G.~De~Palma, A.~Mari, and V.~Giovannetti, ``Classical capacity of Gaussian thermal memory channels,'' \emph{Physical Review A}, vol.~90, no.~4, p. 042312, 2014.\\ {\small\url{http://journals.aps.org/pra/abstract/10.1103/PhysRevA.90.042312}}
\end{enumerate}

\section{Introduction}
Given a physical device acting as a quantum communication channel \cite{caves1994quantum,holevo2013quantum}, an important problem in quantum information theory is to
determine the optimal rate of classical information that can be sent through the channel assuming that one is allowed to use arbitrary
quantum encoding and decoding strategies possibly involving multiple uses of the transmission line ({\it channel uses}). The maximum achievable rate is the {\it classical capacity} associated to the quantum channel \cite{holevo2013quantum,holevo2001evaluating,schumacher1997sending}.
If no memory effects are tampering the communication line (i.e. if the noise affecting the communication acts identically and independently on subsequent channel uses), the classical capacity of the setup can be expressed as the limiting formula \eqref{Ctens}.

Most real communication media are based on electromagnetic signals and are well described within the framework of quantum Gaussian channels \cite{braunstein2005quantum,cerf2007quantum,weedbrook2012gaussian} (see Chapter \ref{GQI}).
The most relevant class is constituted by gauge-covariant channels like attenuators and amplifiers. Such channels reduce or increase the amplitude of the signal
and, at the same time, they add a certain amount of Gaussian noise which depends on the vacuum or thermal fluctuations of the environment.
Recently the proof of the minimum output entropy conjecture \cite{giovannetti2015solution,mari2014quantum} has allowed the determination of the exact classical capacities of these channels \cite{giovannetti2014ultimate} and the respective strong converse theorems \cite{bardhan2015strong}, under the crucial assumption of
their memoryless behavior (see Section \ref{gccapacity}).
As we have seen in Chapter \ref{GQI}, one of the key points of the proof is the additivity of the Holevo information of a memoryless gauge-covariant gaussian channel,
\begin{figure}[t]
\includegraphics[width=\textwidth]{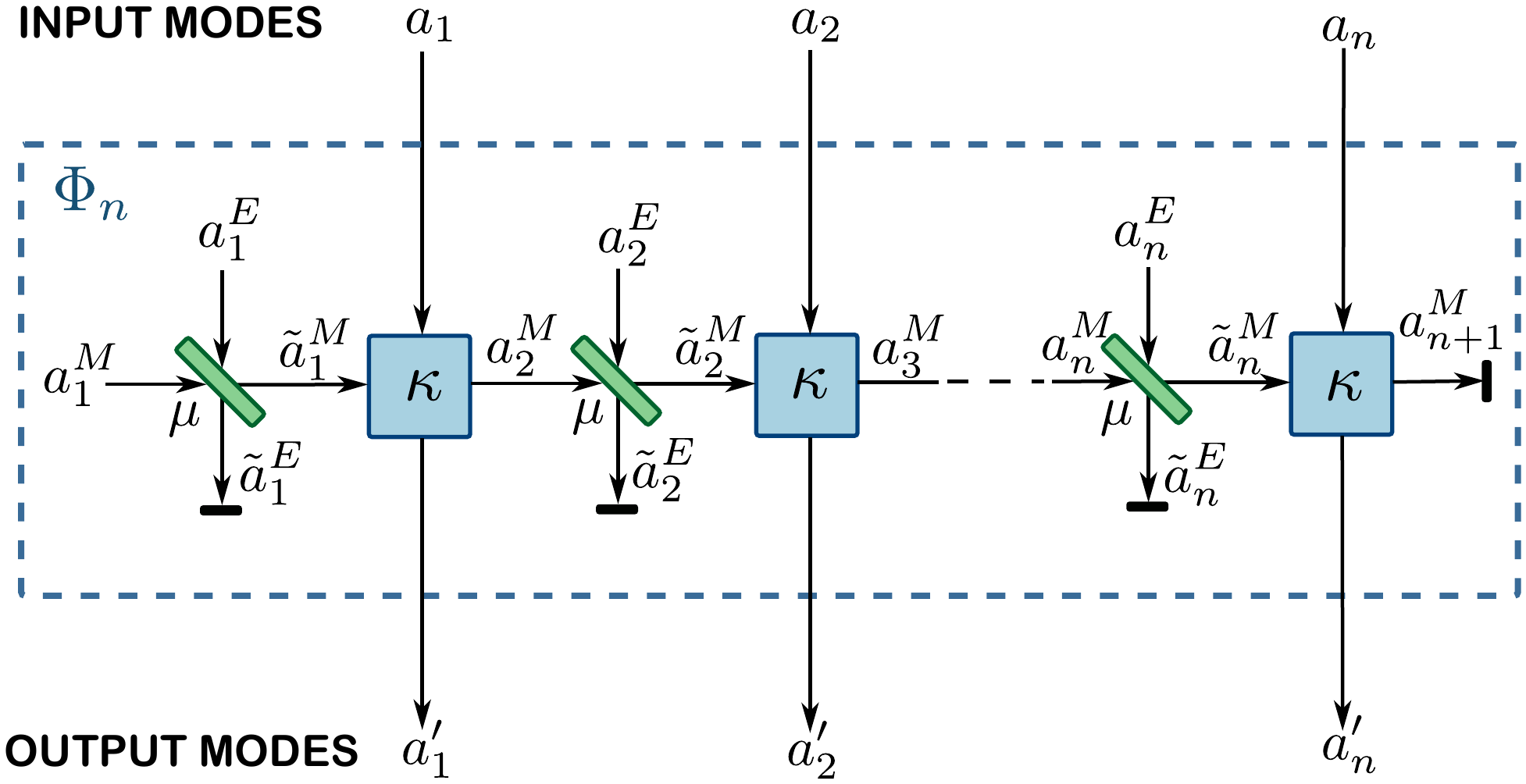}
\caption{Schematic description of a Gaussian memory channel $\Phi_n$ which is iterated $n$ times. The application of a the memory channel to $n$ successive input modes $a_1,\  \dots,a_n $ is described by $n$ gauge-covariant channels $\mathcal E_\kappa$ (thermal attenuators or amplifiers) where each of them is coupled to a Gaussian thermal environment and to a memory mode. The initial memory mode $a^{M}_1$ travels horizontally and correlates the output signals with the previous input signals. A  beamsplitter of transmissivity $\mu$ is used to tune the memory effect of the channel. For $\mu=1$ the memory mode is perfectly preserved while for $\mu=0$ the channel becomes memoryless. A reasonable choice for the initial state of the memory mode is a Gaussian thermal state in equilibrium with the environment, {\it i.e.}\ we make the identification $a^{M}_1=a^{E}_0$. The final state of the memory is assumed to be inaccessible and is traced out. }\label{model}
\end{figure}
which trivializes the limit in \eqref{Ctens}.
Realistic communication lines however, if used at high rates (larger than the relaxation time of the environment), may exhibit
memory effects in which the output states are influenced by the previous input signals \cite{caruso2014quantum,gallager2014information, banaszek2004experimental,demkowicz2007effects,paladino2002decoherence,hu2007controllable}.  In other words, the noise introduced by the channel instead of being independent and identically distributed can  be correlated
with the previous input states preventing one from expressing the input-output mapping  of $n$ successive channel uses as a simple tensor product $\Phi^{\otimes n}$
 and hence from using
 Eq. \eqref{Ctens}. As a matter of fact since the capacity is defined asymptotically in the limit of many repeated channel uses,  memory effects will affect the optimal information rate and the optimal coding strategies.
A characterization of {\it quantum memory channels} can be found in Ref.'s\ \cite{kretschmann2005quantum,datta2007coding,d2007quantum,giovannetti2005dynamical}, while generalizations to infinite dimensional bosonic systems are considered in  \cite{giovannetti2005bosonic,cerf2005quantum,lupo2010capacities,lupo2010memory,pilyavets2012methods,schafer2009capacity,schafer2011gaussian,schafer2012gaussian}.

Here we elaborate on the model of (zero temperature) attenuators and amplifiers with memory effects that was introduced in Ref.'s \ \cite{lupo2010capacities, lupo2010memory} where, in the case of a quantum limited attenuator, the capacity was explicitly determined. We generalize this model to thermal attenuators and thermal amplifiers and we derive the corresponding classical capacities, extending the previous results obtained in the memoryless scenario \cite{giovannetti2014ultimate}. We have also considered the case of the additive noise channel, viewed as a particular limit of an attenuator with large transmissivity and large thermal noise. This limit is essentially equivalent to the model
considered in \cite{lupo2009forgetfulness,schafer2009capacity}, and we have shown that the only effect of the memory is a redistribution of the added noise.
An interesting feature which emerges from our analysis is the presence of a critical environmental temperature which strongly affects
 the distribution of the input
energy  among the various modes of the model. In particular  for temperatures larger than the critical one, only the modes which have a sufficiently high effective  transmissivity are allowed to
contribute to the signaling process, the remaining one being forced to carry no energy nor information.

Given a quantum channel the associated unitary dilation is not unique and one can imagine different models for memory effects.
Nonetheless our paradigm is expected to cover many real devices like optical fibers  \cite{gisin2002quantum,tanzilli2002ppln},  microwave systems \cite{lang2013correlations}, $THz$ lasers \cite{kohler2002terahertz}, free space communication \cite{chan2006free,fedrizzi2009high}, {\it etc}..
All physical implementations are known to exhibit time delay and memory effects whenever used at sufficiently high repetition rates. Moreover, especially in microwave and electrical channels, thermal noise is not
negligible and will affect the classical capacity. In general, our analysis applies to any physical realization of quantum channels in which memory effects and thermal noise are simultaneously present.

We begin in Section \ref{sec:gaus} by recalling some basic facts about the memory channel model of Ref.'s \cite{cerf2007quantum,weedbrook2012gaussian}. In particular we describe its
normal mode decomposition which allows one to express the associated mapping as a tensor product of  not necessarily identical single mode transformations.
In Section \ref{sec:capac} we compute the classical capacity of the setup and discuss some special cases, while in Section \ref{sec:lagr} we analyze how the distribution of the input energy among the various modes is affected by the  presence of a thermal environment.
Conclusions and perspectives are provided in Section \ref{sec:conc}.

\section{Gaussian memory channels}\label{sec:gaus}

In this Section we review the model of {\it Gaussian memory channels} introduced in Ref.'s \cite{lupo2010capacities,lupo2010memory}. We closely  follow their analysis showing that these memory channels can be reduced to a collection of memoryless channels by some appropriate
encoding and decoding unitary operations.

\subsection{ Quantum attenuators and amplifiers}
The  building blocks of our analysis are single mode quantum attenuators and amplifiers \cite{cerf2007quantum,weedbrook2012gaussian}, that we have presented in Section \ref{secattampl}. Let us consider a continuous variable bosonic system \cite{braunstein2005quantum} described by the creation and annihilation
operators $a$ and $a^\dag$ and another mode described by $a^E$ and $a^{E \dag}$ associated to the environment. We focus on two important Gaussian unitaries,
\begin{subequations}
\begin{eqnarray}
U_\kappa&=&e^{\arctan\sqrt{\frac{1-\kappa}{\kappa}}\left(a\, a^{E\dag} - a^\dag a^E\right)}\;,\qquad 0\leq\kappa\leq1\;,\\
U_\kappa &=&e^{\mathrm{arctanh}\sqrt{\frac{\kappa-1}{\kappa}}\left(a^\dag a^{E \dag} - a\, a^E \right)}\;,\qquad\kappa\geq1\;,
\end{eqnarray}
\end{subequations}
corresponding to the beamsplitter and the two-mode squeezing operations, respectively.
Their action on the annihilation operator is
\begin{subequations}
\begin{eqnarray}
U_\kappa ^\dag a U_\kappa &=&\sqrt{\kappa}\; a - \sqrt{1- \kappa}\; a^E\;,\qquad0\leq\kappa\leq1\;, \label{U1} \\
U_\kappa^\dag a U_\kappa&=&\sqrt{\kappa}\; a + \sqrt{\kappa-1}\; a^{E \dag}\;,\qquad\kappa\geq1\;. \label{U2}
\end{eqnarray}
\end{subequations}
If the environment is in a Gaussian thermal state
\begin{equation}
\rho_E= e^{- \beta \hbar \omega a^{E \dag} a^E}\left/ \mathrm{Tr}\left[ e^{- \beta \omega \hbar a^{E \dag} a^E}\right]\right.
\end{equation}
with mean photon number
\begin{equation}
N=\mathrm{Tr}\left[a^{E \dag} a^E \rho\right]= \left(e^{\beta \hbar \omega}-1\right)^{-1}\;,
\end{equation}
applying the unitaries \eqref{U1} and \eqref{U2} and tracing out the environment, we get

\begin{equation}
\mathcal E_\kappa(\rho)= \mathrm{Tr}_E \left[ U_\kappa (\rho  \otimes \rho_E) U_\kappa^\dag\right].
\end{equation}
This generates two different gauge-covariant channels depending on whether $\kappa$ is less or larger than 1. For $\kappa \in [0,1]$  the channel corresponds to a thermal attenuator, while for
$\kappa >1$ the channel is a thermal amplifier. In both cases the classical capacity has been recently determined in \cite{giovannetti2014ultimate}. Under the input energy constraint  $\mathrm{Tr}\left[a^\dag a \rho\right] \leqslant E $, the capacities of the   attenuator and of the amplifier are obtainable via a Gaussian encoding and are given by \cite{giovannetti2014ultimate} (in nats for channel use):
\begin{subequations}
\begin{align}
C_{\kappa\in [0,1]}&=g[\kappa E + (1-\kappa) N] - g[(1-\kappa) N], \label{capbs}\\
C_{\kappa>1}&=g[\kappa E + (\kappa-1) (N+1)]  - g[(\kappa-1) (N+1)]\label{capamp}\;,
\end{align}
\end{subequations}
where $g(x)=(x+1)\ln(x+1)-x\ln(x)$.

\subsection{Gaussian memory channels}

In order to include memory effects we follow the model introduced in \cite{lupo2010capacities, lupo2010memory} and schematically shown in Fig.\ \ref{model}.
In addition to the degrees of freedom of the system and of the thermal environment we introduce a ``memory'' described by the bosonic operators $a^M$ and $a^{M \dag}$.
The channel acts in the following way: as a first step the memory is mixed with the environment via a beamsplitter of transmissivity $\mu$,
\begin{eqnarray}
\tilde a^M&=&\sqrt{\mu}\; a^M + \sqrt{1- \mu}\; a^E.
\end{eqnarray}
The outcome state is used as an effective environment for the quantum attenuator or alternatively the quantum amplifier. More precisely, the second step consists in applying the unitary \eqref{U1} or \eqref{U2} to the product state of the system and of the effective environment,
 \begin{subequations}
 \begin{eqnarray}
 a' &=&\sqrt{\kappa}\; a - \sqrt{1- \kappa}\; \tilde a^M, \quad  \kappa  \in [0,1], \\
 a' &=&\sqrt{\kappa}\; a + \sqrt{\kappa-1}\; \tilde a^{M \dag}, \quad  \kappa>1.
\end{eqnarray}
\end{subequations}
The second port of the attenuator or amplifier is given by the corresponding complementary channel,
 \begin{subequations}
 \begin{eqnarray}
 a^{M'} &=&\sqrt{\kappa} \;\tilde a^M + \sqrt{1- \kappa}\; a, \quad  \kappa  \in [0,1], \\
 a^{M'} &=&\sqrt{\kappa} \;\tilde a^M + \sqrt{\kappa-1}\; \tilde a^\dag, \quad  \kappa>1.
\end{eqnarray}
\end{subequations}
The complementary mode described by the annihilation operator $a^{M'}$ contains a fraction of the amplitudes of the input state, and represents the updated state of the memory, {\it i.e.} in the next use of the channel,
the mode $a^{M'}$ will play the role of the previous memory operator $a^M$. Once the initial states of the memory and
of the environment are specified, the action of the channel after $j$ uses is completely determined and can be computed
recursively. The explicit formula for the $j$th output mode can be found in \cite{lupo2010memory} and is not repeated here.
What is important is just the structure of the equations
\begin{subequations}
\begin{align}
a_j'= \sum_{h=1}^{j-1} A_{jh} \, a_h - \sum_{h=0}^{j} E_{jh} \, a^E_h \, , \quad & \kappa  \in [0,1] \, , \label{outputmodes}\\
a_j' = \sum_{h=1}^{j-1}  A_{jh} \, a_h +\sum_{h=0}^{j}  E_{jh} \, a_h^{E \dag} \, ,
\quad & \kappa > 1 \, ,
\end{align}
\end{subequations}
where $A$, $E$, are real matrices and the initial state of the memory has been identified with an additional mode of the environment $a^M=a^E_0$.
 Moreover  the following identities hold
\begin{subequations}
\begin{align}
\sum_{k=1}^n \left(A_{ik}A_{jk}+E_{ik}E_{jk}\right)&=\delta_{ij},\quad \kappa  \in [0,1],\label{relea}\\
\sum_{k=1}^n \left(A_{ik}A_{jk}-E_{ik}E_{jk}\right)&=\delta_{ij}, \quad \kappa>1.
\end{align}
\end{subequations}
This implies that there exist some orthogonal matrices $O,O',O''$ realizing the following singular value decompositions
\cite{lupo2010memory}:
\begin{subequations}
\begin{eqnarray}
A_{jh} &=& \sum_{j'=1}^n O_{jj'} \, \sqrt{\eta^{(n)}_{j'}} \, O'_{j'h}\, , \\
E_{jh} &=& \sum_{j'=1}^n O_{jj'} \, \sqrt{\left|\eta^{(n)}_{j'}-1\right|} \,
O''_{j'h} \,,
\end{eqnarray}
\end{subequations}
where $\eta^{(n)}_j$ are positive real numbers and the matrix $O$ is the same in both decompositions.
In terms of the following set of collective modes:
\begin{subequations}
\begin{eqnarray}
\mathrm{a'}_j &:=& \sum_{j'=1}^n O_{j'j} \, a'_{j'} \, ,  \\
\mathrm{a}_j &:=& \sum_{j'} O'_{jj'} \, a_{j'} \, , \label{a_collective}\\
\mathrm{a}^E_j &:=& \sum_{j'} O''_{jj'} \, a^E_{j'} \, ,
\end{eqnarray}
\end{subequations}
the memory channel is diagonalized into $n$ independent channels,
\begin{subequations}
\begin{align}
\mathrm{a}'_j  & = \sqrt{\eta_j^{(n)}} \,  \mathrm{a}_j -
\sqrt{1-\eta_j^{(n)}} \, \mathrm{a}^E_j \, , \quad \kappa  \in [0,1] \, , \\
\mathrm{a}'_j  & = \sqrt{\eta_j^{(n)}} \,  \mathrm{a}_j +
\sqrt{\eta_j^{(n)}-1} \, \mathrm{a}_j^{E \,\dag} \, , \quad \kappa > 1
\, .
\end{align}
\end{subequations}
In particular, if we focus on the physically relevant case in which all the modes of the environment (and the initial memory mode) are in the same thermal state with a given mean photon number $N$, the modes $\left\{ \mathrm{a}^E_j \right\}$ remain in factorized thermal states and one can conclude that the memory channel applied $n$ times is unitarily equivalent to $n$ independent memoryless attenuators or amplifiers,
\begin{equation}
\Phi_n= \mathcal E_{\eta_1^{(n)}}^{N} \otimes  \mathcal E_{\eta_2^{(n)}}^{N} \dots \otimes   \mathcal E_{\eta_n^{(n)}}^{N} .  \label{factor}
\end{equation}

An important feature of the canonical transformation \eqref{a_collective} is that annihilation operators $a_j$
are not mixed with creation operators $a_{j'}^\dag$. This means that the operation is passive, {\it i.e.}\ it does not change the total energy of the input modes and so the capacity with constrained input energy is the same for the diagonalized channel and the original one.

\subsection{Limit of infinite iterations}
In order to compute the capacity we need to take the limit of infinite iterations of the memory channel.
In virtue of the previous factorization into independent channels, the capacity will depend only
on the asymptotic distribution of the gain parameters $\eta_j^{(n)}$ appearing in \eqref{factor}, in the limit of $n \rightarrow \infty$.
The set of gain parameters $\eta_j^{(n)}$ can be computed as the eigenvalues of the matrix
\begin{equation}
M^{(n)} := A A^\dag \, .
\end{equation}
The entries of the matrix $M$ can be computed from the explicit values of $A$ \cite{lupo2010memory}, obtaining
\begin{equation}\label{theseq}
M^{(n)}_{jj'} = \delta_{jj'} + \left(\kappa_{jj'}-1\right)
\sqrt{\mu\kappa}^{|j-j'|} \, ,
\end{equation}
where
\begin{equation}
\kappa_{jj'} := \kappa + \mu(\kappa-1)^2
\sum_{h=0}^{\min{\{j,j'\}}-2} (\mu\kappa)^h \, .
\end{equation}
The asymptotic behavior of the eigenvalues is different according to whether the combination $\mu\kappa$ is greater or lower than one.
Below threshold, {\it i.e.}\ for $\mu\kappa<1$ the sequence of matrices $M^{(n)}$ is {\it
asymptotically equivalent} \cite{gray2006toeplitz} to the (infinite)
Toeplitz matrix $M^{(\infty)}$, given by
\begin{equation}
M_{jj'}^{(\infty)} := M_{j-j'}^{(\infty)} = \delta_{jj'} - \frac{(1-\mu)(1-\kappa)}{1-\kappa\mu}\sqrt{\mu\kappa}^{|j-j'|} \, .
\end{equation}

We can now exploit the full power of the Toeplitz matrices theory (see Ref. \cite{gray2006toeplitz} for more details): the Szeg\"o theorem \cite{gray2006toeplitz} states that,
for any smooth function $F$, we have
\begin{equation}\label{szego}
\lim_{n\to\infty} \frac{1}{n} \sum_{j=1}^n F\left[\eta^{(n)}_j\right] = \int_0^{2\pi}
\frac{dz}{2\pi} F[\eta(z)] \, ,
\end{equation}
where the function $\eta(z)$ is the Fourier transform of the elements of the matrix $M^{(\infty)}$, \emph{i.e.}
\begin{equation}
\eta(z) = \sum_{j=-\infty}^\infty M^{(\infty)}_{j} e^{iz j/2} = \frac{\kappa+\mu-2\sqrt{\kappa\mu}\cos\frac{z}{2}}{1+\kappa\mu-2\sqrt{\kappa\mu}\cos\frac{z}{2}} \label{monospectrum}\, ,
\end{equation}
with $z\in[ 0, 2\pi]$ (see Fig.s \ref{effectivetrans1}, \ref{effectivetrans2}).
\begin{figure}[t]
\includegraphics[width=\textwidth]{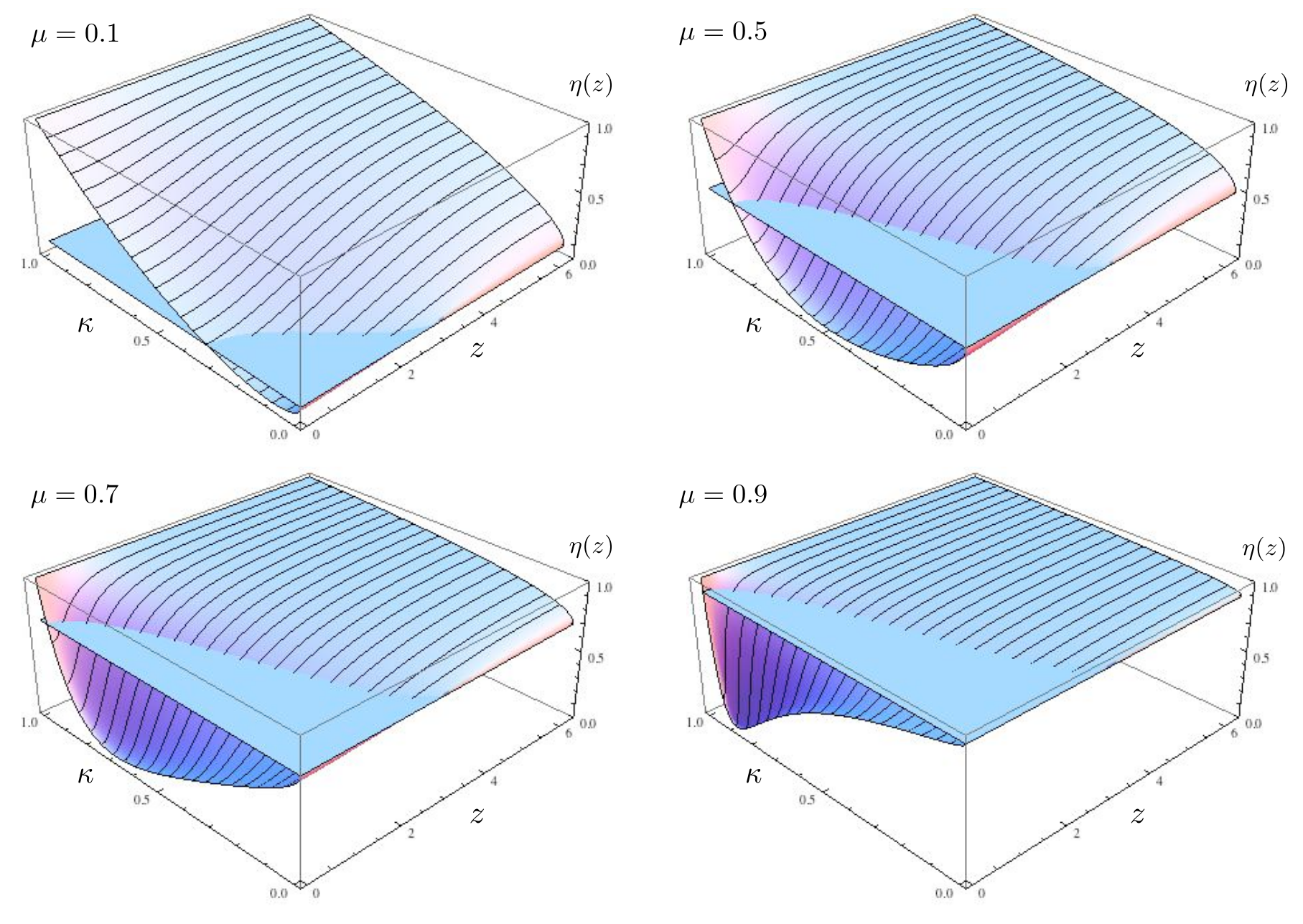}
\caption{Asymptotic spectrum $\eta(z)$ of Eq. \eqref{monospectrum}  for the case where ${\cal E}_{\kappa}$ of Fig. \ref{model}
represents an attenuator
channel (i.e.  $\kappa\in[0,1]$). In this case the system is operated  below the threshold limit $\mu \kappa \leq1$ (no divergency in the spectrum occurs) and the
 values of $\eta(z)$ are always bounded below $1$ (i.e. the channels ${\cal E}_{\eta_j^{(n)}}^{N}$ entering the decomposition \eqref{factor} are attenuators).
In each plot the plane represents the value of $\mu$. We notice that for $\kappa = 0$ one has $\eta(z)= \mu$, while for $\kappa =1$, $\eta(z)=1$ independently from  $\eta$.
}\label{effectivetrans1}
\end{figure}
\begin{figure}[t]
\includegraphics[width=\textwidth]{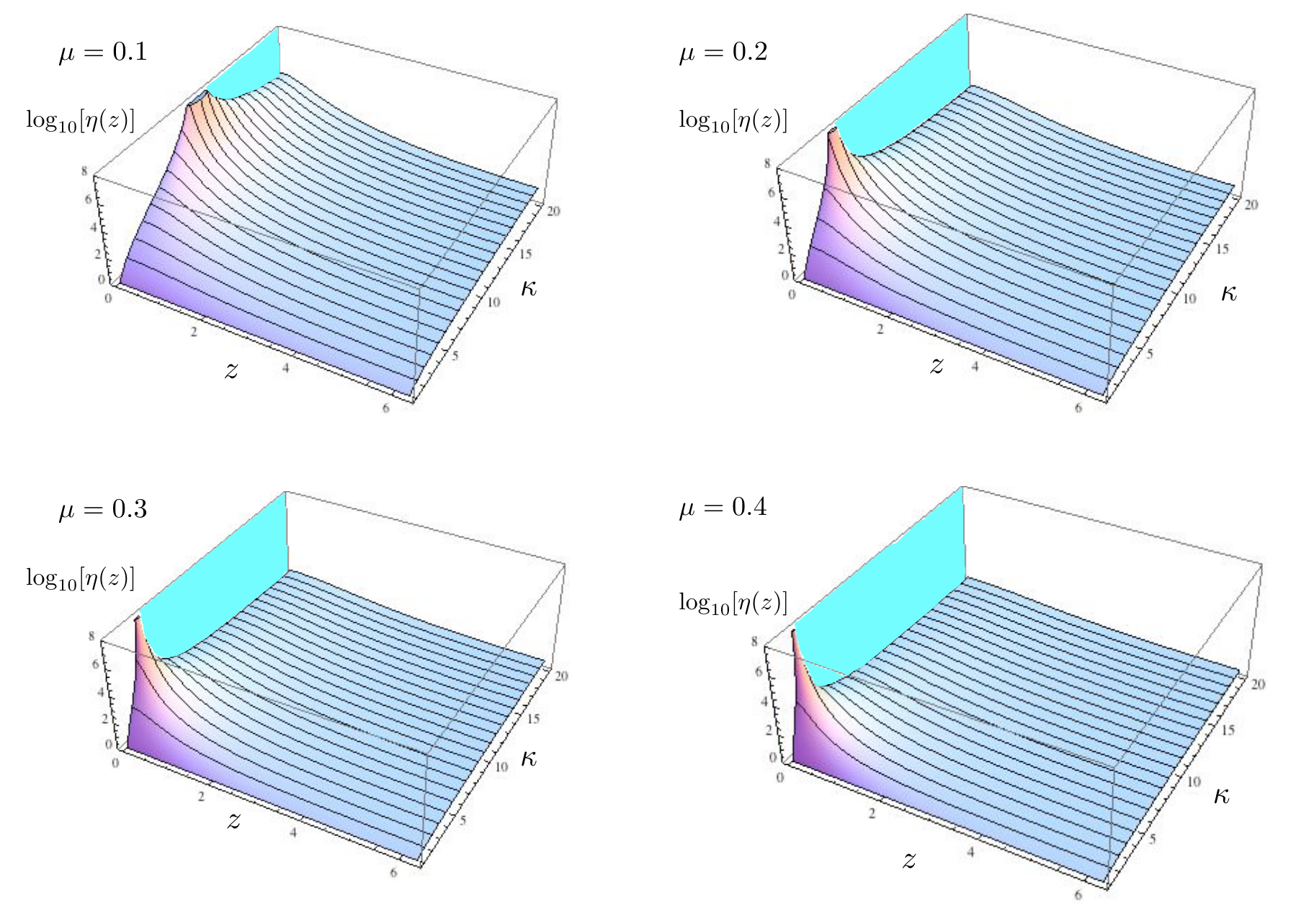}
\caption{Logarithm of the asymptotic spectrum $\eta(z)$ of Eq. \eqref{monospectrum}  for the case where ${\cal E}_{\kappa}$ of Fig. \ref{model}
represents an amplifier
channel (i.e.  $\kappa\geq 1$). In this case the
 values of $\eta(z)$ are always larger than $1$ meaning that the channels ${\cal E}_{\eta_j^{(n)}}^{N}$ entering the decomposition \eqref{factor} describe
 amplifiers. Above threshold (i.e. $\mu \kappa \geq 1$) the system acquires also a divergent, singular eigenvalue represented in the picture by the cyan vertical region.
}\label{effectivetrans2}
\end{figure}

Above threshold, \emph{i.e.} for $\mu\kappa>1$, the sequence of matrices
does not converge. Nonetheless, the divergence can be ascribed to a single diverging eigenvalue, and it is possible to rewrite Eq. \eqref{theseq} as the sum of two terms:
\begin{equation}\label{commuting}
M^{(n)} = c^{(n)} P^{(n)} + \Delta{M}^{(n)} \, ,
\end{equation}
where the $P^{(n)}$ are rank one projectors, $c^{(n)}$ is a
diverging sequence of positive real numbers, and $\Delta{M}^{(n)}$
is a sequence of matrices which asymptotically converges towards the
(infinite) Toeplitz matrix $\Delta{M}^{(\infty)}$, given by
\begin{equation}\label{Mrescaled}
\Delta{M}^{(\infty)}_{jj'} = \delta_{jj'}+\frac{(1-\mu)(\kappa-1)}{\mu\kappa-1}
\frac{1}{\sqrt{ \mu\kappa }^{|j-j'|}} \, .
\end{equation}
(See the Appendix of \cite{lupo2010memory} for the expressions of $P^{(n)}$ and $\Delta{M}^{(n)}$). It is possible to prove that for $n\rightarrow
\infty$, the matrices $P^{(n)}$  and  $\Delta{M}^{(n)}$  commute, and we can conclude that, as anticipated,
 the spectrum of the matrix \eqref{theseq} is
asymptotically composed of only one diverging eigenvalue [corresponding to the
diverging sequence $c^{(n)}$] and of the asymptotic spectrum of the infinite Toeplitz matrix \eqref{Mrescaled}.
As for the below threshold case, the latter is given by the Fourier transform of the matrix elements, where the Fourier transform $\eta(z)$ is given by Eq.\ \eqref{monospectrum} analytically continued to the region $\mu\kappa>1$.

Finally it remains to consider the case $\mu\kappa=1$. At this threshold,  the matrix $M^{(n)}$ can be
expressed as
\begin{equation}\label{theseqTHR}
M^{(n)}_{jj'} = \delta_{jj'} + (1-\mu) + \frac{(1-\mu)^2}{\mu} \min\{j,j'\} \;.
\end{equation}
In this case it appears not feasible to extract the asymptotic spectrum. From a practical point of view however this is not a real problem since any real physical channel will always fall into one of the two classes characterized by $\mu \kappa >1 $ or $\mu \kappa < 1$, respectively.

It is important to stress that for any $\mu\in[0,1]$, in the thermal attenuator case ($\kappa\in[0,1]$) all the channels in the asymptotic diagonal decomposition \eqref{factor} are also thermal attenuators, \emph{i.e.} $\eta(z)\in[0,1]$ for any $z\in[0,2\pi]$. The same happens in the amplifier case, \emph{i.e.} if $\kappa>1$ also $\eta(z)>1$ for any $z\in[0,2\pi]$.

\section{Capacities}\label{sec:capac}
In this Section we will compute the capacity of the memory channel model of Section \ref{sec:gaus}, with the environment in a thermal multimode state with fixed temperature and associated mean photon number per mode $N$.

Let $\Phi_P$ be the mapping describing the input-output relations of the first $P$-channel uses of the model depicted in Fig. \ref{model}.
Since any input influences all the following outputs, its classical capacity cannot be directly computed as in Eq. \eqref{Ctens}.
Still, thanks to the fact that
$\Phi_P$ can be expressed as a tensor product of $P\gg 1$ independent maps
of effective transmissivities $\eta^{(P)}_j$ (see Eq. \eqref{factor}), a close formula for $C$ can be derived.
The fundamental observation here is that, even though in general the  $\eta^{(P)}_j$ will differ from each other,
 for large enough $P$ one can organize them into subgroups each containing a number of elements of order $P$, and characterized by an almost identical
 value of the transmissivity distributed according to the continuous function $\eta(z)$ of Eq. \eqref{monospectrum}.
Let us consider next the channel $\Phi_{2P}$. Its effective transmissivities are different, but they are taken from almost the same distribution, therefore we can write
\begin{equation}
\Phi_{2P}\simeq \Phi_P\otimes\Phi_P\;.
\end{equation}
Iterating, we get
\begin{equation}
\Phi_{\ell P}\simeq\Phi_P^{\otimes \ell}\;,
\end{equation}
and we have managed to express $\Phi_n$ for $n\to\infty$ as the limit of infinite uses of a fixed memoryless channel.

Let's formalize this procedure: we fix $P\gg1$, and take $n=\ell P$. We label the eigenvalues $\eta^{(n)}_j$ in increasing order ($\eta^{(n)}_j\leqslant\eta^{(n)}_{j'}$ if $j<j'$), and divide them into $P$ groups, the $p$th one being made by $\left\{\left.\eta^{(n)}_j\right|(p-1)\ell < j \leqslant
p\ell \right\}$.
Let $\underline{\eta}^{(P)}_p$ and $\overline{\eta}^{(P)}_p$ be respectively the infimum and the supremum of the $p$th group over all $\ell$:
\begin{subequations}
\begin{align}
\underline{\eta}^{(P)}_p = & \inf_\ell \inf_{(p-1)\ell < j \leqslant
p\ell} \eta^{(\ell P)}_j \, , \\
\overline{\eta}^{(P)}_p = & \sup_\ell \sup_{(p-1)\ell < j \leqslant
p\ell} \eta^{(\ell P)}_j \, .
\end{align}
\end{subequations}
Now, the two collections of transmissivities $\underline{\eta}^{(P)}_p$ and $\overline{\eta}^{(P)}_p$ identify two memoryless $P$-mode gaussian channels.
Let $\phi(\eta,N)$ be the Gaussian attenuator / amplifier with transmissivity $\eta\geqslant0$, mixing the input with a thermal state with mean photon number $N$. Remembering that $\phi(\eta,N)\phi(\eta',N)=\phi(\eta\eta',N)$ and that the capacity decreases under composition of channels, if we replace each transmissivity with the supremum or the infimum of its group, the capacity will increase or decrease, respectively.
Each group has exactly $\ell$ eigenvalues, so the $n$ uses of the single mode memory channel can be compared to $\ell$ uses of these two $P$-mode channels, and letting $\ell\to\infty$ we can bound the capacity with
\begin{equation}\label{IMPOC}
\underline{C}^{(P)} \leqslant C \leqslant \overline{C}^{(P)} \, ,
\end{equation}
where $\underline{C}^{(P)}$ and $\overline{C}^{(P)}$ are precisely the capacities of these $P$-mode channels with transmissivities $\left\{ \underline{\eta}^{(P)}_p \right\}$, $\left\{ \overline{\eta}^{(P)}_p \right\}$.
As customary, to keep them finite we impose a constraint on the input mean energy:
\begin{equation}\label{constraint}
\frac{1}{n} \sum_{j=1}^n \mathrm{Tr}\left[ \rho^{(n)} a^\dag_j a_j \right]
\leqslant E \, ,
\end{equation}
where $n$ is the number of uses of the channel and $\rho^{(n)}$ is the joint input density matrix.
As already stressed, this constraint looks identically if expressed in terms of the collective modes \eqref{a_collective}, since they are related to the original ones by an orthogonal matrix.
\subsection{Thermal attenuator}
Let us first consider the case of the attenuating thermal memory channel, \emph{i.e.} $\kappa\leq1$.
It has recently been proven \cite{giovannetti2015solution} that the $\chi$ capacity of successive uses of Gaussian gauge-covariant channels is additive also if they are different:
\begin{equation}
\chi(\Phi_1\otimes\ldots\otimes\Phi_n)=\chi(\Phi_1)+\ldots+\chi(\Phi_n)\;.
\end{equation}
Then the capacity of our two $P$-mode channels can be simply obtained by summing \eqref{capbs} over all modes, yielding the bounds
\begin{subequations}
\label{boundC}
\begin{eqnarray}
\underline{C}^{(P)} & =& \frac{1}{P} \sum_{p=1}^P\left( g\left[\underline{\eta}^{(P)}_p \underline{N}_p+\left(1-\underline{\eta}^{(P)}_p\right)N_T\right]-g\left[\left(1-\underline{\eta}^{(P)}_p\right)N_T\right]\right) \; ,\\
\overline{C}^{(P)}  & =& \frac{1}{P} \sum_{p=1}^P\left(g\left[\overline{\eta}^{(P)}_p \overline{N}_p+\left(1-\overline{\eta}^{(P)}_p\right)N_T\right]-g\left[\left(1-\overline{\eta}^{(P)}_p\right)N_T\right]\right) \; ,
\end{eqnarray}
\end{subequations}
where
\begin{equation}
g(x)=(x+1)\ln(x+1)-x\ln x\;,
\end{equation}
and the parameters $\underline{N}_p$, $\overline{N}_p$ describe the optimal distribution of the mean photon number of the modes and must satisfy the constraints
\begin{subequations}
\begin{align}
&\underline{N}_p,\,\overline{N}_p\geqslant\label{positive}
0\\
&\frac{1}{P}\sum_{p=1}^P \underline{N}_p=\frac{1}{P}\sum_{p=1}^P \overline{N}_p=E\;.
\end{align}
\end{subequations}
If the positivity constraint \eqref{positive} were not there, these optimal values could be computed with the Lagrange multiplier method, yielding
\begin{equation}\label{Lagrange}
\underline{N}_p = \frac{1}{\underline{\eta}^{(P)}_p}\left(\frac{1}{e^{\underline{\lambda}/\underline{\eta}^{(P)}_p} - 1}-\left(1-\underline{\eta}^{(P)}_p\right)N\right)\;,
\end{equation}
and the analog for $\overline{N}_p$.
Taking the limit $P\to\infty$ and applying \eqref{szego}, the two bounds converge to the same quantity and we get
\begin{equation}\label{classical}
C = \int_0^{2\pi} \frac{dz}{2\pi} \left(g\left[\eta(z) N(z)+\left(1-\eta(z)\right)N\right]-g\left[\left(1-\eta(z)\right)N\right]\right) \quad \kappa\in[0,1] \, .
\end{equation}
In the zero temperature case $N=0$ the expression \eqref{Lagrange} is positive definite. As $N$ grows, \eqref{Lagrange} is no more guaranteed to be positive, and we have to impose this constraint by hand. Then, above a certain critical temperature the optimal energy distribution $N(z)$ will vanish for $0\leqslant z\leqslant z_0$. Physically, this means that it is convenient to concentrate all the energy on a fraction $\frac{2\pi-z_0}{2\pi}$ of all the beamsplitters. We will show in Section \ref{sec:lagr} that to determine the optimal energy distribution we can still use the Lagrange multipliers, with the only caveat that $N(z)$ is given now by the positive part of what we would have got without the energy constraint:
\begin{equation}
N(z)=\frac{1}{\eta(z)}\left(\frac{1}{e^\frac{\lambda}{\eta(z)}-1}-(1-\eta(z))N\right)^+\;,\label{nzbs}
\end{equation}
where $f^+(z)=[ f(z) + |f(z)| ]/2$ is the positive part of $f$. The energy constraint reads as expected
\begin{equation}
\int_0^{2\pi} \frac{dz}{2\pi} N(z) = E \, .\label{enconstr}
\end{equation}
We notice that the function $\eta$ is symmetric in $\mu$ and $\kappa$, \emph{i.e.}
\begin{equation}
\eta(\mu,\kappa,z)=\eta(\mu'=\kappa,\;\kappa'=\mu,\;z)\;.
\end{equation}
Since $\mu$ and $\kappa$ appear in the computation of the capacity only through $\eta$, the channel with parameters $(\mu',\kappa')$ has the same capacity of the original one, \emph{i.e.} we can exchange the memory with the transmissivity. Then, varying the memory with fixed transmissivity has the same effect on the capacity as varying the transmissivity for fixed memory. In Fig. \ref{capacityfig} we report the capacity of the channel as a function of the temperature.

\subsection{Thermal amplifier}
The minimum output entropy conjecture lets us compute the capacity also in the amplifier case $\kappa>1$. Now, all the transmissivities are greater than 1, so the capacity decreases as they increase and the two bounds \eqref{boundC} are inverted:
\begin{subequations}
\begin{align}
\overline{C}^{(P)}  =  \frac{1}{P} \sum_{p=1}^P\left( g\left[\underline{\eta}^{(P)}_p \underline{N}_p+\left(\underline{\eta}^{(P)}_p-1\right)\left(N+1\right)\right]-g\left[\left(\underline{\eta}^{(P)}_p-1\right)\left(N+1\right)\right]\right) \; ,\\
\underline{C}^{(P)}   =  \frac{1}{P} \sum_{p=1}^P
\left(g\left[\overline{\eta}^{(P)}_p \overline{N}_p+\left(\overline{\eta}^{(P)}_p-1\right)\left(N+1\right)\right]-g\left[\left(\overline{\eta}^{(P)}_p-1\right)\left(N+1\right)\right]\right) \; .
\end{align}
\end{subequations}
As in the thermal attenuator case, we take the limit $P\to\infty$. Above the threshold ($\mu\kappa>1$) one of the eigenvalues is diverging but, being only one, it does not contribute in the limit, so the capacity is still fully determined by the infinite Toeplitz matrix $\Delta M^{(\infty)}$ yielding
\begin{equation}
C = \int_0^{2\pi} \frac{dz}{2\pi} \left(g\left[\eta(z) N(z)+\left(\eta(z)-1\right)\left(N+1\right)\right]-g\left[\left(\eta(z)-1\right)\left(N+1\right)\right]\right) \label{classicalamp}\; ,
\end{equation}
where as before $N(z)$ is determined by the Lagrange multiplier method, with the caveat of taking the positive part of the resulting function
\begin{equation}
N(z)=\frac{1}{\eta(z)}\left(\frac{1}{e^\frac{\lambda}{\eta(z)}-1}-(\eta(z)-1)(N+1)\right)^+\;,\label{nzamp}
\end{equation}
and with the same constraint on the mean energy
\begin{equation}
\int_0^{2\pi} \frac{dz}{2\pi} N(z) = E \, .
\end{equation}
We notice that in \eqref{nzamp} the positive part is at least in principle necessary also in the case of zero temperature.

Also the amplifier enjoys a sort of duality between $\kappa$ and $\mu$: the function $\eta(\mu,\kappa,z)$ satisfies
\begin{equation}
\eta(\mu,\kappa,z)=\eta\left(\mu'=\frac{1}{\kappa},\;\kappa'=\frac{1}{\mu},\;z\right)\;.
\end{equation}
Noticing that $\kappa'\mu'=\frac{1}{\kappa\mu}$, this relation associates to any channel identified by $(\mu,\kappa)$ above threshold ($\mu\kappa>1$) the new one identified by $(\mu',\kappa')$, which is below threshold.
Then, to investigate the capacity regions as function of the parameters, it is sufficient to consider only the channels below threshold.

In Fig. \ref{capacityfig} we report the capacity of the thermal memory channel as a function of the thermal photon number $N$. As for the thermal attenuator, the capacity is degraded by the temperature and enhanced by the memory.
\begin{figure}[t]
  \includegraphics[width=0.9\textwidth]{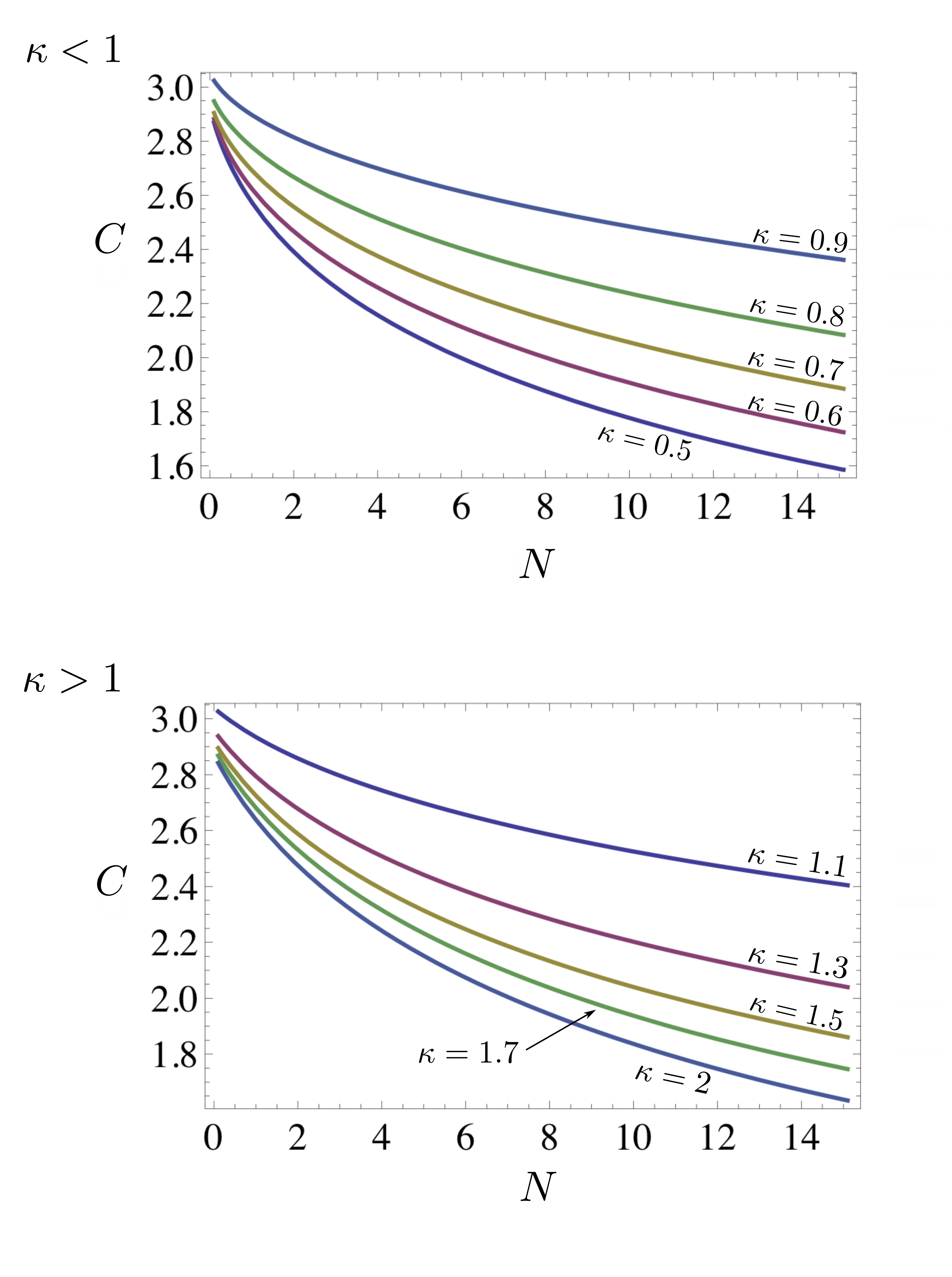}
  \caption{Capacity (in nats / channel use) as a function of the thermal photon number $N$ for $\mu=0.8$ and mean input energy $E=8$ for various values of the transmissivity $\kappa$.  In particular the upper panel refers to the case where the map ${\cal E}_{\kappa}$ of Fig. \ref{model}  is  an attenuator (i.e. $\kappa \in [0,1]$), while the lower panel to the case where  ${\cal E}_{\kappa}$ is an amplifier ($\kappa \geq 1$).    As expected, the capacity is degraded by the temperature and enhanced if the transmissivity is close to unity.}
\label{capacityfig}
\end{figure}

\subsection{Optimal encoding and decoding}
We have seen how the optimal encoding is a coherent-state one with Gaussian weights in the normal mode decomposition $\{\mathrm{a}_j\}$ introduced in Eq. \eqref{a_collective}  in which the channel is diagonal.
They  are related to the input modes $\{ a_j \}$ by a passive orthogonal transformation, and since such transformations send coherent states to coherent states, the latter are also not entangled. However, since the optimal coding requires a non-uniform energy distribution among the $\{\mathrm{a}_j\}$, the  modes $\{ a_j \}$ will be classically correlated. Then this optimal coding can be achieved by independent uses of the channel, but the probabilities of choosing a particular coherent state will be correlated among the various inputs.

Since also in the case of multiple uses of a fixed memoryless channel the optimal decoding requires measures entangled among the various outputs \cite{holevo2013quantum}, in our case the preprocessing with an orthogonal passive transformation to convert the physical basis into the diagonal one does not add further complications to the procedure.

Above threshold ($\mu\kappa>1$), the diverging eigenvalue signals the presence of an input mode that gets amplified by a factor which increases indefinitely with the number of channel uses. Then, even if such mode is left in the vacuum, the corresponding output mode will have a very high energy, and could in principle lead the beamsplitter used in the decoding procedure to a nonlinear regime. The experimentally achievable capacity could then be lower than the theoretical bound, depending on the stability of the decoding device when dealing with high energy inputs.

\subsection{Trivial cases}
There are some particular values of the parameters for which the capacity can be computed analytically.
\begin{itemize}
\item $\kappa=1$ or $\mu=1$

This case corresponds to the identity channel ($\kappa=1$) or to the perfect memory channel ($\mu=1$). In both cases, $\eta(z)=1$ and the capacity is the one of the identity channel with mean energy $E$:
\begin{equation}
C=g(E)\;.
\end{equation}
An intuitive explanation of the result for the perfect memory channel can be given: since $\mu=1$, the first $n$ output modes $\left\{a'_i\right\}$ are a linear combination only of the first $n$ input modes $\left\{a_i\right\}$ and the first memory mode $a_1^M$, and the environment modes $\left\{a_i^E\right\}$ do not play any role. Now we can imagine that in the large $n$ limit the mode $a_1^M$ is no more relevant, and the channel behaves almost as if the output modes were an invertible linear combination of the input ones. This combination can be inverted in the decoding, recovering (almost) the identity channel.

\item $\kappa\to\infty$

This is the case of infinite amplification. Here $\eta(z)=\frac{1}{\mu}$, and the capacity is the one of the amplifier with amplification factor $\frac{1}{\mu}$
\begin{equation}
C=g\left(\frac{E}{\mu}+\frac{1-\mu}{\mu}\left(N+1\right)\right)-g\left(\frac{1-\mu}{\mu}\left(N+1\right)\right)\;.
\end{equation}

\item $\kappa=0$

This is the case of infinite attenuation, in which all the signal is provided by the memory. Here the $n$-th input mode $a_n$ does not influence at all the $n$-th output $a'_n$, but it directly mixes with the $n+1$-th environmental mode $a_{n+1}^E$ through the beamsplitter with transmissivity $\mu$ to give the $n+1$-th output $a'_{n+1}$. Then the only memory effect is a translation of the inputs, and the channel behaves as a thermal attenuator with transmissivity $\mu$. Indeed, as shown in Fig. \ref{effectivetrans1}, here $\eta(z)=\mu$, and the capacity matches the attenuator one \cite{giovannetti2014ultimate}:
\begin{equation}
C=g\left(\mu E+(1-\mu)N\right)-g\left((1-\mu)N\right)\;.
\end{equation}

\item $\mu=0$

This is the memoryless case, and the capacity is the one of the thermal attenuator / amplifier with transmissivity $\kappa$:
\begin{subequations}
\begin{align}
C&=g\left(\kappa E+(1-\kappa)N\right)-g\left((1-\kappa)N\right)\;,\\
C&=g\left(\kappa E+(\kappa-1)\left(N+1\right)\right)-g\left((\kappa-1)\left(N+1\right)\right)\;.
\end{align}
\end{subequations}
\end{itemize}

\subsection{Additive noise channel}
The one--mode additive noise channel adds to the covariance matrix $\sigma$ of the input state a multiple of the identity:
\begin{equation}
\sigma\mapsto\sigma+N_C\mathbb{I}\;.
\end{equation}
A beamsplitter of transmissivity $\eta$, mixing the input with a thermal state with mean photon number $N$, performs instead a convex combination of the corresponding covariance matrices:
\begin{equation}\label{bscm}
\sigma\mapsto\eta\sigma+(1-\eta)\left(N+\frac{1}{2}\right)\mathbb{I}\;.
\end{equation}
The additive noise channel can now be recovered in the limit $\eta\to1^-$ with the second addend of \eqref{bscm} kept fixed, i.e. with
\begin{equation}
(1-\eta)\left(N+\frac{1}{2}\right)=N_C\;,\qquad\eta\to1^-\;,\qquad N\to\infty\;.
\end{equation}
It is then natural to consider what happens to our model for the memory channel in the limit $N\to\infty$, $\kappa\to1^-$ with fixed $(1-\kappa)\left(N+\frac{1}{2}\right)=N_C$. We start from the expression \eqref{outputmodes} which expresses the output modes in terms of the input and the (thermal) environment. From the expressions for the matrices $A$ and $E$ in \cite{lupo2010memory} it is easy to show that, since they do not depend on $N$, their limit for $\kappa\to1$ are $A\to\mathbb{I}$ and $E\to0$, respectively. Physically, this happens because for $\kappa=1$ the channel is the identity and the output is equal to the input. We will now compute the expectation values of all the operators quadratic in the output modes, \emph{i.e.} the output covariance matrix.
We remember that, since the input and the environment are in a completely factorized state,
\begin{subequations}
\begin{eqnarray}
&\left\langle a_i a_j^E\right\rangle=\left\langle a_i^\dag a_j^E\right\rangle=\left\langle a^E_i a_j^E\right\rangle=0\;,\\
&\left\langle {a^E_i}^\dag a_j^E\right\rangle=N\delta_{ij}\;.
\end{eqnarray}
\end{subequations}
We have then
\begin{subequations}
\begin{eqnarray}
\left\langle a_i' a_j'\right\rangle&=&\left\langle a_i a_j\right\rangle\;,\\
\left\langle {a_i'}^\dag a_j'\right\rangle&=&\left\langle {a_i}^\dag a_j\right\rangle+\lim_{N\to\infty}N\sum_k E_{ik}E_{jk}\;,
\end{eqnarray}
\end{subequations}
where the limit is nontrivial since the matrix $E$ depends on $\kappa$, which changes with $N$. Recalling \eqref{relea}
\begin{equation}
AA^T+EE^T=\mathbb{I}\;,
\end{equation}
and from the expression for $AA^T=AA^\dag$ in \cite{lupo2010memory} it is easy to prove that
\begin{equation}
\lim_{N\to\infty}N\sum_k E_{ik}E_{jk}=N_C\mu^\frac{|i-j|}{2}\;,
\end{equation}
so
\begin{equation}
\left\langle {a_i'}^\dag a_j'\right\rangle=\left\langle {a_i}^\dag a_j\right\rangle+N_C\mu^\frac{|i-j|}{2}\label{additivenoise}\;.
\end{equation}
If we look only at a single output mode $a_i'$, throwing away all the others, \eqref{additivenoise} becomes
\begin{equation}
\left\langle {a_i'}^\dag a_i'\right\rangle=\left\langle {a_i}^\dag a_i\right\rangle+N_C\;,
\end{equation}
\emph{i.e.} the reduced channel exactly adds classical noise $N_C$.
However, for nonzero memory ($\mu>0$), $N_C\mu^\frac{|i-j|}{2}$ is nonzero also for $i\neq j$: the added noise is correlated among the various outputs, and the resulting channel is not simply the product of $n$ independent additive noise ones.
We expect this correlation to enhance the capacity: looking at the limit of our formula \eqref{classical}, we will see that it is effectively so. Let's look at this limit in the normal modes variables. Remembering that the environment associated to the operators $\mathrm{a}_j^E$ is still in a factorized thermal state with temperature $N$, we have
\begin{subequations}
\begin{eqnarray}
\left\langle \mathrm{a}_i' \mathrm{a}_j'\right\rangle&=&\left\langle \mathrm{a}_i \mathrm{a}_j\right\rangle\;,\\
\left\langle {\mathrm{a}_i'}^\dag \mathrm{a}_j'\right\rangle&=&\left\langle {\mathrm{a}_i}^\dag \mathrm{a}_j\right\rangle+\delta_{ij}\lim_{N\to\infty}N\left(1-\eta^{(n)}_i\right)\;,
\end{eqnarray}
\end{subequations}
and since
\begin{equation}
\lim_{N\to\infty}(1-\eta(z))N=\frac{N_C(1-\mu)}{1+\mu-2\sqrt{\mu}\cos\frac{z}{2}}\;,\label{distraddnoise}
\end{equation}
in the limit of infinite channel uses we get a factorized additive noise channel, but with the added noise depending on the mode and distributed according to \eqref{distraddnoise}. This model for an additive noise channel with memory coincides with the one considered in \cite{lupo2009forgetfulness,schafer2009capacity}, derived starting from correlated translations with Gaussian weights.

First, we notice that $\eta(z)$ does not depend on $N$, and $\lim_{\kappa\to1}\eta(z)=1$.
Let us compute the limit of the expression for $N(z)$ \eqref{nzbs}:
\begin{equation}
N(z)=\left(\frac{1}{e^\lambda-1}-\lim_{N\to\infty}(1-\eta(z))N\right)^+\;.
\end{equation}
From the expression for $\eta(z)$ \eqref{monospectrum} we can compute the limit

so that
\begin{equation}
N(z)=\left(\frac{1}{e^\lambda-1}-\frac{N_C(1-\mu)}{1+\mu-2\sqrt{\mu}\cos\frac{z}{2}}\right)^+\label{naddit}\;.
\end{equation}
For simplicity, we consider only the case in which the positive part in \eqref{naddit} is not needed.
The mean energy constraint \eqref{enconstr} becomes
\begin{equation}
\frac{1}{e^\lambda-1}=N_C+E\;,
\end{equation}
where we have used that
\begin{equation}
\int_0^{2\pi}\frac{1-\mu}{1+\mu-2\sqrt{\mu}\cos\frac{z}{2}}\frac{dz}{2\pi}=1\;,
\end{equation}
and we have for the positivity constraint on $N(z)$
\begin{equation}
E\geq\frac{2N_C\sqrt{\mu}}{1-\sqrt{\mu}}\;.
\end{equation}
Finally, we can compute the capacity taking the limit of \eqref{classical}:
\begin{equation}
C=g(E+N_C)-\int_0^{2\pi}g\left(\frac{N_C(1-\mu)}{1+\mu-2\sqrt{\mu}\cos\frac{z}{2}}\right)\;.\label{capaddn}
\end{equation}
Since $g(x)$ is concave, the LHS of \eqref{capaddn} decreases if we take the integral inside $g$, so
\begin{equation}
C\geq g(E+N_C)-g(N_C)\;.  \label{boundC2}
\end{equation}
The right-hand-side  of \eqref{boundC2} is exactly the capacity of the single mode additive noise channel, \emph{i.e.} the correlation of the added noise enhances the capacity as expected.

\section{Optimal energy distribution}\label{sec:lagr}
In this Section we will prove that the Lagrange multipliers method with the caveat of taking the positive part in \eqref{nzbs} and \eqref{nzamp} works also with the positivity constraint \eqref{positive}, and we will analyze the resulting optimal energy distribution $N(z)$.

\subsection{The proof}
The function $\eta(z)$ is increasing for the thermal attenuator ($\kappa<1$) and decreasing for the amplifier ($\kappa>1$), \emph{i.e.} the channel with transmissivity $\eta(z)$ always improves as $z$ increases.
For simplicity here we consider only the thermal attenuator case, the amplifier one being completely analogous.

Let $\widetilde{N}(z,w)$ be the Lagrange multipliers solution in the interval $w\leqslant z\leqslant2\pi$ which maximizes the capacity
\begin{equation}
C = \int_w^{2\pi} \frac{dz}{2\pi} \left(g\left[\eta(z) \widetilde{N}(z,w)+\left(\eta(z)-1\right)\left(N+1\right)\right]-g\left[\left(\eta(z)-1\right)\left(N+1\right)\right]\right)\label{capw}
\end{equation}
with the mean energy constraint
\begin{equation}
\int_{w}^{2\pi}\frac{dz}{2\pi}\widetilde{N}(z,w)dz=E\;,\label{constrn}
\end{equation}
where the integrals are restricted to $w\leqslant z\leqslant2\pi$ and we do not care about the positivity of $N(z,w)$.
Such solution is given by
\begin{equation}
\widetilde{N}(z,w)=\frac{1}{\eta(z)}\left(\frac{1}{e^\frac{\lambda}{\eta(z)}-1}-(1-\eta(z))N\right),
\end{equation}
where the multiplier $\lambda$ is determined by the constraint \eqref{constrn} (strictly speaking, with $\widetilde{N}(z,w)$ we mean the function analytically continued to the whole interval $0\leqslant z\leqslant2\pi$).

Let $N(z)$ be the optimal positive distribution of the photons. Since it is better to use more energy in the better channels, $N(z)$ must be increasing: if not, we could move a bit of energy from a bad channel to a better one with less energy, and this would increase the capacity.
Let $N(z)$ be zero for $0\leqslant z<z_0$, and strictly positive for $z_0<z\leqslant2\pi$.
In particular $N(z)$ is the optimal solution among all the functions equal to zero for $0\leqslant z<z_0$ and strictly positive for $z_0<z\leqslant2\pi$.
We consider all the infinitesimal variations $N(z)+\delta N(z)$ satisfying the mean energy constraint and such that $\delta N(z)$ is nonzero only in the interval $z_0<z\leqslant2\pi$. Since $N(z)$ is strictly positive there, $N(z)+\delta N(z)$ is still positive for infinitesimal $\delta N$, so it is a legal positive photon distribution. For its optimality $N(z)$ must be a stationary point of the capacity for all such variations, but this means exactly that $N(z)$ is the solution of the Lagrange multipliers method $\widetilde{N}(z,z_0)$:
\begin{equation}
N(z)=\widetilde{N}(z,z_0)\theta(z-z_0)\;,
\end{equation}
where $\theta(z)$ is the step function.

We now claim that $\widetilde{N}(z_0,z_0)$ must be zero. Let us suppose $\widetilde{N}(z_0,z_0)>0$.
Since $\widetilde{N}(z,w)$ is continuous in $w$, we can choose a $w_0<z_0$ such that $\widetilde{N}(z,w_0)$ is strictly positive in the whole interval $w_0<z\leqslant2\pi$. Then, $\widetilde{N}(z,w_0)\theta(z-w_0)$ is an admissible solution. Since also $N(z)$ has been considered in the maximization problem \eqref{capw} defining $\widetilde{N}(z,w_0)$, the latter must achieve a greater capacity than the former, impossible.

For the same argument used with $N(z)$, $\widetilde{N}(z,z_0)$ must be increasing within each interval where it is positive, and since it is continuous in $z$ it must be negative for $0\leqslant z<z_0$ and positive for $z_0<z\leqslant2\pi$. Then we can finally write as promised $N(z)$ as
\begin{equation}
N(z)=\frac{1}{\eta(z)}\left(\frac{1}{e^\frac{\lambda}{\eta(z)}-1}-(1-\eta(z))N\right)^+\;,\qquad z\in[0,2\pi],\label{nz}
\end{equation}
where $f^+(z)$ is the positive part of $f$.

\subsection{Analysis of the optimal distribution}
The typical behavior of $N(z)$ in the attenuator case is shown in Fig. \ref{n(z)}. It is increasing, as it has to be. We can identify a critical temperature $N_{crit}$, that for our choice of the parameters ($\kappa=0.9$, $\mu=0.8$, $E=8$) is nearly $N_{crit}\sim0.8$. Below this critical value, $N(z)$ approaches a constant positive value for $z\to0$, \emph{i.e.} the optimal configuration exploits all the beamsplitters. Above the critical value, $N(z)$ is zero on a finite interval $[0,z_0]$, \emph{i.e.} the optimal configuration does not use at all a finite fraction $\frac{z_0}{2\pi}$ of the beamsplitters, being more convenient to concentrate all the energy on the other ones.

The behavior of $N(z)$ in the amplifier case is shown in Fig. \ref{n(z)amp}. It is completely analogous to the thermal attenuator, but for our choice of the parameters ($\kappa=1.1$, $\mu=0.8$, $E=8$) the critical temperature is much greater, $N_{crit}\sim 9.8$.

An analysis of the fraction $\frac{z_0}{2\pi}$ (remember that $z_0$ ranges from $0$ to $2\pi$) of the unused beamsplitters is presented in Fig. \ref{fig:Z0}. For fixed $\kappa$ and $\mu$, for zero temperature ($N=0$) all the beamsplitters are exploited and $z_0=0$; then $z_0$ remains zero up to the critical temperature $N_{crit}$, and grows for $N>N_{crit}$.
We can notice that for typical parameters, the critical value $N_{crit}$ for the beamsplitter is much lower than for the amplifier.

We will now show that in the infinite temperature limit ($N\to\infty$), $z_0$ tends to $2\pi$, and the optimal configuration concentrates all the energy on an infinitesimal fraction of the beamsplitters. We first notice that for $N\to\infty$ the multiplier $\lambda$ in \eqref{nz} must tend to zero, and we can approximate $e^{\lambda/\eta}-1\sim\lambda/\eta$, getting
\begin{equation}
N(z)=\left(\frac{1}{\lambda}-\left(\frac{1}{\eta(z)}-1\right)N\right)\theta(z-z_0)+\mathcal{O}(1)\;,
\end{equation}
where $z_0$ is the point where $N(z)$ vanishes, given by
\begin{equation}
\frac{1}{\lambda}=\left(\frac{1}{\eta(z_0)}-1\right)N\;.
\end{equation}
The energy constraint \eqref{enconstr} can be now written as
\begin{equation}
E=N\int_{z_0}^{2\pi}\left(\frac{1}{\eta(z_0)}-\frac{1}{\eta(z)}\right)\frac{dz}{2\pi}+\mathcal{O}(1)\;,
\end{equation}
and since $\eta(z)$ is strictly increasing, the only way to keep $E$ finite for $N\to\infty$ is to let $z_0\to2\pi$, \emph{i.e.} in the high temperature limit all the energy is concentrated on an infinitesimal fraction of the beamsplitters.

The minimum energy $E_{crit}$ for which all the beamsplitters are exploited is shown in Fig. \ref{fig:Encrit} for various values of the temperature $N$.
We know that for $\kappa=0,1$ and $\kappa\to\infty$ no beamsplitter is left unused, and indeed $E_{crit}=0$ at these points.
As expected, $E_{crit}$ always grows with the temperature.
In the attenuator case, we notice a divergence of $E_{crit}$ for $\kappa=\mu$ ($\mu=0.8$ in the plot). Actually, if $\kappa=\mu$ we have $\eta(0)=0$ (while in any other case $\eta(z)$ is always positive), and some normal modes have infinitesimal transmissivity. It is then natural that for any nonzero temperature it is not convenient to send energy into these low-capacity modes. More formally, the argument of the positive part in \eqref{nz} in the case $\kappa=\mu$ in $z=0$ is $-N<0$, so for any $N>0$ the positive part must be taken into account.

\section{Conclusion}\label{sec:conc}

We have studied a model of Gaussian thermal memory channels extending a previous proposal by Lupo {\it et al.} \cite{lupo2010capacities,lupo2010memory} in order to incorporate the disturbance of thermal noise. The memory effects imply that successive uses of a channel cannot be
considered independently but they are potentially correlated \cite{gallager2014information,kretschmann2005quantum}. In our model this correlation is generated by an internal memory mode which is assumed to be unaccessible by the users of the channel.

Exploiting the factorization into independent normal modes \cite{lupo2010memory} and the recent determination of the capacity of memoryless gauge-covariant Gaussian channels (see \cite{giovannetti2014ultimate} and Section \ref{gccapacity}), we explicitly determine the classical capacity of our memory channel model. We find that, as in the memoryless case, coherent states are sufficient for an optimal coding. However, the associated probability distribution is factorized only in the normal mode decomposition that diagonalizes the channel, so in order to fully exploit its intrinsic memory, the  input signals $\{ a_j\}$ (and consequently their  outputs counterparts) must be correlated. Then the optimal transmission rate of information can still be achieved by independent uses of the channel, but the probability distribution of the physical inputs will not be factorized.

Our results can find applications in  bosonic communication channels with memory effects and affected by a non-negligible amount of thermal noise. In particular low frequency communication devices, {\it e.g.}\  GHz communication systems \cite{lang2013correlations}, THz lasers \cite{kohler2002terahertz}, {\it etc.}, are intrinsically subject to black-body thermal noise and thus they fall in the theoretical framework presented here.

\begin{figure}[ht]
  \includegraphics[width=\textwidth]{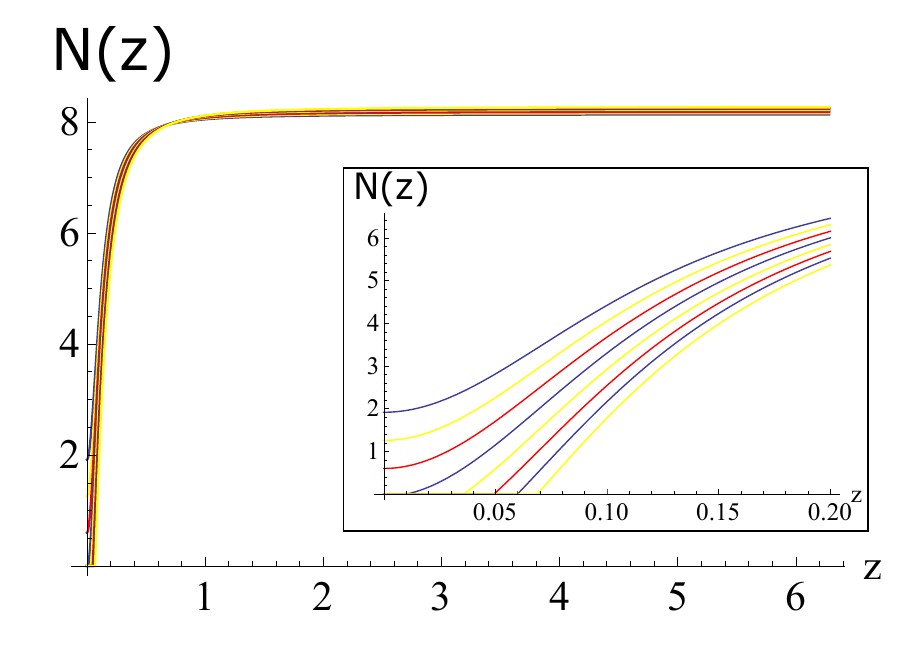}
  \caption{Behavior of the energy density $N(z)$ for $\kappa=0.9$, $\mu=0.8$, $E=8$ and $N$ ranging in steps of 0.1 from top to bottom from 0.5 to 1.2, near to the critical temperature $N_{crit}\sim0.8$. As expected, $N(z)$ is always increasing. If we exclude the region near $z=0$, the functions are almost identical and approach nearly the same constant value for $z\gtrsim1$. Inset:
  Zoom on the region $z\to0$. We can see that above the critical temperature $N(z)$ is zero on a finite interval, while below it $N(z)$ approaches a positive value which strongly depends on the temperature.}\label{n(z)}
\end{figure}
\begin{figure}[ht]
  \includegraphics[width=\textwidth]{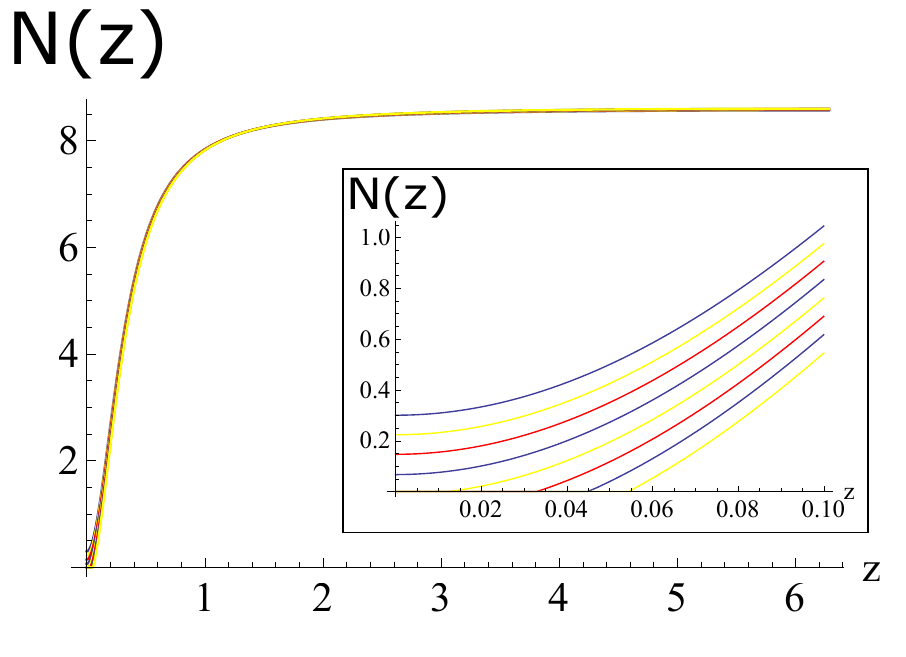}
  \caption{Behavior of the energy density $N(z)$ for $\kappa=1.1$, $\mu=0.8$, $E=8$ and $N$ ranging in steps of 0.1 from top to bottom from 9.4 to 10.1, near to the critical temperature $N_{crit}\sim9.8$. As expected, $N(z)$ is always increasing. If we exclude the region near $z=0$, the functions are almost identical and approach nearly the same constant value for $z\gtrsim1$.
  Inset: Zoom on the region $z\to0$. We can see that above the critical temperature $N(z)$ is zero on a finite interval, while below it $N(z)$ approaches a positive value which strongly depends on the temperature.}\label{n(z)amp}
\end{figure}
\begin{figure}[ht]
  \includegraphics[width=0.8\textwidth]{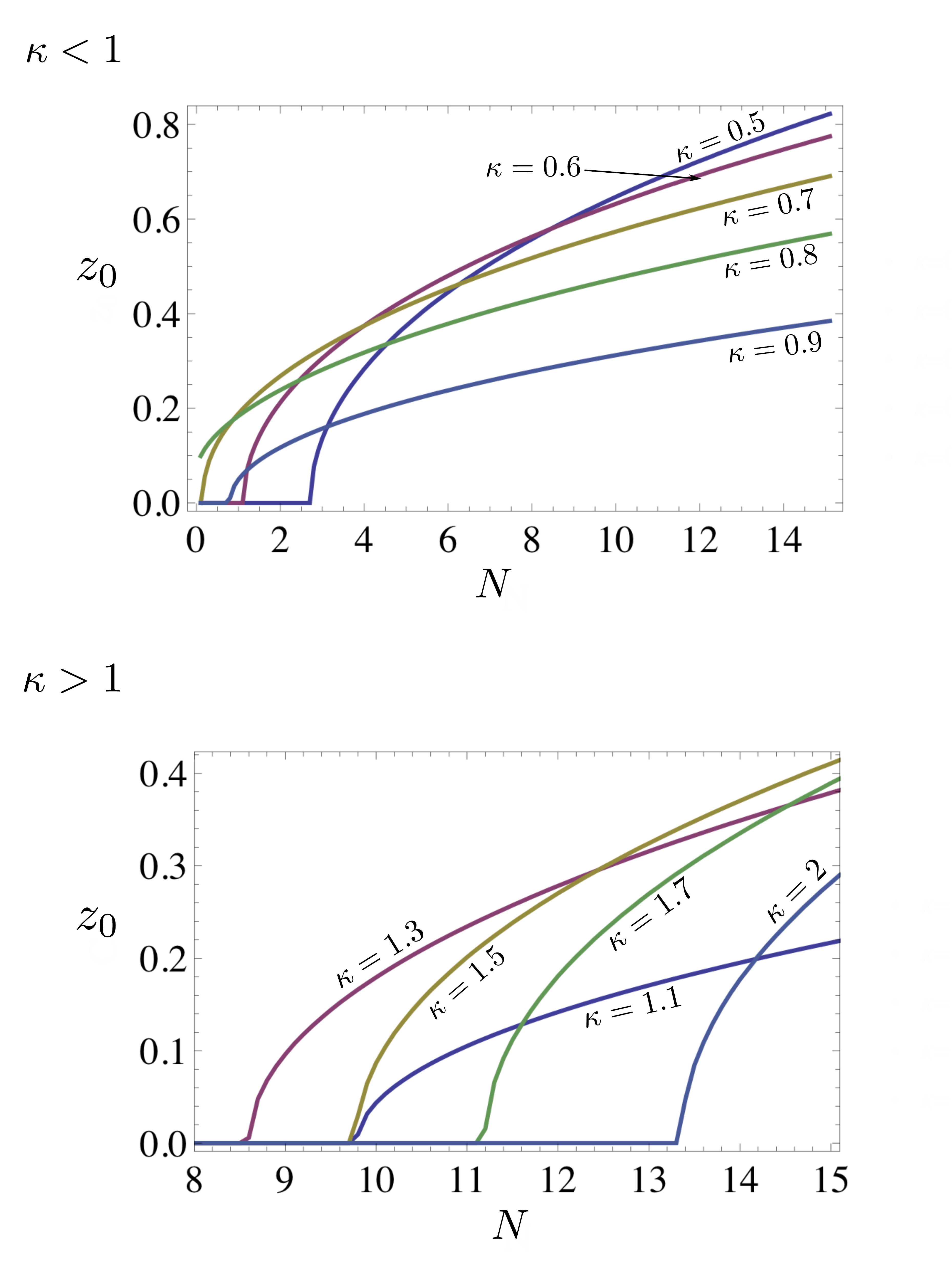}
  \caption{Behavior of the fraction $\frac{z_0}{2\pi}$ ($z_0$ ranges from $0$ to $2\pi$) of unused beamsplitters as a function of the temperature $N$ for $E=8$, $\mu=0.8$ and various values of $\kappa$. At zero temperature ($N=0$) all the beamsplitters are exploited and $z_0=0$; then $z_0$ remains zero up to the critical temperature $N_{crit}$, and grows for $N>N_{crit}$. We notice that for typical values of the parameters $N_{crit}$ is much greater for $\kappa>1$ than for $0<\kappa<1$. In the infinite temperature limit $N\to\infty$ only an infinitesimal fraction of the beamsplitters is used and $z_0$ tends to $2\pi$, even if this is not evident from the plots due to the limited range of $N$.}\label{fig:Z0}
\end{figure}
\begin{figure}[ht]
  \includegraphics[width=0.9\textwidth]{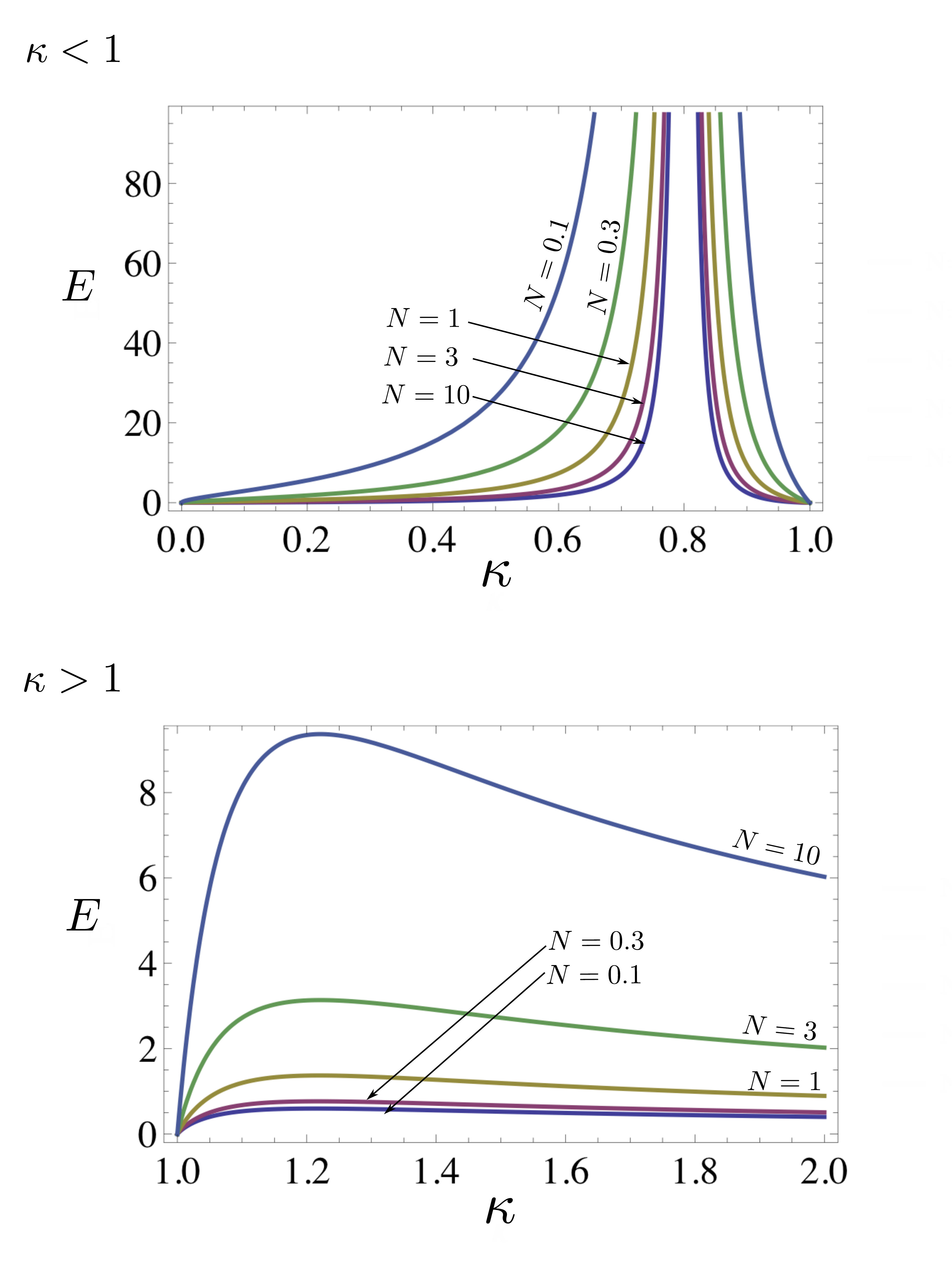}
  \caption{Behavior of the minimal energy for which all the beamsplitters are exploited as a function of $\kappa$ for $\mu=0.8$ and various values of the temperature $N$. As expected, $E$ grows with the temperature, and $E=0$ for $\kappa=0,1$ and $\kappa\to\infty$. In the attenuator case we notice the divergence of $E$ for $\kappa=\mu$ ($=0.8$), due to the fact that $\eta(0)=0$ and for any positive temperature the optimal $N(z)$ must vanish on a finite interval.}\label{fig:Encrit}
\end{figure}

\chapter{Normal form decomposition for Gaussian-to-Gaussian superoperators}\label{normal}
In this Chapter, we explore the set of the linear trace preserving not necessarily positive maps sending the set of quantum Gaussian states into itself.
These maps can be exploited as a test to check whether a given quantum state belongs to the convex hull of Gaussian states, exactly as positive but not completely positive maps are tests for entanglement.
For one mode, we prove that these maps are all built from the phase-space dilatation, that is hence the only relevant test of this kind.

The Chapter is based on
\begin{enumerate}
\item[\cite{de2015normal}] G.~De~Palma, A.~Mari, V.~Giovannetti, and A.~S. Holevo, ``Normal form decomposition for Gaussian-to-Gaussian superoperators,'' \emph{Journal of Mathematical Physics}, vol.~56, no.~5, p. 052202, 2015.\\ {\small\url{http://scitation.aip.org/content/aip/journal/jmp/56/5/10.1063/1.4921265}}
\end{enumerate}

\section{Introduction}

As we have seen in Chapter \ref{GQI}, Gaussian Bosonic States (GBSs) play a fundamental role in the study of
continuous-variable (CV) quantum information processing \cite
{holevo2011probabilistic,ferraro2005gaussian,braunstein2005quantum,weedbrook2012gaussian} with applications in quantum cryptography,
quantum computation and quantum communication where they are known to
provide optimal ensembles for a large class of quantum communication lines
(specifically the gauge-covariant Gaussian Bosonic maps) \cite{holevo2001evaluating,giovannetti2006quantum,giovannetti2004classical,wolf2007quantum,giovannetti2015solution,giovannetti2014ultimate}. GBSs are characterized by the
property of having Gaussian Wigner quasi-distribution (see Section \ref{chia} of Appendix \ref{appG}) and describe Gibbs
states of Hamiltonians which are quadratic in the field operators of the
system. Further, in quantum optics they include coherent, thermal and
squeezed states of light and can be easily created via linear amplification
and loss.

Directly related to the definition of GBSs is the notion of Gaussian
transformations \cite{holevo2011probabilistic,braunstein2005quantum,weedbrook2012gaussian}, i.e. superoperators mapping the
set $\mathfrak{G}$ of GBSs into itself. In the last two decades, a great
deal of attention has been devoted to characterizing these objects. In
particular the community focused on Gaussian Bosonic Channels (GBCs) \cite
{holevo2001evaluating}, i.e. Gaussian transformations which are completely positive (CP)
and provide hence the proper mathematical representation of data-processing
and quantum communication procedures which are physically implementable \cite
{caves1994quantum}. On the contrary, less attention has been devoted to the study of
Gaussian superoperators which are not CP or even non-positive. A typical
example of such mappings is provided by the phase-space dilatation, which,
given the Wigner quasi-distribution $W_{\hat{\rho}}(\mathbf{r})$ (see \eqref{wignerdef}) of a
state $\hat{\rho}$ of $n$ Bosonic modes, yields the function $
W^{(\lambda)}_{\hat{\rho}}( \mathbf{r}) \equiv W_{\hat{\rho}}(\mathbf{r}
/\lambda)/\lambda^{2n}$ as an output, with the real parameter $\lambda$
satisfying the condition $|\lambda|>1$.
On one hand, when acting on $\mathfrak{G}$ the mapping
\begin{eqnarray}
W_{\hat{\rho}}(\mathbf{r})\mapsto W^{(\lambda)}_{\hat{\rho}}( \mathbf{r}) \;,
\label{MAP}
\end{eqnarray}
always outputs proper (Gaussian) states. Specifically, given $\hat{\rho}\in
\mathfrak{G}$ one can identify another Gaussian density operator $\hat{\rho}
^{\prime}$ which admits the function $W^{(\lambda)}_{\hat{\rho}}( \mathbf{r}
) $ as Wigner distribution, i.e. $W_{\hat{\rho}^{\prime}}(\mathbf{r})=
W^{(\lambda)}_{\hat{\rho}}( \mathbf{r})$. On the other hand, there exist
inputs $\hat{\rho}$ for which $W^{(\lambda)}_{\hat{\rho}}( \mathbf{r})$ is
no longer interpretable as the Wigner quasi-distribution of \textit{any}
quantum state: in this case in fact $W^{(\lambda)}_{\hat{\rho}}( \mathbf{r})$
is the Wigner quasi-distribution $W_{\hat{\theta}}(\mathbf{r})$
of an operator $\hat{\theta}$ which is not positive \cite{brocker1995mixed} (for example, any pure non-Gaussian state has this property for any $\lambda\neq\pm1$ \cite{dias2009narcowich}).
Accordingly
phase-space dilatations \eqref{MAP} should be considered as ``unphysical''
transformations, i.e. mappings which do not admit implementations in the
laboratory. Still dilatations and similar exotic Gaussian-to-Gaussian
mappings turn out to be useful mathematical tools that can be employed to
characterize the set of states of CV systems in a way which is not
dissimilar to what happens for positive (but not completely positive)
transformations in the analysis of entanglement \cite{horodecki2009quantum}. In particular
Br\"ocker and Werner \cite{brocker1995mixed} used \eqref{MAP} to study the convex hull
$\mathfrak{C}$ of Gaussian states (i.e. the set of density operators $\hat{
\rho}$ which can be expressed as a convex combination of elements of $
\mathfrak{G}$). The rationale of this analysis is that the set $\mathfrak{F}$
of density operators which are mapped into proper output states by this
transformation includes $\mathfrak{C}$ as a proper subset, see Fig. \ref
{FigSys}. Accordingly if a certain input $\hat{\rho}$ yields a $
W^{(\lambda)}_{\hat{\rho}}( \mathbf{r})$ which is not the Wigner
distribution of a state, we can conclude that $\hat{\rho}$ is not an element
of $\mathfrak{C}$.
Finding mathematical and experimental criteria which help
in identifying the boundaries of $\mathfrak{C}$ is indeed a timely and
important issue which is ultimately related with the characterization of
non-classical behavior in CV systems, see e.g. Ref.'s \cite{filip2011detecting,genoni2013detecting,genoni2013detecting,hughes2014quantum,mari2011directly,kiesel2011nonclassicality,richter2002nonclassicality,jevzek2011experimental,genoni2007measure,genoni2008quantifying,genoni2010quantifying}, and also \cite{vershynina2014complete,vershynina2014complete} for the fermionic case.

In this context a classification of non-positive Gaussian-to-Gaussian
operations is mandatory.
This analysis has been initiated in \cite{giedke2002characterization}, where Gaussian-to-Gaussian maps are characterized through their
Choi-Jamio\l kowski state, under the hypothesis that this state has a Gaussian characteristic function.
One goal is proving this hypothesis: we prove
that the action of such
transformations on the covariance matrix and on the first moment must be
linear, and we write explicitly the transformation properties of the
characteristic function (Theorem \ref{gaussthmq}).
In the classical case, any probability measure can be written as a convex
superposition of Dirac deltas, so the convex hull of the Gaussian measures
coincides with the whole set of measures. A simple consequence of this property is that a linear transformation
sending Gaussian measures into Gaussian (and then positive) measures is
always positive. Nothing of this holds in the more interesting quantum case,
so we focus on it, and use Theorem \ref{gaussthmq} to get a decomposition
which, for single-mode operations, shows that any linear quantum
Gaussian-to-Gaussian transformation can always be
decomposed as a proper combination of a dilatation \eqref{MAP} followed by a
CP Gaussian mapping plus possibly a transposition. We also show that our
decomposition theorem applies to the multimode case, as long as we restrict
the analysis to Gaussian transformations which are homogeneous at the level
of covariance matrix. For completeness we finally discuss the case of
contractions: these are mappings of the form \eqref{MAP} with $|\lambda|<1$.
They are not proper Gaussian transformations because they map some Gaussian
states into non-positive operators. Still some of the results which apply to
the dilatations can be extended to this set.

\begin{figure}[t]
\centering
\includegraphics[width=0.4\textwidth]{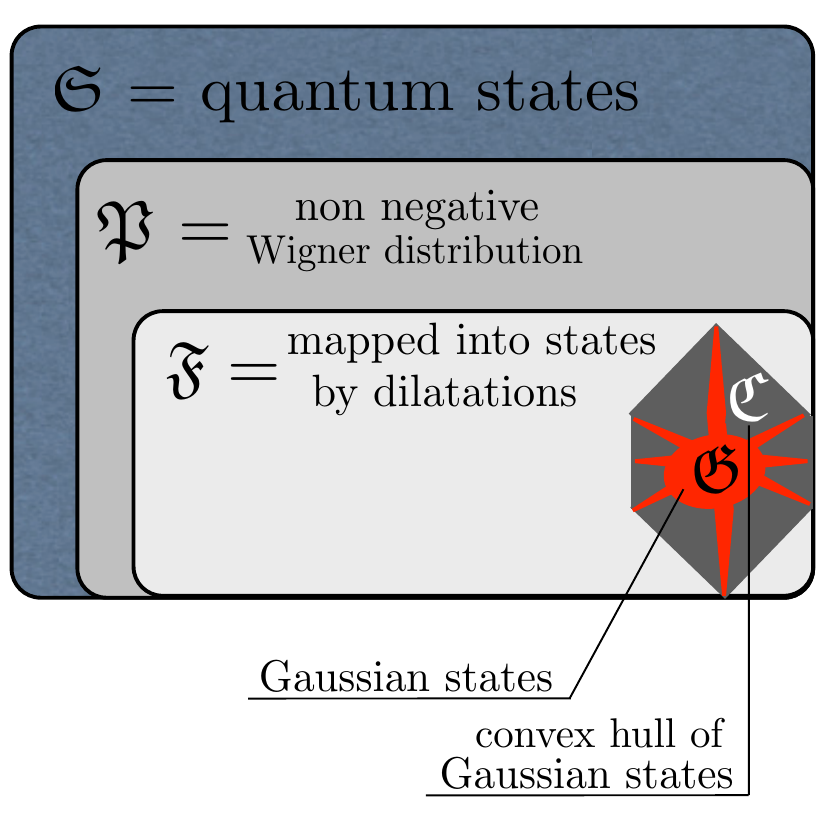}
\caption{Pictorial representation of the structure of the set
of states $\mathfrak{S}$ of a CV system. $\mathfrak{P}$ is the subset of
density operators $\hat{\protect\rho}$ which have non-negative Wigner
distribution (\protect\ref{wignerdef}). $\mathfrak{F}$ is set of states which instead are mapped into proper density operators by an arbitrary
dilatation (\protect\ref{MAP}).
$\mathfrak{G}$ is the set of Gaussian states
and $\mathfrak{C}$ its convex hull.
$\mathfrak{S}$, $\mathfrak{P}$, $
\mathfrak{F}$, and $\mathfrak{C}$ are closed under convex convolution, $
\mathfrak{G}$ is not. For a detailed study of the relations among these sets
see Ref. \protect\cite{brocker1995mixed}. }
\label{FigSys}
\end{figure}

The Chapter is organized as follows. In Section \ref{hullsec} we define the convex hull of Gaussian states. In Section \ref{S:PROBLEM}
we state the problem and prove Theorem \ref{gaussthmq} characterizing the action of Gaussian-to-Gaussian
superoperators on the characteristic functions of quantum states and its variations, including the probabilistic analog. In subsection \ref{contr} we consider the case of contractions. In Section \ref{S:ONE} we present the main result of the
chapter, i.e. the decomposition theorem for single-mode
Gaussian-to-Gaussian transformations. The multimode case is then analyzed
in Section \ref{S:MULTI}.
In Section \ref{app} we prove the unboundedness of phase-space dilatations with respect to the trace norm.
The Chapter ends hence with Section \ref{S:CON} where we
present a brief summary and discuss some possible future developments.

\section{The convex hull of Gaussian states}\label{hullsec}
States with positive Wigner function \eqref{wignerdef} form a convex subset $
\mathfrak{P}$ in the space of the density operators $\mathfrak{S}$ of the
system. The set $\mathfrak{G}$ of Gaussian states is a proper subset of $
\mathfrak{P}$.

Starting from the vacuum, devices as simple as
beamsplitters combined with one-mode squeezers permit (at least in
principle) to realize all the elements of $\mathfrak{G}$. Then, choosing
randomly according to a certain probability distribution which Gaussian
state to produce, it is in principle possible to realize all the states in
the convex hull $\mathfrak{C}$ of the Gaussian ones, i.e. all the states $
\hat{\rho}$ that can be written as
\begin{equation}
\hat{\rho}=\int\hat{\rho}_G(\sigma,\mathbf{x})\;d\mu(\mathbf{x},\sigma)\;,
\label{gaussenv}
\end{equation}
where $\hat{\rho}_G(\sigma,\mathbf{x})$ is the Gaussian state with first moment $\mathbf{x}$ and covariance matrix $\sigma$ (see Section \ref{Gsta}), and $\mu$ is the associated probability measure of the process.

It is easy to verify that $\mathfrak{C}$ is strictly larger than $\mathfrak{G
}$, i.e. there exist states \eqref{gaussenv} which are not Gaussian. On the
other hand, one can observe that \eqref{gaussenv} implies
\begin{equation}
W_{\hat{\rho}}(\mathbf{r})=\int \frac{1}{\sqrt{\det\left(\pi\,\sigma\right)}}
\;e^{-(\mathbf{r}-\mathbf{x})^T\sigma^{-1}(\mathbf{r}-\mathbf{x})}\;d\mu(
\mathbf{x},\sigma)>0\;,
\end{equation}
so also $\mathfrak{C}$ is included into $\mathfrak{P}$, see Fig. \ref{FigSys}
. There are however elements of $\mathfrak{P}$ which are not contained in $
\mathfrak{C}$: for example, any finite mixture of Fock states
\begin{equation}
\hat{\rho}=\sum_{n=0}^Np_n|n\rangle\langle n|\qquad N<\infty\qquad
p_n\geq0\qquad\sum_{n=0}^Np_n=1
\end{equation}
is not even contained in the weak closure of $\mathfrak{C}$, even if some of
them have positive Wigner function \cite{brocker1995mixed}.

\section{Characterization of Gaussian-to-Gaussian maps}

\label{S:PROBLEM}

Determining whether a given state $\hat{\rho}$ belongs to the convex hull $
\mathfrak{C}$ of the Gaussian set is a difficult problem \cite
{filip2011detecting,genoni2013detecting,hughes2014quantum}. Then, there comes the need to find criteria to
certify that $\hat{\rho}$ \emph{cannot} be written in the form
\eqref{gaussenv}. A possible idea is to consider a non-positive
superoperator $\Phi$ sending any Gaussian state into a state \cite{brocker1995mixed}.
By linearity $\Phi$ will also send any state of $\mathfrak{C}$ into a state,
therefore if $\Phi(\hat{\rho})$ is not a state, $\hat{\rho}$ cannot be an
element of $\mathfrak{C}$: in other words, the transformation $\Phi$ acts as
a mathematical \textit{probe} for $\mathfrak{C}$. In what follows we will
focus on those probes which are also Gaussian transformations, i.e. which
not only send $\mathfrak{G}$ into states, but which ensure that the output
states $\Phi(\hat{\rho})$ are again elements of $\mathfrak{G}$. Then the
following characterization theorem holds

\begin{thm}
\label{gaussthmq} Let $\Phi $ be a linear bounded map of
the space $\mathfrak{H}$ of Hilbert-Schmidt operators (see \eqref{HSN} in Appendix \ref{appG} for the definition), sending the set of Gaussian states $\mathfrak{G}$
into itself. Then its action in terms of the characteristic function (see \eqref{chidef}), the first moments and
the covariance matrix (see \eqref{covdef}) is of the form
\begin{eqnarray}
\Phi &:&\chi (\mathbf{k})\rightarrow \chi \left( \mathbf{k}K\right)
\;e^{-\frac{1}{4}\mathbf{k}\alpha \mathbf{k}^T+i\mathbf{k}\mathbf{y}_0
}\;,  \label{channel} \\
\Phi &:&\mathbf{x} \rightarrow K\mathbf{x}+\mathbf{y}_0 \label{Phix}\\
\Phi &:&\sigma \rightarrow K\sigma K^{T}+\alpha \;,  \label{Phicm}
\end{eqnarray}
where $\mathbf{y}_0$ is an $\mathbb{R}^{n}$ vector, and $K$ and $\alpha $ are $
2n\times 2n$ real matrices such that $\alpha $ is symmetric, and for any $
\sigma \geq \pm i\Delta $
\begin{equation}
K\sigma K^{T}+\alpha \geq \pm i\Delta \;,  \label{positivity}
\end{equation}
where the inequalities are meant to hold for both plus and minus signs in the right-hand-sides.
\end{thm}

The condition \eqref{positivity} imposes that $\Phi (\hat{\rho})$ is a
Gaussian state for any Gaussian $\hat{\rho}$.
It is weaker than the condition which guarantees complete positivity \eqref{CP}, which also ensures the mapping of Gaussian states into Gaussian states. An example of not completely positive mapping fulfilling \eqref{positivity} is provided by the dilatations defined in Eq. \eqref{MAP}.
Such mappings in fact,
while explicitly not CP \cite{brocker1995mixed}, correspond to the transformations \eqref{channel} where we set $\mathbf{y}_0=\mathbf{0}$ and take
\begin{equation}
K=\lambda \mathbb{I}_{2n}\;,\qquad \alpha =0\;,  \label{DILATATION}
\end{equation}
with $|\lambda |>1$. At the level of the covariance matrices \eqref{Phicm},
this implies $\sigma ^{\prime }=\lambda ^{2}\sigma $ which clearly still
preserve the Heisenberg inequality \eqref{heis} (indeed ${\lambda }
^{2}\sigma \geq \sigma \geq \pm i\Delta $), ensuring hence the condition \eqref{positivity}. Dilatations are not bounded with respect to the trace norm (see Theorem \ref{unbounded} of Section \ref{app}). This explains why Theorem \ref{gaussthmq} is formulated on the space of Hilbert-Schmidt operators. Indeed, via the Parceval formula (see \eqref{hilbert} in Appendix \ref{appG}) we can prove that dilatations are bounded in this space:
\begin{equation}
\left\Vert \Phi (\hat{\rho})\right\Vert ^{2} = \int \,\left\vert \chi _{\hat{
\rho}}(\lambda \mathbf{k})\right\vert^{2}\;\frac{d\mathbf{k}}{(2\pi )^{n}}
= \int \,\left\vert \chi _{\hat{\rho}}(\mathbf{k})\right\vert^{2}\;\frac{d
\mathbf{k}}{(2\pi \lambda ^{2})^{n}}=\frac{1}{\lambda ^{2n}}\left\Vert \hat{
\rho}\right\Vert ^{2}\;.
\end{equation}
For $\lambda=\frac{1}{\mu}$ with $|\mu|>1$ the transformation \eqref{DILATATION} yields a contraction of the output Wigner quasi-distribution.
In the Hilbert space $\mathfrak{H}$, the contraction by $\lambda$ is $\lambda^{2n}$ times the adjoint of the dilatation
by $\mu=\frac{1}{\lambda}$, as follows from the Parceval formula \eqref{hilbert}. As different from the dilatations, these mappings no longer ensure that all
Gaussian states will be transformed into proper density operators. For
instance, the vacuum state is mapped into a non-positive operator (this shows in particular that the contractions and hence
the adjoint dilatations are non-positive maps).

Another example of transformation not fulfilling the CP requirement \eqref{CPTP} but respecting \eqref{positivity} is the (complete) transposition
\begin{equation}
K=T_{2n}\qquad \alpha =0\;,  \label{TRANS}
\end{equation}
that is well-known not to be CP. Unfortunately, being positive it cannot be
used to certify that a given state is not contained in the convex hull $
\mathfrak{C}$ of the Gaussian ones. Is there anything else? We will prove
that for one mode, any channel satisfying \eqref{positivity} can be written
as a dilatation composed with a completely positive channel, possibly
composed with the transposition \eqref{TRANS}, see Fig. \ref{FigSys1}. We
will also show that in the multimode case this simple decomposition does
not hold in general; however, it still holds if we restrict to the channels
that do not add noise, i.e. with $\alpha =0$.

\begin{proof} Let the Gaussian state $\hat{\rho}_G(\sigma,\mathbf{x})$ be
sent into the Gaussian state $\hat{\rho}_G(\tau,\mathbf{y})$ with
covariance matrix $\tau (\mathbf{x},\,\sigma )$ and first moment $\mathbf{y}(
\mathbf{x},\,\sigma )$, with the characteristic function
\begin{equation}
\chi _{\Phi(\hat{\rho}_G(\sigma,\mathbf{x}))}(\mathbf{k})\equiv
\chi _{y,\tau }(\mathbf{k})=e^{-\frac{1}{4}\mathbf{k}\,\tau \,\mathbf{k}^T
+i\mathbf{k}\,\mathbf{y}}\;.  \label{Phichi}
\end{equation}
We first remark that the functions $\tau (\mathbf{x},\,\sigma )$ and $
\mathbf{y}(\mathbf{x},\,\sigma )$ are continuous. The map $\Phi $ is bounded
and hence continuous in the Hilbert-Schmidt norm. The required continuity
follows from

\begin{lem} The bijection $(\mathbf{x},\,\sigma )\rightarrow \hat{\rho}_G(\sigma,\mathbf{x})$ is bicontinuous in the Hilbert-Schmidt norm.\end{lem}

The proof of the lemma follows from the Parceval formula (see \eqref{hilbert} in Appendix \ref{appG}) by direct computation of the
Gaussian integral
\begin{equation*}
\int \,\left\vert \chi _{\hat{\rho}_G(\sigma,\mathbf{x})}(\mathbf{k})-\chi
_{\hat{\rho}_G(\sigma',\mathbf{x}')}(\mathbf{k}
)\right\vert^{2}\;\frac{d\mathbf{k}}{(2\pi )^{n}}\;.
\end{equation*}
Next, we have the identity
\begin{equation}\label{iden}
\int \hat{\rho}_G(\sigma',\mathbf{x}')\,\mu _{\mathbf{x}
,\,\sigma }(d\mathbf{x}^{\prime })=\hat{\rho}_G(\sigma'+\sigma,\mathbf{x})\;,
\end{equation}
where $\mu _{\mathbf{x},\,\sigma }$ is Gaussian probability measure with the
first moments $\mathbf{x}$ and covariance matrix $\sigma ,$ which is
verified by comparing the quantum characteristic functions of both sides.

Applying to both sides of this identity the continuous map $\Phi $ we obtain
\begin{equation*}
\int \hat{\rho}_G\left(\mathbf{y}(\mathbf{x}^{\prime },\,\sigma ^{\prime }),\,\tau
(\mathbf{x}^{\prime },\,\sigma ^{\prime })\right)\,\mu _{\mathbf{x},\,\sigma
}(d\mathbf{x}^{\prime })=\hat{\rho}_G\left(\mathbf{y}(\mathbf{x},\,\sigma ^{\prime }+\sigma
),\,\tau (\mathbf{x},\,\sigma ^{\prime }+\sigma )\right)\;.
\end{equation*}
By taking the quantum characteristic functions of both sides, we obtain
\begin{eqnarray}
\int \chi _{\mathbf{y}(\mathbf{x}^{\prime },\,\sigma ^{\prime }),\,\tau (
\mathbf{x}^{\prime },\,\sigma ^{\prime })}(\mathbf{k})\;\mu _{\mathbf{x}
,\,\sigma }(d\mathbf{x}^{\prime })=&&\nonumber\\
=\chi _{\mathbf{y}(\mathbf{x},\,\sigma ^{\prime
}+\sigma ),\,\tau (\mathbf{x},\,\sigma ^{\prime }+\sigma )}(\mathbf{k}
)\;,&&\quad \mathbf{k}\in\mathbb{R}^{n}\;.  \label{cauchy}
\end{eqnarray}
We now notice that $\mu _{\mathbf{x},\,\sigma }$ is the
fundamental solution of the diffusion equation:
\begin{equation}
du=\frac{1}{4}\partial _{i}d\sigma ^{ij}\partial _{j}u\;,  \label{dphi}
\end{equation}
where $d$ is the differential with respect to $\sigma $, i.e.
\begin{equation}
d=\sum_{i,\,j=1}^{m}d\sigma ^{ij}\frac{\partial }{\partial \sigma ^{ij}}\;
\end{equation}
and
\begin{equation}
\partial _{i}=\frac{\partial }{\partial x^{i}}\;,
\end{equation}
with the sum over the repeated indices. Relation \eqref{cauchy} means
that for any fixed $\mathbf{k},$ the function
\begin{equation}
u(\mathbf{x},\,\sigma )=\chi
_{\mathbf{y}(\mathbf{x},\,\sigma ^{\prime }+\sigma ),\,\tau (\mathbf{x}
,\,\sigma ^{\prime }+\sigma )}(\mathbf{k})
\end{equation}
is the solution of the Cauchy
problem for the equation \eqref{cauchy} with the initial condition
\begin{equation}
u(\mathbf{x},\,0)=\chi _{\mathbf{y}(\mathbf{x},\,\sigma ^{\prime }),\,\tau (
\mathbf{x},\,\sigma ^{\prime })}(\mathbf{k})\;.
\end{equation}
Since the last function is
bounded and continuous, the solution of the Cauchy problem is infinitely
differentiable in $(\mathbf{x},\,\sigma)$ for $\sigma >0.$ Substituting
\begin{equation*}
u(\mathbf{x},\,\sigma )=\exp \left[ -\frac{1}{4}\mathbf{k}\,\tau (
\mathbf{x},\,\sigma ^{\prime }+\sigma )\,\mathbf{k}^T+i\mathbf{k}\,\mathbf{
y}(\mathbf{x},\,\sigma ^{\prime }+\sigma )\right]
\end{equation*}
into \eqref{dphi} and differentiating the exponent, we obtain the identity
\begin{eqnarray}
-\frac{1}{4}\mathbf{k}\;d\tau \;\mathbf{k}^T+i\mathbf{k}\,d\mathbf{y} &=& \frac{1}{4}\left(\frac{1}{4}\mathbf{k}\,\partial_i\tau\,\mathbf{k}^T-i
\mathbf{k}\,\partial_i
\mathbf{y}\right)d\sigma^{ij}\left(\frac{1}{4}\mathbf{k}\,\partial_j\tau\,
\mathbf{k}^T-i\mathbf{k}\,\partial_j\mathbf{y}\right)+\nonumber\\
&& -{\frac{1}{16}}\mathbf{k}\left(\partial_id\sigma^{ij}\partial_j\tau\right)
\mathbf{k}^T+
\frac{i}{4}\mathbf{k}\,\partial_id\sigma^{ij}\partial_j\mathbf{y}\;.
\end{eqnarray}
We can now compare the two expressions. Since the left hand side
contains only terms at most quadratic in $\mathbf{k}$, we get
\begin{equation}
\partial _{i}\tau =0\;,
\end{equation}
i.e. $\tau $ does not depend on $\mathbf{x}$. Then, the right hand side
simplifies into
\begin{equation}
-\frac{1}{4}\mathbf{k}\left( \partial _{i}\mathbf{y}\,d\sigma
^{ij}\partial _{j}\mathbf{y}^{T}\right) \mathbf{k}^T+\frac{i}{4}\mathbf{k}\,\partial _{i}d\sigma ^{ij}\partial _{j}\mathbf{y}\;.  \label{logphi3}
\end{equation}
Comparing again with the left hand side, we get
\begin{eqnarray}
d\tau (\sigma ) &=&\partial _{i}\mathbf{y}\,d\sigma ^{ij}\partial _{j}
\mathbf{y}^{T}  \label{deltatau} \\
d\mathbf{y}(\mathbf{x},\,\sigma ) &=&\frac{1}{4}\partial _{i}d\sigma
^{ij}\partial _{j}\mathbf{y}\;.  \label{deltay}
\end{eqnarray}
Since $d\tau (\sigma )$ does not depend on $\mathbf{x}$, also $\partial _{i}
\mathbf{y}$ cannot, i.e. $\mathbf{y}$ is a linear function of $\mathbf{x}$:
\begin{equation}
\mathbf{y}(\mathbf{x},\,\sigma )=K(\sigma )\,\mathbf{x}+\mathbf{y}
_{0}(\sigma )\;,
\end{equation}
where $K(\sigma )$ and $\mathbf{y}_{0}(\sigma )$ are still arbitrary
functions. But now \eqref{deltay} becomes
\begin{equation}
d\mathbf{y}(\mathbf{x},\,\sigma )=0\;,
\end{equation}
i.e. $\mathbf{y}$ does not depend on $\sigma $, i.e.
\begin{equation}
\mathbf{y}=K\mathbf{x}+\mathbf{y}_{0}\;,  \label{y}
\end{equation}
with $K$ and $\mathbf{y}_{0}$ constant. Finally, \eqref{deltatau} becomes
\begin{equation}
d\tau (\sigma )=K\,d\sigma \,K^{T}\;,
\end{equation}
that can be integrated into
\begin{equation}
\tau (\sigma )=K\,\sigma \,K^{T}+\alpha \;.  \label{tau}
\end{equation}

Thus we get that the transformation rules for the first and second
moments are given by Eqs. \eqref{Phix} and \eqref{Phicm}. The positivity condition
for quantum Gaussian states implies \eqref{positivity}. The map defined by
\eqref{channel} correctly reproduces \eqref{Phix} and \eqref{Phicm}, so it
coincides with $\Phi $ on the Gaussian states. Since it is linear and
continuous, and the linear span of of Gaussian states is dense in $\mathfrak{
H},$ it coincides with $\Phi $ on the whole  $\mathfrak{H}$.
\end{proof}

\begin{rem} A similar argument can be used to prove that any linear positive map $\Phi$ of the Banach space $\mathfrak{T}$
of trace-class operators, leaving the set of Gaussian states globally invariant, has the form \eqref{channel}.
By Lemma 2.2.1 of \cite{davies1976quantum} any such map is bounded, and the proof of Theorem \ref{gaussthmq}
can be repeated, with $\mathfrak{H}$ replaced by $\mathfrak{T}$. In addition, since the trace of operator is continuous
on $\mathfrak{T}$, the formula \eqref{channel}  implies preservation of trace. However, the positivity condition is difficult to express
in terms of the map parameters $\mathbf{y}_{0}, K, \alpha$.

On the other hand, if $\Phi$ is completely positive then the necessary and sufficient condition is (see \eqref{CP})
\begin{equation}
\alpha \geq \pm i(\Delta -\Delta _{K})\;,  \label{CPTP}
\end{equation}
where
\begin{equation}
\Delta _{K}\equiv K\Delta K^{T}\;.  \label{DELTAK}
\end{equation}
Thus $\Phi$ is a quantum Gaussian channel \cite{holevo2001evaluating},
and the condition Eq. \eqref{positivity} is replaced by the more stringent constraint \eqref{CPTP}.

For automorphisms of the $C^*$-algebra of the Canonical Commutation Relations a similar characterization,
based on a different proof using partial ordering of Gaussian states, was first given in \cite{demoen1977completely,fannes1976quasi}.
\end{rem}

\begin{rem}
There is a counterpart of Theorem \ref{gaussthmq} in probability theory:
\end{rem}
\begin{thm}
Let $\Phi $ be an endomorphism (linear bounded transformation) of the Banach space $
\mathcal{M}(\mathbb{R}^{n})$ of finite signed Borel measures on $\mathbb{R}
^{n}$ (equipped with the total variation norm) having the Feller property (the dual  $\Phi^* $ leaves invariant the space of bounded continuous functions on $\mathbb{R}^{n}$).
Then, if $\Phi$ sends the set of Gaussian probability measures into itself,  $\Phi $ is a Markov operator whose action in terms of characteristic functions is of the form \eqref{channel}, with the condition \eqref{positivity} replaced by $\alpha\geq 0$.
\begin{proof}
The proof is parallel to the proof of Theorem \ref{gaussthmq}, with replacement of \eqref{iden} by the corresponding identity for Gaussian probability measures.
As a result, we obtain that the action of $\Phi$ in terms of characteristic functions is given by \eqref{channel} for any measure $\mu$
which is a linear combination of Gaussian probability measures. For arbitrary measure $\mu\in \mathcal{M}(\mathbb{R}^{n})$ the characteristic function of $\Phi(\mu)$ is
\begin{equation*}
\chi_{\Phi(\mu)}(\mathbf{k})=\int \mathbf{e}^{i\,\mathbf{k}\,\mathbf{x}}\;\Phi(\mu)(d\mathbf{x}) = \int \Phi^*\left(\mathbf{e}^{i\,\mathbf{k}\,\mathbf{x}}\right)\;\mu(d\mathbf{x})\;,
\end{equation*}
where $\Phi^*\left(\mathbf{e}^{i\,\mathbf{k}\,\mathbf{x}}\right)$ is continuous bounded function by the Feller property. Since the linear span of Gaussian probability measures is dense in $\mathcal{M}(\mathbb{R}^{n})$ in the weak topology defined by continuous bounded functions (it suffices to take Dirac's deltas, i.e, probability measures degenerated at the points of  $\mathbb{R}^{n}$) , the formula \eqref{channel} extends to characteristic function of arbitrary finite signed Borel measure on $\mathbb{R}^{n}$. The action of $\Phi$ on the moments is given by \eqref{Phix} and \eqref{Phicm}.
The positivity of the output covariance matrix when the input is a Dirac delta implies $\alpha\geq0$.
\end{proof}
\end{thm}
\subsection{Contractions}

\label{contr}

A contraction
by $\lambda=\frac{1}{\mu}$ behaves properly on the restricted subset $\mathfrak{G}^{(>)}_{\mu^2}$ of $\mathfrak{G}$ formed by the Gaussian states whose
covariance matrix admits symplectic eigenvalues larger than ${\mu}^2$.
Indeed all elements of $\mathfrak{G}^{(>)}_{\mu^2}$ will be mapped into
proper Gaussian output states by the contraction (and by linearity also the
convex hull of $\mathfrak{G}^{(>)}_{\mu^2}$ will be mapped into proper
output density operators). We will prove that any transformation with this property can be written as a contraction of $1/\mu$, followed by a transformation of the kind of Theorem \ref{gaussthmq}.
Let us first notice that:

\begin{lem}
A set $(K,\alpha)$ satisfies \eqref{positivity} for any $\sigma$ with
symplectic eigenvalues greater than $\mu^2$ iff $(\mu K,\;\alpha)$ satisfies
\eqref{positivity} for any $\sigma\geq\pm i\Delta$.
\begin{proof}
$\sigma$ has all the symplectic eigenvalues greater than $\mu^2$ iff $\sigma\geq\pm i\mu^2\Delta$, i.e. iff $\sigma'=\sigma/\mu^2$ is a state. Then \eqref{positivity} is satisfied for any $\sigma\geq\pm i\mu^2\Delta$ iff
\begin{equation}
\mu^2 K\sigma'K^T+\alpha\geq\pm i\Delta\qquad\forall\;\sigma'\geq\pm i\Delta\;,
\end{equation}
i.e. iff $(\mu K,\;\alpha)$ satisfies \eqref{positivity} for any $\sigma\geq\pm i\Delta$.
\end{proof}
\end{lem}

Then we can state the following result:

\begin{cor}
Any transformation associated with $(K,\alpha )$ satisfying
\eqref{positivity} for any state in $\mathfrak{G}_{\mu ^{2}}^{(>)}$ (i.e.
for any $\sigma \geq \pm i\mu ^{2}\Delta $) can be written as a contraction
of $1/\mu $, followed by a transformation satisfying \eqref{positivity} for
any state in $\mathfrak{G}$ (i.e. for any $\sigma \geq \pm i\Delta $).
\end{cor}

\begin{figure}[t]
\includegraphics[width=\textwidth]{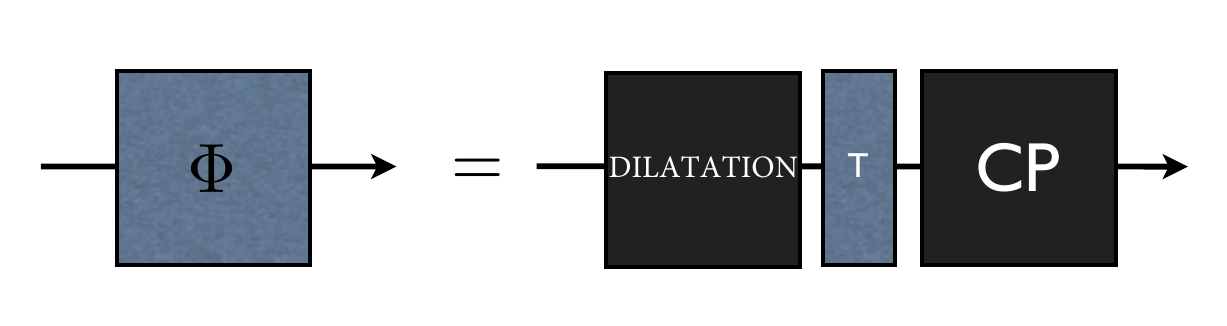}
\caption{Pictorial representation of the decomposition of a
generic (not necessarily positive) Gaussian single-mode transformation $\Phi$
in terms of a dilatation, CP mapping and (possibly) a transposition. The
same decomposition applies also to the case of $n$-mode transformations when
no extra noise is added to the system, see Section \protect\ref{S:MULTI}. }
\label{FigSys1}
\end{figure}

\section{One mode}

\label{S:ONE} Here we will give a complete classification of all the
one-mode maps \eqref{channel} satisfying \eqref{positivity}.

We will need the following

\begin{lem}
A set $(K,\;\alpha)$ satisfies \eqref{positivity} iff
\begin{equation}  \label{det1}
\sqrt{\det\alpha}\geq1-|\det K|\;.
\end{equation}
\end{lem}

\begin{proof}
For one mode, $\sigma\geq0$ satisfies $\sigma\geq\pm i\Delta$ iff $\det\sigma\geq1$, and condition \eqref{positivity} can be rewritten as
\begin{equation}
\det \left( K\sigma K^{T}+\alpha \right) \geq 1,\qquad \forall \;\sigma \geq
0,\;\det \sigma \geq 1\;.  \label{posdet}
\end{equation}
To prove \eqref{posdet} $\Longrightarrow $ \eqref{det1} let us consider first the
case $\det K\neq 0.$ Choosing $\sigma $ such that
\begin{equation}
K\sigma K^{T}=\frac{|\det
K|}{\sqrt{\det \alpha }}\alpha\;,
\end{equation}
we have $\sigma \geq 0$ and $\det \sigma \geq 1$.
Inserting this into \eqref{posdet}, we obtain
\begin{equation*}
\left( 1+\frac{|\det K|}{\sqrt{\det \alpha }}\right) ^{2}\det \alpha \geq 1
\end{equation*}
or, taking square root,
\begin{equation*}
\left( 1+\frac{|\det K|}{\sqrt{\det \alpha }}\right) \sqrt{\det \alpha }\geq
1\;.
\end{equation*}
hence \eqref{det1} follows.

If $\det K=0,$ then there is a unit vector $\mathbf{e}$ such that $K\mathbf{e}=0$. Choose
\begin{equation}
\sigma =\epsilon ^{-1}\;\mathbf{e}\mathbf{e}^{T}+\epsilon\;\mathbf{e}_{1}\mathbf{e}_{1}^{T}\;,
\end{equation}
where $\epsilon >0$,
and $\mathbf{e}_{1}$ is a unit vector orthogonal to $\mathbf{e}$. Then $\sigma \geq 0$, $\det
\sigma =1$, and $K\sigma K^{T}=\epsilon A$, where
\begin{equation}
A=K\mathbf{e}_{1}\mathbf{e}_{1}^{T}K^{T}\geq 0\;.
\end{equation}
Inserting this into \eqref{posdet}, we obtain
\begin{equation}
\det \left( \epsilon A+\alpha \right) \geq 1,\quad \forall \;\epsilon \geq 0,
\end{equation}
hence \eqref{det1} follows.

To prove \eqref{det1} $\Longrightarrow $  \eqref{posdet}, we use Minkowski's
determinant inequality
\begin{equation}
\sqrt{\det (A+B)}\geq \sqrt{\det A}+\sqrt{\det B}\qquad \forall \;A,\,B\geq
0\;.  \label{mink2}
\end{equation}
We have for all $\sigma \geq 0,\det \sigma \geq 1,$
\begin{equation}
\sqrt{\det \left( K\sigma K^{T}+\alpha \right) }\geq\left\vert \det K\right\vert \sqrt{\det \sigma }+\sqrt{\det \alpha }
\geq\left\vert \det K\right\vert +\sqrt{\det \alpha }\;\geq 1\;,
\label{ineq}
\end{equation}
where in the last step we have used \eqref{det1}.
\end{proof}
To compare transformations satisfying \eqref{det1} with CP ones, we need also

\begin{lem}
A set $(K,\;\alpha)$ characterizes a completely positive transformation
(i.e. satisfies \eqref{CPTP}) iff
\begin{equation}
\sqrt{\det\alpha}\geq|1-\det K|\;.  \label{det2}
\end{equation}
\end{lem}

\begin{proof}
For one mode, using \eqref{det1mode},
\begin{equation}
\Delta_K=K\Delta K^T=\det K\;\Delta\;,
\end{equation}
and \eqref{CPTP} becomes
\begin{equation}\label{CPTPdet}
\alpha\geq\pm i(1-\det K)\Delta\;.
\end{equation}
Using Eq. \eqref{state1} of Appendix \ref{appG},  for linearity \eqref{CPTPdet} becomes exactly
\begin{equation}
\det\alpha\geq(1-\det K)^2\;.
\end{equation}
\end{proof}
We recall here that a complete classification of single mode CP
maps has been provided in Ref.'s \cite{holevo2007one,caruso2006one}.

We are now ready to prove the main result of this Section.

\begin{thm}
\label{thm:1mode} Any map $\Phi $ satisfying \eqref{positivity} can be
written as a dilatation possibly composed with the transposition, followed
by a completely positive map. In more detail, given a pair $(K,\;\alpha )$
satisfying \eqref{positivity},

\begin{description}
\item[a1] If
\begin{equation}
0\leq\det K\leq1\;,
\end{equation}
$\Phi$ is completely positive.

\item[a2] If
\begin{equation}
\det K>1\;,
\end{equation}
$\Phi$ can be written as a phase-space dilatation of parameter $\lambda=\sqrt{\det K}>1$, composed with
the symplectic transformation given by
\begin{equation}
S=\frac{K}{\sqrt{\det K}}\;,
\end{equation}
composed with the addition of Gaussian noise given by $\alpha$.

\item[b1] If
\begin{equation}
-1\leq\det K<0\;,
\end{equation}
$\Phi$ can be written as a transposition composed with a completely positive
map.

\item[b2] If
\begin{equation}
\det K<-1\;,
\end{equation}
$\Phi$ can be written as a dilatation of $\sqrt{|\det K|}$ composed with the
transposition, followed by the symplectic transformation given by
\begin{equation}
S=\frac{K}{\sqrt{|\det K|}}\;,
\end{equation}
composed with the addition of Gaussian noise given by $\alpha$.
\end{description}
\end{thm}

\begin{proof}
\begin{description}
\item[a] Let us start from the case
\begin{equation}
\det K\geq0\;.
\end{equation}
\begin{description}
\item[a1] If
\begin{equation}
0\leq\det K\leq1\;,
\end{equation}
\eqref{det1} and \eqref{det2} coincide, so $\Phi$ is completely positive.
\item[a2] If
\begin{equation}
\det K>1\;,
\end{equation}
we can write $K$ as
\begin{equation}
K=S\;\sqrt{\det K}\mathbb{I}_2\;,
\end{equation}
where
\begin{equation}
S=\frac{K}{\sqrt{\det K}}
\end{equation}
is symplectic since $\det S=1$.
Then $\Phi$ can be written as a dilatation of $\sqrt{\det K}>1$, followed by the symplectic transformation given by $S$, composed with the addition of the Gaussian noise given by $\alpha$.
\end{description}
\item[b] If
\begin{equation}
\det K<0\;,
\end{equation}
we can write $K$ as
\begin{equation}
K=K'T\;,
\end{equation}
where $T$ is the one-mode transposition
\begin{equation}
T=\left(
    \begin{array}{cc}
      1 &  \\
       & -1 \\
    \end{array}
  \right)\;,
\end{equation}
and
\begin{equation}
\det K'=-\det K>0\;.
\end{equation}
From \eqref{det1} we can see that also $K'$ satisfies
\begin{equation}
\sqrt{\det\alpha}\geq1-\left|\det K'\right|\;,
\end{equation}
and we can exploit the classification with positive determinant, ending with the same decomposition with the addition of the transposition after (or before, since they commute) the eventual dilatation.
\end{description}
\end{proof}

\section{Multimode case}

\label{S:MULTI} In the multimode case, a classification as simple as the
one of Theorem \ref{thm:1mode} does not exist. However, we will prove that
if $\Phi$ does not add any noise, i.e $\alpha=0$, the only solution to
\eqref{positivity} is a dilatation possibly composed with a (total)
transposition, followed by a symplectic transformation. We will also provide
examples that do not fall in any classification like \ref{thm:1mode}, i.e.
that are not composition of a dilatation, possibly followed by a (total)
transposition, and a completely positive map.

We will need the following lemma:

\begin{lem}
\label{lem:inf}
\begin{equation}
\inf_{\sigma\geq\pm i\Delta}\mathbf{w}^\dag\sigma\mathbf{w}=\left|\mathbf{w}
^\dag\Delta\mathbf{w}\right|\qquad\forall\;\;\mathbf{w}\in\mathbb{C}^{2n}\;.
\label{inf}
\end{equation}
\begin{proof}
~\paragraph{Lower bound}
The lower bound for the LHS is straightforward: for any $\sigma\geq\pm i\Delta$ and $\mathbf{w}\in\mathbb{C}^{2n}$ we have
\begin{equation}
\mathbf{w}^\dag\sigma\mathbf{w}\geq\pm i\mathbf{w}^\dag\Delta\mathbf{w}\;,
\end{equation}
and then
\begin{equation}
\inf_{\sigma\geq\pm i\Delta}\mathbf{w}^\dag\sigma\mathbf{w}\geq\left|\mathbf{w}^\dag\Delta\mathbf{w}\right|\;.
\end{equation}

\paragraph{Upper bound}
To prove the converse, let
\begin{equation*}
\mathbf{w}=\mathbf{w}_{1}+i\mathbf{w}_{2}\;,\qquad \mathbf{w}_{i}\in \mathbb{
R}^{2n}\;,
\end{equation*}
where without lost of generality we assume $\mathbf{w}_1\neq\mathbf{0}$.
Then
\begin{equation*}
\mathbf{w}^{\dag }\sigma \mathbf{w}=\mathbf{w}_{1}^{T}\sigma \mathbf{w}_{1}+\mathbf{w}_{2}^{T}\sigma
\mathbf{w}_{2},\quad \left\vert \mathbf{w}^{\dag }\Delta \mathbf{w}\right\vert
=2\left\vert \mathbf{w}_{1}^{T}\Delta \mathbf{w}_{2}\right\vert .
\end{equation*}
Let us first assume $\mathbf{w}_{1}^{T}\Delta \mathbf{w}_{2}\equiv \epsilon \neq0.$ Then we can
introduce the symplectic basis $\{\mathbf{e}_{j},\;\mathbf{h}_{j}\}_{j=1,\dots ,n}$, where
\begin{equation*}
\mathbf{e}_{1}=\frac{\mathbf{w}_{1}}{\sqrt{|\epsilon|}}\;,\qquad\mathbf{h}_{1}=\frac{\mathrm{sign}(\epsilon)\;\mathbf{w}_{2}}{\sqrt{|\epsilon|}}\;.
\end{equation*}
Expressed in this basis the question
\eqref{inf} reduces to the first mode, and the infimum is attained by the
matrix of the form
\begin{equation*}
\sigma =\left(
\begin{array}{cc}
1 & 0 \\
0 & 1
\end{array}
\right) \oplus \sigma _{n-1}\;,
\end{equation*}
 where $\sigma _{n-1}$ is any quantum correlation matrix in
the rest $n-1$ modes.

Let us consider next the case where $\mathbf{w}_{1}^{T}\Delta \mathbf{w}_{2}=0$ and $\mathbf{w}_{2}$ is not proportional to $\mathbf{w}_{1}$. In this context
we  introduce the symplectic basis $\{\mathbf{e}_{j},\;\mathbf{h}_{j}\}_{j=1,\dots ,n}$,
where
\begin{equation*}
\mathbf{e}_{1}=\mathbf{w}_{1}\;,\qquad \mathbf{e}_{2}=\mathbf{w}_{2}\;.
\end{equation*}
Accordingly the identity \eqref{inf} reduces to the
first two modes, and the infimum is attained by the matrices of the form
\begin{equation*}
\sigma (\epsilon )=\left(
\begin{array}{cc}
\epsilon  & 0 \\
0 & \epsilon ^{-1}
\end{array}
\right) \oplus \left(
\begin{array}{cc}
\epsilon  & 0 \\
0 & \epsilon ^{-1}
\end{array}
\right) \oplus \sigma _{n-2}\;,
\end{equation*}
where $\sigma _{n-2}$ is any quantum correlation matrix in the rest $n-2$
modes, and $\epsilon \rightarrow 0$.

{Finally, if $\mathbf{w}_{2}=c\;\mathbf{w}_1,\;c\in\mathbb{R}$, we  introduce the symplectic basis $\{\mathbf{e}_{j},\;\mathbf{h}_{j}\}_{j=1,\dots ,n}$, where $\mathbf{e}_{1}=\mathbf{w}_{1}$. The question \eqref{inf} reduces to the
first mode, and the infimum is attained by the matrices of the form
\begin{equation*}
\sigma (\epsilon )=\left(
\begin{array}{cc}
\epsilon  & 0 \\
0 & \epsilon ^{-1}
\end{array}
\right) \oplus \sigma _{n-1}\;,
\end{equation*}
where $\sigma _{n-1}$ is any quantum correlation matrix in the rest $n-1$
modes, and $\epsilon \rightarrow 0$.}
\end{proof}
\end{lem}

A simple consequence of lemma \ref{lem:inf} is

\begin{lem}
Any $\alpha$ satisfying \eqref{positivity} for some $K$ is positive
semidefinite.
\end{lem}

\begin{proof}
The constraint \eqref{positivity} implies
\begin{equation}
\left(K^T\mathbf{k}\right)^T\sigma\left(K^T\mathbf{k}\right)+\mathbf{k}^T\alpha\mathbf{k}\geq0\label{kpos}
\end{equation}
for any $\sigma\geq\pm i\Delta$ and $\mathbf{k}\in\mathbb{R}^{2n}$.
Taking the inf over $\sigma\geq\pm i\Delta$, and exploiting lemma \ref{lem:inf} with $\mathbf{w}=K^T\mathbf{k}$, we get
\begin{equation}\label{kak}
\mathbf{k}^T\alpha\mathbf{k}\geq0\;,
\end{equation}
i.e. $\alpha$ is positive semidefinite.
In deriving \eqref{kak} we have used that, since $\Delta$ is antisymmetric, $\mathbf{k}\Delta\mathbf{k}^T=0$ for any real $\mathbf{k}$.
\end{proof}

The Lemma \ref{lem:inf} allows us to rephrase the problem: indeed, the constraint
\eqref{positivity} can be written as
\begin{equation}
(K^T\mathbf{w})^\dag\sigma(K^T\mathbf{w})+\mathbf{w}^\dag\alpha\mathbf{w}
\geq\left|\mathbf{w}^\dag\Delta\mathbf{w}\right|\;,
\end{equation}
$\forall\;\sigma\geq\pm i\Delta$, $\forall \mathbf{w}\in\mathbb{C}^{2n}$.
Taking the inf over $\sigma$ in the LHS we hence get
\begin{equation}  \label{posw}
\left|\mathbf{w}^\dag\Delta_K\mathbf{w}\right|+\mathbf{w}^\dag\alpha\mathbf{w
}\geq\left|\mathbf{w}^\dag\Delta\mathbf{w}\right|\;,\quad\forall\;\mathbf{w}
\in\mathbb{C}^{2n}\;,
\end{equation}
with $\Delta_K$ as in Eq. \eqref{DELTAK}. We notice that, as for the complete
positivity constraint \eqref{CPTP}, since $K$ enters in \eqref{posw} only
through $\left|\mathbf{w}^\dag\Delta_K\mathbf{w}\right|$, whether given $K$
and $\alpha$ satisfy \eqref{positivity} depends not on the entire $K$ but
only on $\Delta_K$.

The easiest way to give a general classification of the channels satisfying
\eqref{posw} (and then \eqref{positivity}) would seem choosing a basis in
which $\Delta$ is in the canonical form of Eq. \eqref{Deltac} of Appendix \ref{appG}, and then try to put the
antisymmetric matrix $\Delta_K$ in some canonical form using symplectic
transformations preserving $\Delta$. However, the complete classification of
antisymmetric matrices under symplectic transformations is very involved
\cite{lancaster2005canonical}, and in the multimode case the problem simplifies only if we
consider maps $\Phi$ that do not add noise, since in this case the
constraint \eqref{posw} rules out almost all the equivalence classes. In the
general case, we will provide examples showing the other possibilities.

\subsection{No noise}

The main result of this Section is the classification of the maps $\Phi$
that do not add noise ($\alpha=0$) and satisfy \eqref{positivity}:

\begin{thm}
\label{thm:nonoise} A map $\Phi$ with $\alpha=0$ satisfying
\eqref{positivity} can always be decomposed as a dilatation \eqref{DILATATION}, possibly composed with the transposition, followed by a
symplectic $S$ transformation: i.e.
\begin{equation}
K=S\;\kappa\mathbb{I}_{2n}\qquad\text{or}\qquad K=S\;T\;\kappa\mathbb{I}
_{2n}\;,
\end{equation}
with $\kappa\geq1$.
\end{thm}

\begin{proof}
With $\alpha=0$ and
\begin{equation}
\mathbf{w}=\mathbf{w}_1+i\mathbf{w}_2\;,\qquad\mathbf{w}_i\in\mathbb{R}^{2n}\;,
\end{equation}
\eqref{posw} becomes
\begin{equation}
\left|\mathbf{w}_1^T\Delta_K\mathbf{w}_2\right|\geq\left|\mathbf{w}_1^T\Delta\mathbf{w}_2\right|\;,\label{Deltamaj}
\end{equation}
i.e. all the matrix elements of $\Delta_K$ are in modulus bigger than the corresponding ones of $\Delta$ in \emph{any} basis. In particular, if some matrix element $\Delta_K^{ij}$ vanishes, also $\Delta^{ij}$ must vanish.
Let us choose a basis in which $\Delta_K$ has the canonical form
\begin{equation}
\Delta_K=\bigoplus_{i=1}^{\frac{r}{2}}\left(
                                         \begin{array}{cc}
                                            & 1 \\
                                           -1 &  \\
                                         \end{array}
                                       \right)\oplus0_{2n-r}\;,
\end{equation}
where
\begin{equation}
r\equiv\mathrm{rank}\,\Delta_K\;.
\end{equation}
For \eqref{Deltamaj}, in this basis $\Delta$ must be of the form
\begin{equation}
\Delta=\bigoplus_{i=1}^{\frac{r}{2}}\left(
                                         \begin{array}{cc}
                                            &  \lambda_i\\
                                           -\lambda_i &  \\
                                         \end{array}
                                       \right)\oplus0_{2n-r}\;,\qquad|\lambda_i|\leq1\;.
\end{equation}
Since $\Delta$ has full rank, there cannot be zeroes in its decomposition, so $r$ must be $2n$.

We will prove that all the eigenvalues $\lambda_i$ must be equal. Let us take two eigenvalues $\lambda$ and $\mu$, and consider the restriction of $\Delta$ and $\Delta_K$ to the subspace associated to them:
\begin{equation}
\Delta_K=\left(
           \begin{array}{cc|cc}
              & 1 &  &  \\
             -1 &  &  & \\
             \hline
              &  &  & 1 \\
              &  & -1 &  \\
           \end{array}
         \right)\qquad\Delta=\left(
           \begin{array}{cc|cc}
              & \lambda &  &  \\
             -\lambda &  &  & \\
             \hline
              &  &  & \mu \\
              &  & -\mu &  \\
           \end{array}
         \right)\;.
\end{equation}
If we change basis with the rotation matrix
\begin{eqnarray}
&R=\left(
    \begin{array}{cc}
      \cos\theta\;\mathbb{I}_2 & -\sin\theta\;\mathbb{I}_2 \\
      \sin\theta\;\mathbb{I}_2 & \cos\theta\;\mathbb{I}_2 \\
    \end{array}
  \right)\;,&\nonumber \\ &\Delta\mapsto R\Delta R^T\;,\qquad  \Delta_K\mapsto R\Delta_K R^T\;,&
\end{eqnarray}
$\Delta_K$ remains of the same form, while $\Delta$ acquires off-diagonal elements proportional to $\lambda-\mu$. Since for \eqref{Deltamaj} the off-diagonal elements of $\Delta$ must vanish also in the new basis, the only possibility is $\lambda=\mu$. Then all the $\lambda_i$ must be equal, and $\Delta_K$ must then be proportional to $\Delta$:
\begin{equation}
\Delta_K=\frac{1}{\lambda}\Delta\;,\qquad0<|\lambda|\leq1\;,\label{deltaprop}
\end{equation}
where we have put all the $\lambda_i$ equal to $\lambda\neq0$ (since $\Delta$ is nonsingular they cannot vanish).
Relation \eqref{deltaprop} means
\begin{equation}
K\Delta K^T=\frac{1}{\lambda}\Delta\;,
\end{equation}
i.e.
\begin{equation}
\left(\sqrt{|\lambda|}\;K\right)\;\Delta\;\left(\sqrt{|\lambda|}\;K\right)^T=\mathrm{sign}(\lambda)\Delta\;.
\end{equation}
If $0<\lambda\leq1$, we can write $K$ as a dilatation of
\begin{equation}
\kappa=\frac{1}{\sqrt{\lambda}}\;,
\end{equation}
composed with a symplectic transformation given by
\begin{equation}
S=\sqrt{\lambda}\;K\;,
\end{equation}
i.e.
\begin{equation}
K=S\;\kappa\mathbb{I}_{2n}\;,\qquad S\Delta S^T=\Delta\;.
\end{equation}
If $-1\leq\lambda<0$, since the total transposition $T$ changes the sign of $\Delta$:
\begin{equation}
T\Delta T^T=-\Delta\;,
\end{equation}
we can write $K$ as a dilatation of
\begin{equation}
\kappa=\frac{1}{\sqrt{|\lambda|}}\;,
\end{equation}
composed with $T$ followed by a symplectic transformation:
\begin{equation}
K=S\;T\;\kappa\mathbb{I}_{2n}\;,\qquad S\Delta S^T=\Delta\;.
\end{equation}
\end{proof}

\subsection{Examples with nontrivial decomposition}

If $\alpha\neq0$, a decomposition as simple as the one of theorem \ref
{thm:nonoise} does no more exist: here we will provide some examples in
which the canonical form of $\Delta_K$ is less trivial, and that do not fall
in any classification like the precedent one. Essentially, they are all
based on this observation:

\begin{prop}
If $\alpha$ is the covariance matrix of a quantum state, i.e. $\alpha\geq\pm
i\Delta$, the constraint \eqref{positivity} is satisfied by any $K$.
\end{prop}

Since for one mode the decomposition of theorem \ref{thm:1mode} holds, we
will provide examples with two-mode systems.

We will always consider bases in which
\begin{equation}
\Delta=\left(
\begin{array}{cc|cc}
& 1 &  &  \\
-1 &  &  &  \\ \hline
&  &  & 1 \\
&  & -1 &
\end{array}
\right)\;.
\end{equation}

\subsubsection{Partial transpose}

The first example is the partial transpose of the second subsystem, composed
with a dilatation of $\sqrt{\nu}$ and the addition of the covariance matrix
of the vacuum as noise:
\begin{equation}
K=\sqrt{\nu}\left(
\begin{array}{cc}
\mathbb{I}_2 &  \\
& T_2
\end{array}
\right)\;,\qquad\nu>0\;,\qquad\alpha=\mathbb{I}_4\;.
\end{equation}
In this case we have
\begin{equation}  \label{deltaKt}
\Delta_K=\left(
\begin{array}{cc|cc}
& \nu &  &  \\
-\nu &  &  &  \\ \hline
&  &  & -\nu \\
&  & \nu &
\end{array}
\right)\;,
\end{equation}
and $i(\Delta-\Delta_K)$ has eigenvalues $\pm(1+\nu)$, $\pm(1-\nu)$, so that one of them is $\left|1+|\nu|\right|>1$, and the complete positivity
requirement \eqref{CPTP}
\begin{equation}
\mathbb{I}_4\geq\pm i(\Delta-\Delta_K)
\end{equation}
cannot be fulfilled by any $\nu\neq0$.

We will prove that this map cannot be written as a dilatation, possibly
composed with the transposition, followed by a completely positive map.
Indeed, let us suppose that we can write $K$ as $K^{\prime}\,\lambda\,\mathbb{I}_4$ or $K^{\prime}\,T_4\,\lambda\,\mathbb{I}_4$ for some $\lambda\geq1$.
Then
\begin{equation}
\Delta_{K^{\prime}}=\pm\frac{1}{\lambda^2}\Delta_K
\end{equation}
is always of the form \eqref{deltaKt} with
\begin{equation}
\nu^{\prime}=\pm\frac{\nu}{\lambda^2}\;,
\end{equation}
and also the transformation with $K^{\prime}$ cannot be completely positive.

\subsubsection{\emph{Q} exchange}

As second example, we take for the added noise $\alpha$ still the covariance
matrix of the vacuum, and for the matrix $K$ the partial transposition of
the first mode composed with the exchange of $Q^1$ and $Q^2$ followed by a
dilatation of $\sqrt{\nu}$:
\begin{equation}
\alpha=\mathbb{I}_4\geq\pm i\Delta\;,\quad K=\sqrt{\nu}\left(
\begin{array}{cc|cc}
&  & 1 &  \\
& -1 &  &  \\ \hline
1 &  &  &  \\
&  &  & 1
\end{array}
\right)\;,\quad\nu>0\;.
\end{equation}
With this choice,
\begin{equation}  \label{deltakex}
\Delta_K=\left(
\begin{array}{cc|cc}
&  &  & \nu \\
&  & \nu &  \\ \hline
& -\nu &  &  \\
-\nu &  &  &
\end{array}
\right)\;.
\end{equation}
The transformation is completely positive iff
\begin{equation}  \label{CPTPex}
\mathbb{I}_4\geq\pm i(\Delta-\Delta_K)\;,
\end{equation}
and since the eigenvalues of $i(\Delta-\Delta_K)$ are $\pm\sqrt{1+\nu^2}$, the condition \eqref{CPTPex} is never fulfilled for any $\nu\neq0$.

As before, we will prove that this map cannot be written as a dilatation,
possibly composed with the transposition, followed by a completely positive
map. Indeed, let us suppose that we can write $K$ as $K^{\prime}\,\lambda\,\mathbb{I}_4$ or $K^{\prime}\,T_4\,\lambda\,\mathbb{I}_4$ for some $\lambda\geq1$.
Then
\begin{equation}
\Delta_{K^{\prime}}=\pm\frac{1}{\lambda^2}\Delta_K
\end{equation}
is always of the form \eqref{deltakex} with
\begin{equation}
\nu^{\prime}=\pm\frac{\nu}{\lambda^2}\;,
\end{equation}
and also the transformation with $K^{\prime}$ cannot be completely positive.

\section{Unboundedness of dilatations}\label{app}
\begin{thm}\label{unbounded}
For any $\lambda\neq\pm1$ the phase-space dilatation by $\lambda$ is not bounded in the Banach space $\mathfrak{T}$ of trace-class operators.
\begin{proof}
Fix $\lambda\neq\pm1$, and let $\Theta$ be the phase-space dilatation by $\lambda$.
Let us suppose that $\Theta$ is bounded, i.e.
\begin{equation}\label{bounded}
\left\|\Theta\left(\hat{X}\right)\right\|_1\leq\left\|\Theta\right\|\;\left\|\hat{X}\right\|_1\qquad\forall\;\hat{X}\in\mathfrak{T}\;.
\end{equation}
Let also
\begin{equation}
p_n^{(m)}:=\langle n|\Theta\left(|m\rangle\langle m|\right)|n\rangle\;.
\end{equation}
Eq. \eqref{bounded} implies
\begin{equation}
\sum_{n=0}^\infty\left|p_n^{(m)}\right|\leq\|\Theta\|\qquad\forall\;m\in\mathbb{N}\;.
\end{equation}
The moment generating function of $p^{(m)}$ is \cite{brocker1995mixed}
\begin{equation}\label{mgf}
g_m(q):=\sum_{n=0}^\infty p_n^{(m)}\;e^{-i\,n\,q}=\frac{1-\tau}{1-\tau\,e^{-i\,q}}\left(\frac{1-\tau\,e^{i\,q}}{e^{i\,q}-\tau}\right)^m\;,
\end{equation}
where $q\in\mathbb{R}$ and
\begin{equation}\label{tau2}
\tau:=\frac{\lambda^2-1}{\lambda^2+1}\;.
\end{equation}
Let us define
\begin{equation}\label{am}
a_m:=\frac{1-\tau}{\sqrt[3]{m\,\tau(1+\tau)}}\;.
\end{equation}
Let $\phi\in C_c^\infty(\mathbb{R})$ be an infinitely differentiable test function with compact support.
We must then have
\begin{equation}\label{cond}
\sum_{n=0}^\infty\phi\left(a_m\left(n-\lambda^2m\right)\right)\;p_n^{(m)}\leq\|\phi\|_\infty\;\|\Theta\|\;.
\end{equation}
Expressed in terms of the Fourier transform of $\phi$
\begin{equation}
\widetilde{\phi}(k)=\int_{-\infty}^\infty\phi(x)\;e^{i\,k\,x}\;dx\;,
\end{equation}
\eqref{cond} becomes
\begin{equation}
\sum_{n=0}^\infty\left(\int_{-\infty}^\infty\widetilde{\phi}(k)\;e^{i\,\lambda^2\,m\,a_m\,k}\;e^{-i\,k\,a_m\,n}\;\frac{dk}{2\pi}\right)p_n^{(m)}\leq \|\phi\|_\infty\;\|\Theta\|\;.
\end{equation}
Since the sum of the integrands is dominated by the integrable function
\begin{equation*}
\frac{\|\Theta\|}{2\pi}\left|\widetilde{\phi}(k)\right|\;,
\end{equation*}
we can bring the sum inside the integral, getting
\begin{equation}
\int_{-\infty}^\infty\widetilde{\phi}(k)\;g_m\left(a_mk\right)\;e^{i\,\lambda^2\,m\,a_m\,k}\;\frac{dk}{2\pi}\leq\|\phi\|_\infty\;\|\Theta\|\;.
\end{equation}
Since for any $k$
\begin{equation}\label{limitm}
\lim_{m\to\infty}\left(g_m\left(a_m\,k\right)\;e^{i\,\lambda^2\,m\,a_m\,k}\right)=e^{\frac{i\,k^3}{3}}
\end{equation}
(see subsection \ref{note}), by the dominated convergence theorem
\begin{equation}
\lim_{m\to\infty}\int_{-\infty}^\infty\widetilde{\phi}(k)\;g_m\left(a_mk\right)\;e^{i\,\lambda^2\,m\,a_m\,k}\;\frac{dk}{2\pi} =\int_{-\infty}^\infty\widetilde{\phi}(k)\;e^\frac{i\,k^3}{3}\;\frac{dk}{2\pi}=\int_{-\infty}^\infty\phi(x)\;\mathrm{Ai}(x)\;dx\;,
\end{equation}
where $\mathrm{Ai}(x)$ is the Airy function.
Now we get
\begin{equation}\label{limit}
\int_{-\infty}^\infty \mathrm{Ai}(x)\;\phi(x)\;dx\leq \|\Theta\|\;\|\phi\|_\infty\qquad\forall\;\phi\in C_c^\infty(\mathbb{R})\;.
\end{equation}
Since the Airy function is continuous and the set of its zeroes has no accumulation points (except $-\infty$), there exists a sequence of test functions $\phi_r\in C_c^\infty(\mathbb{R})$, $r\in\mathbb{N}$ with $\|\phi_r\|_\infty=1$ approximating $\mathrm{sign}\left(\mathrm{Ai}(x)\right)$, i.e. such that
\begin{equation}
\lim_{r\to\infty}\int_{-\infty}^\infty \mathrm{Ai}(x)\;\phi_r(x)\;dx=\int_{-\infty}^\infty\left|\mathrm{Ai}(x)\right|dx=\infty\;,
\end{equation}
implying $\|\Theta\|=\infty$.
\end{proof}
\end{thm}
\subsection{Computation of the limit in (\ref{limitm})}\label{note}
Here we compute explicitly the limit in \eqref{limitm}.
It is better to rephrase it in terms of $q:=a_m\,k$, $q\to0$ (remember that $a_m\sim1/\sqrt[3]{m}$).
Putting together \eqref{limitm}, \eqref{mgf}, \eqref{tau2} and \eqref{am}, we have to compute
\begin{equation} \label{ffd}
\lim_{q\to0}\left(\frac{1-\tau}{1-\tau\,e^{-i\,q}}\left(\frac{1-\tau\,e^{i\,q}}{e^{i\,q}-\tau}\;e^{i\,\frac{1+\tau}{1-\tau}\,q}\right)^\frac{k^3(1-\tau)^3}{q^3\,\tau(1+\tau)}\right)\overset{?}{=}e^\frac{i\,k^3}{3}\;.
\end{equation}
The first term on the left-hand-side tends to one. The second term on the left-hand-side instead can be treated via Taylor expansion, i.e.
\begin{equation}
\frac{1-\tau\,e^{i\,q}}{e^{i\,q}-\tau}\;e^{i\,\frac{1+\tau}{1-\tau}\,q}=1+\frac{i\,q^3\,\tau(1+\tau)}{3(1-\tau)^3}+\mathcal{O}\left(q^5\right)
\end{equation}
for $q\to0$.
This gives
\begin{equation}
\lim_{q\to0}\left(\frac{1-\tau\,e^{i\,q}}{e^{i\,q}-\tau}\;e^{i\,\frac{1+\tau}{1-\tau}\,q}\right)^\frac{k^3(1-\tau)^3}{q^3\,\tau(1+\tau)} = \lim_{q\to0}\left(1+\frac{i\,q^3\,\tau(1+\tau)}{3(1-\tau)^3}+\mathcal{O}\left(q^5\right)\right)^\frac{k^3(1-\tau)^3}{q^3\,\tau(1+\tau)}= e^\frac{i\,k^3}{3}\;,
\end{equation}
which proves the identity of \eqref{ffd}.

\section{Conclusion}

\label{S:CON}

In this Chapter we have explored both at the classical and quantum level the
set of linear transformations sending the set
of Gaussian states into itself without imposing any further requirement,
such as positivity. We have proved that the action on the covariance matrix
and on the first moment must be linear, and we have found the form of the
action on the characteristic function. Focusing on the quantum case, for one
mode we have obtained a complete classification, stating that the only not
CP transformations in the set are actually the total transposition and the
dilatations (and their compositions with CP maps). The same result holds
also in the multimode scenario, but it needs the further hypothesis of
homogeneous action on the covariance matrix, since we have shown the
existence of non-homogeneous transformations belonging to the set but not
falling into our classification.

Despite the set $\mathfrak{F}$ of quantum states that are sent into positive operators by any dilatation is known to strictly contain the convex hull of Gaussian states $\mathfrak{C}$ even in the one-mode case\cite{brocker1995mixed}, the dilatations are then confirmed to be (at least in the single mode or in
the homogeneous action cases) the only transformation in the class
\eqref{channel} that can act as a probe for $\mathfrak{C}$.

\chapter{Necessity of eigenstate thermalization}
\label{chETH}
In this Chapter we prove that if a small quantum system in contact with a large heat bath thermalizes for any initial uncorrelated state with a sharp energy distribution, the system-bath Hamiltonian must satisfy the so-called Eigenstate Thermalization Hypothesis.
This result definitively settles the question of determining whether a quantum system has a thermal behavior, reducing it to checking whether its Hamiltonian satisfies the ETH.

The Chapter is based on
\begin{enumerate}
\item[\cite{de2015necessity}] G.~De~Palma, A.~Serafini, V.~Giovannetti, and M.~Cramer, ``Necessity of Eigenstate Thermalization,'' \emph{Physical Review Letters}, vol. 115, no.~22, p. 220401, 2015.\\ {\small\url{http://journals.aps.org/prl/abstract/10.1103/PhysRevLett.115.220401}}
\end{enumerate}

\section{Introduction}
An ideal heat bath induces thermalization in the sense that, when a physical system is coupled to it, its state will evolve toward a well-defined infinite-time limit which depends only on macroscopic parameters of the bath -- such as its temperature or energy -- and not on any details of the initial state of the system, the bath, or the system-bath interaction.
It is a well-established empirical fact that both classical and quantum systems with a very large number of degrees of freedom exhibit these ideal-bath properties when weakly coupled to much smaller systems, with their temperature being a smooth function of their energy alone.
Yet, rigorous derivations relaying such a ``generic'' behavior to fundamental dynamical laws
seem to require rather sophisticated, and arguably very specific and technical, hypotheses.
Then, understanding the mechanisms lying behind the thermalization of a quantum system has become a hot-debated topic in physics.
The apparent incongruence between the ubiquity of thermalization and the specificity of the hypotheses that seem to imply it has spurred substantial research \cite{deutsch1991quantum,srednicki1994chaos,tasaki1998quantum,calabrese2006time,cazalilla2006effect,rigol2007relaxation,reimann2007typicality,cramer2008exact,rigol2008thermalization,reimann2008foundation,linden2009quantum,rigol2009breakdown,rigol2012alternatives,reimann2010canonical,cho2010emergence,gogolin2011absence,riera2012thermalization,mueller2013thermalization,gogolin2015equilibration,polkovnikov2011colloquium,cazalilla2011one,bloch2008many,eisert2015quantum,deffner2015ten,jarzynski2015diverse,ponte2015many,steinigeweg2014pushing,genway2013dynamics,caux2013time,cassidy2011generalized},
analyzing the dynamical conditions under which a large quantum system behaves as an ideal heat bath and induces thermalization.
Prominent among them is the Eigenstate Thermalization Hypothesis (ETH), which
may be formulated by stating that the partial traces of the eigenstates
of the global Hamiltonian of the bath and the coupled system (including the interaction terms)
are smooth functions of the energy.

It is well known that the ETH is sufficient for thermalization if the initial state has a sufficiently sharp distribution
in energy \cite{deutsch1991quantum,gogolin2015equilibration},
and a lot of effort has been dedicated in checking whether specific quantum systems satisfy the ETH, with both analytical and numerical computations \cite{rigol2008thermalization,ponte2015many,steinigeweg2014pushing,genway2013dynamics,caux2013time,cassidy2011generalized,rigol2009breakdown}.

The converse question, however, of whether the ETH is also necessary for thermalization, i.e. whether there exist quantum systems not fulfilling the ETH but nonetheless exhibiting thermal behavior, is not settled yet, and alternatives to the ETH have been proposed \cite{rigol2012alternatives}.
An answer to this question
has been hinted at, although not proven, in the literature on the subject (see, e.g., the very recent
survey \cite{gogolin2015equilibration}, to which the reader is also referred for a comprehensive overview of the context).
Our goal is to clarify this subtle and somewhat elusive point by providing, for the first time to our knowledge,
a proof that the very definition of ideal bath actually implies the ETH.
Our result then definitively settles the question of determining whether a quantum system has a thermal behavior, reducing it to checking whether its Hamiltonian satisfies the ETH: if the ETH is satisfied, the system always thermalizes, while if it is not satisfied, there certainly exists some reasonable physical initial state not leading to thermalization.

The Chapter is structured as follows.
In Section \ref{sectherm} we state preliminary, rigorous definitions of
thermalization and of an ideal bath.
In Section \ref{seceth} we present our definition of ETH, and we then reconsider its role as a sufficient condition for thermalization on the basis of our definitions.
In Section \ref{secethmain} we proceed to present of our main finding, that the ETH is also necessary for thermalization.
Complete proofs of the needed lemmata may be found in Section \ref{appeth}.
Finally, we conclude in Section \ref{secceth}.

\section{Thermalization and ideal baths}\label{sectherm}
Let us consider a system $S$ coupled to a heat bath $B$, with Hilbert spaces $\mathcal{H}_S$ and $\mathcal{H}_B$ of dimension $d_S$ and $d_B$, respectively.
For convenience, we describe the total Hamiltonian as $\hat{H}=\hat{H}_C+ \hat{H}_B$, composed of a free term $\hat{H}_B$ associated with the bath's inner dynamics, and a term $\hat{H}_C$ that includes
both the free component associated with $S$ and the system-bath coupling component. We only require the norm $\|\hat{H}_C\|$ to be bounded independently of the dimension $d_B$ of the bath\footnote{We denote by $\|\cdot\|$ and $\|\cdot\|_1$ the operator norm (so the largest singular value) and trace norm (so the sum of the singular values) of $\cdot\;$, respectively.}.
Let then the global system start in some state $\hat{\rho}$. At time $t$ it will evolve into the density matrix
\begin{equation}
\hat{\rho}(t)=e^{-i\hat{H}t}\,\hat{\rho}\,e^{i\hat{H}t}\;,
\end{equation}
whose time-averaged counterpart is
the diagonal part of $\hat{\rho}$ in the energy eigenbasis,
\begin{equation}
\Phi\left(\hat{\rho}\right)=\sum_n p_n\,|n\rangle\langle n|\;,
\end{equation}
assuming the spectrum of $\hat{H}$ to be nondegenerate for simplicity. Here, $\Phi$ denotes the time-averaging map and
$p_n=\langle n|\hat{\rho}|n\rangle$ is the probability that the global system has energy $E_n$ \cite{linden2009quantum}.
The time-averaged reduced state of the system $S$ is then obtained by taking the partial trace of
$\Phi\left(\hat{\rho}\right)$ over the bath degrees of freedom,
\begin{equation}\label{Phi}
\Phi_S\left(\hat{\rho}\right)\equiv\mathrm{Tr}_B\Phi\left(\hat{\rho}\right)=\sum_n p_n\;\hat{\tau}_n\;,
\end{equation}
where $\hat{\tau}_n\equiv\mathrm{Tr}_B|n\rangle\langle n|$ is the partial trace of the eigenstate $|n\rangle$.
In this context, thermalization is said to occur when the density matrices $\Phi_S\left(\hat{\rho}\right)$ exhibit a functional dependence only on those properties of the initial states $\hat{\rho}$
which are directly associated with the bath, as the initial properties of $S$ are washed away by the time-average and partial trace operations.

A key point in the study of such processes
is the choice of the set which identifies the initial states $\hat{\rho}$ of the joint system under which thermalization is assumed to occur:
too broad a set being typically too restrictive to describe realistic configurations, too narrow a set leading instead to trivial results.
In many cases of physical interest, one would know the value of only some macroscopic observables of the bath, such as the energy, so a common hypothesis is to impose thermalization when the bath is in the mixed state that maximizes the von Neumann entropy among all the states with given expectation values of the known observables \cite{reimann2010canonical}.
A weakness of this approach is that it does not account for situations where the bath is prepared in a pure state.
Another approach based on typicality has then been proposed.
In Ref. \cite{linden2009quantum}, the initial state of the bath is a pure state chosen randomly according to the Haar measure on the subspace of the bath Hilbert space compatible with the values of the known macroscopic observables.
The reduced system equilibrium state is then proven to be close, with very high probability,
to the equilibrium state resulting from choosing as initial state of the bath the normalized projector over the considered subspace.
A more refined choice would be to modify the notion of typicality by adopting
probability measures that reflect the complexity of the state preparation.
Indeed, the quantum pure states that are more easily built and comparatively stable are the ground states of local Hamiltonians,
so that one may restrict to the uniform measure on the states satisfying the area law \cite{eisert2010colloquium,garnerone2010typicality,garnerone2010statistical},
or introduce a measure arising from applying a local random quantum circuit to a completely factorized initial state \cite{hamma2012quantum,hamma2012ensembles}.
However, these probability measures are much more complicated than the uniform one on the whole Hilbert space,
and the computations may not be feasible.

Besides, asking whether there exist initial states of the bath not leading to thermalization of the system is a legitimate question,
to which these approaches based on typicality do not have an answer.
Here we want to address precisely this question.
Our definition of thermalization is therefore as follows:

\begin{defn}[Thermalization for initial product states]\label{defth}
We say that a subspace $\mathcal{H}_B^{\mathrm{eq}}$ of the bath Hilbert space induces thermalization of the system to a state $\hat{\omega}$ with precision $\epsilon$ if for any initial \emph{product} global state supported on $\mathcal{H}_S\otimes\mathcal{H}_B^{\mathrm{eq}}$ the equilibrium reduced state of the system is close to $\hat{\omega}$. That is,
$\mathcal{H}_B^{\mathrm{eq}}$ is such that\addtocounter{footnote}{-1}\addtocounter{Hfootnote}{-1}\footnotemark
\begin{equation}\label{thomega}
\left\|\Phi_S\left(\hat{\rho}\right)-\hat{\omega}\right\|_1\leq\epsilon
\end{equation}
for all $\hat{\rho}=\hat{\rho}_S\otimes\hat{\rho}_B$ with $\mathrm{Supp}\,\hat{\rho}_B\subset\mathcal{H}_B^{\mathrm{eq}}$.
\end{defn}
To discuss the connection between ETH and thermalization we shall further restrict the analysis to subspaces $\mathcal{H}_B^{\mathrm{eq}}$ corresponding to microcanonical energy shells  $\mathcal{H}_B(E,\Delta_B)$ of the bath free Hamiltonian, i.e., to subspaces
 spanned by those eigenvectors of $\hat{H}_B$ with eigenvalues in the interval $[E-\Delta_B,E+\Delta_B]$.
In this context the  associated equilibrium reduced state $\hat{\omega}$ entering Eq. \eqref{thomega} is  assumed to depend upon $\mathcal{H}_B(E,\Delta_B)$
only via a smooth function $\beta(E)$ of $E$, which effectively defines the inverse temperature $1/T(E)=k\beta(E)$ of the bath, $k$ being the Boltzmann constant. We notice that $\hat{\omega}(\beta(E))$ and $\beta(E)$ are otherwise arbitrary\footnote{In thermodynamics, the inverse temperature $\beta(E)$ is related to the density of energy levels
of $\hat{H}_B$ around $E$, $\Omega(E)$, by
$\beta(E)= \partial_{E}\ln\Omega$, while the density matrices  $\hat{\omega}(\beta)$ are identified with
the Gibbs states associated with the system Hamiltonian $\hat{H}_S$, i.e., $\hat{\omega}(\beta)=e^{-\beta\hat{H}_S}/\mathrm{Tr}(e^{-\beta\hat{H}_S})$.
However, both these assumptions are not necessary to prove our results, and we shall not make them here.}.
Of course, a necessary condition for this to happen  is to have  the width $\Delta_B$ much smaller than the scale over which the mapping $E\mapsto \hat{\omega}(\beta(E))$ varies appreciably.
More precisely, with $C \equiv dE/dT>0$ the bath's heat capacity, we must have that $\hat{\omega}(\beta)$ does not appreciably change for
variations of $\beta$ on the order
\begin{equation}
\delta\beta\approx\Delta_B|d\beta/dE|= k\beta^2\Delta_B/C\;.
\end{equation}
Considering that the largest energy scale that can be associated with the system alone is the operator norm $\|\hat{H}_C\|$, we can conclude that
 thermalization with precision $\epsilon$ is reasonable if $\|\hat{H}_C\|\delta\beta\leq\epsilon$, i.e. if
\begin{equation}\label{condepsilon}
k\,{\beta(E)}^2\,\Delta_B\,\|\hat{H}_C\|\leq \epsilon\, C(\beta(E))\;.
\end{equation}
We are then led to define an ideal heat bath as follows.
\begin{defn}[Ideal heat bath]\label{bath}
We say that a bath is ideal in the energy range $\mathcal{E}_B$\footnote{The restriction to a specific energy range $\mathcal{E}_B$  in the definition of ideal bath  originates from the need to exclude
possible pathological behaviors associated with the use of finite dimensional bath models to describe realistic physical configurations.
Similar considerations apply to the restriction to the energy range $\mathcal{E}$ of the spectrum of $\hat{H}$ in Definition \ref{ETH}.} with energy-dependent inverse temperature $\beta(E)$ if, for any
$\Delta_B$ and $\epsilon$ satisfying Eq. \eqref{condepsilon}
and for any $E\in\mathcal{E}_B$, the microcanonical shell $\mathcal{H}_B(E,\Delta_B)$ induces thermalization to the state $\hat{\omega}(\beta(E))$ with precision $\epsilon$ in the sense of Definition \ref{defth}.
\end{defn}

\section{ETH implies thermalization}\label{seceth}
The ETH roughly states that, given two eigenvalues $E_n$ and $E_m$ of the global Hamiltonian $\hat{H}$ which are close, the  associated reduced density matrices $\hat{\tau}_n$ and $\hat{\tau}_m$ defined in Eq. \eqref{Phi}
must also be close,  i.e., that $\hat{\tau}_n$ is a ``sufficiently continuous'' function of the energy of the joint system.
More precisely,
our working definition is the following.
\begin{defn}[ETH]\label{ETH}
We say that a Hamiltonian $\hat{H}=\sum_nE_n|n\rangle\langle n|$ fulfils the ETH in the region of the spectrum $\mathcal{E}$\addtocounter{footnote}{-1}\addtocounter{Hfootnote}{-1}\footnotemark on a scale $\Delta$ with precision $\epsilon_{ETH}$  if
all $E_n,\,E_m\in\mathcal{E}$ with $\left|E_m-E_n\right|\le 2\Delta$ fulfil  $\left\|\hat{\tau}_m-\hat{\tau}_n\right\|_1\le \epsilon_{ETH}$.
\end{defn}
It is worth observing that
the usual formulation of the ETH \cite{deutsch1991quantum,gogolin2015equilibration} does not split the global system into system and bath. Instead, it identifies a class of relevant macroscopic observables $\mathcal{A}$, and states that for any $\hat{A}\in\mathcal{A}$ the diagonal matrix elements in the energy eigenbasis $\langle n|\hat{A}|n\rangle$ depend ``sufficiently continuously'' on the energy. Upon choosing as $\mathcal{A}$ the set of self-adjoint operators acting on the system alone, our definition is equivalent. Indeed, for any $\hat{A}=\hat{A}_S\otimes\hat{\mathbb{I}}_B$ we have $\langle n|\hat{A}|n\rangle=\mathrm{Tr}_S(\hat{A}_S\,\hat{\tau}_n)$, which are sufficiently continuous functions of the energy for any $\hat{A}_S$ if and only if $\hat{\tau}_n$ is.

It is well established that if the ETH holds for any initial global state with a sharp enough energy distribution, then the time average of the reduced state of the system is a smooth function of its average global energy alone \cite{gogolin2015equilibration}; i.e., different initial global states lead to nearly the same equilibrium reduced state for the system if their average energies are close and their energy distribution is sufficiently sharp.
Moreover, this equilibrium state is close to the one associated with a microcanonical global state.
To make our treatment self-contained, and better emphasize the importance of the ETH in the study of thermalization,
let us state here precisely our version of this implication in terms of the definitions introduced above (see Section \ref{proofP1} for a proof).
\begin{prop}[ETH implies microcanonical thermalization]\label{ETH->th}
Let $\hat{H}$ fulfil the ETH in $\mathcal{E}$ on a scale $\Delta$ with precision $\epsilon_{ETH}$.
Let $\hat{P}$ be the projector onto the energy shell $\mathcal{H}(E,\Delta)$ of the total Hamiltonian, so onto the subspace spanned by those eigenvectors of $\hat{H}$ that have eigenvalues in the interval $[E-\Delta,E+\Delta]$, which is assumed to be contained in $\mathcal{E}$.
Then, for any initial state $\hat{\rho}$ peaked around the energy $E$ in the sense $\mathrm{Tr}[\hat{\rho}(\hat{\mathbb{I}}-\hat{P})]\leq\epsilon_{ETH}$,
the time-averaged reduced state $\Phi_S\left(\hat{\rho}\right)$ of Eq. \eqref{Phi} is close to the reduced microcanonical state associated to $\mathcal{H}(E,\Delta)$,
\begin{equation}\label{omegaE}
\left\|\Phi_S(\hat{\rho})-\mathrm{Tr}_B(\hat{P})/\mathrm{Tr}(\hat{P})\right\|_1\le 3\epsilon_{ETH}\;.
\end{equation}
\end{prop}
Let us stress that this proposition does \emph{not} assume $\hat{\rho}$ to be a product or separable state; i.e., the ETH implies thermalization even if the system and bath are initially entangled.
The link with Definitions \ref{defth} and \ref{bath} is then provided by Lemma \ref{rhoQeth} of Section \ref{appeth}: If $\hat{\rho}$ is a state supported on $\mathcal{H}_S\otimes\mathcal{H}_B(E,\Delta_B)$ then
\begin{equation}
\mathrm{Tr}[\hat{\rho}(\hat{\mathbb{I}}-\hat{P})]\leq\epsilon_{ETH}\;,
\end{equation}
and Eq. \eqref{omegaE} follows from the ETH on a scale $\Delta=(\|\hat{H}_C\|+\Delta_B)/\sqrt{\epsilon_{ETH}}$. Further, for conditions under which the microcanonical state may be replaced by the canonical state, see, e.g., Ref.'s \cite{riera2012thermalization,mueller2013thermalization,popescu2006entanglement,goldstein2006canonical,brandao2015equivalence} and references therein.

\section{Thermalization implies ETH}\label{secethmain}
Proposition \ref{ETH->th} seems to imply that the ETH is too strong a hypothesis and that weaker assumptions might be sufficient to
justify thermalization. It turns out that this is not true. Indeed,
we shall prove  that the ETH must hold for any ideal heat bath satisfying Definition \ref{bath}.
First off, we show that if a subspace of the bath $\mathcal{H}_B^{\mathrm{eq}}$ induces thermalization to a state $\hat{\omega}$ for any initial product state as per Definition \ref{defth}, the property extends to the entangled initial states up to an overhead which is linear in the system dimension.
Our argument relies on the observation that the entanglement of the eigenstates $|n\rangle$ is limited by the system dimension $d_S$, and cannot grow arbitrarily even when the bath dimension is large.
Note that this result is similar in spirit to the main finding of Ref. \cite{gogolin2011absence}, where thermalization is disproved in certain nonintegrable systems by establishing an upper bound on the average system-bath entanglement over random initial bath states.
\begin{lem}\label{therm}
Let $\mathcal{H}_B^{\mathrm{eq}}$ be a subspace of the bath Hilbert space that induces thermalization to a state $\hat{\omega}$ with precision $\epsilon$ in the sense of Definition \ref{defth}.
Then $\mathcal{H}_B^{\mathrm{eq}}$ induces thermalization also on the entangled initial states with precision $4d_S\epsilon$, i.e.
\begin{equation}\label{thm1}
\left\|\Phi_S\left(\hat{\rho}\right)-\hat{\omega}\right\|_1\leq4d_S\epsilon
\end{equation}
for all $\hat{\rho}$ with support contained in $\mathcal{H}_S\otimes\mathcal{H}_B^{\mathrm{eq}}$.
\end{lem}
By virtue of this Lemma,
the equilibration to some fixed state $\hat{\omega}$ of all initial product states in $\mathcal{H}_S\otimes\mathcal{H}_B^{\mathrm{eq}}$ extends to {\it all} initial states in this subspace.
Then, if an eigenstate $|n\rangle$ of the Hamiltonian is almost contained in the same subspace,
the resulting time-averaged reduced state of the system $\Phi_S(|n\rangle\langle n|)$ is also close to $\hat{\omega}$.
However, if we initialize the global system in an eigenstate of the Hamiltonian, it obviously remains there forever,
\begin{equation}\label{Phin}
\Phi_S(|n\rangle\langle n|)=\hat{\tau}_n\;.
\end{equation}
Combining this with the fact that the trace norm is contracting under completely positive trace-preserving maps \cite{perez2006contractivity}, we have under the assumptions of Lemma \ref{therm} that (see Section \ref{main proof} for details)
\begin{equation}
\label{old_lemma_2}
\left\|\hat{\tau}_n-\hat{\omega}\right\|_1
\le 4d_S\epsilon+2\sqrt{\langle n|\hat{Q}|n\rangle},
\end{equation}
where $\hat{Q}$ is the projector onto the subspace orthogonal to $\mathcal{H}_S\otimes\mathcal{H}_B^{\mathrm{eq}}$. It remains to bound $\langle n|\hat{Q}|n\rangle$ for given $\mathcal{H}_B^{\text{eq}}=\mathcal{H}_B(E,\Delta_B)$, which we do in Section \ref{main proof}, to arrive at the statement that whenever $\mathcal{H}_B(E,\Delta_B)$ induces thermalization to $\hat{\omega}$ with precision $\epsilon$ then for all $n$ with $|E_n-E|\le \Delta_B/2$ we have
\begin{equation}
\left\|\hat{\tau}_n-\hat{\omega}\right\|_1\le \frac{8\|\hat{H}_C\|^2}{\Delta_B^2}+4d_S\epsilon,
\end{equation}
which implies our main result (see Section \ref{main proof} for details).
\begin{thm}[Thermalization implies ETH]\label{->ETH} Let the bath be ideal in the energy range $\mathcal{E}_B$ as in Definition \ref{bath}. Let
\begin{equation}
\epsilon_{ETH}=12\sup_{E\in\mathcal{E}_B}
\left(\frac{2\|\hat{H}_C\|^2d_S\,k\,{\beta(E)}^2}{C(\beta(E))}\right)^{2/3}.
\end{equation}
Then $\hat{H}$ fulfils the ETH in the region $\mathcal{E}_B$ on a scale
\begin{equation}
\Delta=2\sqrt{3}\frac{\|\hat{H}_C\|}{\sqrt{\epsilon_{ETH}}}
\end{equation} with precision $\epsilon_{ETH}$.
\end{thm}

Typically,
for any fixed inverse temperature $\beta$, the bath's heat capacity $C(\beta)$ increases with the size of the bath.
On the contrary, $\hat{H}_C$ has been chosen such that it remains bounded.
Then, for fixed $\beta$ and $d_S$, the error $\epsilon$ becomes arbitrarily small (and thus the width $\Delta$ arbitrarily large) as $d_B\rightarrow\infty$.

\section{Proofs}\label{appeth}
Here we provide explicit proofs of the various lemmata and theorems.

\subsection{Proof of Proposition \ref{ETH->th}}\label{proofP1}
Defining
\begin{equation}\label{defC}
\mathcal{C}\equiv\left\{n:|E_n-E|\le\Delta\right\}\subset\mathcal{E}\;,
\end{equation}
the partial trace of the microcanonical shell can be written as
\begin{equation}
\frac{\mathrm{Tr}_B\hat{P}}{\mathrm{Tr}\,\hat{P}}=\frac{1}{|\mathcal{C}|}\sum_{n\in\mathcal{C}}\hat{\tau}_n\;.
\end{equation}
We have then
\begin{eqnarray}
\left\|\Phi_S\left(\hat{\rho}\right)-\frac{\mathrm{Tr}_B\hat{P}}{\mathrm{Tr}\,\hat{P}}\right\|_1&=& \frac{1}{|\mathcal{C}|}\left\|\sum_n\sum_{m\in\mathcal{C}}p_n\left(\hat{\tau}_n-\hat{\tau}_m\right)\right\|_1\leq \frac{1}{|\mathcal{C}|}\sum_n\sum_{m\in\mathcal{C}}p_n\left\|\hat{\tau}_n-\hat{\tau}_m\right\|_1\leq\nonumber\\ &\leq&\frac{1}{|\mathcal{C}|}\sum_{n\in\mathcal{C}}\sum_{m\in\mathcal{C}}p_n\left\|\hat{\tau}_n-\hat{\tau}_m\right\|_1+2\mathrm{Tr}(\hat{\rho}\,(\hat{\mathbb{I}}-\hat{P}))\;,
\end{eqnarray}
where $\mathrm{Tr}(\hat{\rho}\,(\hat{\mathbb{I}}-\hat{P}))\le \epsilon_{ETH}$ and
from \eqref{defC}, for any $m,n\in\mathcal{C}$ we have $|E_n-E_m|\le 2\Delta$ and then $\left\|\hat{\tau}_n-\hat{\tau}_m\right\|\leq\epsilon_{ETH}$.

\subsection{Lemma \ref{rhoQeth}}

One arrives at the statement after Proposition \ref{ETH->th} in the main text by applying the following lemma to $\hat{A}_1=\hat{H}$, $\hat{A}_2=\hat{\mathbb{I}}_S\otimes\hat{H}_B$, $\lambda=E$, $\Delta_1=\Delta=\frac{\|\hat{H}_C\|+\Delta_B}{\sqrt{\epsilon_{ETH}}}$ and $\Delta_2=\Delta_B$.
\begin{lem}\label{rhoQeth}
Let us consider two self-adjoint operators $\hat{A}_1$ and $\hat{A}_2$.
Let $\mathcal{H}_i(\lambda,\Delta_i)$ be the subspace identified by $\lambda-\Delta_i\leq\hat{A}_i\leq\lambda+\Delta_i$, for $i=1,2$.
Let $\hat{\rho}$ be a quantum state with support contained in $\mathcal{H}_2(\lambda,\Delta_2)$, and $\hat{Q}$ the projector onto the subspace orthogonal to $\mathcal{H}_1(\lambda,\Delta_1)$.
Then
\begin{equation}
\label{dd}
\mathrm{Tr}\bigl(\hat{\rho}\,\hat{Q}\bigr)\leq \left(\frac{\|\hat{A}_1-\hat{A}_2\|+\Delta_2}{\Delta_1}\right)^2\;.
\end{equation}
\begin{proof}
Let us consider first a pure state $\hat{\rho}=|\psi\rangle\langle\psi|$ and start from the identity
\begin{equation}
(\hat{A}_1-\lambda)|\psi\rangle=(\hat{A}_1-\hat{A}_2)|\psi\rangle+(\hat{A}_2-\lambda)|\psi\rangle\;.
\end{equation}
On one hand, the square norm of the left-hand-side is
\begin{equation}\label{LHS}
\|(\hat{A}_1-\lambda)|\psi\rangle\|^2=\langle\psi|(\hat{A}_1-\lambda)^2|\psi\rangle\geq\Delta_1^2\langle\psi|\hat{Q}|\psi\rangle\;.
\end{equation}
On the other hand, the norm of the right-hand-side satisfies
\begin{eqnarray}\label{RHS}
\|(\hat{A}_1-\hat{A}_2)|\psi\rangle+(\hat{A}_2-\lambda)|\psi\rangle\| &\leq& \|\hat{A}_1-\hat{A}_2\|+\|(\hat{A}_2-\lambda)|\psi\rangle\|\leq\nonumber\\
&\leq& \|\hat{A}_1-\hat{A}_2\|+\sqrt{\langle\psi|(\hat{A}_2-\lambda)^2|\psi\rangle}\leq\nonumber\\
&\leq& \|\hat{A}_1-\hat{A}_2\|+\Delta_2\;.
\end{eqnarray}
Putting together \eqref{LHS} and \eqref{RHS}, we have Eq. \eqref{dd} for pure states. For mixed states $\hat{\rho}=\sum_k p_k|\psi_k\rangle\langle\psi_k|$ with  $|\psi_k\rangle\in\mathcal{H}_2(\lambda,\Delta_2)$ we have $\mathrm{Tr}(\hat{\rho}\,\hat{Q})=\sum_kp_k\langle\psi_k|\hat{Q}|\psi_k\rangle$ and the assertion follows by applying the pure-state result to each term individually.
\end{proof}
\end{lem}

\subsection{Proof of Lemma \ref{therm}}\label{proofL1}
Let $\hat{\rho}$ a state in the subspace $\mathcal{H}_S\otimes\mathcal{H}_B^{eq}$ and denote by $\hat{P}$ the projector on said subspace.
By the variational characterization of the trace norm we have
\begin{equation}
\frac{1}{2}\|\Phi_S\left(\hat{\rho}\right)-\hat{\omega}\|_1=\bigl|\text{Tr}\bigl[\hat{M}(\Phi_S\left(\hat{\rho}\right)-\hat{\omega})\bigr]\bigr|
=\bigl|\text{Tr}[\hat{\rho}\hat{X}_{\hat{M}}]\bigr|,\;\;\;\hat{X}_{\hat{M}}=\sum_n\text{Tr}[(\hat{\tau}_n-\hat{\omega})\hat{M}]\hat{P}|n\rangle\langle n|\hat{P}\;,
\end{equation}
for some $\hat{M}$ with $0\le \hat{M}\le\hat{\mathbb{I}}$. Assuming w.l.o.g. that $\hat{\varrho}=|\psi\rangle\langle\psi|$ (the mixed case follows by convexity), we may Schmidt-decompose $|\psi\rangle=\sum_{k=1}^{d_S}\psi_k|s_k\rangle| b_k\rangle$. Hence
\begin{equation}
\frac{1}{2}\|\Phi_S\left(|\psi\rangle\langle\psi|\right)-\hat{\omega}\|_1=\sum_{k,l}\psi_k\psi_l\langle s_k|\langle b_k|\hat{X}_{\hat{M}}|s_l\rangle|b_l\rangle
=:\sum_{k,l}\psi_k\psi_l\langle s_k|\hat{S}_{k,l}|s_l\rangle\;,
\end{equation}
where the operator $\hat{S}_{k,l}=\langle b_k|\hat{X}_{\hat{M}}|b_l\rangle$ acts on $\mathcal{H}_S$ and we have
by the triangle and Cauchy--Schwarz inequality
\begin{equation}
2|\langle s_k|\hat{S}_{k,l}|s_l\rangle|\le |\langle s_k|\bigl(\hat{S}_{k,l}+\hat{S}_{k,l}^\dagger\bigr)|s_l\rangle|+|\langle s_k|i\bigl(\hat{S}_{k,l}-\hat{S}^\dagger_{k,l}\bigr)|s_l\rangle|\le \|\hat{S}_{k,l}+\hat{S}_{k,l}^\dagger\|+\|i\bigl(\hat{S}_{k,l}-\hat{S}^\dagger_{k,l}\bigr)\|,
\end{equation}
where $\hat{S}_{k,l}+\hat{S}_{k,l}^\dagger$ and $i(\hat{S}_{k,l}-\hat{S}^\dagger_{k,l})$ are hermitian such that
\begin{eqnarray}
\|\hat{S}_{k,l}+\hat{S}_{k,l}^\dagger\|&=& \max_{\substack{|\psi\rangle\in\mathcal{H}_S\\ \langle\psi|\psi\rangle=1}}|\langle \psi|\bigl(\hat{S}_{k,l}+\hat{S}_{k,l}^\dagger\bigr)|\psi\rangle|
\le\nonumber\\
&\le& \max_{\substack{|\psi\rangle\in\mathcal{H}_S\\ \langle\psi|\psi\rangle=1}}|\langle \psi|\langle b_k|\hat{X}_{\hat{M}}|b_l\rangle|\psi\rangle|
+\max_{\substack{|\psi\rangle\in\mathcal{H}_S\\ \langle\psi|\psi\rangle=1}}|\langle \psi|\langle b_l|\hat{X}_{\hat{M}}|b_k\rangle|\psi\rangle|\le\nonumber\\
&\le& 2\max_{\substack{|\psi\rangle\in\mathcal{H}_S\\ \langle\psi|\psi\rangle=1}}
\max_{\substack{|\phi\rangle\in\mathcal{H}_B^{\text{eq}}\\ \langle\phi|\phi\rangle=1}}
|\langle \psi|\langle \phi|\hat{X}_{\hat{M}}|\phi\rangle|\psi\rangle|,
\end{eqnarray}
where we used the Cauchy-Schwarz inequality and the hermiticity of $\langle \psi|\hat{X}_{\hat{M}}|\psi\rangle$ to obtain the last line. The same upper bound holds for $|\langle s_k|i\bigl(\hat{S}_{k,l}-\hat{S}^\dagger_{k,l}\bigr)|s_l\rangle|$. Further,
\begin{equation}
|\langle \psi|\langle \phi|\hat{X}_{\hat{M}}|\phi\rangle|\psi\rangle|=\bigl|\text{Tr}\bigl[\hat{M}(\Phi_S\left(|\psi\rangle\langle\psi|\otimes|\phi\rangle\langle\phi|\right)-\hat{\omega})\bigr]\bigr|\le
\bigl\|\Phi_S\left(|\psi\rangle\langle\psi|\otimes|\phi\rangle\langle\phi|\right)-\hat{\omega}\bigr\|_1\le \epsilon
\end{equation}
such that
\begin{equation}
\|\Phi_S\left(|\psi\rangle\langle\psi|\right)-\hat{\omega}\|_1\le 4\epsilon\sum_{k,l}\psi_k\psi_l\le 4\epsilon d_S.
\end{equation}

\subsection{Proof of Theorem \ref{->ETH}}
\label{main proof}
We first give the details of how to arrive at Eq. \eqref{old_lemma_2}. Denote the projector onto $\mathcal{H}_S\otimes\mathcal{H}_B^{\mathrm{eq}}$ by $\hat{P}$. Inserting a zero and using the triangle inequality yields
\begin{equation}
\left\|\hat{\tau}_n-\hat{\omega}\right\|_1
\le \left\|\Phi_S\left(|n\rangle\langle n|-\frac{\hat{P}|n\rangle\langle n|\hat{P}}{\langle n|\hat{P}|n\rangle}\right)\right\|_1+\left\|\Phi_S\left(\frac{\hat{P}|n\rangle\langle n|\hat{P}}{\langle n|\hat{P}|n\rangle}\right)-\hat{\omega}\right\|_1\;.
\end{equation}
Making use of the contractivity of the trace norm for the first term and the assumptions of Lemma \ref{therm} for the second term, we have
\begin{equation}
\left\|\hat{\tau}_n-\hat{\omega}\right\|_1
\le \left\||n\rangle\langle n|-\frac{\hat{P}|n\rangle\langle n|\hat{P}}{\langle n|\hat{P}|n\rangle}\right\|_1+4d_S\epsilon=2\sqrt{\langle n|\hat{Q}|n\rangle}+4d_S\epsilon\;,\label{aa}
\end{equation}
where in the second step we have derived the trace norm with an explicit computation of the eigenvalues.

Now let $\hat{H}_B=\sum_{k}e_k|k\rangle\langle k|$ and $\hat{Q}=\sum_{k\notin\mathcal{H}_B^{\text{eq}}}|k\rangle\langle k|$. Then
\begin{equation}
\min_{k\notin \mathcal{H}_B^{\text{eq}}}(e_k-E_n)^2\hat{Q}\le \sum_{k\notin\mathcal{H}_B^{\text{eq}}}(e_k-E_n)^2|k\rangle\langle k|
=\hat{Q}\sum_{k}(e_k-E_n)^2|k\rangle\langle k|
=\hat{Q}\bigl(\hat{H}_B-E_n\bigr)^2.
\end{equation}
Hence, by the Cauchy--Schwarz inequality
\begin{eqnarray}\label{bb}
\min_{k\notin \mathcal{H}_B^{\text{eq}}}(e_k-E_n)^2\langle n|\hat{Q}|n\rangle &\le& \sqrt{\langle n|\hat{Q}|n\rangle\langle n|\bigl(\hat{H}_B-E_n\bigr)^4|n\rangle}=\nonumber\\
&=&\sqrt{\langle n|\hat{Q}|n\rangle\langle n|\bigl(\hat{H}_B-\hat{H}\bigr)^4|n\rangle}\le\nonumber\\
&\le& \sqrt{\langle n|\hat{Q}|n\rangle}\|\hat{H}_B-\hat{H}\|^2 =\sqrt{\langle n|\hat{Q}|n\rangle}\|\hat{H}_C\|^2.
\end{eqnarray}
With $\mathcal{H}_B^{\text{eq}}=\mathcal{H}_B(E,\Delta_B)=\text{span}\{|k\rangle : |e_k-E|\le \Delta_B\}$, we have
\begin{equation}
\min_{k\notin \mathcal{H}_B^{\text{eq}}}(e_k-E_n)^2=\min_{k:|e_k-E|>\Delta_B}(e_k-E_n)^2\;,
\end{equation}
and, combining Eqs. \eqref{aa},\eqref{bb}, we have that if  $\mathcal{H}_B^{\text{eq}}=\mathcal{H}_B(E,\Delta_B)$ induces thermalization to a state $\hat{\omega}$ with precision $\epsilon$ then for all $n$ with $|E_n-E|\le \Delta_B/2$
\begin{equation}
\left\|\hat{\tau}_n-\hat{\omega}\right\|_1\le \frac{2\|\hat{H}_C\|^2}{\min_{|e-E|>\Delta_B}(e-E_n)^2}+4d_S\epsilon\le \frac{8\|\hat{H}_C\|^2}{\Delta_B^2}+4d_S\epsilon.
\end{equation}

If the bath is ideal in the energy range $\mathcal{E}_B$ with inverse temperature $\beta(E)$ then for any $\epsilon$, $\Delta_B$ with
\begin{equation}
\label{cc}
k\,{\beta(E)}^2\,\Delta_B\,\|\hat{H}_C\|\le \epsilon\, C(\beta(E))
\end{equation}
and any $E\in\mathcal{E}_B$ we have that $\mathcal{H}_B(E,\Delta_B)$ induces thermalization to the state $\hat{\omega}(\beta(E))$ with precision $\epsilon$. Hence, setting $\epsilon$ such that we have equality in Eq. \eqref{cc}
 and letting $E\in\mathcal{E}_B$ and $n$ such that $|E_n-E|\le \Delta_B/2$, we have
\begin{equation}
\left\|\hat{\tau}_n-\hat{\omega}\right\|_1\le 4\|\hat{H}_C\|\left(\frac{2\|\hat{H}_C\|}{\Delta_B^2}+d_S\frac{k\,{\beta(E)}^2\,\Delta_B\,}{C(\beta(E))}\right),
\end{equation}
which is minimized by
\begin{equation}
\Delta_B^3=\frac{4\|\hat{H}_C\|C(\beta(E))}{d_Sk\,{\beta(E)}^2}\;.
\end{equation}
Hence, if the bath is ideal in the energy range $\mathcal{E}_B$ then for any $E_n,E_m\in\mathcal{E}_B$ with (we set $E=(E_n+E_m)/2$)
\begin{equation}
|E_n/2-E_m/2|=|E_{n/m}-E|\le  \Delta_B/2:=\Delta=2\|\hat{H}_C\|\sqrt{\frac{3}{\epsilon_{ETH}}}\;,
\end{equation}
 we have
\begin{equation}
\left\|\hat{\tau}_{n/m}-\hat{\omega}\right\|_1\le 24\|\hat{H}_C\|^{2}/\Delta_B^{2}=6\left(\frac{2\|\hat{H}_C\|^2d_Sk\,{\beta(E)}^2}{C(\beta(E))}\right)^{2/3}=:\epsilon_{ETH}/2\;,
\end{equation}
and finally
\begin{equation}
\|\hat{\tau}_m-\hat{\tau}_n\|_1\le \left\|\hat{\tau}_m-\hat{\omega}\right\|_1+\left\|\hat{\omega}-\hat{\tau}_n\right\|_1\le\epsilon_{ETH}\;.
\end{equation}

\section{Conclusion}\label{secceth}
The Eigenstate Thermalization Hypothesis has been central to much of the ongoing discussion
concerning the relaxation of open quantum systems to fixed equilibrium states.
Its role as a sufficient condition for thermalization, which we reviewed in Proposition \ref{ETH->th}, is well established
and has been repeatedly remarked on in several past contributions.
By proving that, conversely, an ideal heat bath must necessarily interact with the system with
a Hamiltonian fulfilling the ETH we have, in a precise and rigorous sense, revealed the full role
such a condition has to play. This result rests on a definition of an ideal bath which is rigorous and yet
broad enough to encompass all practically relevant instances,
and hence sheds considerable light on the very general mechanisms that let open quantum systems thermalize.

\chapter{A universal limit for testing quantum superpositions}
\label{chsuperpos}
In this Chapter we prove that any measurement able to distinguish a coherent superposition of two wavepackets of a massive or charged particle from the corresponding incoherent statistical mixture must require a minimum time.
For a charged particle this bound can be ascribed to the electromagnetic radiation that is unavoidably emitted during the measurement.
Then, for a massive particle this bound provides an indirect evidence for the existence of quantum gravitational radiation.

The Chapter is based on
\begin{enumerate}
\item[\cite{mari2015experiments}] A.~Mari, G.~De~Palma, and V.~Giovannetti, ``Experiments testing macroscopic quantum superpositions must be slow,'' \emph{Scientific Reports}, vol.~6, p. 22777, 2016.\\ {\small\url{http://www.nature.com/articles/srep22777}}
\end{enumerate}
\section{Introduction}

The existence of coherent superpositions is a fundamental postulate of quantum mechanics but, apparently, implies very counterintuitive consequences when extended to macroscopic systems. This problem, already pointed out since the beginning of quantum theory through the famous Schr\"odinger cat paradox \cite{wheeler2014quantum}, has been the subject of a large scientific debate which is still open and very active.

Nowadays there is no doubt about the existence of quantum superpositions.
Indeed this effect has been demonstrated in a number of experiments involving microscopic systems (photons \cite{taylor1909interference,dopfer1998two}, electrons \cite{donati1973experiment, tonomura1989demonstration}, neutrons \cite{zeilinger1988single}, atoms \cite{anderson1998macroscopic, monroe1996schrodinger}, molecules \cite{arndt1999wave, eibenberger2013matter}, {\it etc.}).
However, at least in principle, the standard theory of quantum mechanics is valid at any scale and does not put any limit on the size of the system:  if you can delocalize a molecule then nothing should forbid you to delocalize a cat, apart from technical difficulties. Such difficulties are usually associated with the impossibility of isolating the system from its environment, because it is well known that any weak interaction changing the state of the environment is sufficient to destroy the initial coherence of the system.

We are interested in the ideal situation in which we have a macroscopic mass or a macroscopic charge perfectly isolated from the environment and prepared in a quantum superposition of two spatially separated states. Without using  any speculative theory of quantum gravity or sophisticated tools of quantum field theory, we propose a simple thought experiment based on  particles interacting via semiclassical forces. Surprisingly a simple consistency argument with relativistic causality is enough to obtain a fundamental result which, being related to gravitational and electric fields, indirectly tells us something about quantum gravity and quantum field theory.

The result is the following: assuming that a macroscopic mass $m$ is prepared in a superposition of two states separated by a distance $d$, then any experiment discriminating the coherent superposition from a classical incoherent mixture requires a minimum time $T \propto m \, d$, proportional to the mass and the separation distance. Analogously for a quantum superposition of a macroscopic charge $q$, such minimum time is proportional to the associated electric dipole  $T \propto q\, d$. In a nutshell, experiments testing macroscopic superpositions are possible in principle, but they need to be slow.
For common experiments  involving systems below the Planck mass and the Planck charge this limitation is irrelevant, however such time can become very important  at macroscopic scales. As an extreme example, if the center of mass of the Earth were in a quantum superposition with a separation distance of one micrometer, according to our result  one would need a time equal to the age of the universe in order to distinguish this state from
a classical statistical mixture. Clearly this limitation suggests that at sufficiently macroscopic scales quantum mechanics can be safely replaced by classical statistical mechanics without noticing the difference.

 The fact that large gravitational or electromagnetic fields can be a limitation for the observation of quantum superpositions is not a new idea. In the past decades, several models of spontaneous localization \cite{penrose1996gravity, ghirardi1986unified, diosi1989models,karolyhazy1966gravitation, bassi2013models} have been proposed which, going beyond the standard theory of quantum mechanics,  postulate the existence of a gravity induced collapse at macroscopic scales. Remaining within the domain of standard quantum mechanics, the loss of coherence in interference experiments due to the emission of electromagnetic radiation has been already studied in the literature \cite{breuer2001destruction,baym2009two}. Similarly, the interaction of a massive particle with gravitational waves \cite{blencowe2013effective,suzuki2015environmental,jaekel2006quantum} and the dephasing effect of time dilation on internal degrees of freedom \cite{pikovski2013universal} have been considered as possible origins of quantum decoherence.

  For what concerns our thought experiment, a similar setup can be found in the literature where the interference pattern of an electron passing through a double slit is destroyed by a distant measurement of its electric field. This thought experiment can be traced back to Bohr as quoted in \cite{baym2009two},  was discussed by Hardy interviewed in \cite{schlosshauer2011elegance} and appears as an exercise in the book by Aharanov and Rohrlich \cite{aharonov2008quantum}. Recently different experiments involving interacting test particles have been proposed in order to discriminate the quantum nature of the gravitational field from a potentially classical description \cite{kafri2014classical,bahrami2015gravity}, while some limitations that relativistic causality imposes to the possible measurements in quantum field theory have been investigated in \cite{benincasa2014quantum}.

Our original contribution is that, imposing the consistency with relativistic causality, our thought experiment  allows the derivation of a fundamental minimum time which is valid for {\it any} possible experiment involving macroscopic superpositions. In this sense our bounds represent universal limitations having a role analogous to the Heisenberg uncertainty principle in quantum mechanics. For this reason, while our results could be observable in advanced and specific experimental setups \cite{romero2011large, marshall2003towards, arndt2014testing, pikovski2012probing, bawaj2014probing,schnabel2015einstein,scala2013matter,wan2015tolerance}, their main contribution  is probably a better understanding of the theory of quantum mechanics at macroscopic scales.
For charged particles we propose two different measurements for testing the coherence.
The first requires to accelerate the charge, and our bound on the discrimination time is due to the entanglement with the emitted photons.
In the second, the bound can be instead ascribed to the presence of the vacuum fluctuations of the electromagnetic field.
On the other hand we also find an equivalent bound associated to quantum superposition of large masses.
What is the origin of this limitation?
The analogy suggests that the validity of our bound could be interpreted as an indirect evidence for the existence of  quantum fluctuations of the gravitational field, and of quantum gravitational radiation.

The Chapter is structured as follows.
In Section \ref{experiment} we propose our thought experiment which suggests a minimum discrimination time for any macroscopic quantum superposition. In Section \ref{time} we derive a quantitative bound.
In Sections \ref{ent} and \ref{QED} we check the consistency of our results with an explicit analysis of a charged particle interacting with the electromagnetic field.
Here we propose two different measurements, and show that in both cases they are able to check the coherence only if their duration satisfies our fundamental limit.
Finally, we conclude in Section \ref{csup}

\section{Thought  experiment}\label{experiment}

Let us consider the thought experiment represented in Fig. \ref{gedanken}, and described by the following protocol. The protocol can be equivalently applied to quantum superpositions of large masses or large charges.

\begin{figure}
\includegraphics[width=\textwidth]{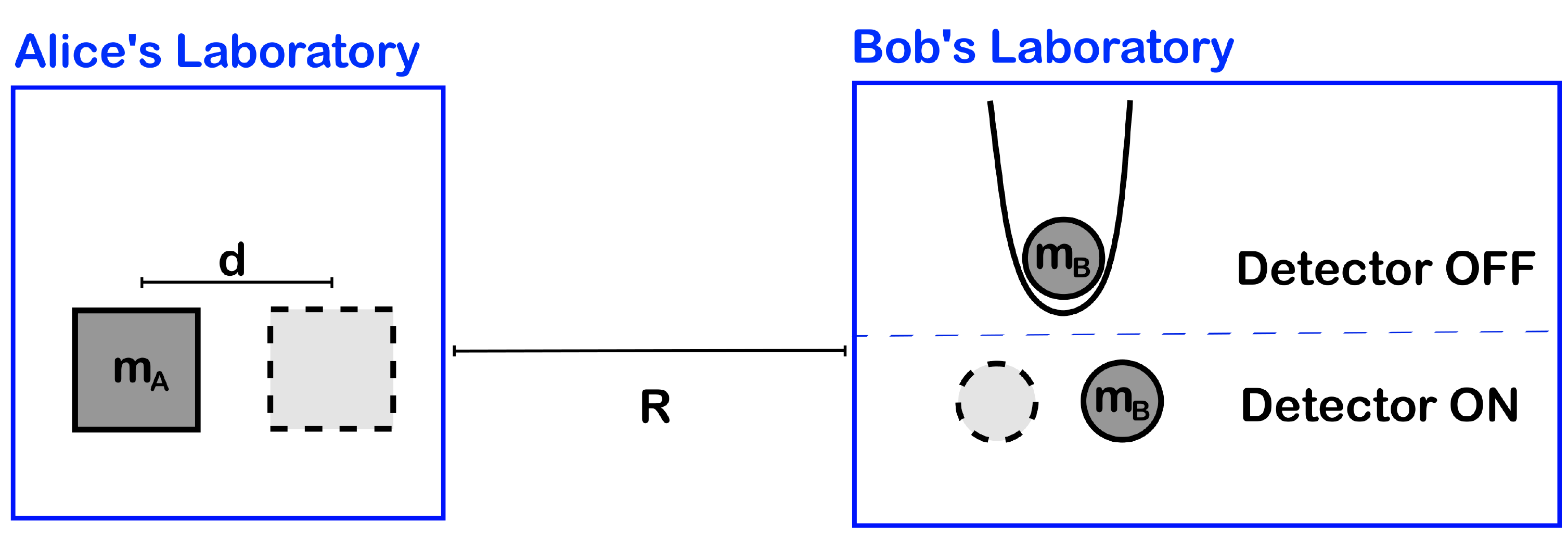}
\caption{Picture of the thought experiment. Alice prepares a macroscopic mass in a quantum spatial superposition. Bob has at disposal a test mass prepared in the ground state of a narrow harmonic trap. Bob can send one bit of information to Alice by choosing between two alternatives: doing nothing (detector {\it off}) or removing the trap (detector {\it on}). Once a time $T_{\mathrm{B}}$ necessary to generate entanglement (if the detector is {\it on}) has passed,
Alice performs a measurement in a time $T_{\mathrm{A}}$ in order to discriminate the coherent superposition from a classical incoherent mixture. In this way, by knowing whether the detector is {\it on} or {\it off}, she gets the information sent by Bob in a time $T_{\mathrm{A}}+T_{\mathrm{B}}$. A completely  equivalent protocol can be obtained by replacing massive particles with charged particles.   } \label{gedanken}
\end{figure}

{\paragraph{Protocol of the thought experiment}
\begin{enumerate}
\item Alice has at disposal, in her laboratory, a massive/charged particle in a macroscopic superposition of  a ``left'' and  a ``right'' state:
\begin{equation}\label{LR}
 | \psi \rangle =\frac{ | {\rm L} \rangle + | {\rm R}\rangle}{\sqrt{2}}.
\end{equation}
The wave functions of the two states are $\langle x | {\rm L}  \rangle = \phi(x)$ and $\langle x | {\rm R}  \rangle = \phi(x-d)$, where  $d>0$ is the relative separation of the superposition.
\item Bob is in a laboratory at a distance $R$ from Alice and containing a massive / charged test particle prepared in the ground state of a very narrow harmonic trap.  Bob freely chooses between two options: doing nothing ( detector = {\it off}), or removing the trap (detector = {\it on}). In the first case the state of test particle remains unchanged while, in the second case,  the dynamics is sensitive to the local Newton / Coulomb field generated by Alice's particle and the global state will eventually become entangled. If the detector is {\it off}, the initial quantum superposition is preserved, while if the detector is {\it on} the generation of entanglement eventually destroys the coherence of the reduced state of Alice.
\item
Alice performs an arbitrary experiment in her laboratory with the task of discriminating the coherent superposition from a statistical incoherent mixture of the two states $|{\rm L}  \rangle$ and $|{\rm R}  \rangle$.
For example, she could make an interference experiment, a measurement of the velocity, or she could measure the gravitational / electromagnetic field (or some spatial average of it) in any point within her laboratory.
The specific details of the experiment are irrelevant. Depending on the result of the experiment, Alice deduces the choice of  Bob ({\it i.e.} if the detector was {\it on} or {\it off}).
\end{enumerate}}

Clearly,  the previous thought experiment constitutes a communication protocol in which Bob can send information to Alice.
Moreover, for a large enough mass $m$ or for a large enough charge $q$, the test particle of Bob can become entangled with Alice's particle in an arbitrarily short time. But then, apparently, Bob can send a message to Alice faster than light  violating the fundamental principle of relativistic causality. How can we solve this  paradox?
Let us make a list of possible solutions:
\begin{enumerate}
\item[a)] It is impossible to prepare a macroscopic superposition state or to preserve its coherence because of some unknown intrinsic effect lying outside the theory of quantum mechanics.

\item[b)]
Once the superposition is created, the particle is entangled with its own static gravitational / electric field, and
Alice's local state is always mixed. Then she cannot distinguish a coherent superposition from the corresponding incoherent statistical mixture with an experiment inside her laboratory, since the probability distribution of the outcomes of any measurement she can perform does not depend on Bob's choice.
We notice that, if this were the solution, the protocol not only would not allow for superluminal communication, but it would not allow for communication at all.

\item[c)]
Alice needs a minimum time to locally discriminate whether the superposition is coherent or not.
More quantitatively we have that, if Bob is able to generate entanglement in a time $T_{\mathrm{B}}$ and if $T_{\mathrm{A}}$ is the time necessary to Alice for performing her discrimination measurement, then relativistic causality requires
\begin{equation}
T_{\mathrm{A}}+T_{\mathrm{B}} \ge \frac{R}{c}\;. \label{causality}
\end{equation}
Therefore, whenever entanglement can be generated in a time $T_{\mathrm{B}} \le R/c$, we get a non-trivial lower bound on $T_{\mathrm{A}}$. Here we are neglecting the time necessary to Bob for switching from {\it off} to {\it on} the detector, {\it i.e.} for removing the trap confining the particle.  In Section \ref{secDST} we justify the validity of this approximation.
\end{enumerate}

Anomalous decoherence effects \cite{penrose1996gravity, ghirardi1986unified, diosi1989models,karolyhazy1966gravitation, bassi2013models} (as {\it e.g.}\ the Penrose spontaneous localization model) are important open problems in the foundations of quantum mechanics and cannot be excluded  a priori. Up to now however their existence was never experimentally demonstrated and therefore, instead of closing our discussion by directly invoking point a), we try to remain within the framework of quantum mechanics and check if points b) or c) are plausible solutions.

The reader who is familiar with the field of open quantum systems may find the option b) very natural. In standard non-relativistic quantum mechanics, the formation of entanglement between a system and its environment is widely accepted as the origin of any observed form of decoherence. Indeed this approach has also been used to explain the decoherence of moving charged particles, mainly focusing to the double-slit interference experiment~\cite{breuer2001destruction}.
It has been recognized by previous works that in a double-slit  experiment there is a limit to the charge of the particle above which photons are emitted due to the acceleration associated to the interference paths~\cite{breuer2001destruction,baym2009two}. For large charges then, the particle entangles with the emitted photons and this effect can destroy the  interference pattern.
The reader can then notice that also in our case the particle needs to be accelerated when it is put in the superposition~\eqref{LR}, and if it is charged it will radiate and can become entangled with the emitted photons.
Similarly, an accelerated mass generates gravitational radiation and can become entangled with the emitted gravitons.
However, in Section \ref{rad}
we prove that, if the accelerations are slow enough, the resulting quantum state of the electromagnetic field has almost overlap one with the vacuum, and therefore the particle does not get entangled with the emitted photons because no photons at all are emitted.
The same argument can be repeated for the gravitational radiation in the linear approximation.

The reader may now think that the particle in the superposition~\eqref{LR} is entangled at least with its static Coulomb electric field.
However, as we show in details in Section \ref{dof}, the static Coulomb electric field is not a propagating degree of freedom (it has zero frequency) and vanishes in absence of electric charges.
In the Coulomb gauge, the Hilbert space associated to the static field is the same Hilbert space of the particle, whose reduced state remains pure.
Indeed, the quantum operator associated to the electric field contains explicitly the operator associated to the position of the particle, and then the expectation value of the electric field can be non-vanishing and depends on the state of the particle even if all the propagating modes of radiation are in their vacuum state and there is no entanglement.

In other gauges entanglement can be present. However, contrarily to what usually happens, the presence of entanglement by itself does not prevent Alice to distinguish a coherent superposition from a statistical mixture.
Indeed, as we will show later, Alice can exploit an internal degree of freedom of the particle to remove this entanglement with a local operation, and then perform an experiment only on the internal degree of freedom to test the coherence. The operation consists in bringing the right wavepacket of the superposition $|R\rangle$ to the left position $|L\rangle$, while leaving the left wavepacket $|L\rangle$ untouched. Independently on the gauge, the particle will then have a definite position and also its own static Coulomb field will be definite. This proves that the protocol allows for communication, and the solution b) is wrong in the sense that if Alice could perform in a sufficiently short time the above local operation, superluminal communication would still be possible.

The reasonable solution to the paradox appears then to be the final option c).
Basically, even if the state of the particle is pure and coherent, Alice cannot instantaneously test this fact with a local experiment in her laboratory.
We notice that the hypothesis c) is weaker than hypothesis a), and the two can logically coexist. Clearly if a) is valid Alice cannot make any useful experiment because decoherence has already happened.
Therefore we conclude that the weaker and most general solution to the paradox is the fundamental limitation exposed in point c).
In the last part of this Chapter, we propose two different measurements and show that they are both consistent with this limitation.

\section{Minimum discrimination time}\label{time}

In Section \ref{experiment} we argued that relativistic causality requires a fundamental limitation: Alice's discrimination experiment must be slow. But how slow it has to be? By construction  any thought experiment of the class described before gives a lower bound on the discrimination time $T_{\mathrm{A}}$ whenever $T_{\mathrm{B}} \le R/c$. In what follows we are going to optimize over this class of experiments. We anticipate that this approach leads to the following two bounds which constitute the main results of this Chapter.

{\paragraph{(i) Minimum discrimination time for quantum superpositions of large masses}

Given a particle of mass $m$ prepared in a macroscopic quantum superposition of two states separated by a distance $d$, it is impossible to locally discriminate the coherent superposition from an incoherent mixture in a time (up to a multiplicative numerical constant) less than
\begin{equation}
T \simeq \frac{m}{m_{\mathrm{P}}}\,\frac{d}{c}\;,   \label{Tmass}
\end{equation}
where $m_{\rm P} $ is the Planck mass
\begin{equation}
m_{\rm P} = \sqrt{\frac{\hbar c}{G}}  \simeq  2.18 \times 10^{-8} {\rm \; kg}\;.
\end{equation}

\paragraph{(ii) Minimum discrimination time for quantum superpositions of large charges}

 Given a particle of charge $q$ prepared in a macroscopic quantum superposition of two states separated by a distance $d$, it is impossible to locally discriminate the coherent superposition from an incoherent mixture in a time (up to a multiplicative numerical constant) less than
\begin{equation}
T \simeq \frac{q}{q_{\mathrm{P}}}\,\frac{d}{c}\;,  \label{Tcharge}
\end{equation}
where $q_{\rm P} $ is the Planck charge
\begin{equation}
q_{\rm P} = \sqrt{4 \pi\epsilon_0 \hbar c}  \simeq  11.7\;  e \simeq 1.88 \times 10^{-18} {\rm \; C}\;.
\end{equation}}

Before giving a derivation of the previous results,  we stress that both the bounds \eqref{Tmass} and \eqref{Tcharge} are relevant only for $q \ge q_{\rm P} $ and $m \ge m_{\rm P} $. Indeed for systems below the Planck mass / charge, even if the bounds are formally correct, their meaning is trivial since any measurement of the state must at least interact with both parts of the superposition and this process requires at least a time $d/c$.

\subsection{Dynamics of Bob's test mass}

Let us  first focus on the superpositions of massive particles and give a proof of the bound \eqref{Tmass} (the proof of \eqref{Tcharge} is analogous and will be given later).
It is easy to check that, for a sufficiently narrow trap (detector = {\it off}) the test mass of Bob is insensitive to the gravitational force of Alice's particle and remains stable in its ground state (see Section \ref{strength}
for details).
On the contrary, if the trap is removed,  the test mass will experience a different force depending on the position of Alice's particle. The two corresponding Hamiltonians are:
\begin{equation}
\hat{H}_{\rm L}= \frac{\hat{P}^2}{2m_{\mathrm{B}}} - F_{\rm L} \hat{X}\;,   \qquad \hat{H}_{\rm R}=\frac{\hat{P}^2}{2m_{\mathrm{B}}} -F_{\rm R} \hat{X}\;,
\end{equation}
where $m_{\mathrm{B}}$ is the mass of Bob's particle, and $F_{\rm L}$ and $F_{\rm R}$ are the different gravitational forces associated to the ``left" and ``right" positions of  Alice's particle. Their difference
\begin{equation}
\Delta F=  F_{\rm L}-F_{\rm R} \simeq \frac{G\, m_{\mathrm{A}}\, m_{\mathrm{B}}\,  d}{ R^3}\;, \label{deltaFmass}
\end{equation}
where $m_{\mathrm{A}}$ is the mass of Alice's particle, determines the dipole force sensitivity that Bob should be able to detect in order to induce the decoherence of the reduced state possessed by Alice.

Given the initial state of the test mass $|\phi\rangle$, it is easy to check that entanglement can be generated in a time $t$ whenever the different time evolutions associated to $\hat{H}_{\rm L}$ and $\hat{H}_{\rm R}$ drive the test mass into almost orthogonal states, {\it i.e.}
\begin{equation}
\left| \langle \phi | e^{\frac{i}{\hbar} \hat{H}_{\rm R} t }  e^{-\frac{i}{\hbar} \hat{H}_{\rm L} t}    | \phi \rangle \right| \ll 1 . \label{echo}
\end{equation}
Such time depends on the initial state $|\phi \rangle$ and on the Loschmidt echo operator
\begin{equation}\label{Losch}
 \hat{L}(t)=e^{\frac{i}{\hbar} \hat{H}_{\mathrm{R}} t}  e^{-\frac{i}{\hbar} \hat{H}_{\mathrm{L}} t},
\end{equation}
which after two iterations of the Baker-Campbell-Hausdorff formula can be written as
\begin{equation}
 \hat{L}(t)=\exp\left[ i\,\frac{\Delta F}{\hbar} \left(\hat{X}\,t + \frac{\hat{P}}{2m_{\mathrm{B}}}\,t^2   +  \frac{F_1+F_2}{12 m_{\mathrm{B}}}\,  t^3 \right)\right]\;.
\end{equation}
Neglecting the complex phase factor  $e ^{\frac{i}{\hbar}\, \Delta F \, \frac{F_1+F_2}{12 m_{\mathrm{B}}}\,  t^3} $,  $\hat{L}(t)$ is essentially a
phase--space displacement operator of the form  $e ^{\frac{i}{\hbar} (\delta_x  \hat{P} - \delta_p \hat{X})}$,
where
\begin{eqnarray}
\delta_x&=& \frac{\Delta F\, t^2 }{ 2 m_{\mathrm{B}}} ,  \label{deltax} \\
\delta_p&=&-\Delta F \,t\label{deltap}
\end{eqnarray}
are the shifts in position and momentum, respectively.

Since the initial state $|\phi\rangle $ of the test mass is the ground state of a very narrow harmonic trap, it will correspond to a localized Gaussian wavepacket which is very noisy in momentum and therefore we may focus only on the position shift  \eqref{deltax} and compare it with the position uncertainty $\Delta X$ of the initial state (see Section \ref{strength}
 for a detailed proof). We can argue that entanglement is generated only after a time $t=T_{\mathrm{B}}$ such that
\begin{equation}
\frac{\delta x}{\Delta X}= \frac{\Delta F\, T_{\mathrm{B}}^2 }{ 2 m_{\mathrm{B}} \Delta X} \simeq 1. \label{ratio}
\end{equation}

Apparently Bob can generate entanglement arbitrarily quickly by reducing the position uncertainty $\Delta X$. However there is a fundamental limit to the localization precision which is set by the Planck length.
It is widely accepted that no reasonable experiment can overcome this limit~\cite{salecker1958quantum,mead1964possible,hossenfelder2013minimal}:
\begin{equation}
\Delta X \ge l_{\rm P}= \sqrt{\frac{\hbar G}{c^3}}. \label{minDeltax}
\end{equation}
From Eq.\ \eqref{ratio}, substituting Eq. \eqref{deltaFmass} and using the minimum $\Delta X$ allowed by the constraint \eqref{minDeltax},  we get
\begin{equation}
\frac{\delta x}{\Delta X}=\frac{1}{2}\,\frac{m_{\mathrm{A}}}{m_{\mathrm{P}}}\,\frac{d\   c^{2}   T_{\mathrm{B}}^2}{R^3 }   \simeq 1. \label{ratio2}
\end{equation}
As we have explained in Section \ref{experiment}, relativistic causality implies the inequality \eqref{causality} involving Alice's measurement time $T_{\mathrm{A}}$ and the entanglement time $T_{\mathrm{B}}$. Such inequality provides a lower bound on $T_{\mathrm{A}}$ only if $T_{\mathrm{B}} < R/c$ while it gives no relevant information for $T_{\mathrm{B}} \ge R/c$. Therefore we parametrize $R$ in terms of $T_{\mathrm{B}}$ and a dimensionless parameter $\eta$:
\begin{equation}
 \eta= \frac{c\,T_{\mathrm{B}}}{R}\;, \qquad 0 \le \eta \le 1.
\end{equation}
Using this parametrization, from Eq.\ \eqref{ratio2}, we get

\begin{equation}
T_{\mathrm{B}} \simeq \frac{1}{2}\,\eta^3\,\frac{m_{\mathrm{A}}}{m_{\mathrm{P}}}\,  \frac{d}{ c }  .
\end{equation}
From the causality inequality \eqref{causality} we have
\begin{equation}
T_{\mathrm{A}} + T_{\mathrm{B}} \ge \frac{R}{c} \; \Longrightarrow \; T_{\mathrm{A}} \ge \frac{T_{\mathrm{B}}}{\eta} - T_{\mathrm{B}}= \frac{1}{2}\,\frac{m_{\mathrm{A}}}{m_{\mathrm{P}}}\,\frac{d}{ c }\,(\eta^2- \eta^3).
\end{equation}
Optimizing over $\eta$ we get

\begin{equation}
T_{\mathrm{A}} \ge  \frac{2}{27}\,\frac{m_{\mathrm{A}}}{m_{\mathrm{P}}}\,\frac{d}{c}\;.
\end{equation}
This is, up to a multiplicative numerical constant, the bound given in Eq. \eqref{Tmass}.

\subsection{Dynamics of Bob's test charge}

The calculation in the case in which we have a test charge instead of a test mass is almost identical. The only difference is that Eq. \eqref{deltaFmass} is replaced by the Coulomb counterpart
\begin{equation}
\Delta F=  F_{\rm L}-F_{\rm R} \simeq \frac{q_{\mathrm{A}}\,q_{\mathrm{B}}\,  d}{ 4 \pi \epsilon_0 R^3}\;, \label{deltaFcharge}
\end{equation}
where $q_{\mathrm{A}}$ and $q_{\mathrm{B}}$ are the charges of Alice's and Bob's particles, respectively, while the localization limit \eqref{minDeltax} is replaced by Bob's particle charge radius~\cite{weinberg1995quantum}
\begin{equation}
\Delta X \ge \frac{q_{\mathrm{B}}}{q_{\mathrm{P}}} \frac{\hbar}{m_{\mathrm{B}} c}. \label{minDeltaxCharge}
\end{equation}
More details on the minimum localization of a macroscopic charge are given in Section \ref{loc}.
From Eq.s \eqref{deltaFcharge} and \eqref{minDeltaxCharge}, repeating exactly the previous argument one finds
\begin{equation}\label{TA}
T_{\mathrm{A}} \ge  \frac{2}{27}\,\frac{q_{\mathrm{A}}}{q_{\mathrm{P}}} \frac{d}{c}\;,
\end{equation}
which is, up to a multiplicative numerical constant, the bound given in Eq. \eqref{Tmass}.

\section{Minimum time from entanglement with radiation}\label{ent}

In Section \ref{time} we have proved that relativistic causality requires that any measurement Alice can perform to test the coherence of her superposition must require a minimum time, depending on the mass or charge of her particle.
Here we focus on the electromagnetic case, and propose two different measurements to check the coherence of the superposition.

The first is a simplified version of the experiment proposed in Ref.~\cite{scala2013matter,wan2015tolerance}.
Let Alice's particle have spin $\frac{1}{2}$, and let us suppose that her superposition is entangled with the spin, i.e.~\eqref{LR} is replaced by
\begin{equation}
|\psi\rangle=\frac{|L\rangle|\uparrow\rangle+|R\rangle|\downarrow\rangle}{\sqrt{2}}\;.
\end{equation}
Let now Alice apply a spin-dependent force, that vanishes if the spin is up, while brings the particle from $|R\rangle$ to $|L\rangle$ if the spin is down.
We notice that, at the end of this operation, the charge has a well-defined position, and also its own static Coulomb electric field is well-defined. In this way, the
original macroscopic superposition has been reduced to a microscopic spin superposition,  on which an instantaneous discrimination experiment can be performed.

If Bob does not perform the measurement, the final state of Alice's particle is
\begin{equation}
\rho_A=|L\rangle\langle L|\otimes|+\rangle\langle +|\;,
\end{equation}
where
\begin{equation}
|+\rangle=\frac{|\uparrow\rangle+|\downarrow\rangle}{\sqrt{2}}\;.
\end{equation}
On the contrary, if Bob induces a collapse of the wavefunction, the final state is
\begin{equation}
\rho_A'=|L\rangle\langle L|\otimes\frac{|\uparrow\rangle\langle\uparrow|+|\downarrow\rangle\langle\downarrow|}{2}\;,
\end{equation}
and Alice can test the coherence measuring the spin.

However, this protocol requires the particle to be accelerated if it has spin down and needs to be moved from $|R\rangle$ to $|L\rangle$.
Then it will radiate, and it can entangle with the emitted photons.
A semiclassical computation of the emitted radiation can be found in Section \ref{rad}.
There we show that such radiation is indistinguishable from the vacuum state of the field only if the motion lasts for at least the time required by our previous bound~\eqref{Tcharge}.

\section{Minimum time from quantum vacuum fluctuations}\label{QED}

In Section \ref{ent} we have provided an example of experiment able to test the coherence.
The protocol requires to accelerate the charge, and if its duration is too short, the charge radiates and entangles with the emitted photons.
The reader could now think that the bound on the time could be beaten with an experiment that does not involve accelerations.
An example of such experiment could seem to be a measurement of the canonical momentum of Alice's particle.
In this Section, we first show that this measurement is indeed able to test the coherence of the superposition
and then we estimate the minimum time necessary to perform it.

The canonical momentum of a charged particle coupled to the electromagnetic field is not gauge invariant, and therefore cannot be directly measured.
Alice can instead measure directly the velocity of her particle, that is gauge invariant.
However, its relation with the canonical momentum now contains the vector potential.
Even if there is no external electromagnetic field, the latter is a quantum-mechanical entity, and is subject to quantum vacuum fluctuations.
Then, the fluctuations of the vector potential enter in the relation between velocity and momentum.
If Alice is not able to measure the field outside her laboratory, she can measure only the velocity of her particle (see Section \ref{dof}
for a detailed discussion), and can reconstruct its canonical momentum only if the fluctuations are small.
We show that in an instantaneous measurement these fluctuations are actually infinite.
However, if Alice measures the average of the velocity over a time $T$, they decrease as $1/T^2$, and can be neglected if $T$ is large enough.
This minimum time is found consistent with the general bound given in Eq.\ \eqref{Tcharge}.

In order to simplify our formulas, in this Section and in the related Sections in the Appendix \ref{appsup} we put as in~\cite{weinberg1995quantum}
\begin{equation}
\hbar=c=\epsilon_0=\mu_0=1\;,\qquad q_{\mathrm{P}}^2=4\pi\;.
\end{equation}
These constants will be put back into the final result.

\subsection{The canonical momentum as a test for coherence}

Let us first show that Alice can test the coherence with a measurement of the canonical momentum of her particle.

Let the particle be in the coherent superposition~\eqref{LR} of two identical wavepackets centered in different points, with wavefunction
\begin{equation}\label{psiS}
\psi(\mathbf{x})=\frac{\phi(\mathbf{x})+e^{i\varphi}\;\phi(\mathbf{x}-\mathbf{d})}{\sqrt{2}}\;,
\end{equation}
where $\varphi$ is an arbitrary phase.

The probability distribution of the canonical momentum $\hat{P}$ is the modulus square of the Fourier transform of the wavefunction:
\begin{equation}\label{osc}
\frac{\left|\psi(\mathbf{k})\right|^2}{(2\pi)^3}=2\cos^2\left(\frac{\mathbf{k}\cdot\mathbf{d}-\varphi}{2}\right)\;\;\frac{\left|\phi(\mathbf{k})\right|^2}{(2\pi)^3}\;,
%p\left(d^3k\right)=\left|\psi(\mathbf{k})\right|^2\;\frac{d^3k}{(2\pi)^3}=2\cos^2\left(\frac{\mathbf{k}\cdot\mathbf{d}-\varphi}{2}\right)\;\left|\phi(\mathbf{k})\right|^2\;\frac{d^3k}{(2\pi)^3}\;,
\end{equation}
and she can test the coherence of the superposition from the interference pattern in momentum space generated by the cosine.
Indeed, an incoherent statistical mixture would be associated to the probability distribution $\left.\left|\phi(\mathbf{k})\right|^2\right/(2\pi)^3\,$, where the cosine squared is replaced by $1/2$, its average over the phase $\varphi$.

We notice from~\eqref{osc} that, in order to be actually able to test the coherence, Alice must measure the canonical momentum with a precision of at least
\begin{equation}\label{prec}
\Delta P\lesssim\frac{\pi}{d}\;,
\end{equation}
where $d=|\mathbf{d}|$.
This precision increases with the separation of the wavepackets, e.g. for $d=1\,\mathrm{m}$, it is $\Delta P\lesssim 10^{-34}\,\mathrm{kg}\cdot \mathrm{m}/\mathrm{s}$.

\subsection{Quantum vacuum fluctuations and minimum time}

Let now Alice's particle carry an electric charge $q$.
We want to take into account the quantum vacuum fluctuations of the electromagnetic field, so quantum electrodynamics is required.
The global Hilbert space is then the tensor product of the Hilbert space of the particle $\mathcal{H}_{\mathrm{A}}$ with the Hilbert space of the field $\mathcal{H}_F$.
The reader can find in Section \ref{emf}
the details of the quantization.

The position and canonical momentum operators of Alice's particle $\hat{\mathbf{X}}$ and $\hat{\mathbf{P}}$ still act in the usual way on the particle Hilbert space alone, so that the argument of the previous Subsection remains unchanged.
The full interacting Hamiltonian of the particle and the electromagnetic field is
\begin{equation}\label{H}
\hat{H}=\frac{1}{2m}\left(\hat{\mathbf{P}}-q\;\hat{\mathbf{A}}\left(\hat{\mathbf{X}}\right)\right)^2+\hat{H}_F\;,
\end{equation}
where
\begin{equation}\label{AX}
\hat{A}^i\left(\hat{\mathbf{X}}\right)=\int\frac{\hat{a}^i(\mathbf{k})\;e^{i\mathbf{k}\cdot\hat{\mathbf{X}}}+\hat{a}^{i\dag}(\mathbf{k})\;e^{-i\mathbf{k}\cdot\hat{\mathbf{X}}}}{\sqrt{2|\mathbf{k}|}}\;\frac{d^3k}{(2\pi)^3}
\end{equation}
is the vector-potential operator $\hat{\mathbf{A}}(\mathbf{x})$
(see Eq. \eqref{Afrx})
with the coordinate $\mathbf{x}$ replaced with the position operator $\hat{\mathbf{X}}$, and $\hat{H}_F$ is the free Hamiltonian of the electromagnetic field defined in \eqref{HF}.

Due to the minimal-coupling substitution, the operator associated to the velocity of the particle is
\begin{equation}\label{Vi}
\hat{\mathbf{V}}\equiv i\left[\hat{H},\;\hat{\mathbf{X}}\right]=\frac{1}{m}\left(\hat{\mathbf{P}}-q\;\hat{\mathbf{A}}\left(\hat{\mathbf{X}}\right)\right)\;,
\end{equation}
that contains the operator vector-potential, and acts also on the Hilbert space of the field.
The canonical momentum can be reconstructed from the velocity with
\begin{equation}\label{PV}
\hat{\mathbf{P}}=m\;\hat{\mathbf{V}}+q\;\hat{\mathbf{A}}\left(\hat{\mathbf{X}}\right)
\end{equation}
if the second term in the right-hand side can be neglected.
With the help of the commutation relations
(see \eqref{CCRs}),
a direct computation of the variance of $\hat{\mathbf{A}}\left(\hat{\mathbf{X}}\right)$ on the vacuum state of the field gives
\begin{equation}\label{A2}
\langle0|{\hat{\mathbf{A}}\left(\hat{\mathbf{X}}\right)}^2|0\rangle=\left(\int\frac{1}{|\mathbf{k}|}\,\frac{d^3k}{(2\pi)^3}\right)\hat{\mathbb{I}}_{\mathrm{A}}\;,
\end{equation}
that has a quadratic divergence for $\mathbf{k}\to\infty$ due to the quantum vacuum fluctuations.
This divergence can be cured averaging the vector potential over time with a smooth function $\varphi(t)$.
We must then move to the Heisenberg picture, where operators explicitly depend on time, and we define it to coincide with the Schr\"odinger picture at $t=0$, the time at which Alice measures the velocity.
Since the divergence in~\eqref{A2} does not depend neither on the mass nor on the charge of Alice's particle and is proportional to the identity operator on the particle Hilbert space $\hat{\mathbb{I}}_{\mathrm{A}}$, it has nothing to do with the interaction of the particle with the field.
Then the leading contribution to the result can be computed evolving the field with the free Hamiltonian $\hat{H}_F$ only, i.e. with
\begin{equation}\label{AXt}
\hat{A}^i\left(\hat{\mathbf{X}},t\right)=\int\frac{\hat{a}^i(\mathbf{k})\;e^{i\left(\mathbf{k}\cdot\hat{\mathbf{X}}-|\mathbf{k}|t\right)}+\hat{a}^{i\dag}(\mathbf{k})\;e^{i\left(|\mathbf{k}|t-\mathbf{k}\cdot\hat{\mathbf{X}}\right)}}{\sqrt{2|\mathbf{k}|}}\;\frac{d^3k}{(2\pi)^3}\;.
\end{equation}
Defining the time-averaged vector potential as
\begin{equation}
\hat{\mathbf{A}}_{av}=\int\hat{\mathbf{A}}\left(\hat{\mathbf{X}},t\right)\;\varphi(t)\;dt\;,
\end{equation}
its variance over the vacuum state of the field is now
\begin{equation}
\langle0|{\hat{\mathbf{A}}_{av}}^2|0\rangle=\left(\frac{1}{2\pi^2}\int_0^\infty \left|\tilde{\varphi}(\omega)\right|^2\omega\;d\omega\right)\hat{\mathbb{I}}_{\mathrm{A}}\;,
\end{equation}
where
\begin{equation}
\tilde{\varphi}(\omega)=\int\varphi(t)\;e^{i\omega t}\;dt
\end{equation}
is the Fourier transform of $\varphi(t)$.
Taking as $\varphi(t)$ a normalized Gaussian function of width $T$ centered at $t=0$:
\begin{equation}
\varphi(t)=\frac{e^{-\frac{t^2}{2T^2}}}{\sqrt{2\pi}\;T}\;,
\end{equation}
we get as promised a finite result proportional to $1/T^2$:
\begin{equation}\label{DA2}
\langle0|{\hat{\mathbf{A}}_{av}}^2|0\rangle=\frac{\hat{\mathbb{I}}_{\mathrm{A}}}{4\pi^2 T^2}\;.
\end{equation}
Then, if Alice estimates one component of the canonical momentum (say the one along the $x$ axis) with the time average of the velocity taken with the function $\varphi(t)$, she commits an error of the order of
\begin{equation}\label{DeltaPf}
\Delta P\simeq \frac{q}{2\pi\sqrt{3}\;T}\;.
\end{equation}
Comparing~\eqref{DeltaPf} with the required precision to test the coherence~\eqref{prec}, the minimum time required is
\begin{equation}
T\gtrsim\frac{1}{\sqrt{3\pi^3}}\,\frac{q}{q_{\mathrm{P}}}\,\frac{d}{c}\simeq 0.10\,\frac{q}{q_{\mathrm{P}}}\,\frac{d}{c}\;,
\end{equation}
in agreement with the bound~\eqref{TA} imposed by relativistic causality alone.

\section{Conclusion}\label{csup}

In this Chapter we have studied the limitations that the gravitational and electric fields produced by a macroscopic particle impose  on quantum superposition experiments.  We have found that, in order to avoid a contradiction between quantum mechanics and relativistic causality, a minimum time is necessary in order to discriminate a coherent superposition from an incoherent statistical mixture. This discrimination time is proportional to the separation distance of the superposition and to the mass (or charge) of the particle.

In the same way as the Heisenberg uncertainty principle inspired the development of a complete theory of quantum mechanics,  our fundamental and quantitative bounds on the discrimination time can be useful for the development of current and future theories of quantum gravity.
Moreover, despite an experimental observation of our results clashes with the difficulty of preparing superpositions of masses above the Planck scale,  the current technological progress on highly massive quantum optomechanical and electromechanical systems provides a promising context \cite{romero2011large, marshall2003towards,arndt2014testing, pikovski2012probing, bawaj2014probing,schnabel2015einstein,scala2013matter,wan2015tolerance} for testing our predictions.

\chapter{Conclusions}\label{concl}
The main theme of this Thesis has been the transposition to Gaussian quantum information of the classical principle ``Gaussian channels have Gaussian optimizers''.
We have focused on the constrained minimum output entropy conjecture (Proposition \ref{CMOE}), stating that Gaussian thermal input states minimize the output von Neumann entropy of any gauge-covariant bosonic Gaussian channel.
This conjecture is necessary to determine the capacity region of the degraded quantum Gaussian broadcast channel \cite{yard2011quantum,savov2015classical,guha2007classicalproc,guha2007classical}, where a sender wants to communicate classical information to two receivers, and the triple trade-off region of the Gaussian quantum attenuator \cite{wilde2012quantum,wilde2012public,wilde2012information}.

In Chapter \ref{epi} we have proved the quantum Entropy Power Inequality, that provides an extremely tight lower bound to this minimum output entropy, resulting in almost optimal bound for the capacity region of the Gaussian broadcast channel.

In Chapter \ref{majorization} we have tackled the exact solution with a generalization of the Gaussian majorization conjecture \cite{giovannetti2015majorization,mari2014quantum} exploiting the notion of passivity.
A passive state is diagonal in the Hamiltonian eigenbasis and its eigenvalues decrease as the energy increases.
We have proved that for any one-mode gauge-covariant quantum Gaussian channel, the output generated by a passive state majorizes (i.e. it is less noisy than) the output generated by any other state with the same spectrum, and in particular it has a lower entropy.
Then, the solution to the constrained minimum output entropy problem has certainly to be found among passive states.
We have exploited this result in Chapter \ref{chepni}.
Here we have proved that Gaussian thermal input states minimize the output entropy of the one-mode Gaussian quantum attenuator for fixed input entropy, i.e. conjecture \ref{CMOE} for this channel.
The proof is based on the isoperimetric inequality \eqref{epnilogs}, whose multimode generalization implies conjecture \ref{CMOE} for the multimode attenuator.

The same ideas can be useful in any other entropic optimization problem involving quantum Gaussian channels.
The quantum capacity \cite{wilde2013quantum,holevo2013quantum} of a channel is the maximum number of qubits that can be faithfully sent per channel use.
The private capacity \cite{wilde2013quantum,holevo2013quantum} is the maximum number of bits per channel use that can be sent and certified not to have been read by any eavesdropper.
So far, both the quantum and the private capacity of quantum Gaussian channels are known only in the degradable case \cite{holevo2013quantum}.
Our ideas can be useful to determine them in the general case.

In Chapter \ref{chlossy} we have extended the majorization results of Chapter \ref{majorization} to a large class of lossy quantum channels, resulting from a weak interaction of a small quantum system with a large bath in its ground state.

In Chapter \ref{memory} we have considered a particular model of quantum Gaussian channel that implements memory effects, and we have explicitly determined its classical information capacity.
In Chapter \ref{normal} we have explored the set of linear trace-preserving not necessarily positive maps preserving the set of Gaussian states.
For one mode, we have proved that any non positive map of this kind is built from the so-called phase-space dilatation.
These maps can be used as tests for certifying that a given quantum state does not belong to the convex hull of Gaussian states, in the same way as positive but not completely positive maps are used as tests for entanglement.
Phase-space dilatations are then proven to be the only relevant test of this kind.

In Chapter \ref{chETH} we have proved that requiring thermalization of a quantum system in contact with a heat bath for any initial uncorrelated state with a well-defined temperature implies the Eigenstate Thermalization Hypothesis for the system-bath Hamiltonian.
Then, the ETH constitutes the unique criterion to decide whether a given system-bath dynamics always leads to thermalization.

Finally, in Chapter \ref{chsuperpos} we have proved that any measurement able to distinguish a coherent superposition of two wavepackets of a massive or charged particle from the corresponding incoherent statistical mixture must require a minimum time.
In the case of an electric charge, the bound can be ascribed to the entanglement with the quantum electromagnetic radiation that is unavoidably emitted during the measurement.
Then, in the case of a mass the bound provides an indirect evidence for the existence of quantum gravitational radiation and in general for the necessity of quantizing gravity.

\appendix

\chapter{Gaussian quantum systems}\label{appG}
In this Appendix we provide some technical details on Gaussian quantum information.
In particular, we introduce Gaussian quantum systems in Section \ref{GQSa}, and the method of characteristic functions in Section \ref{chia}.
Section \ref{secdisp} defines the displacement operators, while Section \ref{secmom} defines the first and second moment of a quantum state.
Section \ref{Gsta} introduces quantum Gaussian states, and Section \ref{sec:Husimi} introduces the method of the Husimi function.
Gaussian quantum channels are defined in Section \ref{QGCa}.
Finally, Section \ref{EPnIG} proves the Entropy Photon-number Inequality \eqref{EPnI} for Gaussian input states.

\section{Quadratures and Hilbert space}\label{GQSa}
We consider an $n$-mode bosonic quantum system with Hilbert space $\mathcal{H}$ and quadrature operators $\hat{Q}^i$ and $\hat{P}^i$, $i=1,\ldots,\,n$, satisfying the canonical commutation relations
\begin{equation}\label{CCRQP}
\left[\hat{Q}^i,\;\hat{P}^j\right]=i\,\delta^{ij}\;.
\end{equation}
As usual, we can define the ladder operators
\begin{equation}
\hat{a}^i=\frac{\hat{Q}^i+i\hat{P}^i}{\sqrt{2}}\;,
\end{equation}
satisfying the commutation relations
\begin{equation}
\left[\hat{a}^i,\;\hat{a}^{j\dag}\right]=\delta^{ij}\;.
\end{equation}
The number operator is defined as
\begin{equation}
\hat{N}=\sum_{i=1}^n\hat{a}^{i\dag}\hat{a}^i\;,
\end{equation}
and it counts the number of excitations.

We can put all the quadratures together in the column vector
\begin{equation}\label{Rdef}
\hat{\mathbf{R}}=\left(
                   \begin{array}{c}
                     \hat{R}^1 \\
                     \vdots \\
                     \hat{R}^{2n} \\
                   \end{array}
                 \right):=\left(
                   \begin{array}{c}
                     \hat{Q}^1 \\
                     \hat{P}^1\\
                     \vdots \\
                     \hat{Q}^{n} \\
                     \hat{P}^n
                   \end{array}
                 \right)\;.
\end{equation}
The commutation relations \eqref{CCRQP} become
\begin{equation}\label{CCR}
\left[\hat{R}^i,\;\hat{R}^j\right]=i\,\Delta^{ij}\;,
\end{equation}
where $\Delta$ is the symplectic form associated with the antisymmetric matrix
\begin{equation}\label{Deltac}
\Delta=\bigoplus_{k=1}^n\left(
                          \begin{array}{cc}
                            0 & 1 \\
                            -1 & 0 \\
                          \end{array}
                        \right)\;.
\end{equation}
An even-dimensional vector space equipped with an nondegenerate antisymmetric bilinear form is called symplectic space, and the bilinear form is called its symplectic form.
It is possible to show \cite{holevo2013quantum} that we can always choose a basis such that the matrix associated with the symplectic form has the form \eqref{Deltac}.

A matrix $S$ preserving the symplectic form $\Delta$, i.e. such that
\begin{equation}\label{symplectic}
S\;\Delta\;S^T=\Delta\;,
\end{equation}
is called symplectic matrix.

The symplectic condition \eqref{symplectic} simplifies in the case of one
mode. Indeed for any $2\times 2$ matrix $M$ we have
\begin{equation}
M\Delta M^{T}=\Delta \det M\;,  \label{det1mode}
\end{equation}
therefore a $2\times 2$ matrix $S$ is symplectic iff
\begin{equation}
\det S=1\;.
\end{equation}

\section{Characteristic and Wigner functions}\label{chia}
Let $\mathfrak{H}$ be the set of the Hilbert-Schmidt operators on $\mathcal{H}$, i.e. the operators with finite Hilbert-Schmidt norm:
\begin{equation}\label{HSN}
\left\|\hat{X}\right\|_2^2=\mathrm{Tr}\left(\hat{X}^\dag\,\hat{X}\right)<\infty\;.
\end{equation}
Given an operator $\hat{X}\in \mathfrak{H}$ we define its characteristic function as
\begin{equation}\label{chidef}
\chi_{\hat{X}}(\mathbf{k}):=\mathrm{Tr}\left(\hat{X}\,e^{i\,\mathbf{k}\,\hat{\mathbf{R}}}\right)\;,\qquad\mathbf{k}\in\mathbb{R}^{2n}\;,
\end{equation}
where
\begin{equation}
\mathbf{k}=\left(k_1,\;\ldots,\;k_{2n}\right)
\end{equation}
is a row vector.
It is possible to prove \cite{holevo2013quantum,weedbrook2012gaussian} that $\hat{X}$ can be reconstructed with
\begin{equation}
\hat{X}=\int\chi_{\hat{X}}(\mathbf{k})\;e^{-i\,\mathbf{k}\,\hat{\mathbf{R}}}\;\frac{d^{2n}k}{(2\pi)^n}\;,\label{rhochi}
\end{equation}
and that the characteristic function provides an isometry between $\mathfrak{H}$ and $L^2\left(\mathbb{R}^{2n}\right)$ (see e.g. Theorem 5.3.3 of \cite{holevo2011probabilistic}):
\begin{equation}
\mathrm{Tr}\left(\hat{X}^\dag\,\hat{Y}\right)=\int{\chi_{\hat{X}}(\mathbf{k})}^*\;\chi_{\hat{Y}}(\mathbf{k})\;\frac{d^{2n}k}{(2\pi)^n}\;.\label{hilbert}
\end{equation}
Eq. \eqref{hilbert} is called noncommutative Parceval's formula.

We define the Wigner function of $\hat{X}\in \mathfrak{H}$ as the Fourier transform of the characteristic function:
\begin{equation}\label{wignerdef}
W_{\hat{X}}(\mathbf{x}):=\int\chi_{\hat{X}}(\mathbf{k})\;e^{-i\,\mathbf{k}\,\mathbf{x}}\;\frac{d^{2n}k}{(2\pi)^{2n}}\;,
\end{equation}
where
\begin{equation}
\mathbf{x}=\left(
             \begin{array}{c}
               x^1 \\
               \vdots \\
               x^{2n} \\
             \end{array}
           \right)\in\mathbb{R}^{2n}
\end{equation}
is a column vector.

\section{Displacement operators}\label{secdisp}
We define the displacement operators with
\begin{equation}
\hat{D}(\mathbf{x}):=e^{i\;\mathbf{x}^T\;\Delta^{-1}\;\hat{\mathbf{R}}}\;,
\end{equation}
acting on the quadratures as
\begin{equation}\label{displ}
{\hat{D}(\mathbf{x})}^\dag\;\hat{\mathbf{R}}\;\hat{D}(\mathbf{x})=\hat{\mathbf{R}}+\mathbf{x}\;.
\end{equation}
Using \eqref{displ}, the displacements act on the characteristic function as
\begin{equation}\label{chiDD}
\chi_{\hat{D}(\mathbf{x})\;\hat{X}\;{\hat{D}(\mathbf{x})}^\dag}(\mathbf{k})=e^{i\;\mathbf{k}\;\mathbf{x}}\;\chi_{\hat{X}}(\mathbf{k})\;.
\end{equation}

\section{Moments and symplectic eigenvalues}\label{secmom}
Given a state $\hat{\rho}\in\mathfrak{S}(\mathcal{H})$ and an observable $\hat{A}$, we define its expectation value
\begin{equation}
\left\langle\hat{A}\right\rangle:=\mathrm{Tr}\left(\hat{\rho}\,\hat{A}\right)\;.
\end{equation}
We can now define the first moments of $\hat{\rho}$ as the expectation values of the quadratures
\begin{equation}
\mathbf{r}:=\left\langle\hat{\mathbf{R}}\right\rangle\;,
\end{equation}
and the covariance matrix as
\begin{equation}\label{covdef}
\sigma^{ij}:=\left\langle\left\{\hat{R}^i-r^i,\;\hat{R}^j-r^j\right\}\right\rangle\;,
\end{equation}
where $\left\{\cdot,\cdot\right\}$ stands for the anti-commutator.

It is possible to prove \cite{holevo2013quantum} that the eigenvalues of the matrix $\sigma\Delta^{-1}$ are pure imaginary, and since the matrix is real, they come in couples of complex conjugates.
Their absolute values $\nu_k,\;k=1,\ldots,n$ are called the symplectic eigenvalues of $\sigma$, and satisfy
\begin{equation}\label{detnu}
\det\sigma=\prod_{k=1}^n\nu_k^2\;.
\end{equation}
According to the Williamson theorem \cite{williamson1936algebraic}, for any strictly positive $\sigma$ there exists a symplectic matrix $S$ such that
\begin{equation}
S\;\sigma\;S^T=\bigoplus_{k=1}^n\nu_k\;\mathbb{I}_2\;,
\end{equation}
where the $\nu_k$ are its symplectic eigenvalues.

The positivity of $\hat{\rho}$ together with the commutation relations \eqref{CCR} imply for $\sigma$ the Robertson-Heisenberg uncertainty relation
\begin{equation}
\sigma\geq\pm i\Delta\;.\label{heis}
\end{equation}
Condition \eqref{heis} is equivalent to imposing all the symplectic eigenvalues of $\sigma$ to be greater or equal than one.

A symmetric positive definite $2\times 2$-matrix $\sigma
$ has a single symplectic eigenvalue, given by
\begin{equation}
\nu ^{2}=\det \sigma \;;
\end{equation}
and it is the covariance matrix of a quantum
state iff
\begin{equation}
\det \sigma \geq 1\;.  \label{state1}
\end{equation}

The first moment $\mathbf{r}$ and the covariance matrix $\sigma$ can both be computed from the characteristic function with
\begin{eqnarray}
\mathbf{r}&=&-i\left.\frac{\partial}{\partial\mathbf{k}}\ln\chi(\mathbf{k})\right|_{\mathbf{k}=\mathbf{0}}\\
\sigma&=&-\left.\frac{\partial}{\partial \mathbf{k}}\;\left(\frac{\partial}{\partial\mathbf{k}}\right)^T\ln\chi(\mathbf{k})\right|_{\mathbf{k}=\mathbf{0}}\;,
\end{eqnarray}
where $\frac{\partial}{\partial\mathbf{k}}$ is the column vector of the derivatives
\begin{equation}
\frac{\partial}{\partial\mathbf{k}}:=\left(
                                       \begin{array}{c}
                                         \frac{\partial}{\partial k_1} \\
                                         \vdots \\
                                         \frac{\partial}{\partial k_{2n}} \\
                                       \end{array}
                                     \right)\;.
\end{equation}

\section{Gaussian states}\label{Gsta}
Given a symmetric matrix $\sigma\geq\pm i\Delta$ and a column vector $\mathbf{r}\in\mathbb{R}^{2n}$, the Gaussian state $\hat{\rho}_G(\sigma,\,\mathbf{r})$ with covariance matrix $\sigma$ and first moment $\mathbf{r}$ is the state
\begin{equation}
\hat{\rho}_G=e^{-\left(\hat{\mathbf{R}}-\mathbf{r}\right)^T J\left(\hat{\mathbf{R}}-\mathbf{r}\right)}\left/\mathrm{Tr}\;e^{-\left(\hat{\mathbf{R}}-\mathbf{r}\right)^T J\left(\hat{\mathbf{R}}-\mathbf{r}\right)}\right.\;,
\end{equation}
where $J$ is the positive real matrix such that
\begin{equation}
\sigma=\Delta\cot\left(J\;\Delta\right)\;.
\end{equation}
It has characteristic function
\begin{equation}\label{chigauss}
\chi(\mathbf{k})=e^{-\frac{1}{4}\mathbf{k}\,\sigma\,\mathbf{k}^T+i\,\mathbf{k}\,\mathbf{r}}
\end{equation}
and Wigner function
\begin{equation}
W(\mathbf{x})=\frac{e^{-\left(\mathbf{x}-\mathbf{r}\right)^T\sigma^{-1}\left(\mathbf{x}-\mathbf{r}\right)}}{\sqrt{\det\left(\pi\,\sigma\right)}}\;.
\end{equation}
For $\sigma = \mathbb{I}_{2n}$ we obtain the family of coherent states $\hat{\rho}_G(\mathbf{r},\,\mathbb{I}_{2n}),\;\mathbf{r}\in \mathbb{R}^{2n}$.

For simplicity, we call $\hat{\rho}_G(\sigma)$ the centered state $\hat{\rho}_G(\sigma,\,\mathbf{0})$.

\subsection{Entropy of Gaussian states}\label{entrGa}
The von Neumann entropy $S\left[\hat{\rho}\right]=-\mathrm{Tr}\left[\hat{\rho}\ln\hat{\rho}\right]$ of the Gaussian state $\hat{\rho}_G(\sigma)$ with covariance matrix $\sigma$ is given by
\begin{equation}
S\left(\hat{\rho}_G(\sigma)\right)=\sum_{k=1}^n h\left(\nu_k\right)\;,\label{sgaus}
\end{equation}
where the $\nu_k$ are the symplectic eigenvalues of $\sigma$ and
\begin{equation}\label{defh}
h(\nu) = \frac{\nu+1}{2}\ln\frac{\nu+1}{2}-\frac{\nu-1}{2}\ln\frac{\nu-1}{2}\;.
\end{equation}
For thermal states with covariance matrix proportional to the identity, i.e. $\sigma=\nu \mathbb{I}_{2n}$, it can be useful to express the entropy in terms of the mean photon number per mode
\begin{equation}
N=\frac{1}{n}\mathrm{Tr}\left[\hat{N}\;\hat{\rho}\right]=\frac{\nu-1}{2}\;.
\end{equation}
For this purpose, it is sufficient to replace the function $h(\nu)$ in \eqref{sgaus} with
\begin{equation}\label{defg}
g(N):=h\left(2N+1\right)=(N+1)\ln(N+1)-N\ln N\;.
\end{equation}

In the proof of the quantum Entropy Power Inequality in Chapter \ref{epi} we have used the asymptotic scaling of the entropy for Gaussian states with large covariance in Eqs. \eqref{lowerboundS} and \eqref{upper}.
Here we prove this result.
For $\nu\to\infty$, the function $h$ is almost a logarithm:
\begin{equation}
h(\nu)=\ln\frac{\nu}{2}+1+\mathcal{O}\left(\frac{1}{\nu^2}\right)\;,
\end{equation}
so the entropy of $\hat{\rho}_G(t\sigma)$ for $t\to\infty$ is
\begin{equation}\label{entas}
S\left(\hat{\rho}_G(t\sigma)\right)=\ln\prod_{k=1}^n\frac{t\,\nu_k}{2}\;+n+\mathcal{O}\left(\frac{1}{t^2}\right)= \frac{1}{2}\ln\det\left(\frac{e\, t\,\sigma}{2}\right)\;+\mathcal{O}\left(\frac{1}{t^2}\right)\;,
\end{equation}
where we have used \eqref{detnu}.

\section{Husimi function}\label{sec:Husimi}
Choose a covariance matrix $\gamma\geq\pm i\Delta$.
The generalized Husimi function $Q_{\hat{\rho}}(\mathbf{x})$ of a state $\hat{\rho}$ \cite{giovannetti2015majorization} is its overlap with the Gaussian state $\hat{\rho}_G(\gamma,\,\mathbf{x})$:
\begin{equation}
Q_{\hat{\rho}}(\mathbf{x}):=\frac{1}{(2\pi)^n}\mathrm{Tr}\left(\hat{\rho}\;\hat{\rho}_G(\gamma,\,\mathbf{x})\right)\;.\label{genHusimi}
\end{equation}
We notice that for $\gamma=\mathbb{I}_{2n}$ we recover the usual Husimi function of \cite{barnett2002methods} based on coherent states.
By construction, $Q_{\hat{\rho}}(\mathbf{x})$ is continuous in $\mathbf{x}$ and positive:
\begin{equation}
Q_{\hat{\rho}}(\mathbf{x})\geq0\;.
\end{equation}
In terms of the characteristic function of $\hat{\rho}$, \eqref{genHusimi} reads
\begin{equation}
Q_{\hat{\rho}}(\mathbf{x})=\int e^{-\frac{1}{4}\mathbf{k}\,\gamma\,\mathbf{k}^T-i\,\mathbf{k}\,\mathbf{x}}\;\chi_{\hat{\rho}}(\mathbf{k})\;\frac{d^{2n}k}{(2\pi)^{2n}}\;,
\end{equation}
where we have used \eqref{chigauss} and \eqref{hilbert}.
The Fourier transform of $Q_{\hat{\rho}}(\mathbf{x})$
\begin{equation}
\widetilde{Q}_{\hat{\rho}}(\mathbf{k})=\int Q_{\hat{\rho}}(\mathbf{x})\;e^{i\,\mathbf{k}\,\mathbf{x}}\;d^{2n}x
\end{equation}
is then given by
\begin{equation}
\widetilde{Q}_{\hat{\rho}}(\mathbf{k})=e^{-\frac{1}{4}\mathbf{k}\,\gamma\,\mathbf{k}^T}\;\chi_{\hat{\rho}}(\mathbf{k})\label{Husimik}\;.
\end{equation}
Computing \eqref{Husimik} in $\mathbf{k}=\mathbf{0}$ and remembering that
\begin{equation}
\chi_{\hat{\rho}}(\mathbf{0})=\mathrm{Tr}\,\hat{\rho}=1
\end{equation}
for any normalized state $\hat{\rho}$, we can see that the generalized Husimi function $Q_{\hat{\rho}}(\mathbf{x})$ is a probability distribution:
\begin{equation}
\int Q_{\hat{\rho}}(\mathbf{x})\;d^{2n}x=1\;.
\end{equation}
Besides, it is possible to show \cite{ferraro2005gaussian} that putting $\gamma=0$ in \eqref{Husimik} we formally recover the Wigner function.
Now we can express $\hat{\rho}$ in terms of $Q_{\hat{\rho}}(\mathbf{x})$: putting together \eqref{rhochi} and \eqref{Husimik} we get
\begin{equation}
\hat{\rho}=\int \widetilde{Q}_{\hat{\rho}}(\mathbf{k})\;e^{\frac{1}{4}\mathbf{k}\,\gamma\,\mathbf{k}^T}\;e^{-i\,\mathbf{k}\,\hat{\mathbf{R}}}\;\frac{d^{2n}k}{(2\pi)^n}=
\int Q_{\hat{\rho}}(\mathbf{x})\left(\int e^{\frac{1}{4}\mathbf{k}\,\gamma\,\mathbf{k}^T+i\,\mathbf{k}\,\mathbf{x}}\;e^{-i\,\mathbf{k}\,\hat{\mathbf{R}}}\;\frac{d^{2n}k}{(2\pi)^n}\right)d^{2n}x\,.\label{rhoQ}
\end{equation}
Comparing with \eqref{rhochi} the integral in parenthesis, it looks like the Gaussian ``state'' with covariance matrix $-\gamma$ displaced by $\mathbf{x}$. Of course, this is not a well-defined state, and it makes sense only if integrated against smooth functions as $Q_\gamma(\mathbf{x})$. However, if we formally define
\begin{equation}
\hat{\rho}_G(-\gamma,\,\mathbf{x}):=\int e^{\frac{1}{4}\mathbf{k}\,\gamma\,\mathbf{k}^T+i\,\mathbf{k}\,\mathbf{x}}\;e^{-i\,\mathbf{k}\,\hat{\mathbf{R}}}\;\frac{d^{2n}k}{(2\pi)^n}\;,
\end{equation}
\eqref{rhoQ} becomes
\begin{equation}
\hat{\rho}=\int Q_{\hat{\rho}}(\mathbf{x})\;\hat{\rho}_G(-\gamma,\,\mathbf{x})\;d^{2n}x\;.\label{husimi}
\end{equation}

Then the Husimi function of any bounded operator $\hat{\rho}$ uniquely defines $\hat{\rho}$ .
It follows that the linear span of the set of coherent states, and hence of all Gaussian states, is dense in the Hilbert space of
Hilbert-Schmidt operators $\mathfrak{H}$. Similarly, these linear spans are dense in the Banach space of trace-class
operators $\mathfrak{T}$.

\section{Quantum Gaussian channels}\label{QGCa}
Let $X$ and $Y$ be two symplectic spaces, with symplectic forms $\Delta_X$ and $\Delta_Y$ and associated sets of trace-class operators $\mathfrak{T}_X$ and $\mathfrak{T}_Y$, respectively.
Given a matrix
\begin{equation}
M:X\longrightarrow Y\;,
\end{equation}
a covariance matrix $\alpha$ on $Y$ and a column vector $\mathbf{y}\in Y$, we define the quantum Gaussian channel
\begin{equation}
\Phi:\mathfrak{T}_X\longrightarrow\mathfrak{T}_Y
\end{equation}
of parameters $(M,\;\alpha,\;\mathbf{y})$ as the channel that acts on the characteristic function as
\begin{equation}\label{channelchi}
\chi_{\Phi\left(\hat{\rho}\right)}(\mathbf{k})=e^{-\frac{1}{4}\mathbf{k}\,\alpha\,\mathbf{k}^T+i\,\mathbf{k}\,\mathbf{y}}\;\chi_{\hat{\rho}}\left(\mathbf{k}\,M\right)\;,
\end{equation}
for any row vector $\mathbf{k}$ in $Y$.
The channel defined in \eqref{channelchi} is completely positive (see \cite{holevo2013quantum}) iff
\begin{equation}\label{CP}
\alpha\geq\pm i\left(\Delta_Y-M\;\Delta_X\;M^T\right)\;.
\end{equation}
$\Phi$ sends Gaussian states in $\mathfrak{S}_X$ into Gaussian states in $\mathfrak{S}_Y$:
\begin{equation}\label{Phigauss}
\Phi\left(\hat{\rho}_G(\sigma,\,\mathbf{r})\right)=\hat{\rho}_G\left(M\,\sigma\,M^T+\alpha,\;M\,\mathbf{r}+\mathbf{y}\right)\;,
\end{equation}
and it acts on the moments as
\begin{eqnarray}
\sigma&\mapsto&M\,\sigma\,M^T+\alpha\\
\mathbf{r}&\mapsto&M\,\mathbf{r}+\mathbf{y}\;.
\end{eqnarray}
From the action on the characteristic function \eqref{Phigauss} and Parceval's formula \eqref{hilbert}, it is easy to prove that any quantum Gaussian channel is continuous in the Hilbert-Schmidt norm \eqref{HSN}.

The Hilbert-Schmidt dual of a linear map $\Phi$ acting on trace-class operators and continuous in the trace norm is the map $\Phi^\dag$ acting on bounded operators and continuous in the operator norm defined by
\begin{equation}\label{dagdef}
\mathrm{Tr}\left[\hat{Y}\;\Phi\left(\hat{X}\right)\right]=\mathrm{Tr}\left[\Phi^\dag\left(\hat{Y}\right)\;\hat{X}\right]
\end{equation}
for any trace-class operator $\hat{X}$ and any bounded operator $\hat{Y}$.

If the matrix $M$ is invertible, with Parceval's formula \eqref{hilbert} it is easy to see that the Hilbert-Schmidt adjoint of the Gaussian quantum channel $\Phi$ with parameters $\left(M,\;\alpha,\;\mathbf{y}\right)$ defined in \eqref{Phigauss} acts on the characteristic function as
\begin{equation}\label{HSchi}
\chi_{\Phi^\dag\left(\hat{Y}\right)}(\mathbf{k})=\frac{1}{\det M}\;e^{-\frac{1}{4}\mathbf{k}\,M^{-1}\,\alpha M^{-T}\,\mathbf{k}^T-i\,\mathbf{k}\,M^{-1}\,\mathbf{y}}\;\chi_{\hat{Y}}\left(\mathbf{k}\,M^{-1}\right)\;.
\end{equation}

It is possible to prove \cite{holevo2013quantum} that any quantum Gaussian channel admits a Gaussian Stinespring dilation, i.e. there exist two symplectic spaces $Z$ and $W$ with
\begin{equation}
\dim X+\dim Z=\dim Y+\dim W\;,
\end{equation}
a Gaussian state $\hat{\rho}_Z$ on $\mathcal{H}_Z$ and a symplectic matrix
\begin{equation}
S:X\oplus Z\to Y\oplus W
\end{equation}
satisfying
\begin{equation}
S\left(\Delta_X\oplus\Delta_Z\right)S^T=\Delta_Y\oplus\Delta_W
\end{equation}
such that for any $\hat{X}\in\mathfrak{T}_X$
\begin{equation}
\Phi\left(\hat{X}\right)=\hat{D}(\mathbf{y})\;\mathrm{Tr}_W\left[\hat{U}_S\left(\hat{X}\otimes\hat{\rho}_Z\right)\hat{U}_S^\dag\right]\;{\hat{D}(\mathbf{y})}^\dag\;.
\end{equation}
Here $\hat{U}_S$ is the isometry
\begin{equation}
\hat{U}_S:\mathcal{H}_X\otimes\mathcal{H}_Z\to\mathcal{H}_Y\otimes\mathcal{H}_W
\end{equation}
that implements the symplectic matrix $S$, i.e.
\begin{equation}
\hat{U}_S^\dag\left(\hat{\mathbf{R}}_Y\oplus\hat{\mathbf{R}}_W\right)\hat{U}_S=S\left(\hat{\mathbf{R}}_X\oplus\hat{\mathbf{R}}_Z\right)\;,
\end{equation}
where $\hat{\mathbf{R}}_X$, $\hat{\mathbf{R}}_Y$, $\hat{\mathbf{R}}_Z$ and $\hat{\mathbf{R}}_W$ are the vectors of the quadrature operators of the four systems.

\subsection{Quantum-limited attenuator and amplifier}\label{secatta}
The Gaussian quantum-limited attenuator $\mathcal{E}_\lambda$, $0\leq\lambda\leq1$ defined in Section \ref{secattampl} can be recovered from \eqref{channelchi} identifying $Y$ with $X$ and setting $\mathbf{y}=0$, $M=\sqrt{\lambda}\;\mathbb{I}$ and $\alpha=(1-\lambda)\mathbb{I}$ \cite{holevo2013quantum}.
Its action on the characteristic function is then
\begin{equation}\label{attchi}
\chi_{\mathcal{E}_\lambda\left(\hat{\rho}\right)}(\mathbf{k})=e^{-\frac{1-\lambda}{4}\left|\mathbf{k}\right|^2}\;\chi_{\hat{\rho}}\left(\sqrt{\lambda}\;\mathbf{k}\right)\;.
\end{equation}
The Gaussian quantum-limited amplifier $\mathcal{A}_\kappa$, $\kappa\geq1$ can be instead recovered with $\mathbf{y}=0$, $M=\sqrt{\kappa}\;\mathbb{I}$ and $\alpha=(\kappa-1)\mathbb{I}$:
\begin{equation}\label{amplchi}
\chi_{\mathcal{A}_\kappa\left(\hat{\rho}\right)}(\mathbf{k})=e^{-\frac{\kappa-1}{4}\left|\mathbf{k}\right|^2}\;\chi_{\hat{\rho}}\left(\sqrt{\kappa}\;\mathbf{k}\right)\;.
\end{equation}
Combining \eqref{HSchi} with \eqref{attchi} and \eqref{amplchi}, it is easy to show that the Hilbert-Schmidt dual of the $n$-mode quantum-limited attenuator of parameter $0<\lambda\leq1$ is $1/\lambda^n$ times the  quantum-limited amplifier of parameter $\kappa=1/\lambda\geq1$, and its restriction to trace-class operators is continuous in the trace-norm (see also \cite{ivan2011operator}).

The quantum-limited attenuator and amplifier are the elementary building blocks for gauge-covariant Gaussian channels.
Indeed, it is possible to prove \cite{garcia2012majorization,giovannetti2015solution,mari2014quantum,holevo2015gaussian} that any gauge-covariant quantum Gaussian channel can be expressed as a quantum-limited amplifier composed with a quantum-limited attenuator.

\section{EPnI for Gaussian states}\label{EPnIG}
In this Section we prove that the Entropy Photon-number Inequality \eqref{EPnI} holds when the two inputs are Gaussian states.

Let us consider two $n$-mode Gaussian states with covariance matrices $\sigma_A$ and $\sigma_B$, respectively.
Let $0\leq\lambda\leq1$ be the parameter of the beamsplitter.
From Eqs. \eqref{cmphi} and \eqref{bsdef}, the covariance matrix of the output $C$ is
\begin{equation}
\sigma_C=\lambda\;\sigma_A+(1-\lambda)\;\sigma_B\;.
\end{equation}
The EPnI \eqref{EPnI} reads
\begin{equation}\label{EPnIGe}
h^{-1}\left(S_C/n\right)\geq \lambda\;h^{-1}\left(S_A/n\right)+\left(1-\lambda\right)h^{-1}\left(S_B/n\right)\;,
\end{equation}
where $S_A$, $S_B$ and $S_C$ are the entropies of the two inputs and the output, respectively, and the function $h$ has been defined in \eqref{defh} ($h^{-1}$ is an affine function of $g^{-1}$, hence the inequality \eqref{EPnIGe} can be indifferently written in terms of $g$ or of $h$).
Let $1\leq a_1\leq\ldots\leq a_n$, $1\leq b_1\leq\ldots\leq b_n$ and $1\leq c_1\leq\ldots\leq c_n$ be the symplectic eigenvalues of $\sigma_A$, $\sigma_B$ and $\sigma_C$, respectively.
Their entropies are given by
\begin{equation}
S_X=\sum_{i=1}^n h\left(x_i\right)\;,\qquad X=A,\,B,\,C\;.
\end{equation}
From Theorem 1 of \cite{hiroshima2006additivity},
\begin{equation}
c\prec^w\lambda\;a+(1-\lambda)\;b\;.
\end{equation}
The symbol $\prec^w$ stands for weak supermajorization \cite{marshall2010inequalities}, the partial order relation between vectors with real components that can be defined as
\begin{equation}
x \prec^w y\quad\text{iff}\quad\sum_{i=1}^n f(x_i)\geq\sum_{i=1}^n f(y_i)\quad\forall\;\text{increasing concave}\;f\;.
\end{equation}
Since $h$ is increasing and concave,
\begin{equation}\label{S_C}
S_C=\sum_{i=1}^n h\left(c_i\right)\geq\sum_{i=1}^nh\left(\lambda a_i+(1-\lambda)b_i\right)\;.
\end{equation}
Let us define
\begin{equation}
S^X_i=h(x_i)\;,\qquad X=A,\,B\;.
\end{equation}
Using \eqref{S_C}, the EPnI \eqref{EPnIGe} is implied by
\begin{equation}
\frac{1}{n}\sum_{i=1}^n h\left(\lambda\;h^{-1}\left(S_i^A\right)+(1-\lambda)\;h^{-1}\left(S_i^B\right)\right)\overset{?}{\geq} h\left(\lambda\;h^{-1}\left(\frac{S_A}{n}\right)+(1-\lambda)\;h^{-1}\left(\frac{S_B}{n}\right)\right)\;,
\end{equation}
that is equivalent to the convexity of the function
\begin{equation}
F\left(S_A,S_B\right)=h\left(\lambda\;h^{-1}\left(S_A\right)+(1-\lambda)\;h^{-1}\left(S_B\right)\right)\;,\qquad S_A,S_B\geq0\;.
\end{equation}
For $\lambda=1$, $F=S_A$, while for $\lambda=0$, $F=S_B$.
We can then restrict to $0<\lambda<1$.
The determinant of the Hessian of $F$ is
\begin{align}
&\det H_F\left(h(a),h(b)\right) =\nonumber\\
&=\lambda(1-\lambda)\;h'(c)\;h''(c)\;\frac{h''(a)\;h''(b)}{{h'(a)}^3\;{h'(b)}^3} \left(\frac{h'(c)}{h''(c)}-\lambda\;\frac{h'(a)}{h''(a)}-(1-\lambda)\;\frac{h'(b)}{h''(b)}\right)\;,
\end{align}
where
\begin{equation}
c=\lambda a+(1-\lambda)b\;.
\end{equation}
Since for Lemma \ref{lemg''} the function $h'/h''$ is strictly convex, for any $a\neq b$ we have
\begin{equation}
\frac{h'(c)}{h''(c)}<\lambda\frac{h'(a)}{h''(a)}+(1-\lambda)\frac{h'(b)}{h''(b)}\;,
\end{equation}
and hence $\det H_F>0$.
Then, in each of the two regions $S_A>S_B>0$ and $S_B>S_A>0$, $H_F$ is either strictly positive or strictly negative.
To determine which is the case, we compute the Laplacian of $F$ for $S_A=S_B=h(a)$:
\begin{equation}
\mathrm{Tr}\;H_F\left(h(a),h(a)\right)=-2\lambda(1-\lambda)\;\frac{h''(a)}{{h'(a)}^2}>0\;.
\end{equation}
Then, $H_F$ cannot be strictly negative in neither of the regions, and $F$ is convex.
\subsection{Useful lemmata}
\begin{lem}\label{lemg''}
For any $x>1$ the function
\begin{equation}
\phi(x)=\frac{h'(x)}{h''(x)}
\end{equation}
is strictly convex.
\begin{proof}
Putting
\begin{equation}
x=\frac{1}{\tanh\theta}\;,\qquad\theta>0\;,
\end{equation}
we have
\begin{equation}
\phi''\left(\frac{1}{\tanh\theta}\right)=2\left(\sinh\theta\;\cosh\theta-\theta\right)>0\;.
\end{equation}
\end{proof}
\end{lem}

\section{Quantum Fisher information}\label{appQFI}
Let us consider a family of quantum states $\left\{\hat{\rho}_\theta\right\}_{\theta\in\mathbb{R}}$.
In Section \ref{secQFI} we have defined its quantum Fisher information as
\begin{equation}
J=\left.\frac{d^2}{d\theta^2}S\left(\hat{\rho}\left\|\hat{\rho}_\theta\right.\right)\right|_{\theta=0}= -\mathrm{Tr}\left[\hat{\rho}\left.\frac{d^2}{d\theta^2}\ln\hat{\rho}_\theta\right|_{\theta=0}\right]\;,
\end{equation}
where $\hat{\rho}=\hat{\rho}_0$.
With the help of the identity
\begin{equation}
-\ln\hat{\rho}_\theta=\int_0^\infty\left(\frac{1}{t+\hat{\rho}_\theta}-\frac{1}{t+1}\right)dt\;,
\end{equation}
we get
\begin{equation}
J=2\int_0^\infty\mathrm{Tr}\left[\frac{\hat{\rho}}{\left(t+\hat{\rho}\right)^2}\;\dot{\hat{\rho}}\;\frac{1}{t+\hat{\rho}}\;\dot{\hat{\rho}}\right]dt\;,
\end{equation}
where $\dot{\hat{\rho}}=\left.\frac{d}{d\theta}\hat{\rho}_\theta\right|_{\theta=0}$.
Diagonalizing $\hat{\rho}$ as
\begin{equation}
\hat{\rho}=\sum_i p_i\;|\psi_i\rangle\langle \psi_i|\;,\quad\langle \psi_i|\psi_j\rangle=\delta_{ij}\;,\quad p_i\geq0\;,\quad\sum_i p_i=1\;,
\end{equation}
we can perform the integral and get
\begin{equation}\label{defJ}
J=\sum_i\frac{\left|\langle\psi_i|\dot{\hat{\rho}}|\psi_i\rangle\right|^2}{p_i}+\sum_{i\neq j}\frac{\ln p_i-\ln p_j}{p_i-p_j}\;\left|\langle\psi_i|\dot{\hat{\rho}}|\psi_j\rangle\right|^2\;.
\end{equation}
This is the Kubo-Mori \cite{bengtsson2007geometry} norm squared of $\dot{\hat{\rho}}$.

The quantum Cram\'er-Rao bound \cite{helstrom1967minimum,paris2009quantum,cramer2016mathematical} provides the ultimate bound to the variance of any estimator $\tilde{\theta}$ of $\theta$:
\begin{equation}
\mathrm{Var}\;\tilde{\theta}\geq\frac{1}{F}\;,
\end{equation}
where $F$ is the Bures norm squared \cite{bengtsson2007geometry} of $\dot{\hat{\rho}}$:
\begin{equation}\label{defF}
F=\sum_{ij}\frac{2\left|\langle\psi_i|\dot{\hat{\rho}}|\psi_j\rangle\right|^2}{p_i+p_j}\;,
\end{equation}
and it is also called quantum Fisher information.

$J$ and $F$ coincide when $\dot{\hat{\rho}}$ is diagonal in the eigenbasis of $\hat{\rho}$.
Since for any $p_i,p_j\geq0$
\begin{equation}
\frac{2}{p_i+p_j}\leq\frac{\ln p_i-\ln p_j}{p_i-p_j}\;,
\end{equation}
in general
\begin{equation}
F\leq J\;.
\end{equation}

\chapter{Quantum electrodynamics}\label{appsup}
In this Appendix we recall some properties of quantum electrodynamics and perform some auxiliary computation whose results are needed in Chapter \ref{chsuperpos}.

In particular, in Section \ref{strength} we show that for Bob a measurement of the position is always better than a measurement of the momentum, {\it i.e.} it allows to distinguish the force difference $\Delta F$ in a shorter time.
In Section \ref{loc}, we prove that the minimum width over which a charge $q$ greater than the Planck charge $q_{\mathrm{P}}$ can be localized is its charge radius $q \hbar/(q_{\mathrm{P}} mc)$
\cite{weinberg1995quantum}.
In Section \ref{emf}, we recall the basics of the quantization of the electromagnetic field (see also \cite{cohen1989photons}).
In Section \ref{coh}, we introduce the set of coherent states of the electromagnetic field (see also \cite{barnett2002methods}).
In Section \ref{rad}, we show that the quantum state of the electromagnetic field coupled to a classical charged particle following a non-relativistic motion remains undistinguishable from the vacuum.
In Section \ref{dof}, we explain in detail why when Alice's charged particle is in the quantum superposition \eqref{LR}, despite it generates a static electric field that depends on its position, it is still not entangled with the field, and the global state is a product with the field part in the vacuum.
In Section \ref{secloc}, we explain in detail why Alice can measure only the velocity of her particle, and not its canonical momentum, if she is constrained to remain in her laboratory, which has the size of the support of the wavefunction of the particle.
Finally, in Section \ref{secDST} we check the validity of the approximation of neglecting the time $T_D$ that Bob needs in order to remove the trapping potential.

\section{Strength of the trap}
\label{strength}
Let $\omega$ be the frequency of the harmonic trap. The spatial width of its ground state is given by
\begin{equation}\label{psi0}
\Delta X^2\simeq\frac{\hbar}{m_{\mathrm{B}}\,\omega}\;.
\end{equation}
This ground state is insensible to the force difference $\Delta F$ iff the displacement that it generates is less than $\Delta X$, i.e.
\begin{equation}\label{ins}
\frac{\Delta F}{m_{\mathrm{B}}\,\omega^2}\lesssim \Delta X\;.
\end{equation}
Eliminating $\omega$ with~\eqref{psi0}, the inequality~\eqref{ins} becomes hence
\begin{equation}\label{ins2}
\Delta X^3\lesssim\frac{\hbar^2}{m_{\mathrm{B}}\,\Delta F}\;,
\end{equation}
which is the condition we have to enforce to ensure that Bob's detector is ineffective when switched off. Let us then suppose that, after switching on the detector, Bob tries to distinguish the two states of Alice by a measurement of $P$: accordingly the momentum spread $\Delta P$ of his initial state must be lower than the displacement in momentum $|\delta_p|=\Delta F\,t$. Recalling that Heisenberg's uncertainty principle $\Delta X\,\Delta P\geq\hbar$ is saturated by a Gaussian pure state, the minimum time after which Bob can distinguish is
\begin{equation}
T_{\mathrm{B}}'=\frac{\hbar}{\Delta F\,\Delta X}\;.
\end{equation}
On the other hand, from the ratio considered in \eqref{ratio}, the minimum discrimination time with a measurement of $x$ is
\begin{equation}
T_{\mathrm{B}}=\sqrt{\frac{m_{\mathrm{B}}\,\Delta X}{\Delta F}}\;.
\end{equation}
The reader can check that~\eqref{ins2} implies $T_{\mathrm{B}}\leq T_{\mathrm{B}}'$, i.e. if the trap is strong enough to be insensible to the force difference, for Bob it is always better to measure the position of his particle rather than its momentum.

\section{Maximum localization of a charge}
\label{loc}
Let us suppose to use a harmonic trap of frequency $\omega$ to localize the charge. One could think that in principle, with a strong enough trap, the charge can be arbitrarily localized. However, from the Larmor formula~\cite{jackson1998classical} we know that a classical particle with charge $q$ following a harmonic motion of frequency $\omega$ and width $\Delta X$ loses into electromagnetic radiation a power
\begin{equation}\label{larmor}
\frac{dE}{dt}\simeq \frac{q^2\,\omega^4\,\Delta X^2}{\epsilon_0\,c^3}\;.
\end{equation}
In the quantum case, the charge radiates until it gets to the ground state of the trap, where it cannot radiate anymore since there are no other states with a lower energy to go. However, if the trap is very strong, its ground state is very localized, and therefore has a great uncertainty in velocity. Since any moving charge generates a magnetic field, this velocity uncertainty generates a large uncertainty in the magnetic field, resulting in a large entanglement between the state of the particle and the state of the field. Qualitatively, this happens when the energy classically radiated in a period becomes greater than $\hbar\omega$, the energy of the first excited state. Combining~\eqref{larmor} with~\eqref{psi0}, this happens exactly when the localization $\Delta X$ becomes smaller that the charge radius:
\begin{equation}\label{DX0}
\Delta X\lesssim \frac{q}{q_{\mathrm{P}}}\,\frac{\hbar}{mc}\;.
\end{equation}
Then, if we want the reduced state of the particle to remain pure, we can localize it only up to the limit in~\eqref{DX0}.

\section{Quantization of the electromagnetic field}\label{emf}
We denote with $\hat{O}$ an operator in the Schr\"odinger picture, and with $\hat{O}(t)$ its counterpart in the Heisenberg picture. The two pictures are defined to coincide for $t=0$, i.e. $\hat{O}(0)=\hat{O}$. We recall that in the Heisenberg picture the operators are evolved with the full interacting Hamiltonian.

An Hamiltonian formulation of electrodynamics requires the introduction of the scalar and vector potentials $V$ and $\mathbf{A}$. It is convenient to Fourier-transform with respect to $\mathbf{x}$. The potentials are related to the electric and magnetic fields by
\begin{eqnarray}\label{Et}
\hat{\mathbf{E}}(\mathbf{k},t)&=&-i\mathbf{k}\;\hat{V}(\mathbf{k},t)-\frac{\partial}{\partial t}\hat{\mathbf{A}}(\mathbf{k},t)\\
\hat{\mathbf{B}}(\mathbf{k},t)&=&i\mathbf{k}\times\hat{\mathbf{A}}(\mathbf{k},t)\;.\label{Bt}
\end{eqnarray}
We choose the Coulomb gauge, in which the divergence of the vector potential is set to zero at the operator level:
\begin{equation}\label{coulomb}
\mathbf{k}\cdot\hat{\mathbf{A}}(\mathbf{k},t)=0\;.
\end{equation}
It is now convenient to define the ladder operators
\begin{equation}\label{ai}
\hat{a}^i(\mathbf{k},t)\equiv\sqrt{\frac{|\mathbf{k}|}{2}}\;\hat{A}^i(\mathbf{k},t)+\frac{i}{\sqrt{2|\mathbf{k}|}}\;\frac{\partial}{\partial t}\hat{A}^i(\mathbf{k},t)\;,
\end{equation}
satisfying the constraint $k_i\,\hat{a}^i(\mathbf{k},t)=0$ as a consequence of~\eqref{coulomb}. The definition in~\eqref{ai} can be inverted:
\begin{equation}\label{Afrx}
\hat{A}^i(\mathbf{x},t)=\int\frac{\hat{a}^i(\mathbf{k},t)\;e^{i\mathbf{k}\cdot\mathbf{x}}+\hat{a}^{i\dag}(\mathbf{k},t)\;e^{-i\mathbf{k}\cdot\mathbf{x}}}{\sqrt{2|\mathbf{k}|}}\;\frac{d^3k}{(2\pi)^3}\;.
\end{equation}
The ladder operators satisfy the equal-time canonical commutation relations
\begin{eqnarray}\label{CCRs}
\left[\hat{a}^i(\mathbf{k},t),\;\hat{a}^{j\dag}(\mathbf{q},t)\right]&=&\Pi^{ij}(\mathbf{k})\;(2\pi)^3\delta^3(\mathbf{k}-\mathbf{q})\\
\left[\hat{a}^i(\mathbf{k},t),\;\hat{a}^{j}(\mathbf{q},t)\right]&=&\left[\hat{a}^{i\dag}(\mathbf{k},t),\;\hat{a}^{j\dag}(\mathbf{q},t)\right]=0\;, \quad
\end{eqnarray}
where $\Pi(\mathbf{k})$ is the projector onto the subspace orthogonal to $\mathbf{k}$:
\begin{equation}
\Pi^{ij}(\mathbf{k})=\delta^{ij}-\frac{k^i\,k^j}{\mathbf{k}^2}\;.
\end{equation}
The vacuum state of the field $|0\rangle$ is defined as the state annihilated by all the Schr\"odinger-picture annihilation operators:
\begin{equation}\label{vac}
\hat{a}^i(\mathbf{k})|0\rangle=0\qquad\forall\;\mathbf{k}\in\mathbb{R}^3\;,\quad i=1,\,2,\,3\;,\qquad|0\rangle\in\mathcal{H}_F\;.
\end{equation}
Besides, the $\hat{a}^i(\mathbf{k})$ together with their hermitian conjugates $\hat{a}^{i\dag}(\mathbf{k})$ generate the whole observable algebra of $\mathcal{H}_F$.

Maxwell's equations determine the time evolution of the ladder operators:
\begin{equation}
\frac{\partial}{\partial t}\hat{a}^i(\mathbf{k},t)+i|\mathbf{k}|\,\hat{a}^i(\mathbf{k},t)=\frac{i\,\Pi^i_{\phantom{i}j}(\mathbf{k})}{\sqrt{2|\mathbf{k}|}}\hat{J}^i(\mathbf{k},t)\;,\label{Ais}
\end{equation}
where $\hat{\mathbf{J}}$ is the operator associated to the current density of the quantum system interacting with the electromagnetic field. Eq.~\eqref{Ais} is easily integrated:
\begin{equation}\label{ait}
\hat{a}^i(\mathbf{k},t)=e^{-i|\mathbf{k}|t}\left(\hat{a}^i(\mathbf{k})+\frac{i\,\Pi^i_{\phantom{i}j}(\mathbf{k})}{\sqrt{2|\mathbf{k}|}}\;\int_0^t e^{i|\mathbf{k}|t'} \hat{J}^i(\mathbf{k},t')\;dt'\right)\;,
\end{equation}
where we have imposed the Heisenberg and Schr\"odinger pictures to coincide at $t=0$, i.e. $\hat{a}^i(\mathbf{k},0)=\hat{a}^i(\mathbf{k})$.

In the free case, i.e. when the current vanishes at the operator level ($\hat{\mathbf{J}}(\mathbf{k},t)=0$), the relation between the two pictures is given by the free Hamiltonian
\begin{equation}\label{HF}
\hat{H}_F\equiv\int|\mathbf{k}|\;\hat{a}_i^\dag(\mathbf{k})\;\hat{a}^i(\mathbf{k})\;\frac{d^3k}{(2\pi)^3}\;,
\end{equation}
i.e.
\begin{equation}\label{afree}
\hat{a}^i(\mathbf{k},t)=e^{-i|\mathbf{k}|t}\hat{a}^i(\mathbf{k})=e^{i\hat{H}_Ft}\;\hat{a}^i(\mathbf{k})\;e^{-i\hat{H}_Ft}\;.
\end{equation}

\section{Coherent states}
\label{coh}
For any function $f:\mathbb{R}^3\to\mathbb{C}^3$ subject to the constraint
\begin{equation}
k_i\,f^i(\mathbf{k})=0\qquad\forall\;\mathbf{k}\in\mathbb{R}^3\;,
\end{equation}
we define the unitary displacement operator
\begin{equation}
\hat{D}[f]\equiv\exp\left(\int\left(f_i(\mathbf{k})\,\hat{a}^{i\dag}(\mathbf{k})-f_i^*(\mathbf{k})\,\hat{a}^i(\mathbf{k})\right) \frac{d^3k}{(2\pi)^3}\right)\;,\qquad\hat{D}^\dag[f]=\hat{D}[-f]\;,
\end{equation}
acting on the ladder operators as
\begin{equation}\label{dispa}
\hat{D}^\dag[f]\;\hat{a}^i(\mathbf{k})\;\hat{D}[f]=\hat{a}^i(\mathbf{k})+f^i(\mathbf{k})\;.
\end{equation}
Their composition rule is
\begin{equation}\label{UU}
\hat{D}[f]\;\hat{D}[g]=\hat{D}[f+g]\;\exp{\left(\frac{1}{2}\int\left(f^i(\mathbf{k})\,g_i^*(\mathbf{k})-f_i^*(\mathbf{k})\,g^i(\mathbf{k})\right)\frac{d^3k}{(2\pi)^3}\right)}\;.
\end{equation}
We can now define the coherent states with a displacement operator acting on the vacuum state of the field:
\begin{equation}
|f\rangle\equiv\hat{D}[f]|0\rangle\in\mathcal{H}_F\;,
\end{equation}
that are eigenstates of the annihilation operators:
\begin{equation}
\hat{a}^i(\mathbf{k})|f\rangle=f^i(\mathbf{k})|f\rangle\;.
\end{equation}
Their overlap is
\begin{equation}\label{overlap}
\left|\langle f|g\rangle\right|^2=\exp\left(-\int\left|f(\mathbf{k})-g(\mathbf{k})\right|^2\frac{d^3k}{(2\pi)^3}\right)\;.
\end{equation}

\section{Radiation emitted by Alice's particle}
\label{rad}
We consider a classical charged particle coupled to the quantum electromagnetic field.
Let $\mathbf{J}$ be the classical current density associated to the trajectory of the particle. Looking at the time evolution equation for the ladder operators~\eqref{ait}, and recalling~\eqref{dispa} and~\eqref{afree}, it is easy to show that such evolution is provided by a displacement operator, i.e.
\begin{equation}
\hat{a}^i(\mathbf{k},t)=\hat{D}^\dag[f]\;e^{i\hat{H}_Ft}\;\hat{a}^i(\mathbf{k})\;e^{-i\hat{H}_Ft}\;\hat{D}[f]\;,
\end{equation}
where
\begin{equation}\label{falpha}
f^i(\mathbf{k})=\frac{i\,\Pi^i_{\phantom{i}j}(\mathbf{k})}{\sqrt{2|\mathbf{k}|}}\;\int_0^t e^{i|\mathbf{k}|t'} J^i(\mathbf{k},t')\;dt'\;.
\end{equation}
Since the ladder operators generate the whole observable algebra of $\mathcal{H}_F$, if the fields starts in the vacuum, its time-evolved state is the coherent state $e^{-i\hat{H}_Ft}|f\rangle$. Its overlap with the vacuum can be computed with~\eqref{overlap}:
\begin{equation}\label{cab}
\left|\langle 0|f\rangle\right|^2=\exp\left(-\int\left|f(\mathbf{k})\right|^2\frac{d^3k}{(2\pi)^3}\right)\;.
\end{equation}

We now consider a point particle carrying charge $q$ that starts in $\mathbf{x}=\mathbf{0}$ at $t=0$, and in a time $t_0$ is brought to the position $\mathbf{x}=\mathbf{d}$ with a trajectory described by $\mathbf{x}(t)$. The current density is then
\begin{equation}\label{Jp}
\mathbf{J}(\mathbf{k},t)=q\;\mathbf{v}(t)\;e^{-i\mathbf{k}\cdot\mathbf{x}(t)}\;,
\end{equation}
where $\mathbf{v}(t)\equiv\frac{d}{dt}\mathbf{x}(t)$ is the particle velocity. For wavelengths large with respect to the extension of the motion, i.e. for
\begin{equation}\label{longw}
|\mathbf{k}|\ll\frac{1}{d}\;,
\end{equation}
the phase factor in~\eqref{Jp} can be discarded, getting
\begin{equation}
\mathbf{J}(\mathbf{k},t)\simeq q\;\mathbf{v}(t)\;.
\end{equation}
We want to look at the state of the field after the particle has reached the new position $\mathbf{x}=\mathbf{d}$, i.e. for $t>t_0$. Since the velocity $\mathbf{v}(t)$ vanishes for $t\leq0$ and $t\geq t_0$, the displacement of~\eqref{falpha} becomes
\begin{equation}\label{fq}
f^i(\mathbf{k})=\frac{i\,q}{\sqrt{2|\mathbf{k}|}}\;\Pi^i_{\phantom{i}j}(\mathbf{k})\;v^j(\omega=|\mathbf{k}|)\;,
\end{equation}
where
\begin{equation}
\mathbf{v}(\omega)=\int_{-\infty}^\infty \mathbf{v}(t)\;e^{i\omega t}\;dt
\end{equation}
is the Fourier transform of the velocity. Putting~\eqref{fq} into~\eqref{cab}, the overlap becomes
\begin{equation}\label{ovw}
\left|\langle0|f\rangle\right|^2=\exp\left(-\frac{q^2}{6\pi^2}\int_0^\infty|\mathbf{v}(\omega)|^2\;\omega\;d\omega\right)\;.
\end{equation}
For simplicity we consider a one-dimensional motion, and we put the $x$ axis in the direction of $\mathbf{d}$. As an example, we take
\begin{equation}
x(t)=d\;\sin^2\left(\frac{\pi}{2}\,\frac{t}{t_0}\right)\qquad\text{for}\;0\leq t\leq t_0\;,
\end{equation}
satisfying the conditions
\begin{equation}
x(0)=0\;,\qquad x(t_0)=d\;,\qquad v(0)=v(t_0)=0\;.
\end{equation}
The Fourier transform of the velocity is
\begin{equation}\label{wif}
v(\omega)=e^\frac{i\omega t_0}{2}\;\frac{d\cos\frac{\omega t_0}{2}}{1-\frac{\omega^2t_0^2}{\pi^2}}\;,
\end{equation}
and the overlap
\begin{equation}\label{result}
\left|\langle0|f\rangle\right|^2 =\exp\left(-\pi\,\frac{\pi\mathrm{Si}(\pi)-2}{6}\;\frac{q^2}{q_{\mathrm{P}}^2}\;\frac{d^2}{c^2\,t_0^2}\right)\simeq \exp\left(-2\;\frac{q^2}{q_{\mathrm{P}}^2}\;\frac{d^2}{c^2\,t_0^2}\right)\;,
\end{equation}
where $\mathrm{Si}(x)$ is the sine integral function
\begin{equation}
\mathrm{Si}(x)\equiv\int_0^x\frac{\sin y}{y}dy\;.
\end{equation}
Looking at~\eqref{wif}, the dominant contribution to the integral in~\eqref{ovw} comes from the region $\omega\,t_0\lessapprox1$. The approximation in~\eqref{longw} is then valid iff $d\ll c\,t_0$, i.e. if the motion is not relativistic.

The final result~\eqref{result} tells us that for a fixed distance $d$, no photons are radiated if the motion lasts for at least
\begin{equation}\label{trad}
t_0\gtrsim \sqrt{2}\;\frac{q}{q_{\mathrm{P}}}\;\frac{d}{c}\;.
\end{equation}
Then, Alice can always create the coherent superposition used in the thought experiment without entangling her particle with the emitted photons provided she has enough time to do it. Besides, if Alice wants to perform the particular spin-dependent measurement described in the main text, she needs at least a time~\eqref{trad} to bring the state $|R\rangle$ back to $|L\rangle$ if she does not want to entangle with the emitted photons.

\section{Absence of entanglement with the static electric field}
\label{dof}
The first Maxwell's equation reads
\begin{equation}\label{A0}
\mathbf{k}^2\hat{V}(\mathbf{k})=\hat{\rho}(\mathbf{k})\;,
\end{equation}
and completely determines the electric potential operator $\hat{V}$ in terms of the charge density operator $\hat{\rho}$:
\begin{equation}\label{A0J2}
\hat{V}(\mathbf{k})=\frac{1}{\mathbf{k}^2}\,\hat{\rho}(\mathbf{k})\;.
\end{equation}
Putting together~\eqref{A0J2} and~\eqref{Et}, the electric field is given by
\begin{equation}\label{EA}
\hat{\mathbf{E}}(\mathbf{k})=-\frac{i\mathbf{k}}{\mathbf{k}^2}\;\hat{\rho}(\mathbf{k})-\frac{\partial}{\partial t}\hat{\mathbf{A}}(\mathbf{k})\;.
\end{equation}
Then, the longitudinal (i.e. proportional to $\mathbf{k}$) component of the electric field operator is determined by the charge-density operator, and acts on the Hilbert space of the particle alone. Therefore, even if the field is in its vacuum state~\eqref{vac}, the expectation value of the electric field is the static Coulomb electric field generated by the expectation value of the charge density, and hence depends on the particle wavefunction. This means that the state of the field alone does not contain all the information on the electric field, since its longitudinal component is encoded into the state of the particle.

Seen from a different perspective, the longitudinal component of the electric field is not a dynamical propagating degree of freedom, since it vanishes in absence of external charges and is completely determined by them, so there is no Hilbert space associated to it. The Hilbert space of the field contains only the degrees of freedom associated to the electromagnetic radiation, i.e. the magnetic field and the transverse (orthogonal to $\mathbf{k}$) component of the electric field. Then in a product state with the field part in the vacuum, only these components are in the vacuum mode, while there can be a static electric field depending on the state of the particle.

A final remark should be made about the choice of the gauge. Strictly speaking, the concept of particle-field entanglement is a gauge dependent concept and the previous discussion about the absence of entanglement between a static charge and the electromagnetic field is valid only in the Coulomb gauge. However, despite the initial state we consider would look entangled in a different gauge, obviously the estimation of the minimum discrimination time would not change.

\section{Locality}\label{secloc}
The wavefunction $\psi(\mathbf{x},t)$ of a particle carrying electric charge $q$ coupled to an electromagnetic field is invariant under the joint gauge transformation~\cite{cohen1989photons}
\begin{eqnarray}
\psi'(\mathbf{x},t) &=& e^{iq\Lambda(\mathbf{x},t)}\;\psi(\mathbf{x},t)\\
\mathbf{A}'(\mathbf{x},t) &=& \mathbf{A}(\mathbf{x},t)+\nabla\Lambda(\mathbf{x},t)\label{deltaA}\\
{V}'(\mathbf{x},t) &=& V(\mathbf{x},t)-\frac{\partial}{\partial t}\Lambda(\mathbf{x},t)\;.
\end{eqnarray}
The canonical momentum $\hat{\mathbf{P}}=-i\nabla$ is not gauge invariant, but transforms in the Heisenberg picture as
\begin{equation}\label{deltaP}
\hat{\mathbf{P}}'(t)=\hat{\mathbf{P}}(t)+q\nabla\Lambda\left(\hat{\mathbf{X}}(t),t\right)\;,
\end{equation}
and therefore Alice cannot measure it directly. The reader can easily check using~\eqref{deltaA} and~\eqref{deltaP} that the velocity operator given in the main text is gauge invariant, as it has to be. Alice can then measure directly the velocity, and reconstruct from it the canonical momentum. However, the relation between them
\begin{equation}\label{PV2}
\hat{\mathbf{P}}=m\;\hat{\mathbf{V}}+q\;\hat{\mathbf{A}}\left(\hat{\mathbf{X}}\right)
\end{equation}
contains the vector potential, that from~\eqref{deltaA} is not gauge invariant, and cannot be directly measured. In the Coulomb gauge, it is possible to invert~\eqref{Bt} and express the vector potential in terms of the magnetic field, that is gauge invariant and can be actually measured by Alice:
\begin{equation}\label{AB}
\hat{\mathbf{A}}(\mathbf{x})=\frac{1}{4\pi}\int\frac{\nabla\times\hat{\mathbf{B}}(\mathbf{y})}{\left|\mathbf{x}-\mathbf{y}\right|}\;d^3y\;.
\end{equation}
Putting together~\eqref{AB} and~\eqref{PV2}, we get
\begin{equation}\label{PVB}
\hat{\mathbf{P}}=m\;\hat{\mathbf{V}}+\frac{q}{4\pi}\int\frac{\nabla\times\hat{\mathbf{B}}(\mathbf{y})}{\left|\hat{\mathbf{X}}-\mathbf{y}\right|}\;d^3y\;.
\end{equation}
However, reconstructing the canonical momentum from the velocity with~\eqref{PVB} requires Alice to measure the magnetic field in the whole space. Even if she can allow for some error in the reconstruction, the region in which she has to measure the field increases with the charge $q$, and can extend well outside the support of the wavefunction.

\section{Detector switching time}\label{secDST}
The switching time could be due to many specific technical difficulties but, in principle, the only unavoidable limitation is imposed by relativistic causality. If the linear size of the trap is $L$, then Bob cannot remove the trap instantaneously since he needs at least a time $L/c$ to induce any change in the apparatus. Moreover the size of the trap cannot be arbitrarily small but should be at least larger than the position uncertainty of the test mass, {\it i.e.} $L \ge \Delta X$. Thus, a rough estimate of the minimum switching time is given up to numerical factors by
\begin{equation}
T_D \simeq \frac{\Delta X}{c}. \label{TD}
\end{equation}
If we want to take into account also this finite response time of the experimental apparatus, the causality equation
\eqref{causality} should be changed into
\begin{equation}
T_A + T_B +T_D \ge \frac{R}{c}, \label{causalityTD}
\end{equation}
where $T_A$ is Alice measurement time and $T_B$ is the entanglement generation time. We remind that the lower bounds for the measurement time that we obtained (for $T_D=0$) are:
\begin{equation}
 T_A \ge \frac{q \, d}{q_{\rm P} c }\;, \quad  q> q_{\rm P} \qquad  \left| \qquad T_A \ge \frac{m \, d}{m_{\rm P} c }\;, \quad m > m_{\rm P}\;,\right.
\end{equation}
as given by Eq.s \eqref{Tmass} and \eqref{Tcharge} and corresponding to a quantum superposition of a charge $q$ and mass $m$ respectively. For a finite switching time $T_D$, from
\eqref{TD} and \eqref{causalityTD} we get the weaker bounds
\begin{equation}
 T_A \ge \frac{q \, d}{q_{\rm P} c } - \frac{\Delta X}{c}\;, \quad  q> q_{\rm P}\;, \qquad  \left| \qquad T_A \ge \frac{m \, d}{m_{\rm P} c } - \frac{\Delta X}{c}\;, \quad m > m_{\rm P}\;.\right. \label{correction}
\end{equation}
However we remind that in the derivation of the optimal detection experiment we deduced that the trap should be as narrow as possible in order to minimize $\Delta X$ (ideally down to the charge radius or to the Planck length). In this regime we always have $d \gg
\Delta X$ and the correction terms appearing in Eq.s
\eqref{correction} are negligible.

\bibliographystyle{IEEEtran}
\bibliography{biblio}
\end{document}